\documentclass[12pt]{article}
\usepackage{amsmath}
\usepackage{graphicx}
\usepackage{enumerate}
\usepackage{xr-hyper}
\usepackage{natbib}
\usepackage{url} 
\usepackage{subfig}

\makeatletter
\newcommand*{\addFileDependency}[1]{
  \typeout{(#1)}
  \@addtofilelist{#1}
  \IfFileExists{#1}{}{\typeout{No file #1.}}
}
\makeatother

\usepackage{amssymb,amsfonts,mathrsfs,amsmath,amsthm,mathtools,bm,dsfont, thmtools,bbm}\allowdisplaybreaks
\usepackage[dvipsnames]{xcolor}
\usepackage{array, graphicx, graphics,float,multirow,tabularx,tikz,pgfplots, longtable,threeparttable}
\usepackage{setspace}
\usepackage[colorlinks,citecolor=blue,urlcolor=blue, runcolor=blue, filecolor=cyan]{hyperref}
\usepackage{footnotebackref}
\usepackage{bibentry, natbib} 
\usepackage{paralist}  
\usepackage{appendix} 
\usepackage{mathrsfs}
\usepackage{enumitem}
\usepackage{authblk}
\usepackage{caption}
\usepackage{titlesec}
\usepackage{textcase,relsize}
\usepackage{booktabs}

\declaretheoremstyle[notefont=\bfseries,notebraces={}{},%
    headpunct={},postheadspace=1em]{mystyle}
\declaretheorem[style=mystyle,numbered=no,name=Assumption]{asmp-hand}
\declaretheorem[style=mystyle,numbered=no,name=Condition]{cond-hand}
\declaretheorem[style=mystyle,numbered=no,name=Example]{exmp-hand}

\newcommand{\overbar}[1]{\mkern 2.0mu\overline{\mkern-2.0mu#1\mkern-2.0mu}\mkern 2.0mu}

	\def \calA {\mathcal{A}}

\def \bbE {\mathbb{E}}	\def \calE {\mathcal{E}}		
	\def \calF {\mathcal{F}}		
	\def \calG {\mathcal{G}}		
			
	\def \calI {\mathcal{I}}

	\def \calN {\mathcal{N}}		
	\def \calO {\mathcal{O}}		
\def \bbP {\mathbb{P}}	\def \calP {\mathcal{P}}		
	\def \calQ {\mathcal{Q}}		
\def \bbR {\mathbb{R}}	\def \calR {\mathcal{R}}		
	\def \calS {\mathcal{S}}		
	\def \calT {\mathcal{T}}		
			
	\def \calV {\mathcal{V}}		
			
	\def \calX {\mathcal{X}}		
	\def \calY {\mathcal{Y}}

\def \subE {\text{subE}}

\def \Var {\text{Var}}
\def \Cov {\text{Cov}}

\def \conP {\overset{\bbP}\longrightarrow}
\def \conD {\overset{D}\longrightarrow}

\def \rank {\textrm{rank}}

\def \inj {\textrm{inj}}

\def \aug {\textrm{aug}}
\def \diff {\textrm{diff}}

\newcommand{\norm}[1]{\left\Vert#1\right\Vert}
\newcommand{\abs}[1]{\left\vert#1\right\vert}

\DeclareMathOperator*{\argmin}{arg\,min}

\usepackage{subfig}

\numberwithin{equation}{section}

\theoremstyle{definition}

\newtheorem{assumption}{Assumption}[section]
\newtheorem*{assumption*}{Assumption}

\theoremstyle{plain}
\newtheorem{theorem}{Theorem}[section]
\newtheorem{proposition}[theorem]{Proposition}
\newtheorem{lemma}[theorem]{Lemma}
\newtheorem{corollary}[theorem]{Corollary}
\newtheorem{claim}[theorem]{Claim}
\newtheorem{condition}[theorem]{Condition}

\def \conP {\overset{\bbP}\longrightarrow}
\def \conD {\overset{D}\longrightarrow}
\def \bbE {\mathbb{E}}	\def \calE {\mathcal{E}}		
\def \calP {\mathcal{P}}

\allowdisplaybreaks

\usepackage{algorithm,algpseudocode}



\def\spacingset#1{\renewcommand{\baselinestretch}%
	{#1}\small\normalsize} \spacingset{1}

\begin{document}

\def\spacingset#1{\renewcommand{\baselinestretch}%
{#1}\small\normalsize} \spacingset{1}

%
  \title{Matrix Completion When Missing Is Not at Random and Its Applications in Causal Panel Data Models
  	\footnote{This research was supported by NSF Grants DMS-2015285 and DMS-2052955, and NIH Grant R01-HG01073.}}

\date{(\today)}

  \author{Jungjun Choi and Ming Yuan \\
    Department of Statistics\\
    Columbia University}
  \maketitle
%
%
%

\begin{abstract}
This paper develops an inferential framework for matrix completion when missing is not at random and without the requirement of strong signals. Our development is based on the observation that if the number of missing entries is small enough compared to the panel size, then they can be estimated well even when missing is not at random. Taking advantage of this fact, we divide the missing entries into smaller groups and estimate each group via nuclear norm regularization. In addition, we show that with appropriate debiasing, our proposed estimate is asymptotically normal even for fairly weak signals. Our work is motivated by recent research on the Tick Size Pilot Program, an experiment conducted by the Security and Exchange Commission (SEC) to evaluate the impact of widening the tick size on the market quality of stocks from 2016 to 2018. While previous studies were based on traditional regression or difference-in-difference methods by assuming that the treatment effect is invariant with respect to time and unit, our analyses suggest significant heterogeneity across units and intriguing dynamics over time during the pilot program.
\end{abstract}

\vfill

\newpage
\spacingset{1.4} 

\setlength{\abovedisplayskip}{8pt}
\setlength{\belowdisplayskip}{8pt}
\setlength\intextsep{8pt}
\setlength{\abovecaptionskip}{4pt}

\section{Introduction} \label{sec:intro}
The problem of noisy matrix completion in which we are interested in reconstructing a low-rank matrix from partial and noisy observations of its entries arises naturally in numerous applications. It has attracted a considerable amount of attention in recent years, and a lot of impressive results have been obtained from both statistical and computational perspectives. See, e.g., \cite{candes2010matrix,mazumder:2010, koltchinskii:2011, negahban2012restricted,chen2019inference, chen2020noisy, jin2021factor, xia2021statistical,bhattacharya2022matrix} among many others. A common and crucial premise underlying these developments is that observations of the entries are missing at random. Although this is a reasonable assumption for some applications, it could be problematic for many others. In the past several years, there has been growing interest to investigate how to deal with situations where missing is not at random and to what extent the techniques and insights that are initially developed assuming missing at random can be extended to these cases. See, e.g. \cite{agarwal2020synthetic, agarwal2021causal, athey2021matrix, bai2021matrix, chernozhukov2021inference,cahan2023factor, xiong2023large} among others. 

This fruitful line of research is largely inspired by the development of synthetic control methods in causal inference. See, e.g., \cite{abadie2003economic, abadie2010synthetic,abadie2021using}. The close connection between noisy matrix completion and synthetic control methods for panel data was first made formal by \cite{athey2021matrix} who showed that powerful matrix completion techniques such as nuclear norm regularization can be very useful for many causal panel data models where missing is not at random. It also helps bring together two complementary perspectives of noisy matrix completion: one focuses on statistical inferences assuming a strong factor structure and the other aims at recovery guarantees with minimum signal strength requirement. The main objective of this work is to further bridge the gap between these two schools of ideas and develop a general and flexible inferential framework for matrix completion when missing is not at random and without the requirement of strong factors.

In particular, we shall follow \cite{athey2021matrix} and investigate how the technique of nuclear norm regularization can be used to infer individual treatment effects under a variety of missing mechanisms. One of the key observations to our development is the fact that if the number of missing entries is sufficiently small when compared to the panel size, then they can be estimated well even when missing is not at random. For more general missing patterns with an arbitrary proportion of missingness, we can judicially divide the missing entries into smaller groups and leverage this fact by applying the nuclear norm regularization to a submatrix with a small number of missing entries. This is where our approach differs from that of \cite{athey2021matrix} who suggest applying the nuclear norm regularized estimation to the full matrix. We shall show that subgrouping is essential in producing more accurate estimates and more efficient inferences about individual treatment effects. It is worth noting that it is computationally more efficient to estimate all missing entries together, as suggested by \cite{athey2021matrix}. But estimating too many missing entries simultaneously can be statistically suboptimal. In a way, our results suggest how to trade-off between the computational cost and statistical efficiency.

Our proposal of subgrouping is similar in spirit to the approach taken by \cite{agarwal2021causal} who suggested estimating the missing entries one at a time. For estimating a single missing entry, they propose a matching scheme that constructs multiple ``synthetic" neighbors and averages the observed outcomes associated with each synthetic neighbor. Separating the observations into different sets of neighbors, however, could lead to a loss in efficiency. For example, when estimating the mean of an $N\times N$ matrix with one missing entry, the estimation error of the approach from \cite{agarwal2021causal} for the missing entry converges at the rate of $N^{-1/4}$, which is far slower than the rate of $N^{-1/2}$ attained by our method. 

Furthermore, we show that, with appropriate debiasing, our proposed estimate is asymptotically normal even with fairly weak signals. More specifically, the asymptotic normality holds if $\psi_{\min}^2\gg \sigma^2N$ where $\psi_{\min}$ is the smallest nonzero singular value of the mean of an $N\times N$ matrix and $\sigma^2$ is the variance of the observed entries. Our development builds upon and complements a series of recent works that show that statistical inference for matrix completion is possible with a low signal-to-noise ratio when the data are missing uniformly at random. See, e.g., \cite{chen2019inference, chen2020noisy, xia2021statistical}. Our results also draw an immediate comparison with the recent works by \cite{bai2021matrix, cahan2023factor} who developed an inferential theory for the asymptotic principle component (APC) based approaches when the signal is much stronger, e.g., $\psi_{\min}^2\gtrsim \sigma^2N^2$. It is worth pointing out that the nuclear norm regularization and APC-based approach each has its own merits and requires different treatment. For example, APC-based methods usually assume that the factors are random and impose moment conditions to ensure that the factor structure is strong and identifiable, whereas our development assumes that the factors are deterministic but incoherent and allows for weaker signals.

Our work is motivated by a number of recent studies on the Tick Size Pilot Program, an experiment conducted by the Security and Exchange Commission (SEC) to evaluate the impact of widening the tick size on the market quality of small and illiquid stocks from 2016 to 2018. See, e.g., \cite{albuquerque2020price, chung2020tick, werner2022tick}. The pilot consisted of three treatment groups with a control group: 1) The first treatment group was quoted in \$0.05 increments but still traded in \$0.01 increments (only Q rule), 2) The second treatment group was quoted and traded in \$0.05 increments (Q+T rule), 3) The third treatment group was quoted and traded in \$0.05 increments, and also subject to the trade-at rule (Q+T+TA rule). The trade-at rule, in general, prevents price matching by exchanges that are not displaying the best price. The control group was quoted and traded in \$0.01 increments. Previous studies \citep[see, e.g.,][]{chung2020tick} on the effects of the quote rule (Q), trade rule (T), and trade rule (TA) on the liquidity measure are based on traditional regression or difference-in-difference methods and assume that the treatment effect is invariant with respect to time and unit. As we shall demonstrate, this assumption is problematic for the Tick Size Pilot Program data and there is significant heterogeneity in the treatment effect across both time and units. Indeed, more insights can be obtained using a potential outcome model with interactive fixed effects to capture such heterogeneity. To do so, we extend our methodology from estimating a single matrix to the simultaneous completion of multiple matrices, accounting for the multiple potential situations.

The remainder of this paper is organized as follows. Section \ref{sec:convergencerate} introduces the method of using the nuclear norm penalized estimation when missing is not at random and provides the convergence rates of the estimator. Section \ref{sec:inference} discusses how to reduce bias and provides inferential theory using the debiased estimator. Section \ref{sec:ticksizepilot} shows how our proposed methodology can be applied to infer the treatment effect in the Tick Size Pilot Program and presents the empirical findings of our analysis. Section \ref{sec:sim} examines the finite sample performance of our estimators using simulation studies. Finally, we conclude with a few remarks in Section \ref{sec:conclusion}. All proofs are relegated to the Appendix due to the space limit.

In what follows, we use $\|\cdot\|_{\rm F}$, $\|\cdot\|$, and $\|\cdot\|_*$ to denote the matrix Frobenius norm, spectral norm, and nuclear norm, respectively. In addition, $\|\cdot\|_\infty$ denotes the entrywise $\ell_\infty$ norm, and $\|\cdot\|_{2,\infty}$ the largest $\ell_2$ norm of all rows of the matrix, i.e., $\|A\|_{2,\infty}=\max_i (\sum_j a_{ij}^2)^{1/2}$. For any vector $a$, $\|a\|$ denotes its $\ell_2$ norm. For any set $\calA$, $|\calA|$ is the number of elements in $\calA$. We use $\circ$ to denote the Hadamard product or the entry-by-entry product between matrices of conformable dimensions. $a \lesssim b$ means $|a|/|b| \leq C_1$ for some constant $C_1>0$ and $a \gtrsim b$ means $|a|/|b| \geq C_2$ for some constant $C_2>0$. $c \asymp d$ means that both $c/d $ and $d/c$ are bounded. $a \ll b$ indicates $|a| \leq c_1|b|$ for some sufficiently small constant $c_1>0$ and $a \gg b$ indicates $c_2|a| \geq  |b|$ for some sufficiently small constant $c_2>0$. In addition, $[K] = \{1, \dots , K\}$.

\section{Noisy Matrix Completion}\label{sec:convergencerate}

Consider a panel data setting where $M=(m_{it})_{1\le i\le N, 1\le t\le T}$ is a $N\times T$ matrix of rank $r$ ($\ll \min\{N, T\}$). We use $i$ as the cross-section index and $t$ as the time index. Following the convention of the matrix completion literature, we shall assume that the singular vectors of $M$ are incoherent in that there is a $\mu\geq 1$ such that $\|U_{M}\|_{2,\infty} \leq \sqrt{{\mu r}/{N}}$, $\|V_{M}\|_{2,\infty} \leq \sqrt{{\mu r}/{T}}$ where $U_M$ and $V_M$ denote the left and right singular vectors of $M$, respectively. The incoherence condition requires the singular vectors to be de-localized, in the sense that entries are not dominated by a small number of rows or columns.

Instead of $M$, we observe a subset of the entries of $Y=M+E$ where $E$ is a noise matrix whose entries are independent and identically distributed zero-mean, sub-Gaussian random variable, i.e., $\mathbb{E}[\epsilon_{it}^2] = \sigma^2$, $ \mathbb{E} [\exp(s \epsilon_{it})] \leq \exp(C s^2 \sigma^2)$, $\forall s \in \mathbb{R}$ and some constant $C>0$. Let $\Omega=(\omega_{it})_{1\le i\le N, 1\le t \le T}\in \{0,1\}^{N\times T}$ indicate the observed entries: $\omega_{it}=1$ if and only if $y_{it}$ is observed. The goal of noisy matrix completion is to estimate $M$ from $Y_\Omega:=\{y_{it}: \omega_{it}=1\}$. A popular approach to do so is the nuclear norm penalization:
$$
\widetilde{M}=\argmin_{A\in {\mathbb R}^{N\times T}}\left\{\|\Omega\circ (Y-A)\|_{\rm F}^2+\lambda\|A\|_\ast\right\},
$$
where $\lambda\ge 0$ is a tuning parameter. The properties of $\widetilde{M}$ are by now well understood in the case of missing completely at random, especially when the entries of $\Omega$ are independently sampled from a Bernoulli distribution. See, e.g., \cite{koltchinskii:2011, chen2020noisy}. Instead, we are interested here in the situation where $\Omega$ is not random.

Situations when missing is not at random arise naturally in many causal panel models. Consider, for example, the evaluation of a program that takes effect after time $T_0$ for the last $N-N_0$ units. If $M$ is the potential outcome under the control, then we do not have observations of its entries for $i > N_0$ and $t> T_0$, e.g., $\Omega=1\{t\le T_0 {\rm \ or\ } i\le N_0\}$, yielding a block missing pattern as shown in the left panel of Figure \ref{fig:submatrix_new}. A more general setting that often arises in causal panel data is the staggered adoption where units may differ in the time they are first exposed to the treatment, yielding a missing pattern as shown in the right panel of Figure \ref{fig:submatrix_new}. See \cite{athey2021matrix,agarwal2021causal} for other similar missing patterns that are common in the context of recommendation systems and A / B testing. 
\begin{figure}[htbp]
	\centering
	\includegraphics[width=0.8\textwidth]{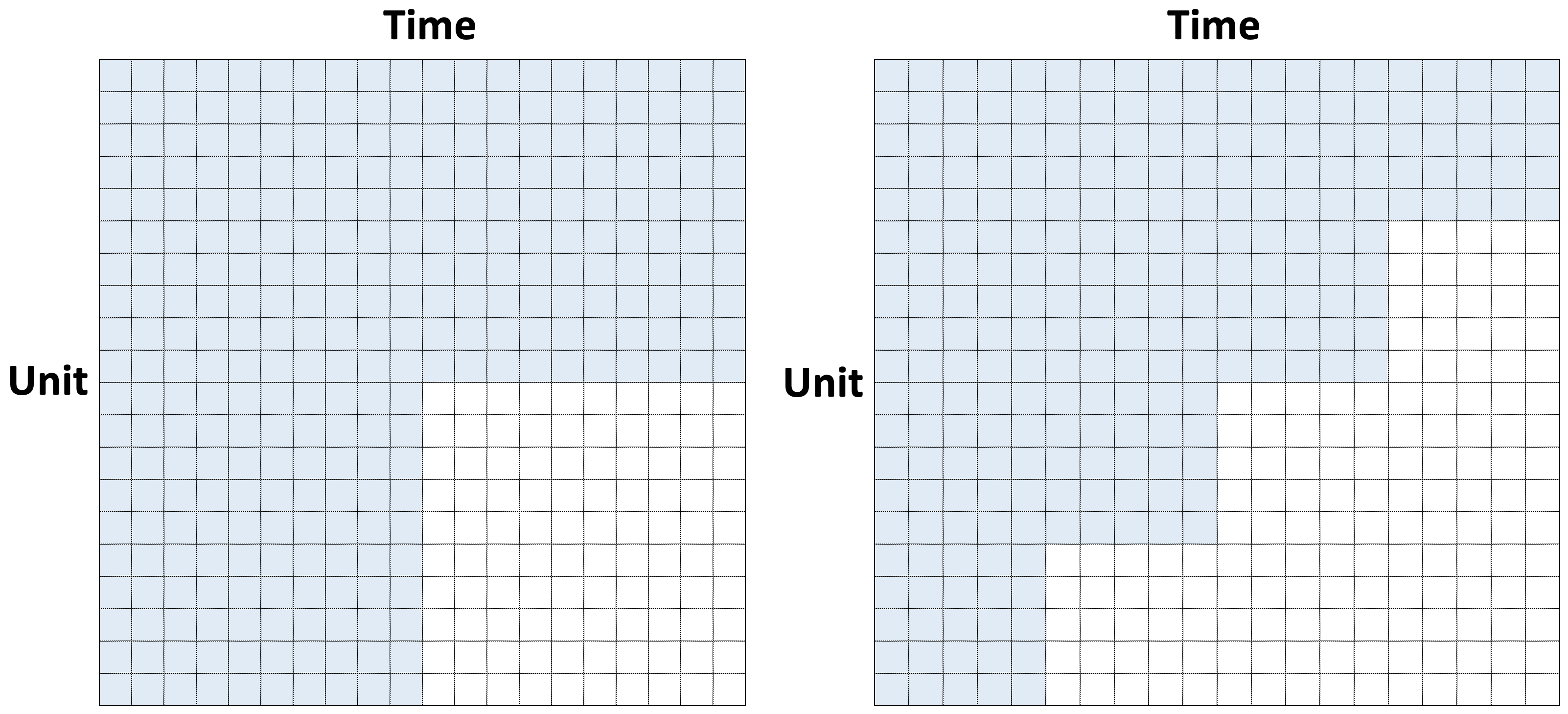}
	\caption{Two typical observation patterns of the potential outcomes under the control in the causal panel model: Here, the blue area is the observed area, and the white area is the missing area. Missingness occurs because we cannot observe the potential outcomes under the control for the treated entries.}
 \label{fig:submatrix_new}
\end{figure}

Note that if the entries are observed uniformly at random, then
$$\|\Omega\circ (Y-A)\|_{\rm F}^2\approx {|\Omega|\over NT}{\mathbb E}\|Y-A\|_{\rm F}^2$$
for sufficiently large $N$ and $T$. The right-hand side is minimized by $M$, which justifies $\widetilde{M}$ as a plausible estimate of $M$. This intuition, however, no longer applies when $\Omega$ is not random and has more structured patterns. Our proposal to overcome this problem is dividing the missing entries into smaller groups and estimating each group via nuclear norm regularization. The main inspiration behind our method is the observation that $\widetilde{M}$ is a good estimate of $M$ when there are only a few missing entries, even if they are missing not at random.

It is instructive to start with a single treated period, e.g., $\Omega=1\{t \leq T-1 {\rm \ or\ } i\le N_0\}$. In this case, the number of missing entries is $|\Omega^c|=N-N_0$. Denote by $\psi_{\max}$ and $\psi_{\min}$ the largest and smallest nonzero singular value of $M$, respectively, and $\kappa = {\psi_{\max}}/{\psi_{\min}}$ its condition number. The following theorem provides bounds for the estimation error of $\widetilde{M}$.

\begin{theorem}\label{thm:consistency}
	Assume that
	\begin{itemize}
		\item[(i)] $\sigma  \kappa^2 \mu^{\frac{1}{2}} r^{\frac{1}{2}} \max\{{N\sqrt{\log{N}}},{T\sqrt{\log{T}}}\} \ll 
		\psi_{\min} \min\{\sqrt{N},\sqrt{T}\}$;
		\item[(ii)] $\kappa^4 \mu^{2} r^{2}  \max\{{N\log^3{N}},{T\log^3{T}}\}  \ll
		\min\{N^{2},T^{2}\}$;
		\item[(iii)] $|\Omega^c| \kappa^2 \mu r   \ll \min\{N,T\}$.
	\end{itemize}      
	Then, with probability at least $1 - O(\min\{N^{-9},T^{-9}\})$, we have
$$\norm{\widetilde{M} - M}_\infty \leq \frac{C\sigma \mu r^{\frac{3}{2}} \kappa^2 \max\{ \sqrt{\log N}, \sqrt{\log T}\}}{\min\{\sqrt{N},\sqrt{T}\}},$$
for some absolute constant $C>0$.
\end{theorem}

Some immediate remarks are in order. Consider the situation where $\kappa,\mu,r = O(1)$, and $N \asymp T$. Ignoring the logarithmic term, the signal-to-noise ratio requirement given by Assumption (i) reduces to ${\psi_{\min}}\gg \sigma N^{1/2}$ which is significantly weaker than those in the existing literature. More specifically, if there is a single missing entry, e.g., $N_0=N-1$, \cite{agarwal2021causal} suggest to partition the submatrix $(m_{it})_{1\le i<N, 1\le t<T}$ into $K$ smaller matrices. In particular, their Theorem 2 states that the best estimation error for their estimate is given by
$$
|\widehat{m}^{\rm ADSS}_{NT} - m_{NT}| = O_p \left( \frac{1}{N^{1/4}} + \frac{1}{T^{1/4}} \right)
$$ 
by setting $K \asymp N^{1/2}$. In contrast, under the assumptions of \cite{agarwal2021causal},  $\sigma,\kappa,\mu,r$ are bounded and hence the convergence rate of our estimator is
$$
|\widetilde{m}_{NT} - m_{NT}|=O_p \left( \left(\frac{1}{N^{1/2}} + \frac{1}{T^{1/2}}\right)\sqrt{\log(NT)} \right).
$$

Theorem \ref{thm:consistency} serves as our building block for dealing with more general and common missing patterns, which we shall now discuss in detail.

\paragraph{Single Treated Period.}

Note that Assumption (iii) of Theorem \ref{thm:consistency} restricts the number of missing entries not to be large compared to $N$ and $T$. In particular, if $\kappa, \mu, r = O(1)$ and $N\asymp T$, then it requires that $|\Omega^c| = o\left(N\right)$. To deal with a larger number of missing entries, we shall leverage this result by splitting the missing entries into small groups and estimating them separately, as illustrated in Figure \ref{fig:howtoconstruct_1_new}.
\begin{figure}[htbp]
\centering
\includegraphics[width=0.75\textwidth]{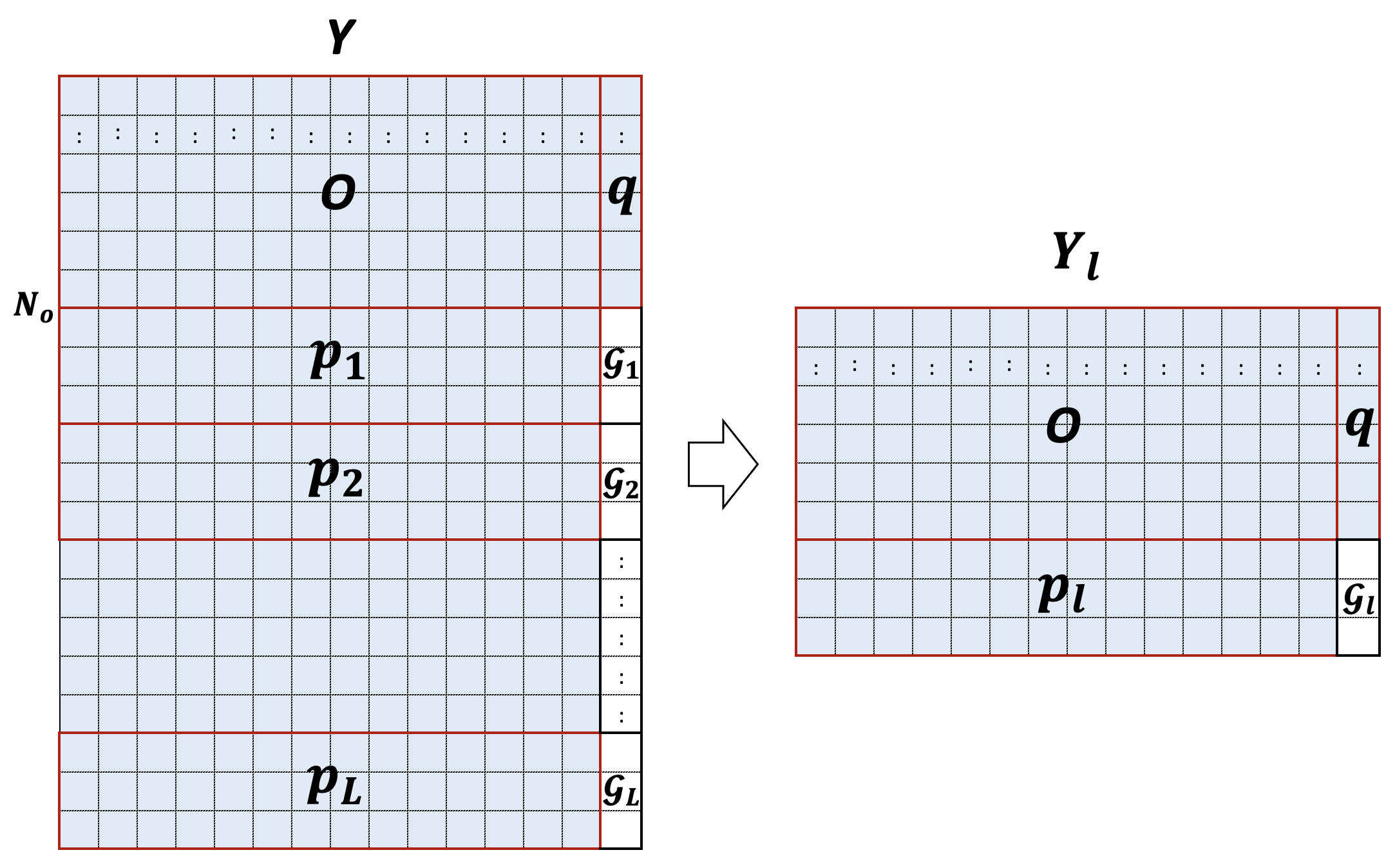}
\caption{How to construct the submatrix: We divide the missing entries into $L$ groups. For each $1\leq l \leq L$, we estimate the entries in $\calG_l$ using the nuclear norm penalized estimation on the submatrix $Y_l$ after making the submatrix $Y_l$ as described in the right panel.}
\label{fig:howtoconstruct_1_new}
\end{figure}

Specifically, we split the missing entries into small groups, denoted by $\{\calG_l \}_{1\leq l \leq L}$, and construct the submatrices $\{Y_l\}_{1\leq l \leq L}$ as illustrated in Figure \ref{fig:howtoconstruct_1_new}. For each $1 \leq l \leq L$, we estimate $M_l$, the corresponding submatrix of $M$, using the nuclear norm penalization:
\begin{gather}\label{eq:nuclear_split}
\widetilde{M}_l=\argmin_{A\in {\mathbb R}^{N_l\times T}}\left\{\|\Omega_l\circ (Y_l-A)\|_{\rm F}^2+\lambda_l\|A\|_\ast\right\},    
\end{gather}
where $N_l = N_0 + |\calG_l|$ and $\Omega_l$ is the corresponding submatrix of $\Omega$. We shall then assemble these estimated submatrices into an estimate $\widetilde{M}$ of $M$. Note that each missing entry appears in one and only one of the submatrices and can therefore be estimated accordingly. The entries from $O$ in Figure \ref{fig:howtoconstruct_1_new}, e.g., the $N_0 \times (T-1)$ principle submatrix of $M$, on the other hand, are estimated for all groups. We can estimate these entries by averaging all of these estimates. Let the smallest nonzero singular value of $M_O$ be $\psi_{\min,O}$, where $M_O$ is the submatrix of $M$ corresponding to $O$. Denote by $u_i^\top$ and $v_t^\top$ the $i$-th row of $U_M$ and $t$-th row of $V_M$, respectively. We can then derive the following bounds from Theorem \ref{thm:consistency}.

\begin{corollary}\label{cor:consistency_subgroup}
	  Assume that
	\begin{itemize}
		\item[(i)] $\sigma  \kappa^{\frac{9}{4}} \mu^{\frac{1}{2}} r^{\frac{1}{2}} \max\{{N_0\sqrt{\log{N_0}}},{T\sqrt{\log{T}}}\} \ll 
		\psi_{\min,O} \min\{\sqrt{N_0},\sqrt{T}\}$;
		\item[(ii)] $\kappa^5 \mu^{2} r^{2}  \max\{{N_0\log^3{N_0}},{T\log^3{T}}\}  \ll
		\min\{N_0^{2},T^{2}\}$;
		\item[(iii)] $|\calG_l| \kappa^{\frac{5}{2}} \mu r   \ll \min\{N_0,T\}$, $l=1,\ldots, L$;
        \item[(iv)] There are constants $C,c > 0 $ such that
        $$
        c \leq \lambda_{\min} \left( \frac{N}{N_0} \sum_{i \leq N_0}  u_{i}u_{i}^\top  \right) \leq \lambda_{\max} \left( \frac{N}{N_0} \sum_{i \leq N_0}  u_{i}u_{i}^\top  \right) \leq C,
        $$ 
    where $\lambda_{\max}(A)$ and $\lambda_{\min}(A)$ are the largest and smallest singular value of $A$, respectively.
  \end{itemize}   
        Then, with probability at least $1 - O(\min\{N_0^{-9},T^{-9}\}L)$, we have
$$\norm{\widetilde{M} - M}_\infty \leq C \frac{\sigma \kappa^\frac{5}{2}\mu r^\frac{3}{2} \max\{ \sqrt{\log N_0} , \sqrt{\log T} \}}{\min\{\sqrt{N_0} ,\sqrt{ T}\}},$$
for some absolute constant $C>0$.
\end{corollary}
The main difference from Theorem \ref{thm:consistency} lies in Assumptions (iii) and (iv) of Corollary \ref{cor:consistency_subgroup}. Assumption (iii) specifies how large a block can be. In principle, we can always take $|\calG_l|=1$, that is, recovering one entry at a time so that this condition is trivially satisfied with sufficiently large $N_0$ and $T$. However, there could be enormous computational advantages in creating groups as large as possible because the number of $\widetilde{M}_l$s that need to be computed decreases with increasing group size.

Assumption (iv) can be viewed as an incoherence condition to ensure that the singular vectors of $M$ are not dominated by either the treated or untreated units. It is easy to see that when there are few missing entries, e.g., $N_0\approx N$, the condition is satisfied by virtue of the incoherence of $u_i$s. In general, if $\{u_i\}_{i \in [N]}$ is exchangeable or if the treated units are uniformly selected, then this condition is satisfied with high probability, at least for sufficiently large $N_0$, since $\frac{N}{N_0} \sum_{i \leq N_0}  u_{i}u_{i}^\top\approx \sum_{i \leq N}  u_{i}u_{i}^\top=I_r$ by means of matrix concentration inequalities \citep[see, e.g.,][]{tropp2015introduction}.

\paragraph{Single Treated Unit.}\label{sec:block_consistency}
A similar estimating strategy can also be used to deal with a single treated unit. Without loss of generality, let $\Omega=1\{t \leq T_0 {\rm \ or\ } i\leq N - 1\}$. Then the fully observed submatrix is $O = (y_{it})_{1\leq i \leq N-1, 1\leq t \leq T_0}$. As in the case of a single treated period, we split the missing entries into smaller groups, denoted by $\calG_1,\ldots, \calG_L$, by periods, and estimate them separately as before. Similar to Theorem \ref{cor:consistency_subgroup}, we have the following bounds for the resulting estimate.

\begin{corollary}\label{cor:groupclt_split_unit}
 Assume that
	\begin{itemize}
		\item[(i)] $\sigma  \kappa^{\frac{9}{4}} \mu^{\frac{1}{2}} r^{\frac{1}{2}} \max\{{N\sqrt{\log{N}}},{T_0\sqrt{\log{T_0}}}\} \ll 
		\psi_{\min,O} \min\{\sqrt{N},\sqrt{T_0}\}$;
		\item[(ii)] $\kappa^5 \mu^{2} r^{2}  \max\{{N\log^3{N}},{T_0\log^3{T_0}}\}  \ll
		\min\{N^{2},T_0^{2}\}$;
		\item[(iii)] $|\calG_l| \kappa^{\frac{5}{2}} \mu r   \ll \min\{N,T_0\}$, $l=1,\ldots, L$;
        \item[(iv)] There are constants $C,c > 0 $ such that
        $$
        c \leq \lambda_{\min} \left( \frac{T}{T_0} \sum_{t \leq T_0} v_{t}v_{t}^\top  \right) \leq \lambda_{\max} \left( \frac{T}{T_0} \sum_{t \leq T_0}  v_{t}v_{t}^\top  \right) \leq C.
        $$ 
  \end{itemize}   
        Then, with probability at least $1 - O(\min\{N^{-9},T_0^{-9}\}L)$, we have
$$\norm{\widetilde{M} - M}_\infty \leq C  \frac{\sigma \kappa^\frac{5}{2}\mu r^\frac{3}{2} \max\{ \sqrt{\log N} , \sqrt{\log T_0} \}}{\min\{\sqrt{N} , \sqrt{T_0}\}},$$
for some absolute constant $C>0$. 
\end{corollary}

\paragraph{General Block Missing Pattern.}
We can also apply the grouping and estimating procedure to general block missing structures such as that depicted in the left panel of Figure \ref{fig:submatrix_new}, e.g., $\Omega=1\{t\le T_0 {\rm \ or\ } i\le N_0\}$, by estimating missing entries one period at a time (or one unit at a time). Denote by $\calG_1, \calG_2,\ldots, \calG_L$ the groups of missing units (or periods). The following result again follows from Theorem \ref{thm:consistency}:

\begin{corollary}\label{cor:consistency_subgroup_block}
 Assume that
	\begin{itemize}
		\item[(i)] $\sigma  \kappa^{\frac{9}{4}} \mu^{\frac{1}{2}} r^{\frac{1}{2}} \max\{{N_0\sqrt{\log{N_0}}},{T_0\sqrt{\log{T_0}}}\} \ll 
		\psi_{\min,O} \min\{\sqrt{N_0},\sqrt{T_0}\}$;
		\item[(ii)] $\kappa^5 \mu^{2} r^{2}  \max\{{N_0\log^3{N_0}},{T_0\log^3{T_0}}\}  \ll
		\min\{N_0^{2},T_0^{2}\}$;
		\item[(iii)] $|\calG_l| \kappa^{\frac{5}{2}} \mu r   \ll \min\{N_0,T_0\}$, $l=1,\ldots, L$;
            \item[(iv)] There are constants $C,c > 0 $ such that
        \begin{align*}
             &c \leq \lambda_{\min} \left( \frac{N}{N_0} \sum_{i \leq N_0} u_{i}u_{i}^\top  \right) \leq  \lambda_{\max} \left( \frac{N}{N_0} \sum_{i \leq N_0} u_{i}u_{i}^\top  \right) \leq C, \\
             &c \leq \lambda_{\min} \left( \frac{T}{T_0} \sum_{t \leq T_0 }v_{t}v_{t}^\top  \right) \leq \lambda_{\max} \left( \frac{T}{T_0} \sum_{t \leq T_0 } v_{t} v_{t}^\top  \right) \leq C.
        \end{align*}
	\end{itemize}   
        Then, with probability at least $1 - O(\min\{N_0^{-9},T_0^{-9}\}L(T-T_0))$, we have
$$\norm{\widetilde{M} - M}_\infty \leq C  \frac{\sigma \kappa^\frac{5}{2}\mu r^\frac{3}{2} \max\{ \sqrt{ \log N_0} , \sqrt{ \log T_0} \}}{\min\{\sqrt{N_0} , \sqrt{T_0}\}},$$
for some absolute constant $C>0$.    
\end{corollary}

It is worth noting that both Corollary \ref{cor:consistency_subgroup} and Corollary \ref{cor:groupclt_split_unit} can be viewed as special cases of Corollary \ref{cor:consistency_subgroup_block}. It is also of interest to compare the rates of convergence with those of \cite{athey2021matrix}. \cite{athey2021matrix} considered a direct application of the nuclear norm penalized estimation to the full matrix. Their Theorem 2 states that
$$
\frac{1}{\sqrt{NT}} \norm{\widetilde{M} - M}_{\rm F} = O_p \left( \sqrt{\frac{T}{N}} + \sqrt{\frac{1}{T}}\right),
$$
ignoring the logarithmic factors and $\sigma$, $r$, and $\norm{M}_\infty$. In other words, the estimate could be inconsistent when $N = O(T)$. On the other hand, the convergence rate of our estimator is given by
$$\norm{\widetilde{M} - M}_\infty = O_p \left( \sqrt{\frac{1}{N_0}} + \sqrt{\frac{1}{T_0}} \right), $$
up to a logarithmic factor when we assume $\kappa,\mu = O_p(1)$. Hence, our estimator is consistent as long as $\min\{N_0,T_0\}$ diverges. Furthermore, the simulation results in Section \ref{sec:sim} also show that applying the nuclear norm penalized estimation to the submatrix indeed performs much better than applying it to the full matrix as long as $N_0$ and $T_0$ are not too small.

\paragraph{Staggered Adoption.} More generally, we can take advantage of our estimation strategy for staggered adoption where there are $D$ number of adoption time points, says $T_1< \cdots< T_D$, and $D$ number of corresponding groups of treated units, says $G_1, \dots , G_D$. That is, for each $d \in [D]$, the units in $G_d$ adopt the treatment in the time period $T_d$. We can utilize the strategy for block missing patterns to estimate the missing entries. More specifically, denote by $M_{d,d'}$ the submatrix with missing entries corresponding to units in $G_{d}$ and time periods in $[T_{d'},T_{d'+1})$, with the convention that $T_{D+1}=T+1$, where $d \leq d' \leq D$. To estimate these missing entries, we can assemble a submatrix, denoted by $Y_{d,d'}$, with units untreated prior to $T_{d'+1}$ and time periods in $[1,T_d)\cup [T_{d'},T_{d'+1})$, as well as units in $G_d$ and time periods in $[1,T_d)$. As shown in Figure \ref{fig:propergrouping_new}, $M_{d,d'}$ is now the missing block of $Y_{d,d'}$, and can be estimated as described in the previous case. 

Denote by $\calG_1, \calG_2,\ldots, \calG_L$ the groups for missing units in $M_{d,d'}$ such as $\cup_{l\in [L]} \calG_l = G_d$, $N_{d'}$ the number of units that are untreated prior to $T_{d'+1}$, and $\psi_{\min,O_{d,d'}}$ the smallest singular value of the submatrix $M_{O_{d,d'}}=(m_{it})_{1\leq i \leq N_{d'}, 1 \leq t \leq T_d}$. The performance of the resulting estimate is given by Corollary \ref{cor:consistency_staggered_adoption}.

\begin{figure}[htbp]
\centering
\includegraphics[width=0.7\textwidth]{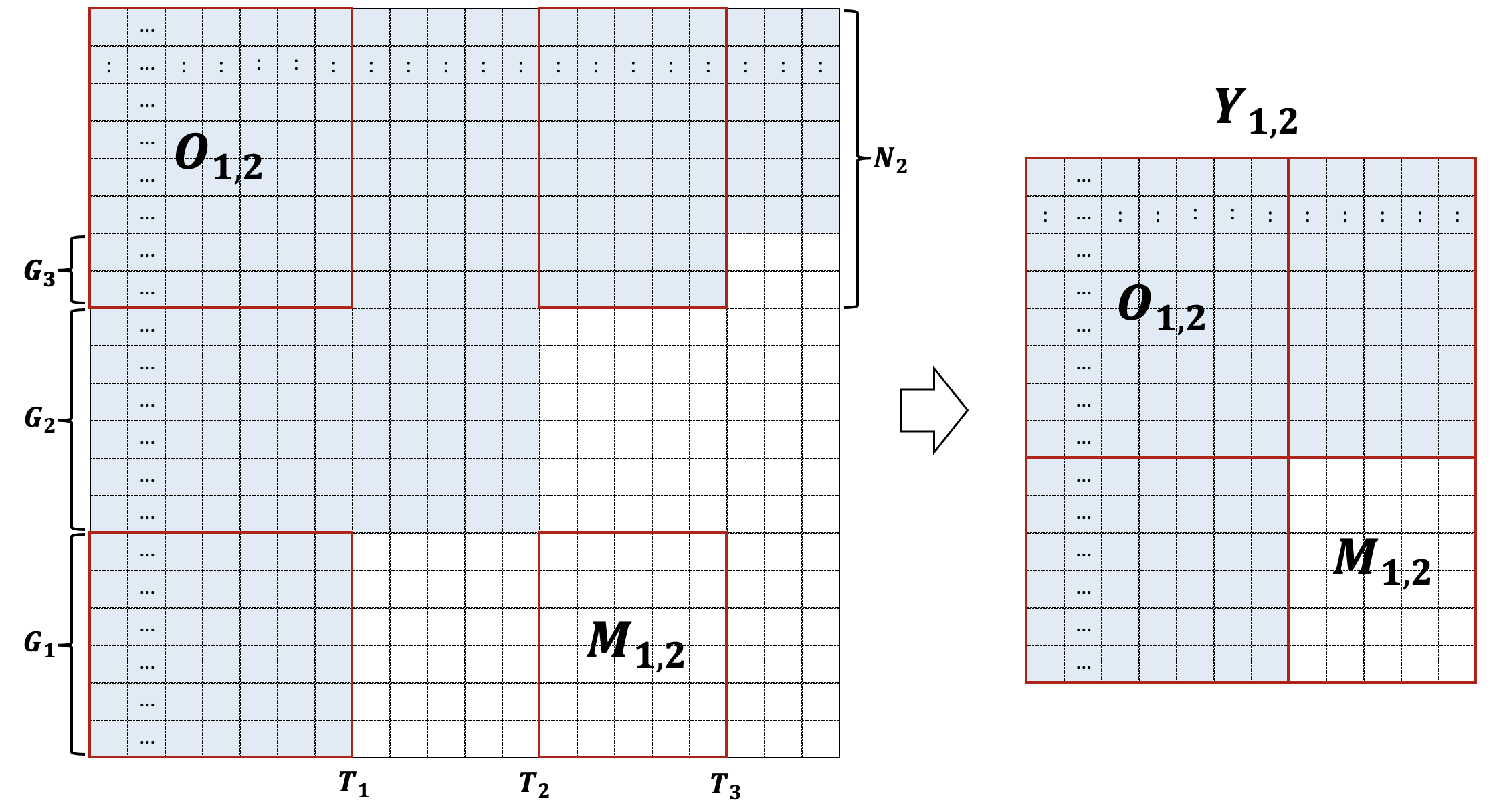}
\caption{How to construct the general block missing pattern: Consider the case of $d=1$ and $d'=2$. When we estimate the missing entries in $M_{1,2}$, we make the block missing matrix $Y_{1,2}$ by assembling four red matrices. Then, we can estimate the missing entries in $M_{1,2}$ using the estimation method for the general block missing pattern.}
\label{fig:propergrouping_new}
\end{figure}

\begin{corollary}\label{cor:consistency_staggered_adoption}
 Assume that
	\begin{itemize}
		\item[(i)] $\sigma  \kappa^{\frac{9}{4}} \mu^{\frac{1}{2}} r^{\frac{1}{2}} \max\{{N_{d'}\sqrt{\log{N_{d'}}}},{T_{d}\sqrt{\log{T_{d}}}}\} \ll 
		\psi_{\min,O_{d,d'}} \min\{\sqrt{N_{d'}},\sqrt{T_{d}}\}$;
		\item[(ii)] $\kappa^5 \mu^{2} r^{2}  \max\{{N_{d'}\log^3{N_{d'}}},{T_{d}\log^3{T_{d}}}\}  \ll
		\min\{N_{d'}^{2},T_{d}^{2}\}$;
		\item[(iii)] $|\calG_l| \kappa^{\frac{5}{2}} \mu r   \ll \min\{N_{d'},T_{d}\}$, $l=1,\ldots, L$;
	   \item[(iv)] There are constants $C,c > 0 $ such that
        \begin{align*}
             &c \leq \lambda_{\min} \left( \frac{N}{N_{d'}} \sum_{i \leq N_{d'}} u_{i}u_{i}^\top  \right) \leq  \lambda_{\max} \left( \frac{N}{N_{d'}} \sum_{i \leq N_{d'}} u_{i}u_{i}^\top  \right) \leq C, \\
             &c \leq \lambda_{\min} \left( \frac{T}{T_{d}} \sum_{t \leq T_{d} } v_{t}v_{t}^\top  \right) \leq \lambda_{\max} \left( \frac{T}{T_{d}} \sum_{t \leq T_{d}} v_{t}v_{t}^\top  \right) \leq C.
        \end{align*}
	\end{itemize}   
	Then, with probability at least $1 - O(\min\{N_{d'}^{-9},T_{d}^{-9}\}L(T_{d'+1} - T_{d'}))$, we have
$$\norm{\widetilde{M}_{d,d'} - M_{d,d'}}_\infty \leq C  \frac{\sigma \kappa^\frac{5}{2}\mu r^\frac{3}{2} \max\{ \sqrt{\log N_{d'}} , \sqrt{\log T_{d}} \}}{\min\{\sqrt{N_{d'}} , \sqrt{T_{d}}\}},$$
for some absolute constant $C>0$.    
\end{corollary}

It is worth comparing the rates of convergence with those of \cite{bai2021matrix} which apply their TW algorithm to the full matrix. For all missing entries, the convergence rates of the estimators in \cite{bai2021matrix} are $O_p\left(\frac{1}{\sqrt{N_D}} + \frac{1}{\sqrt{T_1}}\right)$. On the other hand, if we assume $\kappa,\mu = O_p(1)$, the convergence rate of our estimator is $O_p\left(\frac{1}{\sqrt{N_{d'}}} + \frac{1}{\sqrt{T_d}}\right)$ up to a logarithmic factor. Since $N_{d'} > N_D$ and $T_{d} > T_1$ for all $d' < D$ and $d > 1$, our convergence rate is faster than that of \cite{bai2021matrix} except for the estimation of missing entries in part $M_{1,D}$ for which both estimates have similar rates of convergence. This shows the advantage of exploiting submatrices for the imputation of missing entries.

\section{Debiasing and Statistical Inferences}\label{sec:inference}

We now turn our attention to inferences. While the nuclear norm regularized estimator $\widetilde{M}$ enjoys good rates of convergence, it is not directly suitable for statistical inferences due to the bias induced by the penalty. To overcome this challenge, we propose an additional projection step after applying the nuclear norm penalization in recovering missing entries from group $\calG_l$:
\begin{equation}
\label{eq:biascorrection}
\widehat{M}_l = \calP_{r} \left( \Omega^c_l \circ \widetilde{M}_l + \Omega_l \circ Y_l \right),
\end{equation}
where $\calP_{r}(B) = \argmin_{A:\rank(A)\leq r} \norm{A-B}_F$ is the best rank-$r$ approximation of $B$. We now discuss how this enables us to develop an inferential theory for estimating the missing entries. To fix ideas, we shall focus on inferences about the average of a group of entries at a given time period, e.g., $\sum_{i \in \calG} m_{it_0}/|\calG|$, where $\calG \subseteq [N]$.

\paragraph{Block Missing Patterns.} We shall begin with general block missing patterns, e.g., $\omega_{it}=1$ if $t\le T_0$ or $i\le N_0$. Note that both the single treated period and single treated unit examples from the previous section can be viewed as special cases with $T_0=T-1$ and $N_0=N-1$, respectively. 

Suppose that we are interested in the inference of the average of a group of entries at the time $t_0$, $\sum_{i\in \calG} m_{it_0}/|\calG|$, where $\calG \subseteq \{ 1 , \cdots, N\}$ and $t_0>T_0$. Similar to before, we split the interesting group, $\calG$, into smaller subgroups, denoted by $\{\calG_l \}_{0 \leq l \leq L}$ with the convention that $\calG_0 = \calG \cap \{1, \cdots , N_0\}$, and construct the corresponding submatrices $\{Y_l\}_{1 \leq l \leq L}$ as illustrated in Figure \ref{fig:howtoconstruct_5_new}, and construct $Y_0 = [ (y_{it})_{i\leq N_0,t \leq T_0} \ \ (y_{it})_{i\leq N_0,t = t_0}]$ if $\calG_0 \neq \emptyset$.

\begin{figure}[htbp]
\centering
\includegraphics[width=0.6\textwidth]{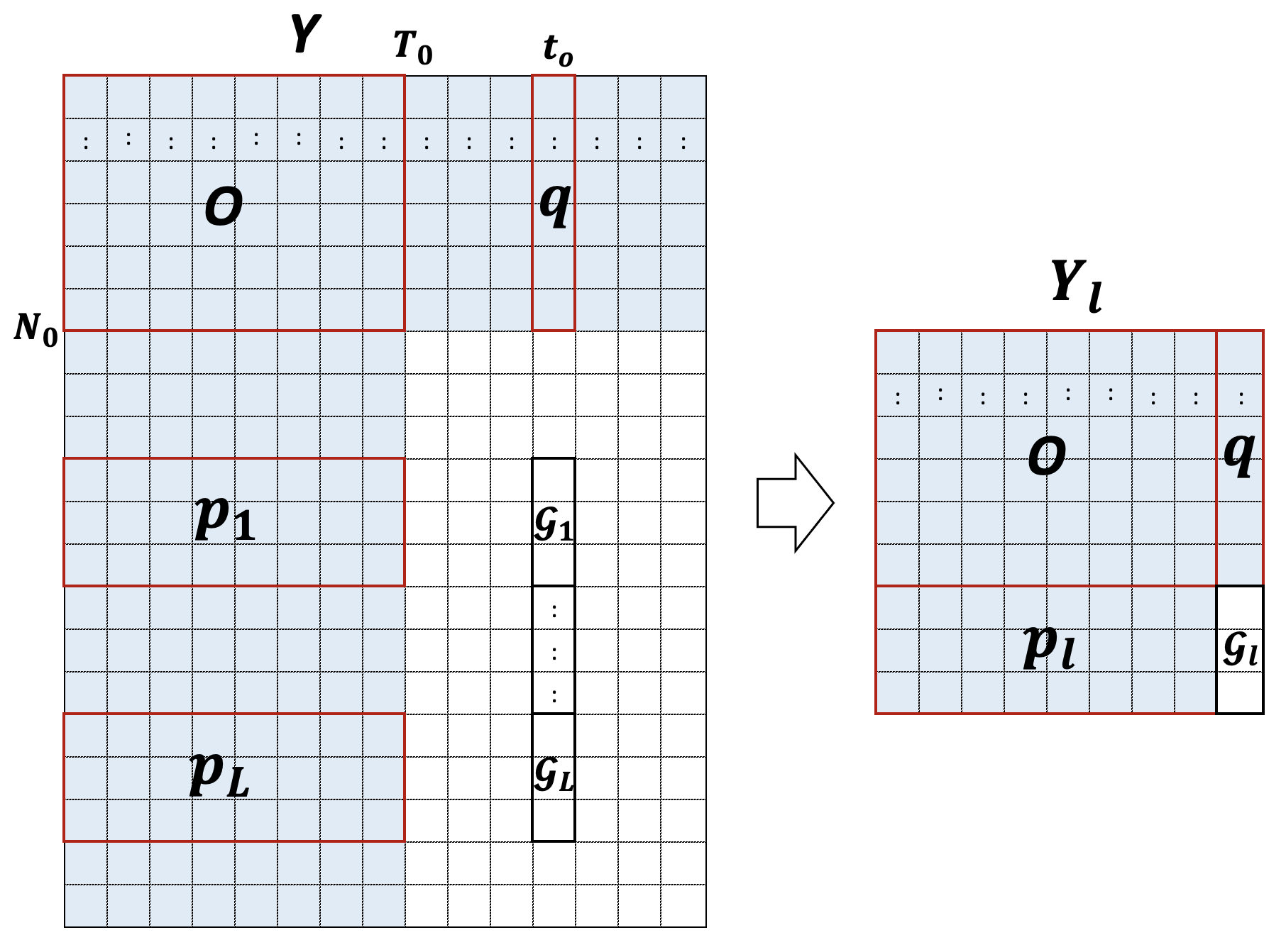}
\caption{How to construct the submatrix: The blue area is the observed area and the white area is the missing area. We estimate the entries in $\calG_l$ using the submatrix $Y_l$ as described in the figure.}
\label{fig:howtoconstruct_5_new}
\end{figure}

Recall that $\psi_{\min,O}$ is the smallest nonzero singular value of the $N_0 \times T_0$ matrix $M_O = (m_{it})_{1\leq i \leq N_0, 1 \leq t \leq T_0}$. The following theorem establishes the asymptotic normality of the group average estimator, $\sum_{ i \in \mathcal{G}}\widehat{m}_{it_0}/|\mathcal{G}|$.

\begin{theorem}\label{thm:groupclt_split_block}
 Assume that
\begin{itemize} 
\item[(i)] $\sigma \kappa^\frac{23}{4} \mu^\frac{3}{2} r^\frac{3}{2} \min\{\sqrt{N_0},\sqrt{|\calG| T_0}\} \max\{{N_0\sqrt{\log{N_0}}},{T_0\sqrt{\log{T_0}}}\}=o_p\left(\psi_{\min,O} \min\{N_0,T_0\}\right)$;
\item[(ii)] $\kappa^\frac{11}{2} \mu^{3} r^{3} \min\{\sqrt{N_0},\sqrt{|\calG| T_0}\} \max\{{\sqrt{N_0\log^3{N_0}}},{\sqrt{T_0\log^3{T_0}}}\}=o_p\left( \min\{N_0^\frac{3}{2},T_0^\frac{3}{2}\}\right)$;
\item[(iii)] $|\calG_l| \kappa^\frac{17}{4} \mu^\frac{5}{2} r^\frac{5}{2} \max\{\sqrt{N_0\log{N_0}},\sqrt{T_0\log{T_0}}\}=o_p\left(\sqrt{N_0} \min\{N_0,T_0\}\right)$, $l=1,\ldots, L$;
\item[(iv)] There are constants $C,c > 0 $ such that
        \begin{align*}
             &c \leq \lambda_{\min} \left( \frac{N}{N_0} \sum_{i \leq N_0} u_{i}u_{i}^\top  \right) \leq  \lambda_{\max} \left( \frac{N}{N_0} \sum_{i \leq N_0} u_{i}u_{i}^\top  \right) \leq C, \\
             &c \leq \lambda_{\min} \left( \frac{T}{T_0} \sum_{t \leq T_0 } v_{t}v_{t}^\top  \right) \leq \lambda_{\max} \left( \frac{T}{T_0} \sum_{t \leq T_0 } v_{t}v_{t}^\top  \right) \leq C;
        \end{align*}
\item[(v)] $\sqrt{N}\norm{\bar{u}_{\calG}} \geq c$ and $\sqrt{T}\norm{v_{t_0}} \geq c$ for some constant $c>0$ where $\bar{u}_{\calG} = |\calG|^{-1}\sum_{i \in \calG} u_{i}$.
\end{itemize}
Then, we have
	\begin{gather*}
		\mathcal{V}_{\mathcal{G}}^{-\frac{1}{2}}\left( \frac{1}{|\mathcal{G}|}\sum_{ i \in \mathcal{G}}\widehat{m}_{it_0} -   \frac{1}{|\mathcal{G}|}\sum_{ i \in \mathcal{G}}m_{it_0} \right) \overset{D}{\longrightarrow} \mathcal{N}(0,1),
	\end{gather*}
where
	$$
\mathcal{V}_{\mathcal{G}}=
	\sigma^2 \left( \bar{u}_{\calG}^\top\left( \sum_{j\le N_0} u_j u_j^\top\right)^{-1}\bar{u}_{\calG}
	+ \frac{1}{|\calG|}v_{t_0}^\top\left( \sum_{s \leq T_0} v_s v_s^\top\right)^{-1}v_{t_0} \right).
	$$   
\end{theorem} 

\paragraph{Staggered Adaption.} More generally, consider the case of staggered adoption when there are $D$ number of adoption time points, $T_1<T_2< \cdots<T_D$, and $D$ number of corresponding groups of treated units, $G_1, \ldots , G_D$. As in the previous situation, suppose that we are interested in inference for the group average at time $t_0$. Denote by $N_{0}$ the number of units that are untreated until $t_0$, and by $T_0$ the number of time periods where $\{1,\dots,N_0\}$ is untreated, respectively.

We proceed by first splitting $\calG$ into smaller groups, denoted by $\{\calG_l\}_{0\leq l \leq L}$ with the convention that $\calG_0 = \calG \cap \{1, \cdots , N_0\}$. In doing so, we want to make sure that all units in each subgroup $\{\calG_l\}_{1\leq l \leq L}$ have the same adoption time point, e.g., $\calG_l \subseteq G_{d_l}$, as illustrated in Figure \ref{fig:howtoconstruct_6_new}. Denote by $D_\calG=\{d_l: 1\le l\le L\}$ and by $\psi_{\min,O_{d}}$ the smallest singular value of the submatrix $M_{O_d}=(m_{it})_{1\leq i \leq N_0, 1\leq t \leq T_d}$.

\begin{figure}[htbp]
\centering
\includegraphics[width=0.7\textwidth]{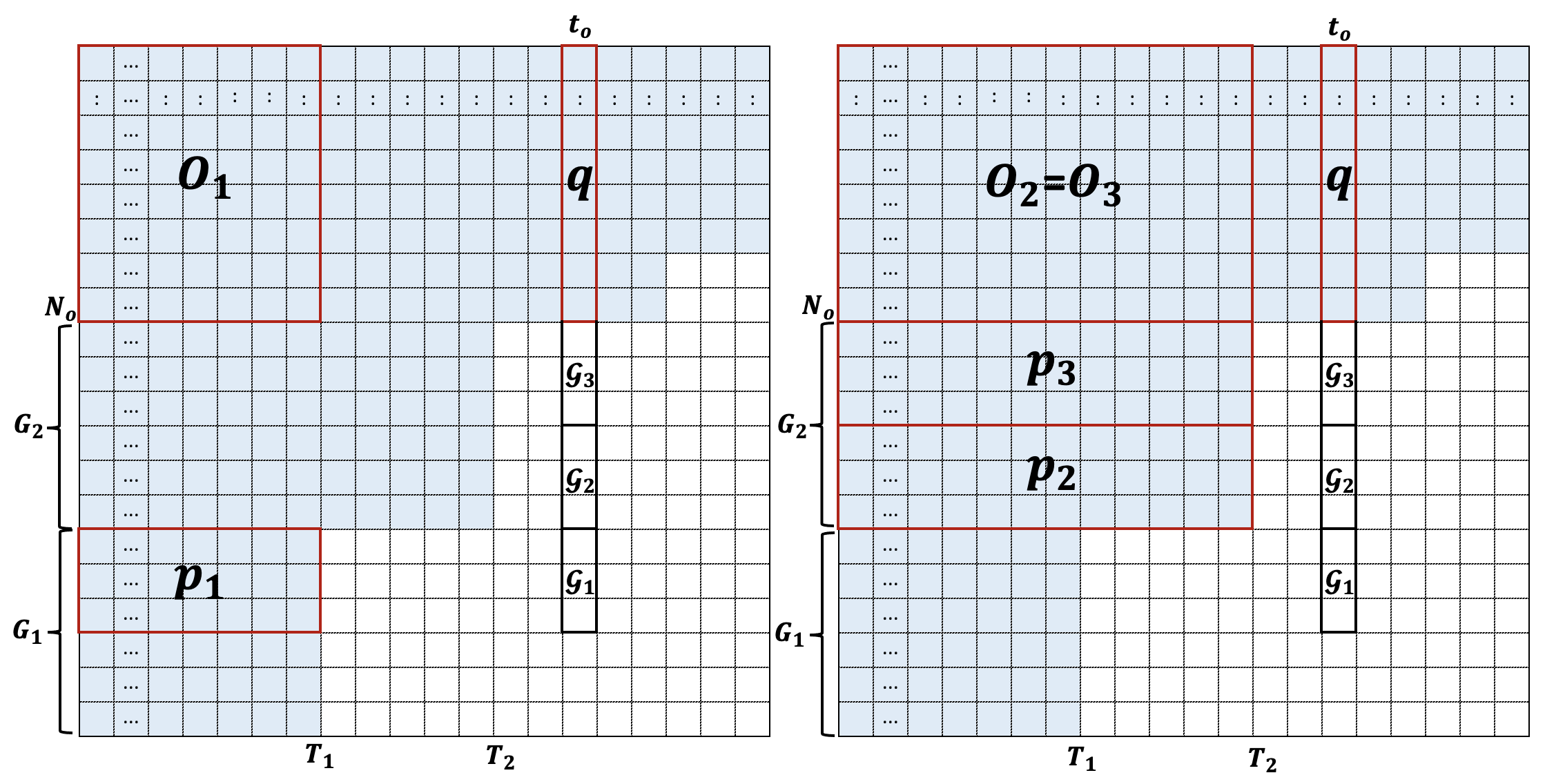}
\caption{
Submatrix construction: For each $1\leq l \leq 3$, we make the submatrix $Y_l$ by putting $O_l$, $p_l$, $q$, and $\calG_l$ together. In addition, we estimate the entries in $\calG_0$ using the fully observed part $Y_0 = (y_{it})_{1\leq i \leq N_0, 1 \leq t \leq T_0}$.}
\label{fig:howtoconstruct_6_new}
\end{figure}

\begin{theorem}\label{thm:groupclt_split_staggered_adoption}
 Assume that for any $d \in D_\calG \cup \{0\}$ and $l=1,\ldots, L$,
\begin{itemize} 
\item[(i)] $\sigma \kappa^\frac{23}{4} \mu^\frac{3}{2} r^\frac{3}{2} \sqrt{N_0} \max\{{N_0\sqrt{\log{N_0}}},{T_d\sqrt{\log{T_d}}}\}=o_p\left(\psi_{\min,O_d} \min\{N_0,T_d\}\right)$;
\item[(ii)] $\kappa^\frac{11}{2} \mu^{3} r^{3} \sqrt{N_0}\max\{{\sqrt{N_0\log^3{N_0}}},{\sqrt{T_d\log^3{T_d}}}\}=o_p\left( \min\{N_0^\frac{3}{2},T_d^{\frac{3}{2}}\}\right)$;
\item[(iii)] $|\calG_l| \kappa^\frac{17}{4} \mu^\frac{5}{2} r^\frac{5}{2} \max\{\sqrt{N_0\log{N_0}},\sqrt{T_{d_l}\log{T_{d_l}}}\}=o_p\left(\sqrt{N_0} \min\{N_0,T_{d_l}\}\right)$;
\item[(iv)] there are constants $C,c > 0 $ such that
        \begin{align*}
             &c \leq \lambda_{\min} \left( \frac{N}{N_0} \sum_{i \leq N_0} u_{i}u_{i}^\top  \right) \leq  \lambda_{\max} \left( \frac{N}{N_0} \sum_{i \leq N_0}u_{i}u_{i}^\top  \right) \leq C, \\
             &c \leq \lambda_{\min} \left( \frac{T}{T_d} \sum_{t \leq T_d}v_{t}v_{t}^\top  \right) \leq \lambda_{\max} \left( \frac{T}{T_d} \sum_{t \leq T_d}v_{t}v_{t}^\top  \right) \leq C;
        \end{align*}
\item[(v)] $\sqrt{N}\norm{\bar{u}_{\calG}} \geq c$ and $\sqrt{T}\norm{v_{t_0}} \geq c$ for some constant $c>0$.
\end{itemize}
Then, we have
	\begin{gather*}
		\mathcal{V}_{\mathcal{G}}^{-\frac{1}{2}}\left( \frac{1}{|\mathcal{G}|}\sum_{ i \in \mathcal{G}}\widehat{m}_{it_0} -   \frac{1}{|\mathcal{G}|}\sum_{ i \in \mathcal{G}}m_{it_0} \right) \overset{D}{\longrightarrow} \mathcal{N}(0,1),
	\end{gather*}
where
	$$
\mathcal{V}_{\mathcal{G}}=
	\sigma^2 \left( \bar{u}_{\calG}^\top\left( \sum_{j\le N_0} u_j u_j^\top\right)^{-1}\bar{u}_{\calG}
	+ \frac{1}{|\calG|}v_{t_0}^\top\left[\sum_{d \in D_\calG \cup \{0\}} \frac{|G_d \cap \calG|}{|\calG|} \left( \sum_{s \leq T_d} v_s v_s^\top\right)^{-1} \right] v_{t_0} \right),
	$$   
 with the convention that $G_0 = \{1,\dots,N_0\}$.
\end{theorem} 

\paragraph{Variance Estimation.} In practice, to use the results above for inferences, we also need to estimate the variance. To this end, let $\widetilde{U}_l\widetilde{D}_l\widetilde{V}_l^\top$ be the SVD of $\calP_r(\widetilde{M}_l)$. Denote by $\widetilde{X}_l = \widetilde{U}_l \widetilde{D}_l^{1/2}$ and $\widetilde{Z}_l = \widetilde{V}_l \widetilde{D}_l^{1/2}$. They can be viewed as estimates of rescaled left and right singular vectors. However, as such, they are significantly biased and the bias can be reduced by considering instead
$$
\widehat{X}_l = \widetilde{X}_l \left(I_r + \lambda_l (\widetilde{X}_l^\top \widetilde{X}_l) \right)^{1/2}, \ \ \widehat{Z}_l = \widetilde{Z}_l \left(I_r + \lambda_l (\widetilde{Z}_l^\top \widetilde{Z}_l) \right)^{1/2}.
$$
We can then use $\widehat{X}_l$ and $\widehat{Z}_l$ in place of the left and right singular vector in defining $\mathcal{V}_{\mathcal{G}}$, leading to the following variance estimate
\begin{eqnarray*}
\widehat{\mathcal{V}}_{\mathcal{G}}=\widehat{\sigma}^2 \sum_{i \leq N_0} \left( \sum_{0 \leq l \leq L} \frac{|\calG_l|}{|\calG|}  \widehat{\bar{X}}_{\calG_l}^\top \left( \sum_{j \leq N_0 } \widehat{X}_{l,j} \widehat{X}_{l,j}^\top \right)^{-1} \widehat{X}_{l,i} \right)^2\\
 +\frac{\widehat{\sigma}^2}{|\calG|} \sum_{0\leq l \leq L} \frac{|\calG_l|}{|\calG|} \widehat{Z}_{t_0}^\top \left( \sum_{s \leq T_{d_l} } \widehat{Z}_{l,s} \widehat{Z}_{l,s}^\top \right)^{-1}  \widehat{Z}_{l,t_0},
\end{eqnarray*}
where $\widehat{\bar{X}}_{\calG_l} = \frac{1}{|\calG_l|} \sum_{j \in \calG_l}\widehat{X}_{l,j}$, $\widehat{\sigma}^2 = \frac{1}{N_0T_0} \sum_{i \leq N_0, t \leq T_0} \widehat{\epsilon}_{it}^2$, and $\widehat{\epsilon}_{it} = y_{it} - \widehat{m}_{it}$. The following corollary shows that asymptotic normality established in Theorem \ref{thm:groupclt_split_staggered_adoption} continues to hold if we use this variance estimate.

\begin{corollary}\label{coro:feasibleclt}
Suppose that the assumptions in Theorem \ref{thm:groupclt_split_staggered_adoption} hold. In addition, suppose that for any $d \in D_{\calG} \cup \{0\}$,
$$\sigma\kappa^5 \mu^3 r^3 N_0 \max\{\sqrt{N_{0} \log N_{0}} ,\sqrt{T_{d} \log T_{d}} \}=o_p\left(\psi_{\min,O_{d}}\min\{ N_{0}, T_{d}\}\right).$$ 
Then
$$
\widehat{\mathcal{V}}_{\mathcal{G}}^{-\frac{1}{2}}\left( \frac{1}{|\mathcal{G}|}\sum_{ i \in \mathcal{G}}\widehat{m}_{it_0} -   \frac{1}{|\mathcal{G}|}\sum_{ i \in \mathcal{G}}m_{it_0} \right) \overset{D}{\longrightarrow} \mathcal{N}(0,1).
$$
\end{corollary}

Since Theorem \ref{thm:groupclt_split_block} is a special case of Theorem \ref{thm:groupclt_split_staggered_adoption}, the variance estimator can also be used for Theorem \ref{thm:groupclt_split_block}. Specifically, it is enough to change from $T_{d_l}$ in $\widehat{\calV}_\calG$ to $T_0$ for Theorem \ref{thm:groupclt_split_block}.

\section{Application to Tick Size Pilot Program}\label{sec:ticksizepilot}

Our work was motivated by the analysis of the Tick Size Pilot Program, which we shall now discuss in detail to demonstrate how the proposed methodology can be applied in causal panel data models.

\subsection{Data and Methods}

\paragraph{Background.}
In October 2016, the SEC launched the Tick Size Pilot Program to evaluate the impact of an increase in tick sizes on the market quality of stocks. As noted before, the pilot consisted of a control group and three treatment groups: 
\begin{enumerate}
\item[{\sf Control.}] stocks in the control group was quoted and traded in \$0.01 increments;
\item[{\sf Q rule}.] stocks in the Q rule group was quoted in \$0.05 increments but still traded in \$0.01 increments;
\item[{\sf Q+T rule}.] stocks in this rule group was quoted and traded in \$0.05 increments;
\item[{\sf Q+T+TA rule.}] stocks in this group are also subject to the additional trade-at rule, a regulation which makes exchanges display the NBBO (National Best Bid and Offer) when they execute a trade at the NBBO.
\end{enumerate}
This pilot program has attracted considerable attention, and there are a growing number of studies on the impact of these changes on market quality, often represented by a liquidity measure such as the effective spread since its conclusion in 2018. See, e.g., \cite{albuquerque2020price, chung2020tick, griffith2019making, rindi2019us, werner2022tick}.

\paragraph{Data.}
Data for control variables were obtained from the Center for Research in Security Prices (CRSP) and the daily share-weighted dollar effective spread data from the Millisecond Intraday Indicators by Wharton Research Data Services (WRDS). A key control variable introduced by \cite{chung2020tick} is TBC which measures the extent to which the new tick size (\$0.05) is a binding constraint on the quoted spreads in the pilot periods and is estimated by the percentage of quoted spreads during the day that are equal to or less than 5 cents, which is the new minimum quoted tick size under the Q rule. Specifically, we calculate the percentage of NBBO updates with quoted spread less than or equal to 5 cents for each day. Using the TBC variable, we can check the effect of an increase in the minimum quoted spread (from 1 cent to 5 cents) on the effective spread.

A data-cleaning process similar to \cite{chung2020tick} yields a total of $N=1,461$ stocks with $N_0=735$ in the control group, $N_1=254$ in the Q group, $N_2=244$ in the Q+T group, and $N_3=228$ in the Q+T+TA group. Following \cite{chung2020tick}, data from Oct 1, 2015 to Sep 30, 2016 were used as the pre-pilot periods and Nov 1, 2016 to Oct 31, 2017 as the pilot periods, i.e., $T_0 = 253$ and $T_1 = 252$ for daily data. See \cite{chung2020tick} for further discussion of data collection. As is common in previous studies, we consider the daily effective spread in cents as a measure of liquidity. Denote by $y_{it}^{(d)}$ the potential outcome for stock $i$ at time $t$ under treatment $d$ with the convention that $d=0, 1,2,3$ corresponds to the control, the Q rule, the Q + T rule, and the Q + T + TA rule, respectively. The four matrices $Y^{(d)}=(y_{it}^{(d)})_{1\le i\le N, 1\le t\le T}$ have block missing patterns, as shown in Figure \ref{fig:severaltreatment_new}.

\begin{figure}[htbp]
	\centering
	\includegraphics[width=0.9\textwidth]{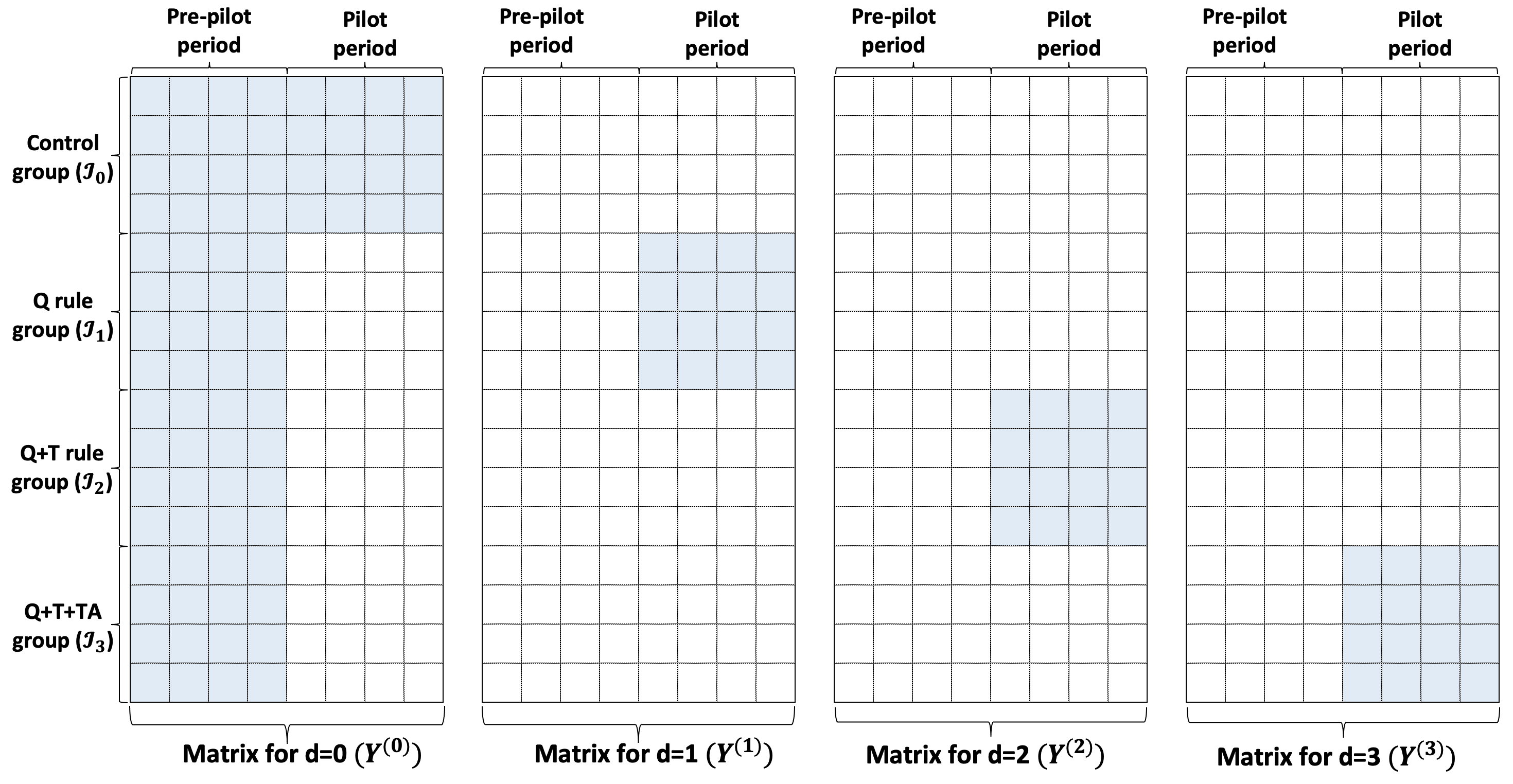}
	\caption{Missing pattern in the pilot program: The blue area is the observed area and the white area is the missing area. In the case of the controlled situation ($d=0$), we can observe the outcomes of all units in the pre-pilot periods and those of the control group in the pilot periods. In the case of the treated situation by the treatment $d$, we can only observe the outcomes of the treatment group $\calI_d$ in the pilot periods.}
\label{fig:severaltreatment_new}
\end{figure}

\paragraph{Model.} Previous studies of the effects of the quote (Q) rule, the trade (T) rule, and the trade-at (TA) rule on the liquidity measure are usually based on traditional regression or difference-in-difference methods by assuming that the treatment effect is constant across all units and time periods. For instance, \cite{chung2020tick} postulated $y_{it} = y_{it}^{(d)}$ if unit $i$ receives treatment $d$ at time $t$ where the potential outcomes
$$
y_{it}^{(d)}=m_{it}^{(d)}+x_{it}^\top \beta+\epsilon_{it}
$$
and
\begin{equation}\label{eq:theirmodel}
m_{it}^{(d)}=\mu^{(d)}+\alpha_i+\delta_t.
\end{equation}
Here, $\mu^{(0)}=0$, $\mu^{(1)}, \mu^{(2)},\mu^{(3)}$, $\alpha_i$s and $\delta_t$s are unknown parameters, and $x_{it}$ is a set of control variables that includes typical stock characteristics like stock prices and trading volumes, and TBC, a variable measuring the extent to which the new tick size (\$0.05) is a binding constraint on the quoted spreads in the pilot period. See Section \ref{sec:additionalempirical} in the Appendix for further details. It is worth noting that, in addition to the treatment effects ($\mu^{(1)}$, $\mu^{(2)}$ and $\mu^{(3)}$), their differences $\theta^{(d)}:=\mu^{(d)}-\mu^{(d-1)}$ are also of interest, as they represent the treatment effects of quote rule, trade rule, and trade-at rule, respectively.

However, \eqref{eq:theirmodel} fails to account for the significant heterogeneity in the treatment effects across units and time periods. To this end, we shall consider a more flexible model:
\begin{equation}\label{eq:ourmodel}
m_{it}^{(d)}=\zeta_i^\top \eta_t^{(d)},\qquad {d=0,1,2,3},
\end{equation}
where $\zeta_i$ is a $r$-dimensional vector of (latent) unit specific characteristics and $\eta_{t}^{(d)}$ is the corresponding coefficients of $\zeta_i$ at time $t$ in the potential situation $d$. As we shall see later in this section, \eqref{eq:ourmodel} allows us to get more insights into the treatment effects of the pilot program.

One of the key assumptions of Model \eqref{eq:ourmodel} is that the subspace spanned by the left singular vector of $M^{(d)}=(m_{it}^{(d)})_{1\le i\le N, 1\le t\le T}$ for all $d=1,2,3$ is included in the subspace spanned by the left singular vector of $M^{(0)}$. \cite{agarwal2020synthetic} propose a subspace inclusion test to check the validity of this assumption. We carried out this test on the pilot data, which confirms this is a reasonable assumption.

We note that similar low-rank models have also been considered by \cite{agarwal2020synthetic} and \cite{chernozhukov2021inference} earlier. However, it is unclear how their methodology can be adapted for the analysis of the Tick Size Program. For example, \cite{chernozhukov2021inference} impose conditions on the missing pattern that are clearly violated by the pilot data; \cite{agarwal2020synthetic} only study the average treatment effect and so cannot be used to assess the heterogeneity or dynamics of the treatment effects across units and time periods, respectively.

\paragraph{Estimation.}
We now discuss how we can apply the methodology in the previous sections to analyze the tick size program, and in particular to estimate and make inferences about \eqref{eq:ourmodel}. More specifically, we are interested in estimating the group-averaged treatment effects: for an interesting group of treated units $\calG$,
$$
\mu^{(d)}_t:={1\over |\calG|}\sum_{i\in\calG} [m^{(d)}_{it}-m^{(0)}_{it}],
$$
and their differences:
$$
\theta^{(d)}_t:=\mu^{(d)}_t-\mu^{(d-1)}_t,
$$
for $t>T_0$. Especially, when $\calG$ is a certain unit, it reduces to the individual treatment effect and if $\calG$ is the group of all treated units, it becomes the cross-sectional averaged treatment effect. To this end, we shall derive estimates for $m^{(d)}_{it}$ under Model \eqref{eq:ourmodel}.

First, note that, for this particular application, one of the covariates (TBC) is only present for the pilot periods. Therefore, we cannot hope to estimate the regression coefficient $\beta$ using the pre-pilot data alone, as suggested by \cite{bai2021matrix}. Nonetheless, under \eqref{eq:ourmodel}, $y_{it}$s follow an interactive fixed effect model:
$$y_{it} = x_{it}^\top \beta + L_{it} + \epsilon_{it}$$
for some low rank components $L_{it}$ and therefore the regression coefficient $\beta$ can be estimated at the rate of $O_p(1/\sqrt{NT})$. See \cite{bai2009panel} for details. This is much faster than that of the estimates of $m^{(d)}_{it}$. For brevity, we shall, therefore, treat the regression coefficient $\beta$ as known in what follows, without loss of generality.

For $d=0$, we can apply the method proposed in the previous sections to the potential outcome panel $\tilde{Y}^{(0)}_{it}=(y_{it}^{(0)}-x_{it}^\top\beta)_{1\le i\le N, 1\le t\le T}$. As illustrated in Figure \ref{fig:severaltreatment_new}, it has a block missing pattern with $\omega_{it}^{(0)}=1$ if and only if $t\le T_0$ or $i\le N_0$. As such, we can derive estimates $\widehat{m}_{it}^{(0)}$ for $t > T_0$.

When $d>0$, we can only observe $y_{it}^{(d)}$ if unit $i$ receives treatment $d$ and $t>T_0$, so our method cannot be applied directly. Instead, we shall combine all observations from prepilot periods and these observations to form a panel $\tilde{Y}^{(d)}$ whose $(i,t)$ entry is $y_{it}^{(d)}-x_{it}^\top\beta$ if $i$ receives treatment $d$ and $t>T_0$, is $y_{it}^{(0)}-x_{it}^\top\beta$ if $t\le T_0$, and is missing otherwise. Let $\tilde{M}^{(d)}$ be a $N\times T$ matrix whose $(i,t)$ entry is $m^{(0)}_{it}$ if $t\le T_0$, and $m^{(d)}_{it}$ otherwise. $\tilde{Y}^{(d)}$ can be viewed as the noisy observation of $\tilde{M}^{(d)}$ with a block missing pattern: $\omega_{it}^{(d)}=1$ if and only if unit $i$ receives treatment $d$ or $t\le T_0$. Under \eqref{eq:ourmodel}, $\tilde{m}^{(d)}_{it}=\zeta_i^\top\tilde{\eta}^{(d)}_t$ where $\tilde{\eta}^{(d)}_t=\eta_t^{(0)}$ if $t\le T_0$ and $\eta^{(d)}_t$ otherwise. Therefore, we can again apply our method to $\tilde{Y}^{(d)}$ to obtain estimates $\hat{m}_{it}^{(d)}$ for $t>T_0$.

We shall then proceed to estimate the treatment effects by
$$
\widehat{\mu}^{(d)}_t:={1\over |\calG|}\sum_{i\in \calG} [\widehat{m}^{(d)}_{it}-\widehat{m}^{(0)}_{it}]
\qquad {\rm and}\qquad
\widehat{\theta}^{(d)}_t:={1\over |\calG|}\sum_{i\in\calG} [\widehat{m}^{(d)}_{it}-\widehat{m}^{(d-1)}_{it}].
$$

\paragraph{Inferences.} We can also use the results from the last section to derive the asymptotic distribution for $\widehat{\mu}^{(d)}_t$ and $\widehat{\theta}^{(d)}_t$. More specifically, let $M$ be a $N\times (T+3T_1)$ matrix that combines all observed outcomes: the first $T$ columns of $M$ consist of the potential outcomes under the control for the whole periods $(m_{it}^{(0)})_{i\leq N, t\leq T}$, the next $T_1$ columns the potential outcomes under the Q rule for the pilot periods $(m_{it}^{(1)})_{i\leq N, t > T_0}$, followed by those under the Q+T rule again for the pilot periods $(m_{it}^{(2)})_{i\leq N, t > T_0}$, and finally those under the Q+T+TA rule $(m_{it}^{(3)})_{i\leq N, t > T_0}$. Note that $M$ is also a rank-$r$ matrix. Let $M=UDV^\top$ be its singular value decomposition. Denote by $u_i^\top$ and $v_t^\top$ the $i$-th row vector of $U$ and $t$-th row vector of $V$, respectively. In addition, denote by $\calI_d$ the group of units treated by treatment $d$ with the convention that $\calI_0$ is the control group. Then, under suitable conditions, we have
	\begin{align*}
			\mathcal{V}_{\mu}^{-\frac{1}{2}}\left( \widehat{\mu}^{(d)}_{t_0} -   \mu^{(d)}_{t_0} \right) \overset{D}{\longrightarrow} \mathcal{N}(0,1),\ \ \mathcal{V}_{\theta}^{-\frac{1}{2}}\left( \widehat{\theta}^{(d)}_{t_0} -   \theta^{(d)}_{t_0} \right) \overset{D}{\longrightarrow} \mathcal{N}(0,1),
		\end{align*}
  $\mathcal{V}_{\mu} = \mathcal{V}_{\calG}(d,0)$ and $\mathcal{V}_{\theta} = \mathcal{V}_{\calG}(d,d-1)$ where
	\begin{align*}
	\mathcal{V}_{\mathcal{G}}(d,d') =&
		\sigma^2 \bar{u}_\calG^\top \left(\sum_{j \in \calI_d} u_j u_j^\top\right)^{-1}  \bar{u}_\calG +  \sigma^2 \bar{u}_\calG^\top \left(\sum_{j \in \calI_{d'}} u_j u_j^\top\right)^{-1}  \bar{u}_\calG\\
		 & \ \  + \frac{\sigma^2}{|\calG|} \left(v_{(d \cdot T_1+t_0)}-v_{(d' \cdot T_1+t_0)}\right)^\top \left(\sum_{s \leq T_0} v_{s} v_{s}^\top \right)^{-1} \left(v_{(d \cdot T_1+t_0)}-v_{(d' \cdot T_1+t_0)}\right).
	\end{align*}   
Similar to before, the variance can be replaced by its estimate. Due to the space limit, we shall defer the formal statements and proofs, as well as derivations of the variance estimator to the Appendix.

\subsection{Empirical Findings} \label{sec:empirical_findings}

\paragraph{Fixed Effects vs Interactive Effects.}
We begin with some exploratory analyses to illustrate the impact of the pilot program. The top left panel of Figure \ref{fig:exploratory} gives the boxplots of difference in the effective spread, averaged over time, after and before the pilot. There are a few units with differences that are much larger in magnitude than usual. For better visualization, the top right panel zooms in with a difference between -10 cents and 10 cents. Taken together, it is clear that the three treatment groups have a significant impact on the effective spread.

\begin{figure}[htbp]
	\centering
	\includegraphics[width=\textwidth]{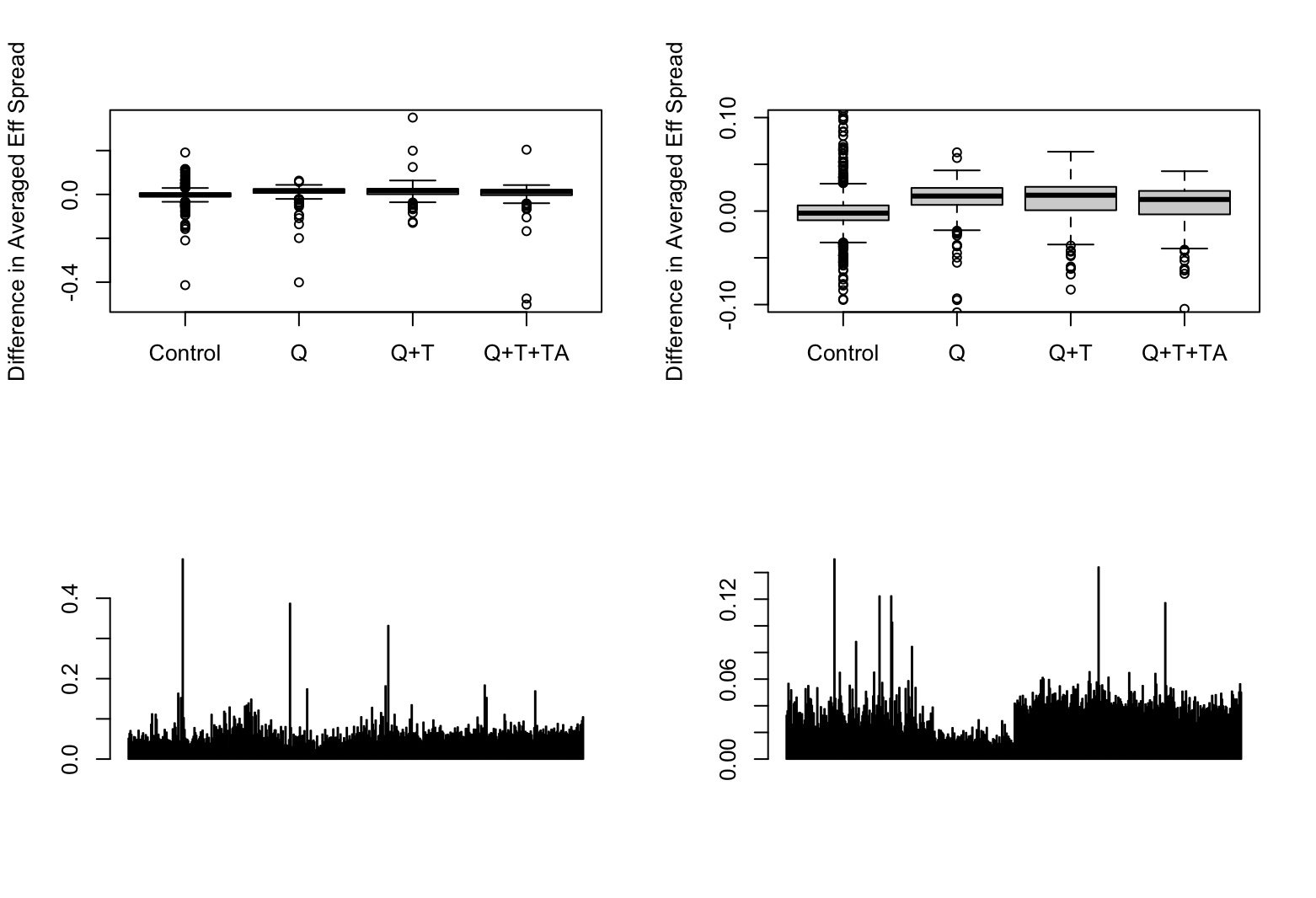}
	\caption{Top panels: Boxplot of difference in averaged effective spread after and before the tick size program. Bottom panels: two stocks treated with Q rule and with different treatment effects.}
\label{fig:exploratory}
\end{figure}

The treatment effect of the pilot, however, differs between units. The bottom panels of Figure \ref{fig:exploratory} show barplots of the time series of the effective spread of two typical stocks. The impact of the treatment is much clearer for the stock depicted in the bottom right panel.

The difference in treatment effect among the units suggests that the interactive effect model is more suitable than the fixed effect model used in the previous studies. Note that the fixed effect model \eqref{eq:theirmodel} can be viewed as a special case of the interactive effect model \eqref{eq:ourmodel} with $\zeta_i = [1 \ \ \alpha_i]^\top$, $\eta_t^{(d)} = [\delta_t+ \mu^{(d)} \ \ 1]^\top$. We conducted a Hausman-type model specification test to further show that the fixed effect model is inadequate in capturing the heterogeneity of the treatment effect. More specifically, denote our estimator of $\theta_{it}^{(d)} \coloneqq m_{it}^{(d)} - m_{it}^{(d-1)}$ by $\hat{\theta}^{(d)}_{it}$ and the two-way fixed effect estimator of $\theta^{(d)} \coloneqq \mu^{(d)} - \mu^{(d-1)} (= m_{it}^{(d)} - m_{it}^{(d-1)})$ in Model \eqref{eq:theirmodel} by $\tilde{\theta}^{(d)}$. We considered the following test statistic for model specification:
$$T-stat_{\rm ms} =  \max_{i \in \calN_{tr}, T_0 < t \leq T} \max_{1 \leq d \leq 3} |\hat{\tau}^{(d)}_{it}|$$
where
$\calN_{tr}$ is the group of all treated stocks, $\hat{\tau}^{(d)}_{it} = \hat{\calV}_{d,it}^{-1/2} (\hat{\theta}^{(d)}_{it} - \tilde{\theta}^{(d)})$, and $\hat{\calV}_{d,it}$ is the estimator of the asymptotic variance of $\hat{\theta}^{(d)}_{it} - \tilde{\theta}^{(d)}$.
Moreover, to test whether $\theta_{it}^{(d)}$ is time and unit invariant or not, we also considered the test statistic such that
$$
T-stat_{(d)} =   \max_{i \in \calN_{tr}, T_0 < t \leq T} \abs{\hat{\calV}_{d,it}^{-1/2} (\hat{\theta}_{it}^{(d)} - \bar{\hat{\theta}}^{(d)})}
$$
where $\bar{\hat{\theta}}^{(d)} =\frac{1}{|\calN_{tr}|T_{1}}\sum_{i\in \calN_{tr}, T_0 < t \leq T}\hat{\theta}_{it}^{(d)}$. 

We derived the large sample distributions of the test statistics under the null and corresponding critical values using the Gaussian bootstrap method \citep[see, e.g.,][]{belloni2018high}. And the null hypothesis that Model \eqref{eq:theirmodel} is well specified and the null hypotheses that $\{\theta_{it}^{(d)}\}_{1\leq d \leq 3}$ are time and unit invariant are all rejected at 1\% significance level, again indicating that Model \eqref{eq:theirmodel} is misspecified and $\{\theta_{it}^{(d)}\}_{1\leq d \leq 3}$ are time and unit variant.

To further illustrate the heterogeneity of the treatment effect, we compute the estimated unit-specific treatment effect averaged over time: $\bar{\hat{\theta}}^{(d)}_i:=T_1^{-1}\sum_{t>T_0} \hat{\theta}^{(d)}_{it}$ and Figure \ref{fig:density} gives the kernel density estimates of these unit-specific treatment effects for the Q rule, T rule and TA rule respectively. It is evident from these density plots that there is considerable amount of variation and skewness among the estimated treatment effects across units.

\begin{figure}[htbp]
	\centering
	\includegraphics[width=0.8\textwidth]{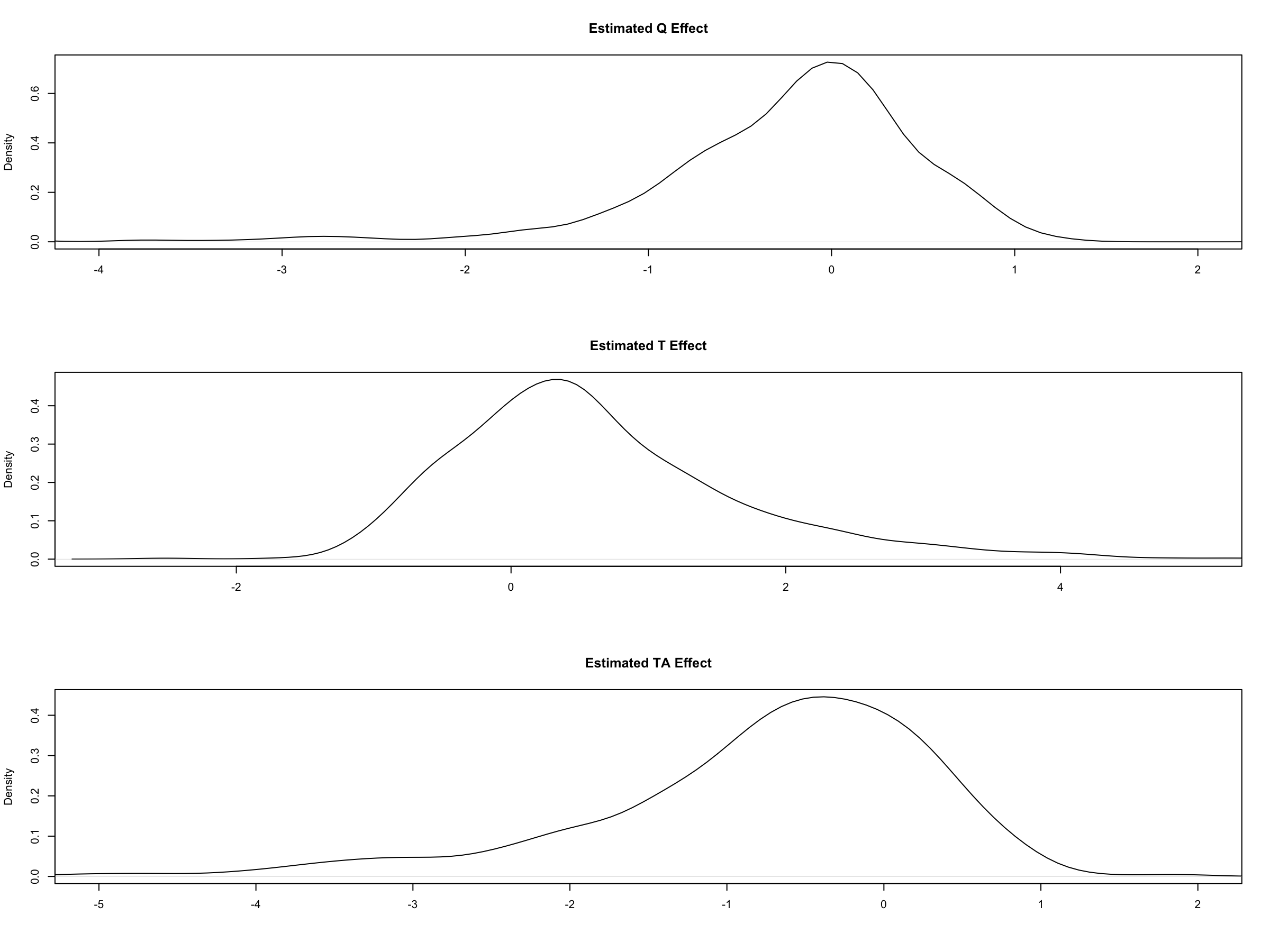}
	\caption{Kernel density estimates of the estimated unit-specific treatment effect averaged over time.}
	\label{fig:density}
\end{figure}

Note that a key assumption behind the interactive effect model is that the unit specific characteristic $\zeta_i$ remains the same across all treatment groups as well as the control group so that they can be learned from the pre-pilot periods and utilized for the estimation of $m^{(d)}_{it}$ during the pilot period. This amounts to the assumption that the left singular space of $M^{(d)}$ is included in that of $M^{(0)}$. To check the validity of the assumption, we carry out the subspace inclusion test for $d=1,2,3$ introduced in \cite{agarwal2020synthetic}, and the test statistics are $0.15$, $0.19$ and $0.11$ with corresponding critical values at 95\% level $0.43$, $0.48$ and $0.28$. Additionally, we also confirm that the ranks of $(m^{(0)}_{it})_{i \in \calI_d , t \leq T_0}$ and $[(m^{(0)}_{it})_{i \in \calI_d , t \leq T_0} \ \ (m^{(d)}_{it})_{i \in \calI_d , t > T_0}]$ are the same for all $1 \leq d \leq 3$ using the typical rank estimation method \citep[e.g.,][]{ahn2013eigenvalue}, which implies the validity of this assumption.

The rank test also indicates that $r=1$ is an appropriate choice for the pilot data. The associated $R^2$ is 0.79. This is to compared with the fixed effect model \eqref{eq:theirmodel} whose $R^2$ is 0.67 with the same degrees of freedom. This again suggests that the interactive effect model \eqref{eq:ourmodel} is preferable.

\paragraph{Dynamics of Treatment Effects.}
Next, we examine the dynamics of the treatment effects of the Q rule, the T rule, and the TA rule.

\begin{figure}[h!]
\centering
\includegraphics[width=0.9\textwidth]{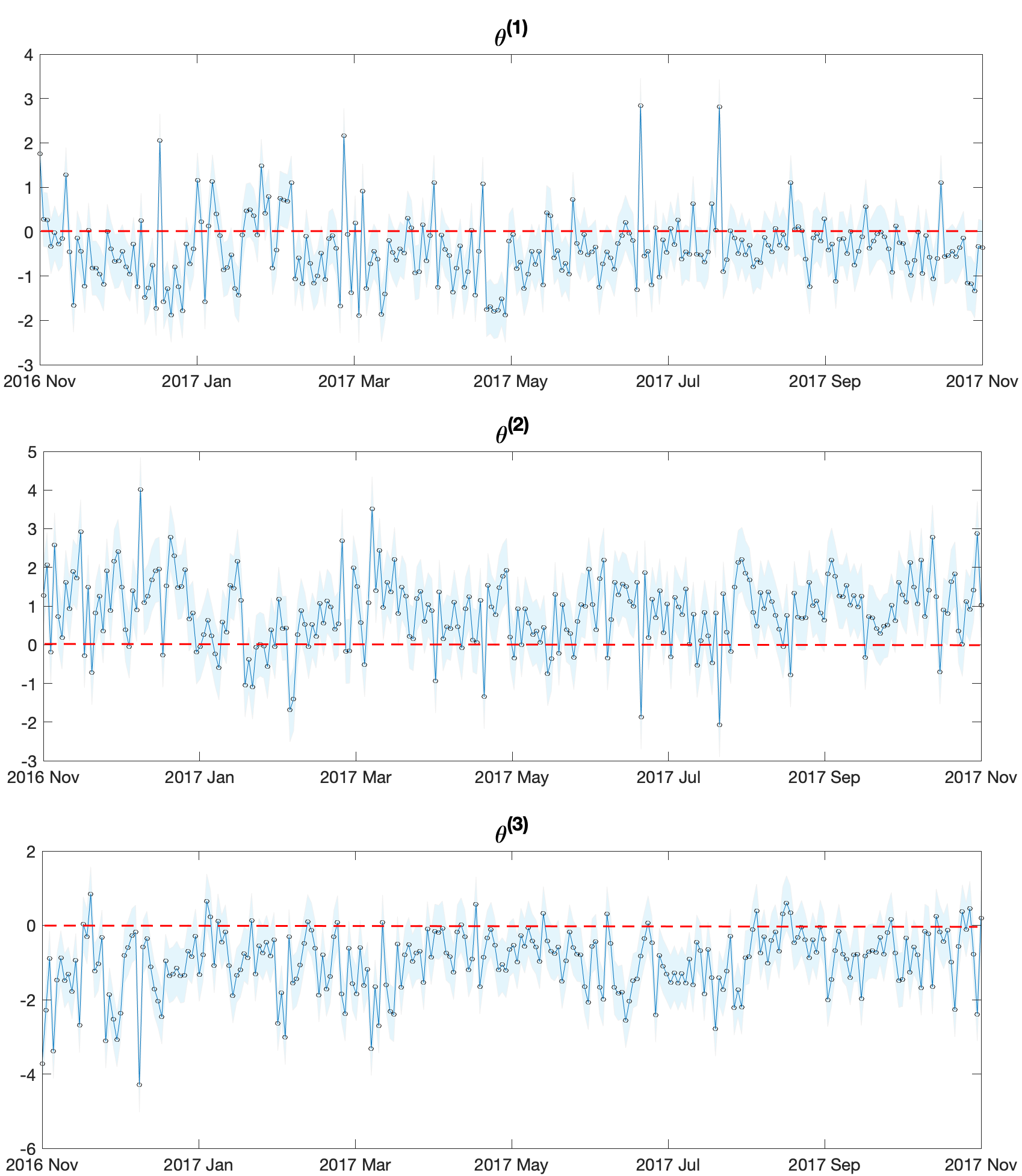}
\caption{The dynamics of the daily cross-sectional average of $\theta_{it}^{(d)}$: For the confidence band, we use the 95\% uniform critical value, $\Phi^{-1}(1 - 0.025/252)$. The dots denote the daily cross-sectional average of $\theta_{it}^{(d)}$.}
\label{fig:theta_dynamics_daily}
\end{figure}

To better visualize the dynamics, we plot in Figure \ref{fig:theta_dynamics_daily} the estimated daily treatment effects along with their 95\% confidence interval, adjusted with Bonferoni correction. To gain further insights, we also plot in Figure \ref{fig:theta_dynamics_weekly} the weekly average of the estimated daily treatment effects, again with their 95\% confidence interval adjusted with Bonferoni correction. Note that to do so, we need to consider the estimator of the form $$\frac{1}{|\calS|}\frac{1}{|\calN_{tr}|} \sum_{t \in \calS} \sum_{i \in \calN_{tr}} \hat{\theta}^{(d)}_{it}$$ 
where $\calS$ is a week of interest. We can generalize the inferential theory from the previous section straightforwardly with the new variance:
$$
	 \sum_{\rho \in \{d, d-1 \}} \left[
		\frac{\sigma^2}{|\calS|} \bar{u}_{\calN_{tr}}^\top \left(\sum_{j \in \calI_\rho} u_j u_j^\top\right)^{-1}  \bar{u}_{\calN_{tr}} \right]  + \frac{\sigma^2}{|\calN_{tr}|} \bar{v}_{\diff}^\top \left(\sum_{s \leq T_0} v_{s} v_{s}^\top \right)^{-1} \bar{v}_{\diff} ,
$$
where $$\bar{v}_{\diff} = \frac{1}{|\calS|} \sum_{t \in \calS} v_{(d \cdot T_1+t)}-v_{((d-1) \cdot T_1+t)}.$$

\begin{figure}[h!]
\centering
\includegraphics[width=0.9\textwidth]{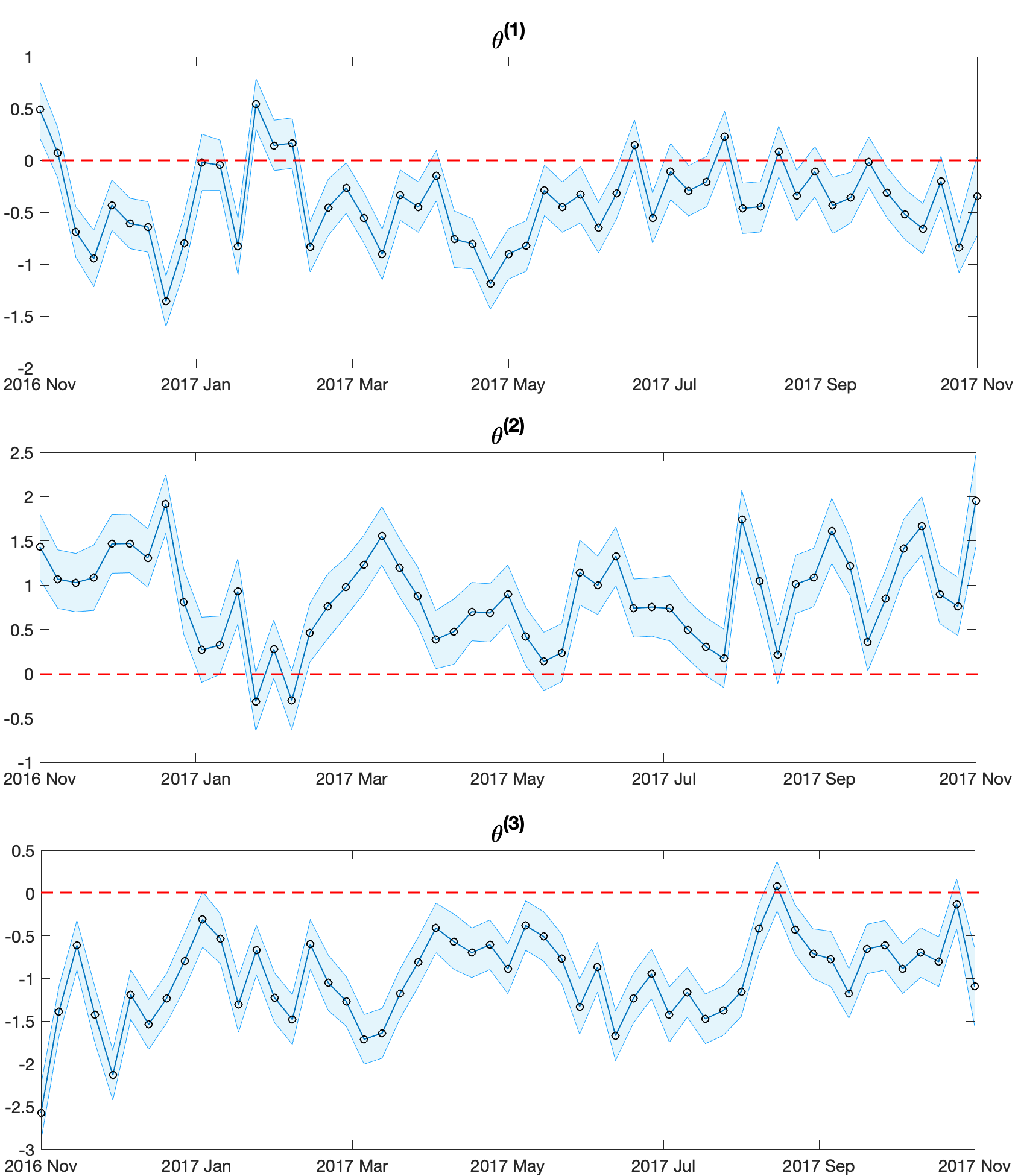}
\caption{The dynamics of the weekly cross-sectional average of $\theta^{(d)}_{it}$: For the confidence band, we use the 95\% uniform critical value, $\Phi^{-1}(1 - 0.025/53)$. The dots denote the weekly cross-sectional average of $\theta^{(d)}_{it}$.}
\label{fig:theta_dynamics_weekly}
\end{figure}

$\theta_{it}^{(2)}$ and $\theta_{it}^{(3)}$ can be interpreted as the treatment effects of the T rule and the TA rule. As expected by theory in the literature, we have the positive treatment effects of T rule most of the time. The T rule has a negative effect on price improvements, as liquidity providers are less likely to offer them when the minimum possible price improvement is larger. For example, if the T rule makes the minimum possible price improvement to be 5 cents, liquidity providers who would have been willing to provide less than 5 cents of price improvements are unlikely to offer any price improvement at all. Since the effective spread is ``quoted spread - price improvement'', we can expect that treatment effects of the T rule is positive. Here, we use the following definitions: $\texttt{Quoted Spread}_t = A_t - B_t$, $\texttt{Effective Spread}_{t} = 2(P_t -\frac{A_t + B_t}{2})$, and $\texttt{Price Improvement}_{t} = 2(A_t - P_t)$, where $A_t$ is the national best ask price at time $t$, $B_t$ is the national best bid price at time $t$, and $P_t$ is the transaction price.

Interestingly, one can observe that the periods associated with large effects of the T rule usually correspond to large trading volumes. In particular, there were large trading volumes in November, early and mid-December in 2016, March, mid and late June, early August, early September, and late October in 2017, and, by and large, these periods coincide with periods with larger impact of the T rule. In general, the correlation coefficients between the estimated effect of the T rule and the trading volume is $0.33$. This suggests that the effect of the T rule becomes stronger when transactions are more active. This agrees with the well-known fact that price improvement is more likely to occur when stocks are actively traded, and therefore the effect of the T rule through price improvement will become amplified and strong when trades are active.

Moreover, we find that the treatment effects of the TA rule are negative most of the time. The TA rule increases visible liquidity by exposing hidden liquidity because, under the TA rule, a venue should display the best bid or ask to execute incoming market orders at the NBBO. It implies a decrease in the quoted spread and a smaller room for price improvements. \cite{chung2020tick} expect that the effect on the quoted spread is likely to be greater than the effect on price improvements, and so the TA rule decreases the effective spread. Our result corroborates with their conjecture. Further discussion about the empirical findings is given in Section \ref{sec:additionalempirical} in the Appendix.

\section{Simulated Experiments} \label{sec:sim}

To further demonstrate the practical merits and finite sample performance of our methodology, we conducted several sets of simulation experiments.

\subsection{Basic Setting}

The first set of simulations was designed to compare the performance of the proposed estimator with that of other existing estimators in a staggered adoption setting. Here, the size of ``no adoption'' group (G0) was set to 200. There are three adoption groups (G1, G2, G3), and the size of each adoption group was set to 100. The number of time points was 500 with G1 adopting the intervention at the 201st time period, G2 at the 301st time period, and G3 at the 401st time period. The potential outcome under the control follows a low-rank model $y_{it}^{(0)}=\zeta_i^\top\eta_t^{(0)} + \varepsilon_{it}$ where the noise $\varepsilon_{it}$ was sampled independently from the standard normal distribution. The unit specific characteristics $\zeta_i$s were sampled independently from $\calN((2.5/\sqrt{2},2.5/\sqrt{2})^\top, I_2)$ for G0, $\calN((1/\sqrt{2},1/\sqrt{2})^\top, I_2)$ for G1, $\calN((1.5/\sqrt{2},1.5/\sqrt{2})^\top, I_2)$ for G2, and $\calN((\sqrt{2},\sqrt{2})^\top, I_2)$ for G3. In addition, the corresponding coefficient $\eta_t^{(0)}$s were sampled independently from $\calN((1/\sqrt{2},1/\sqrt{2})^\top, I_2)$.

To fix ideas, we consider estimating the missing potential outcome $m^{(0)}_{it}$ of a randomly chosen unit in G2 during the last time period ($t=500$) using different estimators including ours (\verb+CY+) along with those from \cite{bai2021matrix} (\verb+BN+), \cite{agarwal2021causal} (\verb+ADSS+) and \cite{athey2021matrix} (\verb+ABDIK+). For ADSS, following the recommendation in \cite{agarwal2021causal}, we set the number of sub-subgroup $K$ to be $K \asymp |AR^{(k)}|_o^{1/3}$. Table \ref{tab:RMSE_sim} reports the RMSE, summarized from 1,000 simulation runs.  The performance of CY, BN, and ADSS are superior to that of ABDIK with CY slightly better than BN and ADSS.

\begin{table}[htbp]
	\centering
	\begin{tabular}{c|cccc}
		\hline
		\hline
		& CY    & BN & ADSS  & ABDIK \\
		\hline
		RMSE  & 0.1157 & 0.1176 & 0.1193 & 0.3507 \\
		\hline
	\end{tabular}%
	\caption{Root mean square error for different methods. }
  \label{tab:RMSE_sim}
\end{table}%

In addition, we recorded the coverage probabilities of the (asymptotic) confidence intervals associated with each method, with the exception of ABDIK for which such inferential tools have not been developed in the literature. From Table \ref{tab:coverage_comparison}, we can see that the coverage probabilities of ADSS are not close to the nominal level, indicating that the asymptotic distributional properties may not provide good approximations in this setting. On the other hand, our method and BN are more accurate, with ours more closely following the target probabilities.

\begin{table}[htbp]
	\centering
	\begin{tabular}{c|ccc}
		\hline
		\hline
		& \multicolumn{3}{c}{Target probability} \\
		Estimator & 90\%  & 95\%  & 99\% \\
		\hline
		CY    & 90.50\% & 95.90\% & 99.30\% \\
		BN & 94.20\% & 97.50\% & 99.50\% \\
		ADSS  & 68.90\% & 76.10\% & 84.80\% \\
		\hline
	\end{tabular}%
	\caption{Coverage probability of the confidence interval.}
	\label{tab:coverage_comparison}%
\end{table}%

\subsection{Interactive Effect Model}
Our next set of simulations mimics the setting of the pilot program studied in the previous section. More specifically, we considered Model \eqref{eq:ourmodel} with two treatment groups, $\calI_1$ and $\calI_2$, and a control group, $\calI_0$. Each treatment group receives a different treatment in the pilot periods. 

We set $r=2$ and generated the unit specific characteristics from $\zeta_{i} \sim \calN((1/\sqrt{2},1/\sqrt{2})^\top, I_2)$, $\varepsilon_{it} \sim \calN(0, 1)$, $\eta_{it}^{(0)} \sim \calN((1/\sqrt{2},1/\sqrt{2})^\top, I_2)$, $\eta_{it}^{(1)} \sim \calN((1.5/\sqrt{2},1.5/\sqrt{2})^\top, I_2)$, and $\eta_{it}^{(2)} \sim \calN((\sqrt{2},\sqrt{2})^\top, I_2)$. In addition, two control variables were included: $x_{1,it}$ is generated from $\calN(0, 1)$ while $x_{2,it}$ is generated from $\calN(0, 1)$ if $t \in \text{Pilot period}$ and 0 otherwise. We set the regression coefficient $\beta = (1,1)^\top$ and estimated it using the interactive fixed effect estimation with data of whole periods. The numbers of $\calI_0$, $\calI_1$, and $\calI_2$ were set to 250 and the numbers of pre-pilot periods and pilot periods were both set to 250.

As before, we estimated $\mu_{it}^{(d)}$ and $\theta_{it}^{(d)}$ for $1 \leq d \leq 2$ of a randomly chosen unit in $\calI_2$ at the last period ($t = 500$). Table \ref{tab:coverage_treatment} reports the coverage probabilities of our methods for $\mu_{it}^{(1)}$, $\mu_{it}^{(2)}$, and $\theta_{it}^{(2)}$, summarized from 1,000 simulation runs. It is evident that our coverage probabilities are quite close to the corresponding target probabilities. This is complemented by  Figure \ref{fig:histogram} that shows the histograms of the standardized estimates (t-statistics) along with the standard normal distribution, which again confirms the asymptotic normality of our estimates.

\begin{table}[htbp]
	\centering
	\begin{tabular}{c|ccc}
		\hline
		\hline
		& \multicolumn{3}{c}{Target probability} \\
		Parameter & 90\%  & 95\%  & 99\% \\
		\hline
		$\mu_{it}^{(1)}$ ($= \theta^{(1)}_{it}$)& 90.20\% & 95.60\% & 98.70\% \\
		$\mu_{it}^{(2)}$ & 90.70\% & 95.80\% & 99.00\% \\
		$\theta^{(2)}_{it}$ & 89.20\% & 94.20\% & 98.50\% \\
		\hline
	\end{tabular}%
	\caption{Coverage probability of the confidence interval}
	\label{tab:coverage_treatment}%
\end{table}%

\begin{figure}[htbp]
\centering
\includegraphics[width=\textwidth]{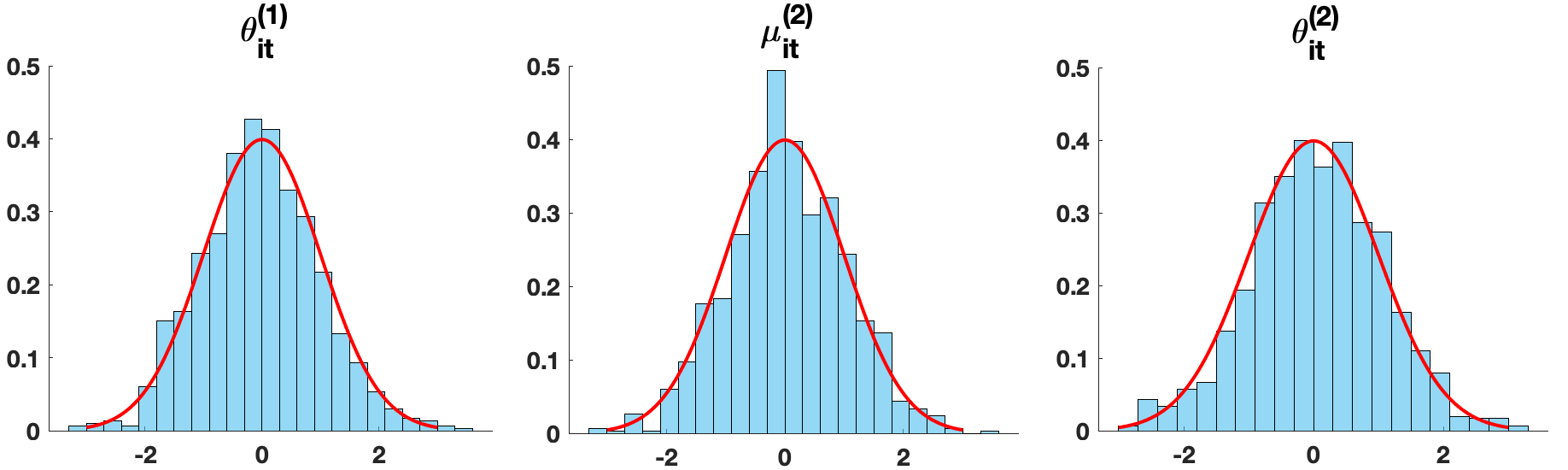}
\caption{Histograms for standardized estimates (t-statistics)}\label{fig:histogram}
\end{figure}

\subsection{Simulated Tabacco Sales Experiments}

Our final experiment is similar to that from \cite{agarwal2021causal} and \cite{athey2021matrix} and is based on the tobacco sales data of \cite{abadie2010synthetic}. In 1988, California introduced the first anti-tobacco legislation in the United States (Proposition 99) and to study the effect of this legislation on tobacco sales, \cite{abadie2010synthetic} used the per capita cigarette sales data which was collected across 39 U.S. states from 1970 to 2000. We considered the time horizon of $n = 31$ years and restricted our focus to the $m = 38$ untreated states (excluding California) in their dataset. This data was encoded into a $38 \times 31$ matrix, $Y$, where the entry $y_{it}$ represents the potential outcome of per capita cigarette sales (in packs) for state $i$ in year $t$ under control, i.e., without any intervention in place. 

To generate MNAR data, we artificially introduced interventions to a subset of states where the probability that a state adopts an intervention (e.g., tobacco control program) depends on their change in cigarette sales pre-1986 and post-1986. More specifically, we considered the following adoption protocol: First, we clustered states into four categories — severe, moderate, mild, and good — based on their percentage change in average cigarette sales during 1986-2000 compared to that during 1970-1985. The severe states are the states where average cigarette sales are hardly reduced ($-0\% \sim -10\%$, MO,WV,SC,AL,AR,TN), and the moderate states are the states whose percentage change is between $-10\%$ and $-15\%$ (KY,DE,GA,IN,OH,MS). The mild states are the states where the percentage change is between $-15\%$ and $-20\%$ (NE,LA,IA,SD,WI,PA). The rest are good states ($-20\% \sim$).

We then designated the timing and probability of intervention for mild, moderate, severe, and good states differently. Half of the severe states adopt an intervention in 1986 and the other half in 1991. Half of the moderate states adopt the intervention at 1991 and the other half in 1996. Half of the mild states adopt the intervention in 1996, and the other half do not adopt the intervention. In addition, the good states do not adopt the intervention at all. This setup reflects the scenario in which a state whose average sales may not be reduced sufficiently without the intervention is more likely to adopt the intervention early.

Table \ref{tab:RMSE_experiment} shows the average RMSE of missing components caused by the intervention in 10 experiments. Here, the missing components mean the potential ``control (no adoption)'' outcomes in the intervention period. The only randomization lies in the resampling of the observation patterns. We can check that ABDIK performs relatively poorly. In addition, the performance of our estimator is slightly better than that of BN and ADSS.

\begin{table}[htbp]
	\centering
	\begin{tabular}{c|cccc}
		\hline
		\hline
		& CY    & BN & ADSS  & ABDIK \\
		\hline
		average RMSE & 18.362 (0.431) & 19.692 (0.400) & 19.619 (0.432) & 25.522 (0.414) \\
		\hline
	\end{tabular}%
	\caption{Average RMSE: The values inside brackets are the standard errors.}
	\label{tab:RMSE_experiment}%
\end{table}%

\section{Concluding Remarks}\label{sec:conclusion}

This article develops an inference framework for the matrix completion when missing is not at random and without the need for strong signals. One of the key observations to our development is that if the number of missing entries is small enough compared to the size of the panel, they can be well estimated even if missing is not at random. We judicially divide the missing entries into smaller groups and use this observation to provide accurate estimates and efficient inferences. Moreover, we showed that our proposed estimate, even with fairly weak signals, is asymptotically normal with suitable debiasing. As an application, we studied the treatment effects in the tick size pilot program, an experiment conducted by the SEC to assess the impact of tick size extension on the market quality of small and illiquid stocks from 2016 to 2018. While previous studies on this program were based on traditional regression or difference-in-difference methods by assuming that the treatment effect is invariant with respect to time and unit, we observed significant heterogeneity in treatment effects and gained further insights about treatment effects in the pilot program using our estimation method. Lastly, we conducted simulation experiments to further demonstrate the practical merits of our methodology.


\newpage

\appendix

{\LARGE 
\begin{center}
    APPENDIX
\end{center}
}

\setcounter{table}{0}
\renewcommand{\thetable}{\thesection.\arabic{table}}

\renewcommand\thefigure{\thesection.\arabic{figure}}    
\setcounter{figure}{0}

\section{Estimation of submatrix where missing occurs only at one column}\label{sec:maintool}

We shall first present the statistical properties of our estimators when missing occurs only at one column, since the estimation in this case serves as the main tool for dealing with more general and common missing patterns. More specifically, we consider the estimation of an arbitrary $N_o \times T_o$ submatrix of $M$ that is constructed using the indices $\calI_o \subseteq [N]$ and $\calT_o \subseteq [T]$. Without loss of generality, assume that $\calI_o = \{1,\cdots, N_o\}$ and $\calT_o = \{1,\cdots, T_o\}$. The model we consider is the following:
$$
Y_o = M_o + \calE_o = X_o Z_o^\top   + \calE_o,
$$
$X_o = U_o D_o^{\frac{1}{2}}$ and $X_o = V_o D_o^{\frac{1}{2}}$ where $U_o D_o V_o^\top $ is the SVD of $M_o = (m_{it})_{i \in \calI_o, t \in \calT_o}$. Denote by $\Omega_o=(\omega_{it})_{i \in \calI_o, t \in \calT_o}$ and we treat it as a given one. Importantly, missing occurs only in the column $t_o \in \calT_o$: $\omega_{it} = 0$ if $i \in \calQ_o \subset \calI_o$ and $t=t_o$, $\omega_{it} = 1$ otherwise. Denote the number of missing entries by $|\calQ_o| = \vartheta_o$. In addition, we put the subscript `$o$' in all parameters regarding the submatrix $M_o$ to distinguish them from the parameters of the full matrix $M$.

\subsection{Definitions of estimators}\label{sec:notation}

Our proof follows a general strategy recently developed by \cite{chen2020nonconvex,chen:2019inference,chen2020noisy}: we first establish the statistical properties of a certain non-convex estimator and then show that it is close to the nuclear norm penalized estimator. There are two main reasons why this approach is more suitable for our purpose than the usual the restricted strong convexity (RSC) condition based techniques. See, e.g., \cite{negahban2012restricted,klopp2014noisy,athey2021matrix,hamdi2022low}. First, this approach is more amenable for deriving estimation error in max norm. Moreover, RSC based approach has difficulty in handling situations where the observation probabilities of some entries are deterministically zero. We shall show that even though the strategy was developed for missing at random, it can be used to deal with deterministic missing patterns and in particular when some entries are missing with probability one.

Recall that the nuclear norm penalized estimator is
$$\widetilde{M}_o \coloneqq \underset{A \in \bbR^{N_o \times T_o}}{\arg\min} \ \ \frac{1}{2} ||\Omega_o \circ (A - Y_o)||_F^2 + \lambda_o ||A||_*,
$$
and the corresponding debiased estimator is
$$
 \widehat{M}_o \coloneqq \calP_{r}\left[\calP_{\Omega_o^c}(\widetilde{M}_o) + \calP_{\Omega_o}(Y_o) \right].
$$
Here, $\calP_{\Omega_o}(B) = \Omega_o \circ B$, and $\calP_{\Omega_o^c}(B) = \Omega_o^c \circ B$ where $\Omega_o^c = \mathbf{11}^\top  - \Omega_o $. The estimators for $X_o$ and $Z_o$ are defined as $\widetilde{X}_o \coloneqq \widetilde{U}_o \widetilde{D}_o^{\frac{1}{2}}$ and $\widetilde{Z}_o \coloneqq \widetilde{V}_o \widetilde{D}_o^{\frac{1}{2}}$ where $\widetilde{U}_o \widetilde{D}_o \widetilde{V}^\top _o$ is the SVD of $\calP_{r}(\widetilde{M}_o)$. In addition, their corresponding debiased estimators are defined as
\begin{align*}
\widehat{X}_o \coloneqq \widetilde{X}_o \left( I_r + \lambda_o (\widetilde{X}_o^\top \widetilde{X}_o)^{-1} \right)^{\frac{1}{2}},\ \
\widehat{Z}_o \coloneqq \widetilde{Z}_o \left( I_r + \lambda_o (\widetilde{Z}_o^\top \widetilde{Z}_o)^{-1} \right)^{\frac{1}{2}}.
\end{align*}
These quantities will also be useful in defining the variance estimation later on.

We now introduce the non-convex estimators. We start with defining the following two loss functions, one for the typical non-convex estimator and the other for the leave-one-out estimator:
\begin{align}
	&\label{eq:nonconvex}  f(X,Z)   \coloneqq  \frac{1}{2}\|\calP_{\Omega_o}\left( XZ^{\top} - Y_o \right)\|_F^2 + \frac{\lambda_o}{2}\|X\|_F^2 + \frac{\lambda_o}{2}\|Z\|_F^2, \\
	& \nonumber f^{(m)}(X,Z) \\ 
 \nonumber&\coloneqq 
 \begin{cases}
		\frac{1}{2} \norm{ \mathcal{P}_{\Omega_{-m,\cdot}}(XZ^\top -Y_o)}_F^2 + \frac{1}{2} \norm{\mathcal{P}_{m,\cdot}(XZ^\top -M_o)}_F^2  + \frac{\lambda_o}{2}\norm{X}_F^2+\frac{\lambda_o}{2}\norm{Z}_F^2, \ \  \text{if $1 \leq m \leq N_o$,} \\
		\frac{1}{2} \norm{ \mathcal{P}_{\Omega_{\cdot, -(m-N_o)}}(XZ^\top -Y_o)}_F^2 + \frac{1}{2} \norm{\mathcal{P}_{\cdot, (m-N_o)}(XZ^\top -M_o)}_F^2  + \frac{\lambda_o}{2}\norm{X}_F^2+\frac{\lambda_o}{2}\norm{Z}_F^2,  \end{cases}\\
	& \label{eq:leaveoneout}\quad \quad \quad \quad \quad \quad \quad \quad \quad \quad \quad \quad \quad \quad \quad \quad \quad \quad \quad \quad \quad \quad \quad   \quad   \text{if $N_o+1 \leq m \leq N_o+T_o$,}
	\end{align}
where $X$ and $Z$ are $N_o \times r$ and $T_o \times r$ matrices, respectively. Here, for each $1 \leq m \leq N_o$, $\mathcal{P}_{\Omega_{-m,\cdot}}(B)\coloneqq \Omega_{-m,\cdot}\circ B$ where $\Omega_{-m,\cdot}\coloneqq(\omega_{js}1\{j \neq m\})_{j \leq N_o,s \leq T_o}$. Also, $\mathcal{P}_{m,\cdot}(B)\coloneqq E_{m,\cdot} \circ B$ where $E_{m,\cdot} \coloneqq (1\{j = m\})_{j \leq N_o,s \leq T_o}$. Note that the estimator constructed from the loss function $ f^{(m)}$ is independent of $\{\epsilon_{ms}\}_{s\leq T_o}$. Similarly, for each $N_o+1 \leq m \leq N_o+T_o$, we define $\mathcal{P}_{\Omega_{\cdot, -(m-N_o)}}(B)\coloneqq \Omega_{\cdot, -(m-N_o)}\circ B$ where $\Omega_{\cdot, -(m-N_o)}\coloneqq (\omega_{js}1\{s \neq m-N_o\})_{j \leq N_o,s \leq T_o}$, and $\mathcal{P}_{\cdot, (m-N_o)}(B)\coloneqq E_{\cdot, (m-N_o)} \circ B$ where $E_{\cdot, (m-N_o)} \coloneqq (1\{s=m-N_o\})_{j \leq N_o,s \leq T_o}$. In this case, the estimator is constructed from $ f^{(m)}$, which is independent of $\{ \epsilon_{j,(m-N_o)}\}_{j\leq N_o}$. Then, based on \eqref{eq:nonconvex}, we define the following gradient descent iterates:
\begin{align}\label{alg:nonconvex}
	\begin{bmatrix}
		X_o^{\tau+1} \\ Z_o^{\tau+1}
	\end{bmatrix}
	=\begin{bmatrix}
		X_o^{\tau} - \eta_o \nabla_X f(X_o^{\tau}, Z_o^{\tau}) \\ Z_o^{\tau} - \eta_o \nabla_Z f(X_o^{\tau}, Z_o^{\tau})
	\end{bmatrix}
\end{align}
where $X_o^0= X_o$, $Z_o^0=Z_o$, $\tau=0,1, \ldots, \overbar{\tau}-1$, and $\overbar{\tau}=\max\{N_o^{23}, T_o^{23}\}$. Here, $\eta_o>0$ is the step size. Similarly, for \eqref{eq:leaveoneout}, we define
\begin{align}\label{alg:loo}
	\begin{bmatrix}
		X_o^{\tau+1, (m)} \\
		Z_o^{\tau+1, (m)}
	\end{bmatrix}
	= \begin{bmatrix}
		X_o^{\tau,(m)} - \eta_o \nabla_X f^{(m)}(X_o^{\tau,(m)}, Z_o^{\tau,(m)}) \\
		Z_o^{\tau,(m)} - \eta_o \nabla_Z f^{(m)}(X_o^{\tau,(m)}, Z_o^{\tau,(m)})
	\end{bmatrix}
\end{align}
where $X_o^{0,(m)}=X_o$, $Z_o^{0,(m)}=Z_o$. Note that the gradient descent iterates in \eqref{alg:nonconvex} and \eqref{alg:loo} are not computable because the initial value $(X_o,Z_o)$ is unknown. However, it does not cause any problems in the paper since we do not need to actually compute $X_o^{\tau}$, $Z_o^{\tau}$, $X_o^{\tau, (m)}$, and $Z_o^{\tau, (m)}$ and only use their existence and theoretical properties for the proof. In addition, we define the corresponding debiased iterates:
\begin{align*}
&X_o^{d,\tau} \coloneqq X_o^{\tau} \left( I_r + \lambda_o (X_o^{\tau\top} X_o^{\tau})^{-1} \right)^{\frac{1}{2}}, \ \
Z_o^{d,\tau} \coloneqq Z_o^{\tau} \left( I_r + \lambda_o (Z_o^{\tau\top} Z_o^{\tau})^{-1} \right)^{\frac{1}{2}},\\
&X_o^{d,\tau,(m)} \coloneqq X_o^{\tau,(m)} \left( I_r + \lambda_o (X_o^{\tau,(m)\top} X_o^{\tau,(m)})^{-1} \right)^{\frac{1}{2}}, \ \
Z_o^{d,\tau,(m)} \coloneqq Z_o^{\tau,(m)} \left( I_r + \lambda_o (Z_o^{\tau,(m)\top} Z_o^{\tau,(m)})^{-1} \right)^{\frac{1}{2}}.
\end{align*}
Moreover, we define corresponding rotation matrices:
\begin{align*}
	&H_o^{\tau} \coloneqq \argmin_{R\in \calO^{r \times r}}\norm{
		\calF_o^{\tau} 
		R - 
		\calF_o}_F ,  \ \ 
	H_o^{\tau, (m)} \coloneqq \argmin_{R \in \mathcal{O}^{r \times r}}\norm{
		\calF_o^{\tau, (m)}
		R - 
		\calF_o}_F, \\
	&Q_o^{\tau, (m)} \coloneqq \argmin_{R \in \mathcal{O}^{r \times r}}\norm{
		\calF_o^{\tau, (m)}
		R - 
		\calF_o^{\tau} H_o^{\tau}}_F,\ \ 
	H_o^{d,\tau} \coloneqq \argmin_{R\in \calO^{r \times r}}\norm{
		\calF_o^{d,\tau} 
		R - 
		\calF_o}_F ,  \\
	&H_o^{d, \tau, (m)} \coloneqq \argmin_{R \in \mathcal{O}^{r \times r}}\norm{
		\calF_o^{d, \tau, (m)}
		R - 
		\calF_o}_F, \text{  where  }\\
	& \calF_o^{\tau} \coloneqq 
	\begin{bmatrix}
		X_o^{\tau} \\
		Z_o^{\tau}
	\end{bmatrix},\ \
	\calF_o^{\tau, (m)}\coloneqq 
	\begin{bmatrix}
		X_o^{\tau,(m)} \\
		Z_o^{\tau,(m)}
	\end{bmatrix}, \ \
		\calF_o^{d,\tau} \coloneqq 
	\begin{bmatrix}
		X_o^{d,\tau} \\
		Z_o^{d,\tau}
	\end{bmatrix},\ \
	\calF_o^{d,\tau, (m)}\coloneqq 
	\begin{bmatrix}
		X_o^{d,\tau,(m)} \\
		Z_o^{d,\tau,(m)}
	\end{bmatrix}, \ \
	\calF_o \coloneqq 
	\begin{bmatrix}
		X_o \\
		Z_o
	\end{bmatrix},
\end{align*}
and $\calO^{r \times r}$ is the set of $r \times r$ orthogonal matrix. 

Finally, we define the non-convex estimators using the gradient descent iterates. Let 
$$
\tau_o^* \coloneqq \argmin_{0\leq \tau < \overbar{\tau}}\norm{\nabla f(X_o^{\tau},Z_o^{\tau})}_F.
$$
Then, the non-convex estimators are defined as:
$$
( \breve{X}_o, \breve{Z}_o ) \coloneqq (X_o^{\tau_o^*},Z_o^{\tau_o^*}) \quad \text{from \eqref{alg:nonconvex}}, \quad
( \breve{X}_o^{(m)},\breve{Z}_o^{(m)}) \coloneqq (X_o^{\tau_o^*,(m)},Z_o^{\tau_o^*,(m)}) \quad \text{from \eqref{alg:loo}},
$$
and the corresponding debiased estimators are defined as:
$$
( \breve{X}_o^d, \breve{Z}_o^d ) \coloneqq (X_o^{d,\tau_o^*},Z_o^{d,\tau_o^*}), \quad
( \breve{X}_o^{d,(m)},\breve{Z}_o^{d,(m)}) \coloneqq (X_o^{d,\tau_o^*,(m)},Z_o^{d,\tau_o^*,(m)}),
$$
with the corresponding rotation matrices $\breve{H}_o \coloneqq H_o^{\tau_o^{*}}$, $\breve{H}_o^{(m)}\coloneqq H_o^{\tau_o^{*},(m)}$, $\breve{H}_o^{d} \coloneqq H_o^{d,\tau_o^{*}}$, and $\breve{H}_o^{d,(m)}\coloneqq H_o^{d,\tau_o^{*},(m)}$. Lastly, we define the rotation matrix for $(\widehat{X}_o,\widehat{Z}_o)$ as $\widehat{H}_o = B_o \breve{H}_o^d$ where $B_o = \argmin_{R \in \calO^{r \times r}} ||\widehat{X}_o R  - \breve{X}^d_o||_F^2 + ||\widehat{Z}_o R  - \breve{Z}^d_o||_F^2$.

\subsection{Key propositions for inferential theory}\label{sec:propositionforinference}

This subsection provides several key propositions for developing the inferential theory of our debiased estimator $\widehat{M}_o$. First, we derive a suitable decomposition for the asymptotic normality of the debiased estimator $(\widehat{X}_o,\widehat{Z}_o)$ (Propositions \ref{pro:xclt} and \ref{pro:zclt}). By using the proximity between $\widehat{M}_o$ and $\widehat{X}_o\widehat{Z}_o^\top $ (Proposition \ref{pro:proximity}) with this decomposition, we derive a decomposition of $\widehat{m}_{o,it} - m_{o,it}$, which is used to show the asymptotic normality of $\widehat{m}_{o,it}$ (Proposition \ref{pro:decomposition}). We begin by introducing several assumptions.

\begin{assumption}[Noise]\label{asp:apdx_error}
	\textit{$\epsilon_{it}$ is i.i.d. zero mean sub-Gaussian random variable such that $\mathbb{E}[\epsilon_{it}] = 0$, $\mathbb{E}[\epsilon_{it}^2 ] = \sigma^2$, $ \mathbb{E} [\exp(s \epsilon_{it})] \leq \exp(C s^2 \sigma^2)$, $\forall s \in \mathbb{R}$, for some constant $C>0$.}
\end{assumption}

\begin{assumption}[Incoherence]\label{asp:apdx_uniformincoherence}
	\textit{There is $\mu_o\geq 1$ such that $||U_{M_o}||_{2,\infty} \leq \sqrt{\frac{\mu_o r}{N_o}}$, $||V_{M_o}||_{2,\infty} \leq \sqrt{\frac{\mu_o r}{T_o}}$. Here, $U_A$ and $V_A$ denote the left and right singular vector of $A$, respectively.}
\end{assumption}

\begin{assumption}[Signal to noise ratio]\label{asp:apdx_signaltonoise}
		\textit{
		$$      
		\sigma  \kappa_o^2 \mu_o^{\frac{1}{2}} r^{\frac{1}{2}} \max\{{N_o\sqrt{\log{N_o}}},{T_o\sqrt{\log{T_o}}}\} \ll 
		\psi_{\min,o} \min\{\sqrt{N_o},\sqrt{T_o}\},
		$$
where $\psi_{\min,o}$ is the smallest nonzero singular value of $M_o$.		}
\end{assumption}

\begin{assumption}[Size of $\vartheta_o$ and parameters]\label{asp:apdx_groupandparameterssize}
		\textit{(i) $\kappa_o^4 \mu_o^{2} r^{2}  \max\{{N_o\log^3{N_o}},{T_o\log^3{T_o}}\}  \ll 
		\min\{N_o^{2},T_o^{2}\}$ and (ii) $\vartheta_o \kappa_o^2 \mu_o r   \ll \min\{N_o,T_o\}$.
		}
\end{assumption}

Denote by $\Omega_{o,i}$ the diagonal matrix consisting of $\{\omega_{is}\}_{1\leq s \leq T_o} $ and by $\Omega_{o,t}$ the diagonal matrix consisting of $\{\omega_{jt}\}_{1 \leq j \leq N_o} $.

\begin{proposition}\label{pro:xclt}
Suppose that Assumptions \ref{asp:apdx_error} - \ref{asp:apdx_groupandparameterssize} hold. Then, with probability at least $1-O(\min\{N_o^{-9},T_o^{-9}\})$, we have for all $1 \leq i \leq N_o$,
$$
e_i^\top (\widehat{X}_o \widehat{H}_o - X_o ) = e_i^\top \calP_{\Omega_o}( \calE_o)Z_o(Z_o^\top \Omega_{o,i} Z_o)^{-1} + \calR^X_{o,i},
$$
where
\begin{align*}
&  \max_i ||\calR^X_{o,i}|| \\
&\leq C_X \frac{\sigma}{\sqrt{\psi_{\min,o}}}\left( \frac{\sigma}{\psi_{\min,o}} \sqrt{\frac{\kappa_o^9 \mu_o r \max\{N_o^2 \log N_o,T_o^2 \log T_o \}}{\min\{N_o , T_o \}}} + \sqrt{\frac{\kappa_o^7 \mu_o^3 r^3 \max \{N_o^2 \log N_o, T_o^2 \log T_o \}}{ N_o\min \{N_o^2, T_o^2 \}}} \right)
\end{align*}
for an absolute constant $C_X > 0$.
\end{proposition}

\begin{proposition}\label{pro:zclt}
Suppose that Assumptions \ref{asp:apdx_error} - \ref{asp:apdx_groupandparameterssize} hold. Then, with probability at least $1-O(\min\{N_o^{-9},T_o^{-9}\})$, we have for all $1 \leq t \leq T_o$,
$$
e_t^\top (\widehat{Z}_o \widehat{H}_o - Z_o ) = e_t^\top \calP_{\Omega_o}( \calE_o)^\top X_o(X_o^\top \Omega_{o,t} X_o)^{-1} + \calR^Z_{o,t},
$$
where
\begin{align*}
&\max_t ||\calR^Z_{o,t}|| \\
&
\leq C_Z \frac{\sigma}{\sqrt{\psi_{\min,o}}}\left(
\frac{\sigma}{\psi_{\min,o}} \sqrt{\frac{\kappa_o^9 \mu_o r \max\{N_o^2 \log N_o,T_o^2 \log T_o \}}{\min\{N_o , T_o \}}} +
 \sqrt{\frac{\kappa_o^7 \mu_o^3 r^3 \max \{N_o^2 \log N_o, T_o^2 \log T_o \}}{ T_o \min \{N_o^2, T_o^2 \}}} \right.\\
&\left. \qquad\qquad\qquad \qquad + \vartheta_o  \sqrt{\frac{\mu_o^3 r^3 \kappa_o^5 \max \{N_o \log N_o, T_o \log T_o \}}{N_o \min \{N_o^2, T_o^2 \}}}
\right)
\end{align*}
for an absolute constant $C_Z > 0$.
\end{proposition}

\begin{proposition}\label{pro:proximity}
Suppose that Assumptions \ref{asp:apdx_error} - \ref{asp:apdx_groupandparameterssize} hold. With probability at least $1-O(\min\{N_o^{-10},T_o^{-10}\})$, we have 
$$
\norm{\widehat{M}_o - \widehat{X}_o\widehat{Z}_o^\top }_F \leq C_{prx} \frac{\sigma}{\max \{N_o^{7/2} , T_o^{7/2} \}}
$$
for an absolute constant $C_{prx} > 0$.
\end{proposition}

\begin{proposition}\label{pro:decomposition}
Suppose that Assumptions \ref{asp:apdx_error} - \ref{asp:apdx_groupandparameterssize} hold. With probability at least $1-O(\min\{N_o^{-9},T_o^{-9}\})$, we have 
\begin{align*}
&\widehat{m}_{o,it_o} - m_{{o,it_o}} \\
&=  X_{o,i}^\top  \left( \sum_{j \in \calI_o } \omega_{j t_o} X_{o,j} X_{o,j}^\top  \right)^{-1}  \sum_{j \in \calI_o } \omega_{j t_o} \epsilon_{jt_o} X_{o,j}  +  Z_{o,t_o}^\top  \left( \sum_{s \in \calT_o } \omega_{is} Z_{o,s} Z_{o,s}^\top  \right)^{-1}  \sum_{s \in \calT_o } \omega_{is} \epsilon_{is} Z_{o,s} + \calR_{o,i}^{M},
\end{align*}
where
\begin{align*}
\max_i||\calR_{o,i}^{M}|| &\leq C_M 
\left(
\frac{\sigma^2}{\psi_{\min,o}} \frac{\kappa_o^5 \mu_o r \max\{N_o \log N_o,T_o \log T_o \}}{\min\{N_o , T_o \}} +
 \sigma \frac{\kappa_o^4 \mu_o^2 r^2 \max \{ \sqrt{N_o \log N_o}, \sqrt{T_o \log T_o} \}}{ \min \{N_o^{\frac{3}{2}}, T_o^{\frac{3}{2}} \}} \right.\\
&\qquad \qquad \left. + \sigma \frac{\vartheta_o \mu_o^2 r^2 \kappa_o^3 \max \{\sqrt{N_o \log N_o}, \sqrt{T_o \log T_o} \}}{ N_o \min \{N_o, T_o \}},
\right)
\end{align*}
$C_M > 0$ is an absolute constant.
\end{proposition}

\subsection{Proofs of Propositions \ref{pro:xclt}-\ref{pro:decomposition}}

\begin{proof}[Proof of Proposition \ref{pro:xclt}]
We first derive a decomposition of $e_i^\top (\breve{X}^d_o \breve{H}^d_o - X_o )$. From the definition of the gradient $\nabla_X f(\breve{X}_o,\breve{Z}_o) = \calP_{\Omega_o}(\breve{X}_o\breve{Z}_o^\top  - Y_o)\breve{Z}_o + \lambda_o \breve{X}_o$ with the decomposition $$\calP_{\Omega_o}(\breve{X}_o\breve{Z}_o^\top  - Y_o) = \breve{X}_o\breve{Z}_o^\top  - X_oZ_o^\top  + A - \calP_{\Omega_o}(\calE_o),$$ where $A \coloneqq \Omega_o \circ (\breve{X}_o\breve{Z}_o^\top  - X_o Z_o^\top ) - (\breve{X}_o\breve{Z}_o^\top  - X_o Z_o^\top )$, we have
$$
\breve{X}_o \left( \breve{Z}_o^\top  \breve{Z}_o + \lambda_o I_r \right) = X_oZ_o^\top \breve{Z}_o + \calP_{\Omega_o}(\calE_o)\breve{Z}_o  - A\breve{Z}_o + \nabla_X f(\breve{X}_o,\breve{Z}_o).
$$
In addition, a simple calculation shows that $ \breve{Z}^{d \top}_o \breve{Z}^{d}_o = \breve{Z}_o^\top  \breve{Z}_o + \lambda_o I_r$. Then, by combining these two equations, we have 
$$
\breve{X}_o \breve{Z}^{d \top}_o \breve{Z}^{d}_o  = X_oZ_o^\top \breve{Z}_o + \calP_{\Omega_o}(\calE_o)\breve{Z}_o  - A\breve{Z}_o + \nabla_X f(\breve{X}_o,\breve{Z}_o).
$$
Multiplying both sides by $(I_r + \lambda_o (\breve{Z}_o^\top  \breve{Z}_o )^{-1})^{1/2}$, we have
$$
\breve{X}_o \breve{Z}^{d \top}_o \breve{Z}^{d}_o (I_r + \lambda_o (\breve{Z}_o^\top  \breve{Z}_o )^{-1})^{1/2} = X_oZ_o^\top \breve{Z}^{(d)}_o + \calP_{\Omega_o}(\calE_o)\breve{Z}^{(d)}_o  - A\breve{Z}^{(d)}_o + \nabla_X f(\breve{X}_o,\breve{Z}_o)(I_r + \lambda_o (\breve{Z}_o^\top  \breve{Z}_o )^{-1})^{1/2}.
$$
Moreover, because the left hand side can be also represented as
\begin{align*}
\breve{X}_o \breve{Z}^{d \top}_o \breve{Z}^{d}_o (I_r + \lambda_o (\breve{Z}_o^\top  \breve{Z}_o )^{-1})^{1/2} 
&= \breve{X}_o(I_r + \lambda_o (\breve{Z}_o^\top  \breve{Z}_o )^{-1})^{1/2}\breve{Z}^{d \top}_o \breve{Z}^{d}_o \\
&= \breve{X}^d_o (\breve{Z}^{d\top}_o\breve{Z}^d_o) -  \breve{X}_o \Delta_{balance}\breve{Z}^{d \top}_o \breve{Z}^{d}_o,
\end{align*}
where $\Delta_{balance} \coloneqq (I_r + \lambda_o (\breve{X}^{\top}_o \breve{X}_o)^{-1})^{\frac{1}{2}} - (I_r + \lambda_o (\breve{Z}^{\top}_o \breve{Z}_o)^{-1})^{\frac{1}{2}}$, we have
\begin{align*}
 \breve{X}^d_o & =  \calP_{\Omega_o}( \calE_o) \breve{Z}^d_o (\breve{Z}^{d\top}_o\breve{Z}^d_o)^{-1}  + X_o Z_o^\top  \breve{Z}^d_o (\breve{Z}^{d\top}_o\breve{Z}^d_o)^{-1} - A \breve{Z}^d_o (\breve{Z}^{d\top}_o\breve{Z}^d_o)^{-1} \\
 & \ \ + \nabla_X f(\breve{X}_o,\breve{Z}_o) \left( I_r + \lambda_o (\breve{Z}^{\top}_o\breve{Z}_o)^{-1} \right)^{1/2} (\breve{Z}^{d\top}_o\breve{Z}^d_o)^{-1} + \breve{X}_o \Delta_{balance},
\end{align*}
by multiplying $(\breve{Z}^{d\top}_o\breve{Z}^d_o)^{-1}$. Then, using the identity
$\breve{Z}^d_o (\breve{Z}^{d\top}_o\breve{Z}^d_o)^{-1} \breve{H}^d_o = \bar{\breve{Z}}^d_o (\bar{\breve{Z}}^{d\top}_o\bar{\breve{Z}}^d_o)^{-1}$ where $\bar{\breve{Z}}^d_o = \breve{Z}^d_o\breve{H}^d_o$, we have the following decomposition:
\begin{align*}
 & e_i^\top (\breve{X}^d_o \breve{H}^d_o - X_o ) = e_i^\top \calP_{\Omega_o}( \calE_o)Z_o(Z_o^\top \Omega_{o,i} Z_o)^{-1} + \sum_{k=1}^5 \delta_{k,i}, \\
& \delta_{1,i} = e_i^\top \calP_{\Omega_o}( \calE_o) (\bar{\breve{Z}}_o (\bar{\breve{Z}}^{\top}_o\bar{\breve{Z}}_o)^{-1} - Z_o(Z_o^\top  Z_o)^{-1}),\\
&\delta_{2,i} = e_i^\top \calP_{\Omega_o}( \calE_o) \left( Z_o(Z_o^\top  Z_o)^{-1} - Z_o(Z_o^\top \Omega_{o,i} Z_o)^{-1} \right),\\
& \delta_{3,i} = e_i^\top X_o[Z_o^\top \bar{\breve{Z}}^d_o (\bar{\breve{Z}}^{d\top}_o\bar{\breve{Z}}^d_o)^{-1} - I_r],\\
&\delta_{4,i} = e_i^\top A \bar{\breve{Z}}^d_o (\bar{\breve{Z}}^{d\top}_o\bar{\breve{Z}}^d_o)^{-1},\\
& \delta_{5,i} = e_i^\top \nabla_X f(\breve{X}_o,\breve{Z}_o) \left( I_r + \lambda_o (\breve{Z}^{\top}_o\breve{Z}_o)^{-1} \right)^{1/2} (\breve{Z}^{d\top}_o\breve{Z}^d_o)^{-1} \breve{H}^d_o 
+ e_i^\top \breve{X}_o \Delta_{balance} \breve{H}^d_o.
\end{align*}
Furthermore, by defining $\delta_{6,i} = e_i^\top (\widehat{X}_o B_o  - \breve{X}^d_o )\breve{H}^d_o$ where
$$B_o = \argmin_{R \in \calO^{r \times r}} ||\widehat{X}_o R  - \breve{X}^d_o||_F^2 + ||\widehat{Z}_o R  - \breve{Z}^d_o||_F^2,$$ 
we have the following decomposition for $e_i^\top (\widehat{X}_o \widehat{H}_o- X_o )$:
$$
e_i^\top (\widehat{X}_o \widehat{H}_o- X_o ) = e_i^\top \calP_{\Omega_o}( \calE_o)Z_o(Z_o^\top \Omega_{o,i} Z_o)^{-1} + \sum_{k=1}^6 \delta_{k,i}
$$
where $\widehat{H}_o = B_o \breve{H}^d_o$.

\paragraph{Part 1.} First, bound the part $\delta_{1,i}$. By defining $\bar{\breve{Z}}^{d,(i)}_o = \breve{Z}^{d,(i)}_o\breve{H}^{d,(i)}_o$, we have
\begin{align*}
&||\delta_{1,i}||_2 \leq \norm{e_i^\top  \calP_{\Omega_o}(\calE_o)\left[ \bar{\breve{Z}}^{d,(i)}_o \left( \bar{\breve{Z}}^{d,(i)\top}_o \bar{\breve{Z}}^{d,(i)}_o \right)^{-1} - Z_o \left( Z_o^{\top} Z_o \right)^{-1} \right] }_2 \\
&\qquad \qquad  + \norm{e_i^\top  \calP_{\Omega_o}(\calE_o)\left[ \bar{\breve{Z}}^{d}_o \left( \bar{\breve{Z}}^{d\top}_o \bar{\breve{Z}}^{d}_o \right)^{-1} - \bar{\breve{Z}}^{d,(i)}_o \left( \bar{\breve{Z}}^{d,(i)\top}_o \bar{\breve{Z}}^{d,(i)}_o \right)^{-1}\right] }_2.
\end{align*}
The first part is bounded in Lemma \ref{lem:leaveoneoutindependence}. For the second part, note that
\begin{align*}
  &\norm{e_i^\top  \calP_{\Omega_o}(\calE_o)\left[ \bar{\breve{Z}}^{d}_o \left( \bar{\breve{Z}}^{d\top}_o \bar{\breve{Z}}^{d}_o \right)^{-1} - \bar{\breve{Z}}^{d,(i)}_o \left( \bar{\breve{Z}}^{d,(i)\top}_o \bar{\breve{Z}}^{d,(i)}_o \right)^{-1}\right] }_2    \\
  &\ \ \leq \norm{\calP_{\Omega_o}(\calE_o)} \norm{ \bar{\breve{Z}}^{d}_o \left( \bar{\breve{Z}}^{d\top}_o \bar{\breve{Z}}^{d}_o \right)^{-1} - \bar{\breve{Z}}^{d,(i)}_o \left( \bar{\breve{Z}}^{d,(i)\top}_o \bar{\breve{Z}}^{d,(i)}_o \right)^{-1}}\\
   &\ \ \lesssim \sigma \sqrt{\max\{N_o,T_o\}} \frac{1}{\psi_{\min,o}} \norm{\bar{\breve{Z}}^{d}_o - \bar{\breve{Z}}^{d,(i)}_o}\\
  &\ \ \lesssim \sigma \sqrt{\max\{N_o,T_o\}} \frac{1}{\psi_{\min,o}} \kappa_o \frac{\sigma \sqrt{\max\{N_o \log N_o,T_o \log T_o\}}}{\psi_{\min,o}} \norm{\calF_o}_{2, \infty}
\end{align*}
by Lemmas \ref{lem:deshrunken} and \ref{LemmaA3}. Hence, we have with probability at least $1-O(\min\{N_o^{-9},T_o^{-9}\})$,
$$
\max_i||\delta_{1,i}||_2 \leq C_{\delta,1} \frac{\sigma}{\sqrt{\psi_{\min,o}}} \frac{\sigma}{\psi_{\min,o}} \sqrt{\frac{\kappa_o^3 \mu_o r \max\{N_o^2 \log N_o,T_o^2 \log T_o \}}{\min\{N_o , T_o \}}}
$$
for some absolute constant $C_{\delta,1} > 0 $.

\paragraph{Part 2.} Note that $$\delta_{2,i} = \sum_{s=1}^{T_o} \omega_{is} \epsilon_{is} Z_{o,s} \left( (Z_o^\top  Z_o)^{-1} - (Z_o^\top \Omega_{o,i} Z_o)^{-1} \right).$$ Because 
$$\norm{Z_o^\top  Z_o - Z_o^\top \Omega_{o,i} Z_o } = ||Z_{o,t_o}Z_{o,t_o}^\top || \leq \frac{\kappa_o \mu_o r}{T_o} \psi_{\min,o}$$ and $ ||(Z_o^\top  Z_o)^{-1}|| = \psi_{\min,o}^{-1},$ 
we have
$$
\norm{(Z_o^\top  Z_o)^{-1} - (Z_o^\top \Omega_{o,i} Z_o)^{-1}} \lesssim \norm{Z_o^\top  Z_o - Z_o^\top \Omega_{o,i} Z_o } ||(Z_o^\top  Z_o)^{-1}||^2 \leq \frac{\kappa_o \mu_o r}{T_o}  \psi_{\min,o}^{-1}.
$$
In addition, by the matrix Berstein inequality, we have
$$\norm{\sum_{s=1}^{T_o} \omega_{is} \epsilon_{is} Z_{o,s}} \lesssim \sigma \sqrt{\log T_o} ||Z_{o}||_F \lesssim \sigma \sqrt{\log T_o} \kappa_o^{\frac{1}{2}} r^{\frac{1}{2}} \psi_{\min,o}^{\frac{1}{2}} $$ with probability at least $1-O(\min\{N_o^{-100},T_o^{-100}\})$. So, we have with probability at least $1-O(\min\{N_o^{-9},T_o^{-9}\})$,
$$
\max_i||\delta_{2,i}||_2 \leq C_{\delta,2} 
\frac{\sigma}{\sqrt{ \psi_{\min,o}}} \frac{\kappa_o^{\frac{3}{2}} \mu_o r^{\frac{3}{2}}\sqrt{\log T_o}}{T_o}.
$$

\paragraph{Part 3.} Note that $$\norm{e_i^\top X_o}_2 \leq \sqrt{\frac{\kappa_o \mu_o r}{N_o}\psi_{\min,o}}$$ by the incoherence condition. By Lemma \ref{lem:deshrunken} and the fact that $\norm{(\bar{\breve{Z}}^{d\top}_o \bar{\breve{Z}}^d_o)^{-1}} \lesssim \psi_{\min,o}^{-1}$, we have
\begin{align*}
||\delta_{3,i}||_2 
&= \norm{e_i^\top X_o[Z_o^\top \bar{\breve{Z}}^d_o (\bar{\breve{Z}}^{d\top}_o\bar{\breve{Z}}^d_o)^{-1} - \bar{\breve{Z}}^{d\top}_o \bar{\breve{Z}}^d_o (\bar{\breve{Z}}^{d\top}_o \bar{\breve{Z}}^d_o)^{-1}]}_2 \\
&\leq \norm{e_i^\top X_o}_2 \norm{(Z_o - \bar{\breve{Z}}^d_o)^\top  \bar{\breve{Z}}^d_o} \norm{(\bar{\breve{Z}}^{d\top}_o\bar{\breve{Z}}^d_o)^{-1}}\\
& \lesssim \sqrt{\frac{\kappa_o \mu r}{N_o}}\frac{1}{\sqrt{\psi_{\min,o}}}\norm{(Z_o - \bar{\breve{Z}}^d_o)^\top  \bar{\breve{Z}}^d_o}.
\end{align*}
Next, we bound $\norm{(Z_o - \bar{\breve{Z}}^d_o)^\top  \bar{\breve{Z}}^d_o}$. Let $\Delta_X \coloneqq \bar{\breve{X}}^{d}_o - X_o$ and $\Delta_Z \coloneqq \bar{\breve{Z}}^{d}_o - Z_o$. Then, $(Z_o - \bar{\breve{Z}}^d_o)^\top  \bar{\breve{Z}}^d_o  = \Delta_Z^\top Z_o + \Delta_Z^\top \Delta_Z$. Following the proof of Lemma 6 in \cite{chen:2019inference}, we can reach
\begin{align*}
&\norm{(Z_o - \bar{\breve{Z}}^d_o)^\top  \bar{\breve{Z}}^d_o}\\
& \leq \norm{\Delta_Z^\top Z_o} + \norm{\Delta_Z^\top \Delta_Z} \\
&\lesssim \frac{1}{\psi_{\min,o}}\underbrace{\norm{\bar{\breve{X}}^{d \top}_o \calP_{\Omega_o}(\calE_o)Z_o }}_{=\alpha_1} + \frac{1}{\psi_{\min,o}}\underbrace{\norm{\bar{\breve{X}}^{d \top}_o A Z_o }}_{=\alpha_2} + \kappa_o \underbrace{\left( \norm{\Delta_X^\top  \Delta_X} + \norm{\Delta_Z^\top  \Delta_Z} \right)}_{=\alpha_3}\\
& \ \ + \frac{1}{\psi_{\min,o}} \underbrace{\norm{ \breve{H}_o^{d \top} \left( I_r + \lambda_o (\breve{X}_o^\top \breve{X}_o)^{-1} \right)^{1/2} [\nabla_Z f(\breve{X},\breve{Z})]^\top  Z_o - \bar{\breve{X}}^{d \top}_o\bar{\breve{X}}^{d }_o\breve{H}_o^{d \top} \Delta_{balance} \breve{Z}_o^\top  Z_o + \Delta_{XZ}^d D_o}}_{=\alpha_4},
\end{align*}
where $\Delta_{XZ}^d \coloneqq \frac{1}{2} \breve{H}_o^{d \top} (\breve{Z}^{d\top}_o \breve{Z}^{d}_o - \breve{X}^{d\top}_o \breve{X}^{d}_o)\breve{H}_o^{d}$. First, we bound $\alpha_1$. Note that $$\alpha_1 \leq \norm{X_o^\top  \calP_{\Omega_o}(\calE_o)Z_o} + \norm{\Delta_X^\top  \calP_{\Omega_o}(\calE_o)Z_o}.$$ By the Bernstein inequality, we have
\begin{align*}
&\norm{X_o^\top  \calP_{\Omega_o}(\calE_o)Z_o}
= \norm{\sum_{i\in\calI_o t\in \calT_o}\omega_{it}\epsilon_{it}X_{o,i} Z_{o,t}} \lesssim \sigma r \kappa_o \psi_{\min,o} \sqrt{\max \{\log N_o, \log T_o \}}.
\end{align*}
In addition, we have by Lemmas \ref{lem:deshrunken} and \ref{LemmaA3} that $\norm{\Delta_X^\top  \calP_{\Omega_o}(\calE_o)Z_o} \leq \sigma^2 \kappa_o^2 \max \{ N_o,  T_o \}$. Hence, we have
$$
\alpha_1 \lesssim \sigma r \kappa_o \psi_{\min,o} \sqrt{\max \{\log N_o, \log T_o \}} + \sigma^2 \kappa_o^2 \max \{ N_o,  T_o \}.
$$
Moreover, since $$||A|| \lesssim \sigma \sqrt{\max\{N_o,T_o\}} \sqrt{\frac{ \kappa_o^4 \mu_o^2 r^2 \max\{N_o \log N_o,T_o \log T_o\}}{ \min\{N_o^2,T_o^2\}}}$$ by Lemma \ref{LemmaA6}, we have
$$
\alpha_2 \lesssim \sigma \sqrt{\max\{N_o,T_o\}} \sqrt{\frac{ \kappa_o^6 \mu_o^2 r^2 \max\{N_o \log N_o,T_o \log T_o\}}{ \min\{N_o^2,T_o^2\}}}\psi_{\min,o}.
$$
By Lemma \ref{lem:deshrunken}, we know $$\alpha_3 \lesssim \max\{||\Delta_X||^2, ||\Delta_Z||^2 \} \leq \sigma^2\frac{\kappa_o^3 \max\{N_o, T_o \}}{\psi_{\min,o}}.$$ Lastly, the term $\alpha_4$ is bounded like
\begin{align*}
\alpha_4 &\leq \norm{ \left( I_r + \lambda_o (\breve{X}_o^\top \breve{X}_o)^{-1} \right)^{1/2}} \norm{\nabla_Z f(\breve{X}_o,\breve{Z}_o)} \norm{Z_o} - \norm{\bar{\breve{X}}^{d \top}_o\bar{\breve{X}}^{d }_o} \norm{ \Delta_{balance}} \norm{\breve{Z}_o^\top  Z_o} + \norm{\Delta_{XZ}^d } \norm{D_o}\\ 
&\lesssim \sigma \frac{\kappa_o^2}{\max \{N_o^{\frac{9}{2}}, T_o^{\frac{9}{2}} \}}  \psi_{\min,o},
\end{align*}
due to Lemmas \ref{lem:deshrunken} and \ref{lem:smallgradient}, and the relation \eqref{eq:balance}. Therefore, we have
\begin{align*}
\max_i||\delta_{3,i}||_2  
&\lesssim \sqrt{\frac{\kappa_o \mu_o r}{N_o}}\frac{1}{\sqrt{\psi_{\min,o}}}\norm{(Z_o - \bar{\breve{Z}}^d_o)^\top  \bar{\breve{Z}}^d_o}\\
&\lesssim \frac{\sigma}{\sqrt{\psi_{\min,o}}} \left( \kappa_o \frac{\sigma}{\psi_{\min,o}} \sqrt{\frac{\kappa_o^7 \mu_o r \max \{N_o^2, T_o^2 \}}{N_o}} + \sqrt{\frac{\kappa_o^7 \mu_o^3 r^3 \max \{N_o^2 \log N_o, T_o^2 \log T_o \}}{N_o \min \{N_o^2, T_o^2 \}}}
\right).
\end{align*}
\paragraph{Part 4.} Note that
$$
\norm{\delta_{4,i}}_2 = \norm{e_i^\top A \bar{\breve{Z}}^d_o (\bar{\breve{Z}}^{d\top}_o\bar{\breve{Z}}^d_o)^{-1}}_2 \leq \norm{e_i^\top A \bar{\breve{Z}}^d_o}_2 \norm{(\bar{\breve{Z}}^{d\top}_o\bar{\breve{Z}}^d_o)^{-1}}
\lesssim \frac{1}{\psi_{\min,o}}\norm{e_i^\top A \bar{\breve{Z}}^d_o}_2.
$$
Let $\nu = [\nu_1 , \dots , \nu_{T_o}] \coloneqq e_i^\top  (\breve{X}_o\breve{Z}_o^\top  - X_o Z_o^\top )$. Then, we have by Lemma \ref{lem:CCFMY_noncovex}
\begin{align*}
\norm{\nu}_{\infty} 
&= \norm{\breve{X}_o\breve{Z}_o^\top  - X_o Z_o^\top }_{\infty}\\
&\leq \norm{\breve{X}_o\breve{H}_o - X_o}_{2,\infty}\norm{\breve{Z}_o}_{2,\infty} + 
\norm{X_o}_{2,\infty} \norm{\breve{Z}_o\breve{H}_o - Z_o}_{2,\infty}\\
 &\lesssim \sigma  \kappa_o^2 \sqrt{\frac{\mu_o^2 r^2 \max \{N_o \log N_o , T_o \log T_o \}}{\min \{N_o^2, T_o^2 \}}}.
\end{align*}
Note that 
\begin{align*}
\norm{e_i^\top A \bar{\breve{Z}}^d_o}_2 
 = \norm{\sum_{s=1}^{T_o} (\omega_{is} - 1) \nu_s  \bar{\breve{Z}}^d_{o,s,\cdot} }_2 = \norm{  (\omega_{it_o} - 1) \nu_{t_o}  \bar{\breve{Z}}^d_{o,t_o,\cdot} }_2 \leq ||\nu||_{\infty} ||Z_o||_{2,\infty}. \end{align*}
Then, since
\begin{align*}
||\nu||_{\infty} ||Z_o||_{2,\infty}
 \lesssim \sigma \sqrt{\psi_{\min,o}} \sqrt{\frac{\mu_o^3 r^3 \kappa_o^5 \max \{N_o \log N_o, T_o \log T_o \}}{T_o \min \{N_o^2, T_o^2 \}}},
\end{align*}
we reach
$$\max_i\norm{\delta_{4,i}}_2  \lesssim  \frac{\sigma}{\sqrt{\psi_{\min,o}}} \sqrt{\frac{\mu_o^3 r^3 \kappa_o^5 \max \{N_o \log N_o, T_o \log T_o \}}{T_o \min \{N_o^2, T_o^2 \}}}.$$

\paragraph{Part 5.} It is easy to check from Lemmas \ref{lem:deshrunken} and \ref{lem:smallgradient}, and the relation \eqref{eq:balance} that
\begin{align*}
 &\norm{e_i^\top \nabla_X f(\breve{X}_o,\breve{Z}_o) \left( I_r + \lambda_o (\breve{Z}^{\top}_o\breve{Z}_o)^{-1} \right)^{1/2} (\breve{Z}^{d\top}_o\breve{Z}^d_o)^{-1} \breve{H}^d_o }\\
 &\quad \leq \norm{\nabla_X f(\breve{X}_o,\breve{Z}_o)} \norm{\left( I_r + \lambda_o (\breve{Z}^{\top}_o\breve{Z}_o)^{-1} \right)^{1/2}} \norm{(\breve{Z}^{d\top}_o\breve{Z}^d_o)^{-1}}\\
 &\quad \lesssim \frac{\sigma}{\sqrt{\psi_{\min,o}}} \frac{1}{\max \{N_o^4,T_o^4\}},\\
 & \norm{e_i^\top \breve{X}_o \Delta_{balance} \breve{H}^d_o} 
 \leq \norm{\breve{X}_o} \norm{\Delta_{balance}}
 \lesssim \frac{\sigma}{\sqrt{\psi_{\min,o}}} \sqrt{\frac{\kappa_o^3 \mu_o r}{\max \{N_o^9, T_o^9 \} \min \{N_o, T_o \}}}.
\end{align*}
Hence, we have $\max_i\norm{\delta_{5,i}}_2  \lesssim  \frac{\sigma}{\sqrt{\psi_{\min,o}}} \frac{1}{\max \{N_o^4, T_o^4 \}}$.

\paragraph{Part 6.} Lastly, we check the proximity between the non-convex debiased estimator and the convex debiased estimator to bound $\max_i|| \delta_{6,i}||$. The proof is basically the same as Section C.2 of \cite{chen:2019inference}. Denote the SVD of $\breve{X}_o\breve{Z}_o^\top $ by $L_o \Sigma_o R_o^\top $. First, we show that $\breve{X}^d_o$ is close to $ L_o (\Sigma_o + \lambda_o I_r )^{\frac{1}{2}}$. By Lemma 20 of \cite{chen2020noisy}, there is an invertible matrix $G$ such that $\breve{X}_o = L_o \Sigma_o^{1/2}G$ and $\breve{Z}_o = R_o \Sigma_o^{1/2}G^{-1^\top }$. Denote the SVD of $G$ by $L_G \Sigma_G R_G^\top $. Then, we have by Lemma 20 of \cite{chen2020noisy} that
\begin{align*}
\norm{\breve{X}_o - L_o \Sigma_o^{1/2} L_G R_G^\top } 
& = \norm{L_o \Sigma_o^{1/2} L_G  \Sigma_G R_G^\top  -  L_o \Sigma_o^{1/2} L_G R_G^\top } \\
&\leq \norm{\Sigma_o^{1/2}} \norm{\Sigma_G - I_r}\\
& \lesssim \sqrt{\psi_{\max,o}} \frac{1}{\psi_{\min,o}} \norm{\breve{X}_o^\top \breve{X}_o - \breve{Z}_o^\top \breve{Z}_o}_F \\
&\lesssim \frac{\sigma}{\max \{N_o^{\frac{7}{2}} , T_o^{\frac{7}{2}} \}} \sqrt{\frac{\kappa_o}{\psi_{\min,o}}}.
\end{align*}
Here, we use the fact $\norm{\Sigma_G - I_r} \lesssim  \norm{\Sigma_G - \Sigma_G^{-1}}_F$ and Lemma \ref{lem:CCFMY_noncovex}. Let $\dddot{X} \coloneqq  L_o \Sigma_o^{1/2} L_G R_G^\top $. Then, we have by Lemma 13 of \cite{chen:2019inference} with the above result
\begin{align*}
&\norm{\breve{X}^d_o - \dddot{X} \left( I_r + \lambda_o( \dddot{X}^\top \dddot{X})^{-1}\right)^{1/2}} \\
&\leq \norm{\breve{X}_o - \dddot{X}}\norm{\left(I_r + \lambda_o(\breve{X}_o^\top \breve{X}_o )^{-1} \right)^{1/2}}
+ \norm{\dddot{X}} \norm{\left(I_r + \lambda_o(\breve{X}_o^\top \breve{X}_o )^{-1} \right)^{1/2} - \left(I_r + \lambda_o(\dddot{X}_o^\top \dddot{X}_o )^{-1} \right)^{1/2}}\\
&\lesssim \frac{\sigma}{\max \{N_o^{\frac{7}{2}} , T_o^{\frac{7}{2}} \}} \sqrt{\frac{\kappa_o}{\psi_{\min,o}}}.
\end{align*}
A similar bound holds for $\breve{Y}^d_o $. Note that $$\dddot{X} \left( I_r + \lambda_o( \dddot{X}^\top \dddot{X})^{-1}\right)^{1/2} = L_o (\Sigma_o + \lambda_o I_r)^{1/2} L_G R_G^\top .$$ Hence, we have
\begin{align*}
&\min_{O \in \calO^{r \times r}} \sqrt{\norm{\breve{X}^d_o O - L_o (\Sigma_o + \lambda_o I_r )^{\frac{1}{2}} }^2_F + \norm{\breve{Z}^d_o O - R_o (\Sigma_o + \lambda_o I_r )^{\frac{1}{2}} }^2_F}\\
&\ \ \leq  \sqrt{\norm{\breve{X}^d_o  -L_o (\Sigma_o + \lambda_o I_r)^{1/2} L_G R_G^\top }^2_F + \norm{\breve{Z}^d_o - R_o (\Sigma_o + \lambda_o I_r )^{\frac{1}{2}} L_G R_G^\top  }^2_F}\\
&\ \ \lesssim  \frac{\sigma}{\max \{N_o^{\frac{7}{2}} , T_o^{\frac{7}{2}} \}} \sqrt{\frac{\kappa_o r}{\psi_{\min,o}}}.
\end{align*}
Next, we show that $\widehat{X}_o$ is also close to $ L_o (\Sigma_o + \lambda_o I_r )^{\frac{1}{2}}$. Because $(\widetilde{X}_o,\widetilde{Z}_o)$ is a balanced factorization of $\calP_{r}(\widetilde{M}_o)$, and $(L_o \Sigma_o^{\frac{1}{2}}, R_o \Sigma_o^{\frac{1}{2}})$ is that of $\breve{X}_o\breve{Z}_o^\top $, we have by the theory for the perturbation bounds on the balanced factorization (Appendix B.7 of \cite{ma2020implicit}, Appendix B.2.1 of \cite{chen2020nonconvex}), 
\begin{align}\label{eq:beforeshrunken}
\nonumber \min_{O \in \calO^{r \times r}} \sqrt{\norm{\widetilde{X}_o O - L_o \Sigma_o^{\frac{1}{2}} }^2_F + \norm{\widetilde{Z}_o O - R_o \Sigma_o^{\frac{1}{2}} }^2_F} &\lesssim \sqrt{\frac{\kappa_o^4 r}{\psi_{\min,o}}} \norm{\calP_{r}(\widetilde{M}_o) - \breve{X}_o\breve{Z}_o^\top }_F \\
&\leq  \sqrt{\frac{\kappa_o^4 r}{\psi_{\min,o}}} \frac{\sigma}{\max \{N_o^{\frac{9}{2}}, T_o^{\frac{9}{2}}\}}.
\end{align}
Then, by repeating the same argument as above, we can conclude from \eqref{eq:beforeshrunken} that
\begin{align}\label{eq:proximityofdeshrunken}
\min_{O \in \calO^{r \times r}} \sqrt{\norm{\widehat{X}_o O - L_o (\Sigma_o + \lambda_o I_r )^{\frac{1}{2}} }^2_F + \norm{\widehat{Z}_o O - R_o (\Sigma_o + \lambda_o I_r )^{\frac{1}{2}} }^2_F} \lesssim  \sqrt{\frac{\kappa_o^4 r}{\psi_{\min,o}}} \frac{\sigma}{\max \{N_o^{\frac{9}{2}}, T_o^{\frac{9}{2}}\}}.
\end{align}
Hence, we have
$$
\max_i||\delta_{6,i}||  \leq \norm{\widehat{X}_o B_o  - \breve{X}^d_o} \norm{\breve{H}^d_o} \lesssim \sqrt{\frac{\kappa_o^4 r}{\psi_{\min,o}}} \frac{\sigma}{\max \{N_o^{\frac{7}{2}}, T_o^{\frac{7}{2}}\}}.
$$
\end{proof}
\bigskip

\begin{proof}[Proof of Proposition \ref{pro:zclt}]
The proof is basically same as that of Proposition \ref{pro:xclt} except for some parts. Here, we check the parts which are different from that of Proposition \ref{pro:xclt}.

\paragraph{Part 2.} In this case, we have
$$
\norm{(X_o^\top  X_o)^{-1} - (X_o^\top \Omega_{o,t} X_o)^{-1}} \lesssim \norm{X_o^\top  X_o - X_o^\top \Omega_{o,t} X_o } ||(X_o^\top  X_o)^{-1}||^2 \leq \frac{\vartheta_o \kappa_o \mu_o r}{N_o}  \psi_{\min,o}^{-1},
$$
because $\norm{X_o^\top  X_o - X_o^\top \Omega_{o,t} X_o} \leq \norm{\sum_{j \in \calQ_o} X_{o,j}X_{o,j}^\top } .$ So, we have with probability at least $1-O(\min\{N_o^{-9},T_o^{-9}\})$ that
$$
\max_t||\delta_{2,t}||_2 \leq C_{\delta,2} 
\frac{\sigma}{\sqrt{ \psi_{\min,o}}} \frac{\vartheta_o \kappa_o^{\frac{3}{2}} \mu_o r^{\frac{3}{2}}\sqrt{\log N_o}}{N_o}.
$$

\paragraph{Part 4.} Note that
$$
\norm{\delta_{4,t}}_2 = \norm{e_t^\top A^\top  \bar{\breve{X}}^d_o (\bar{\breve{X}}^{d,\top}_o\bar{\breve{X}}^d_o)^{-1}}_2 \leq \norm{e_t^\top A^\top  \bar{\breve{X}}^d_o}_2 \norm{(\bar{\breve{X}}^{d,\top}_o\bar{\breve{X}}^d_o)^{-1}}
\lesssim \frac{1}{\psi_{\min,o}}\norm{e_t^\top A^\top  \bar{\breve{X}}^d_o}_2.
$$
Let $\nu = [\nu_1 , \dots , \nu_{N_o}] \coloneqq e_t^\top  (\breve{Z}_o\breve{X}_o^\top  - Z_o X_o^\top )$. Then, because
\begin{align*}
&\norm{e_t^\top A^\top  \bar{\breve{X}}^d_o}_2 
 = \norm{\sum_{j=1}^{N_o} (\omega_{jt} - 1) \nu_j  \bar{\breve{X}}^d_{o,j,\cdot} }_2 
= \norm{\sum_{j\in \calQ_o} (\omega_{jt} - 1) \nu_j  \bar{\breve{X}}^d_{o,j,\cdot} }_2 \text{and  }\\
& \norm{\sum_{j\in \calQ_o} (\omega_{jt} - 1) \nu_j  \bar{\breve{X}}^d_{o,j,\cdot} }_2
\leq \vartheta_o ||\nu||_{\infty} ||X_o||_{2,\infty} 
  \lesssim \sigma \vartheta_o \sqrt{\psi_{\min,o}} \sqrt{\frac{\mu^3 r^3 \kappa_o^5 \max \{N_o \log N_o, T_o \log T_o \}}{N_o \min \{N_o^2, T_o^2 \}}},
\end{align*}
we have $$\max_t\norm{\delta_{4,t}}_2  \lesssim  \frac{\sigma \vartheta_o}{\sqrt{\psi_{\min,o}}} \sqrt{\frac{\mu_o^3 r^3 \kappa_o^5 \max \{N_o \log N_o, T_o \log T_o \}}{N_o \min \{N_o^2, T_o^2 \}}}.$$
Other parts are the same as that of the proof of Proposition \ref{pro:xclt}.
\end{proof}
\bigskip

\begin{proof}[Proof of Proposition \ref{pro:proximity}]
Note that 
$$\widehat{M}_o = \calP_{r}\left[ \calP_{\Omega_o^c}(\widetilde{M}_o) + \calP_{\Omega_o}(Y_o) \right].$$ 
Replacing $\widetilde{M}_o$ by $\breve{X}_o \breve{Z}_o^\top $ results in
$$
\calP_{\Omega_o^c}(\widetilde{M}_o) + \calP_{\Omega_o}(Y_o)  = \calP_{\Omega_o^c}(\breve{X}_o \breve{Z}_o^\top ) + \calP_{\Omega_o}(Y_o) + \Delta_Y,
$$
where $\Delta_Y = \calP_{\Omega_o^c}(\widetilde{M}_o - \breve{X}_o \breve{Z}_o^\top )$. Then, by Lemma \ref{lem:smallgradient}, we can bound 
$$
||\Delta_Y||_F \leq \norm{\widetilde{M}_o - \breve{X}_o \breve{Z}_o^\top }_F \lesssim \frac{\lambda_o}{8}.
$$
Denote the SVD of $\breve{X}_o\breve{Z}_o^\top $ by $L_o \Sigma_o R_o^\top $. By the simple modification of Claim 2 in \cite{chen2020noisy} for our missing pattern, we can have 
$$
\mathcal{P}_{\Omega_o}(\breve{X}_o\breve{Z}_o^\top -Y_o) = -\lambda_o L_o R_o^\top +\mathfrak{R}
$$ 
where $\mathfrak{R}$ is a residual matrix such that
$$
\norm{\mathcal{P}_T(\mathfrak{R})}_F \leq 72 \kappa_o \frac{1}{\sqrt{\psi_{\min,o}}} \norm{\nabla f(\breve{X}_o,\breve{Z}_o)}_F \leq \frac{1}{8}\lambda_o, \ \ \norm{\mathcal{P}_{T^{\perp}}(\mathfrak{R})} \leq \frac{1}{2}\lambda_o
$$
with probability at least $1-O(\min\{N_o^{-10},T_o^{-10}\})$. Here, $T$ is the tangent space of $\breve{X}_o\breve{Z}_o^\top $. Then, we have
\begin{align*}
\widehat{M}_o &= \calP_{r}\left[ \calP_{\Omega_o^c}(\breve{X}_o \breve{Z}_o^\top ) + \calP_{\Omega_o}(Y_o) + \Delta_Y \right] \\
&= \calP_{r}\left[ \breve{X}_o \breve{Z}_o^\top  +  \lambda_o L_o R_o^\top  + \Delta_Y  - \mathfrak{R} \right]\\
&= \calP_{r}\left[ L_o ( \Sigma_o + \lambda_o I_r) R_o^\top  + \Delta_Y  - \mathfrak{R} \right] \\
&= \calP_{r}\left[ \underbrace{L_o ( \Sigma_o + \lambda_o I_r) R_o^\top  + \calP_{T^{\perp}} (\Delta_Y  - \mathfrak{R} )}_{\coloneqq C} +  \underbrace{\calP_{T} (\Delta_Y  - \mathfrak{R} )}_{\coloneqq \Delta} \right].
\end{align*}
Note that $\psi_{k}\left( L_o ( \Sigma_o + \lambda_o I_r) R_o^\top  \right) \geq \lambda_o$ for all $1\leq k \leq r$ and
$$
\norm{\calP_{T^{\perp}} (\Delta_Y  - \mathfrak{R} )} \leq ||\Delta_Y||_F + \norm{\calP_{T^{\perp}}(\mathfrak{R})} \leq \frac{5}{8} \lambda_o
$$
where $\psi_k(A)$ is the $k$-th largest singular value of $A$. Then, because $ L_o ( \Sigma_o + \lambda_o I_r) R_o^\top  $ and $\calP_{T^{\perp}}(\Delta_Y - \mathfrak{R})$ are orthogonal to each other, we know $ L_o ( \Sigma_o + \lambda_o I_r) R_o^\top  $ is the top-$r$ SVD of $C$, $\psi_k(C) = \psi_{k}\left( L_o ( \Sigma_o + \lambda_o I_r) R_o^\top  \right)$ for all $1\leq k \leq r$, and $\psi_{r+1}(C) = \norm{\calP_{T^{\perp}} (\Delta_Y  - \mathfrak{R} )} $. In addition, denote the top-$r$ SVD of $C + \Delta$ by $\check{L}_o \check{\Sigma}_o \check{R}_o^\top $. Note that
$$
\psi_{r+1}(C + \Delta) \leq \psi_{r+1}(C) + ||\Delta|| \leq  \norm{\calP_{T^{\perp}} (\Delta_Y  - \mathfrak{R} )} + ||\Delta||_F \leq \frac{5}{8} \lambda_o + \frac{\lambda_o}{\max\{N_o^4,  T_o^4 \}} \leq \frac{3}{4} \lambda_o
$$
since $||\Delta||_F \leq ||\Delta_Y||_F + \norm{\calP_{T} (\mathfrak{R} )} \leq \frac{\lambda_o}{\max\{N_o^4,  T_o^4 \}}$ by Lemma \ref{lem:smallgradient}. Hence, we have
$$
  \psi_{r}(C) - \psi_{r+1}(C + \Delta) \geq \psi_{r}( L_o ( \Sigma_o + \lambda_o I_r) R_o^\top ) - \frac{3}{4} \lambda_o \geq \psi_{r}( \Sigma_o) + \frac{1}{4} \lambda_o \geq \frac{\psi_{\min,o}}{2}.
$$
Then, because $\widehat{M}_o = \check{L}_o \check{\Sigma}_o \check{R}_o^\top $, we can apply Lemma 14 of \cite{chen:2019inference} to obtain
\begin{gather}\label{eq:prox_1}
\norm{\widehat{M}_o - L_o ( \Sigma_o + \lambda_o I_r) R_o^\top }_F \leq \left(  \frac{12||\Sigma_o + \lambda_o I_r||}{\psi_{\min,o}} + 1 \right) || \Delta ||_F \lesssim \kappa_o || \Delta ||_F \lesssim \frac{\lambda_o}{\max\{N_o^4,  T_o^4 \}}.
\end{gather}
Moreover, we can also obtain from \eqref{eq:proximityofdeshrunken} that
\begin{align}\label{eq:prox_2}
 \nonumber\norm{\widehat{X}_o\widehat{Z}_o^\top  - L_o ( \Sigma_o + \lambda_o I_r) R_o^\top }_F 
&\lesssim \norm{\widehat{X}_o O_o - L_o ( \Sigma_o + \lambda_o I_r)^{\frac{1}{2}}}_F ||\widehat{Z}_o||
+ \norm{\widehat{Z}_o O_o - R_o ( \Sigma_o + \lambda_o I_r)^{\frac{1}{2}} }_F
||\widehat{X}_o|| \\
&\lesssim \sqrt{\kappa_o^5 r} \frac{\sigma}{\max \{N_o^{\frac{9}{2}}, T_o^{\frac{9}{2}}\}},
\end{align}
where $$O_o = \argmin_{O \in \calO^{r \times r}} \sqrt{\norm{\widehat{X}_o O - L_o (\Sigma_o + \lambda_o I_r )^{\frac{1}{2}} }^2_F + \norm{\widehat{Z}_o O - R_o (\Sigma_o + \lambda_o I_r )^{\frac{1}{2}} }^2_F}.$$ Then, we get the desired result from \eqref{eq:prox_1} and \eqref{eq:prox_2}.
\end{proof}
\bigskip

\begin{proof}[Proof of Proposition \ref{pro:decomposition}]
Thanks to Propositions \ref{pro:xclt}, \ref{pro:zclt}, and \ref{pro:proximity}, we have the following decomposition:
\begin{align*}
\widehat{m}_{o,it_o} - m_{{o,it_o}} &= (\widehat{X}_{o,i}^\top \widehat{Z}_{o,t_o} -  X_{o,i}^\top Z_{o,t_o}) + (\widehat{m}_{o,it_o} - \widehat{X}_{o,i}^\top \widehat{Z}_{o,t_o})\\
& =  X_{o,i}^\top (X_o^\top \Omega_{o,{t_o}} X_o)^{-1}X_o^\top \calP_{\Omega_o}( \calE_o) e_{t_o}
+  Z_{o,t_o}^\top  (Z_o^\top \Omega_{o,i} Z_o)^{-1} Z_o^\top \calP_{\Omega_o}( \calE_o)^\top  e_i  \\
& \ \ + \calR_i^{X\top}Z_{o,t_o} + X_{o,i}^\top \calR_{t_o}^Z
+ e_i^\top (\widehat{X}_o \widehat{H}_o - X_o )(\widehat{Z}_o \widehat{H}_o - Z_o )^\top e_{t_o} + (\widehat{m}_{o,it_o} - \widehat{X}_{o,i}^\top \widehat{Z}_{o,t_o}).
\end{align*}
First, because of Proposition \ref{pro:xclt} and the inequality $\norm{Z_{o,t_o}} \leq \sqrt{\frac{\kappa_o \mu_o r}{T_o}\psi_{\min,o}}$, we have
\begin{align*}
\max_i\norm{\calR_i^{X\top}Z_{o,t_o}} & \leq \max_i\norm{\calR_i^{X}} \norm{Z_{o,t_o}} \\
&\leq C_X \left( \frac{\sigma^2}{\psi_{\min,o}} \sqrt{\frac{\kappa_o^{10} \mu_o^2 r^2 \max\{N_o^2 \log N_o,T_o^2 \log T_o \}}{\min\{N^2_o , T^2_o \}}} \right. \\
&\qquad \qquad \left. +
 \sigma \sqrt{\frac{\kappa_o^8 \mu_o^4 r^4 \max \{N_o \log N_o, T_o \log T_o \}}{ \min \{N_o^3, T_o^3 \}}} \right).
\end{align*}
Similarly, due to Proposition \ref{pro:zclt}, we have
\begin{align*}
\max_i\norm{X_{o,i}^\top \calR_{t_o}^Z} &\leq  \max_i\norm{X_{o,i}} \norm{\calR_{t_o}^Z} \\
&\leq C_Z 
\left(
\frac{\sigma^2}{\psi_{\min,o}} \sqrt{\frac{\kappa_o^{10} \mu_o^2 r^2 \max\{N_o^2 \log N_o,T_o^2 \log T_o \}}{\min\{N^2_o , T^2_o \}}} \right.\\
&\left. \quad +
 \sigma \sqrt{\frac{\kappa_o^8 \mu_o^4 r^4 \max \{N_o \log N_o, T_o \log T_o \}}{  \min \{N_o^3, T_o^3 \}}}  + \sigma \frac{\vartheta_o}{N_o}  \sqrt{\frac{\mu_o^4 r^4 \kappa_o^6 \max \{N_o \log N_o, T_o \log T_o \}}{ \min \{N_o^2, T_o^2 \}}}
\right).
\end{align*}
In addition, by Lemma \ref{lem:deshrunken} with the assertion in Part 6 of the proof for Proposition \ref{pro:xclt} that $$\max\{ \norm{\widehat{X}_o B_o  - \breve{X}^d_o} , \norm{\widehat{Z}_o B_o  - \breve{Z}^d_o}\} \lesssim \sqrt{\frac{\kappa_o^4 r}{\psi_{\min,o}}} \frac{\sigma}{\max \{N_o^{\frac{7}{2}}, T_o^{\frac{7}{2}}\}},$$ we obtain
\begin{align*}
\max\norm{ e_i^\top (\widehat{X}_o \widehat{H}_o - X_o )(\widehat{Z}_o \widehat{H}_o - Z_o )^\top e_{t_o} } &\leq \norm{\widehat{X}_o \widehat{H}_o - X_o}_{2,\infty} \norm{\widehat{Z}_o \widehat{H}_o - Z_o}_{2,\infty}\\
& \leq 2 C_{d,\infty}^2 \frac{\sigma^2}{\psi_{\min,o}}\frac{\kappa_o^3 \mu_o r \max \{N_o \log N_o, T_o \log T_o \}}{\min \{N_o ,T_o \}}.
\end{align*}
Lastly, we have $$||\widehat{m}_{o,it_o} - \widehat{X}_{o,i}^\top \widehat{Z}_{o,t_o}|| \leq C_{prx} \frac{\sigma}{\max \{N_o^{7/2} , T_o^{7/2} \}}$$ by Proposition \ref{pro:proximity}. This completes the proof.
\end{proof}
\bigskip

\subsection{Technical lemmas: Statistical properties of the debiased estimators}

This section presents the statistical properties of the debiased estimators. Although this section studies the convergence rates of the nonconvex debiased estimators $(\breve{X}_o^{d},\breve{Z}_o^{d})$, since the nonconvex debiased estimators are very close to the convex debiased estimators $(\widehat{X}_o,\widehat{Z}_o)$, as noted in Part 6 of the proof of Proposition \ref{pro:xclt}, these results are frequently used when we prove the propositions in Section \ref{sec:propositionforinference}. Remind that
\begin{align*}
	\calF_o^{d,\tau} \coloneqq \begin{bmatrix}
		X_o^{d,\tau}\\
		Z_o^{d,\tau}
	\end{bmatrix} \in \mathbb{R}^{(N_o+T_o) \times r},
	\quad
		\calF_o^{d,\tau,(m)} \coloneqq \begin{bmatrix}
		X_o^{d,\tau,(m)}\\
		Z_o^{d,\tau,(m)}
	\end{bmatrix} \in \mathbb{R}^{(N_o+T_o) \times r},
	\quad 
	\calF_o \coloneqq \begin{bmatrix}
		X_o\\
		Z_o
	\end{bmatrix} \in \mathbb{R}^{(N_o+T_o) \times r}.
\end{align*}

\begin{lemma}\label{lem:deshrunken}
Suppose that Assumptions \ref{asp:apdx_error} - \ref{asp:apdx_groupandparameterssize} hold. With probability at least $1-O(\min\{N_o^{-10},T_o^{-10}\})$, the iterates $\{\calF_o^{d,\tau}\}_{0 \leq \tau \leq \overbar{\tau}}$ and $\{\calF_o^{d,\tau,(m)}\}_{0 \leq \tau \leq \overbar{\tau}}$ satisfy
\begin{align}
	&\norm{\calF_o^{d,\tau} H_o^{\tau}-\calF_o} \leq C_{d,op1} \frac{\sigma \sqrt{\max\{N_o,T_o\}}}{\psi_{\min,o}} \norm{X_o}, \label{eq:deshrunken_oper_1}  \\
	&\norm{\calF_o^{d,\tau} H_o^{d,\tau}-\calF_o} \leq C_{d,op2}  \frac{\kappa_o\sigma \sqrt{\max\{N_o,T_o\}}}{\psi_{\min,o}} \norm{X_o},  \label{eq:deshrunken_oper_2} \\
	&\norm{\calF_o^{d,\tau} H_o^{d,\tau}-\calF_o}_F \leq C_{d,F} \frac{\sigma \sqrt{\max\{N_o,T_o\}}}{\psi_{\min,o}} \norm{X_o}_F, \label{eq:deshrunken_frob}  \\
	&\norm{\calF_o^{d,\tau} H_o^{d,\tau}-\calF_o}_{2, \infty} \leq  C_{d,\infty}  \kappa_o \frac{\sigma \sqrt{\max\{N_o \log N_o,T_o \log T_o\}}}{\psi_{\min,o}} \norm{\calF_o}_{2, \infty},  \label{eq:deshrunken_inco} \\
	&\norm{X_o^{d,\tau \top}X_o^{d,\tau}-Z_o^{d,\tau \top}Z_o^{d,\tau}} \leq C_{d,B} \frac{\kappa_o^2 \sigma}{\max\{N_o^{9/2},T_o^{9/2}\}} , \label{eq:deshrunken_prox} \\
	 &\max_{1\leq m\leq N_o+T_o} \norm{\calF_o^{d,\tau,(m)}H_o^{d,\tau,(m)} - \calF_o} \leq 2 C_{d,op2}  \frac{\kappa_o\sigma \sqrt{\max\{N_o,T_o\}}}{\psi_{\min,o}} \norm{X_o},\label{eq:deshrunken_leave_oper} \\
	 &\max_{1\leq m\leq N_o+T_o} \norm{\calF_o^{d,\tau,(m)}H_o^{d,\tau,(m)} - \calF_o}_{2, \infty} \leq  2 C_{d,\infty} \kappa_o \frac{\sigma \sqrt{\max\{N_o \log N_o,T_o \log T_o\}}}{\psi_{\min,o}} \norm{\calF_o}_{2, \infty},\label{eq:deshrunken_leave_inco} \\
	&\max_{1\leq m\leq N_o+T_o} \norm{\calF_o^{d,\tau} H_o^{d,\tau}-\calF_o^{d,\tau,(m)}H_o^{d,\tau,(m)}} \leq C_{d,3} \kappa_o \frac{\sigma \sqrt{\max\{N_o \log N_o,T_o \log T_o\}}}{\psi_{\min,o}} \norm{\calF_o}_{2, \infty}, \label{eq:deshrunken_leave_prox}
\end{align}
where $C_{d,F}$, $C_{d,op1}$, $C_{d,op2}$, $C_{d,\infty}$, $C_{d,3}$, $C_{d,B}>0$ are absolute constants, provided that $\eta_o \overset{c}{\asymp}  \frac{1}{\max\{N_o^6,T_o^6\}\kappa_o^3\psi_{\max,o}}$ and that $\overbar{\tau} = \max\{N_o^{23},T_o^{23}\}$.
\end{lemma}

Additionally, the following lemma is exploited in Part 1 of the proof of Proposition \ref{pro:xclt} to bound some residual term.

\begin{lemma}\label{lem:leaveoneoutindependence}
Suppose that Assumptions \ref{asp:apdx_error} - \ref{asp:apdx_groupandparameterssize} hold. With probability at least $1-O(\min\{N_o^{-10},T_o^{-10}\})$, the iterates $\{(X_o^{d,\tau},Z_o^{d,\tau})\}_{0 \leq \tau \leq \overbar{\tau}}$ and $\{(X_o^{d,\tau,(m)},Z_o^{d,\tau,(m)})\}_{0 \leq \tau \leq \overbar{\tau}}$ satisfy
\begin{align*}
&\max_{1 \leq m \leq N_o}\norm{e_m^\top  \calP_{\Omega_o}(\calE_o)\left[ \bar{Z}_o^{d,\tau,(m)} \left( \bar{Z}_o^{d,\tau,(m)\top} \bar{Z}_o^{d,\tau,(m)} \right)^{-1} - Z_o \left( Z_o^{\top} Z_o \right)^{-1} \right] }_2\\
&\ \ \lesssim  \frac{\sigma^2}{\psi_{\min,o}^{3/2}} \sqrt{r \kappa_o^3 \max \{N_o \log N_o, T_o \log T_o \}},\\
&\max_{1 \leq m \leq T_o}\norm{e_m^\top  \calP_{\Omega_o}(\calE_o)^\top \left[ \bar{X}_o^{d,\tau,(N_o + m)} \left( \bar{X}_o^{d,\tau,(N_o + m)\top} \bar{X}_o^{d,\tau,(N_o + m)} \right)^{-1} - X_o \left( X_o^{\top} X_o \right)^{-1} \right] }_2 \\
&\ \ \lesssim \frac{\sigma^2}{\psi_{\min,o}^{3/2}} \sqrt{r \kappa_o^3 \max \{N_o \log N_o, T_o \log T_o \}},
\end{align*}
where $\bar{Z}_o^{d,\tau,(m)} = Z_o^{d,\tau,(m)} H_o^{d,\tau,(m)}$ and $\bar{X}_o^{d,\tau,(m)} = X_o^{d,\tau,(m)} H_o^{d,\tau,(m)}$.
\end{lemma}

The proofs of the lemmas above are as follows.

\begin{proof}[Proof of Lemma \ref{lem:deshrunken}]
Basically, the proof is similar to that in Section I.2 of \cite{chen:2019inference}. Note that $$\norm{\calF_o^{d,\tau} H_o^{\tau}-\calF_o} 
\leq \norm{\calF_o^{d,\tau}  - \calF_o^{\tau}} + \norm{\calF_o^{\tau} H_o^{\tau}-\calF_o}.$$ 
In addition, we have
\begin{align*}
\norm{\calF_o^{d,\tau}  - \calF_o^{\tau}} 
&\leq \norm{\calF_o^{\tau}} \norm{(I_r + \lambda_o (X^{\tau\top}_o X^{\tau}_o)^{-1})^{\frac{1}{2}} - I_r} \\
&\ \ + \norm{Z_o^{\tau}} \norm{(I_r + \lambda_o (Z^{\tau\top}_o Z^{\tau}_o)^{-1})^{\frac{1}{2}} - (I_r + \lambda_o (X^{\tau\top}_o X^{\tau}_o)^{-1})^{\frac{1}{2}}}.    
\end{align*}
Define $$\Delta_{balance}^{\tau} \coloneqq (I_r + \lambda_o (Z^{\tau\top}_o Z^{\tau}_o)^{-1})^{\frac{1}{2}} - (I_r + \lambda_o (X^{\tau\top}_o X^{\tau}_o)^{-1})^{\frac{1}{2}}.$$ Then, using Lemma 13 of \cite{chen:2019inference} with Lemma \ref{lem:CCFMY_noncovex},
\begin{align}\label{eq:balance}
\nonumber\norm{\Delta_{balance}^{\tau}} &\lesssim \lambda_o \norm{(X_o^{\tau \top}X_o^{\tau})^{-1}-(Z_o^{\tau \top}Z_o^{\tau})^{-1}} \\
\nonumber &\leq \lambda_o \norm{(X_o^{\tau \top}X_o^{\tau})^{-1}}\norm{X_o^{\tau \top}X_o^{\tau}-Z_o^{\tau \top}Z_o^{\tau}}\norm{(Z_o^{\tau \top}Z_o^{\tau})^{-1}}\\
&\lesssim \frac{\lambda_o}{\max\{N_o^5, T_o^5\}} \frac{\kappa_o}{\psi_{\min,o}}.
\end{align}
In addition, by Lemma 13 of \cite{chen:2019inference}, we have
$$
\norm{(I_r + \lambda_o (X^{\tau\top}_o X^{\tau}_o)^{-1})^{\frac{1}{2}} - I_r} \lesssim \frac{\lambda_o}{\psi_{\min,o}} \leq \frac{\sigma \sqrt{\max{N_o,T_o}}}{\psi_{\min,o}}.
$$
Hence, we have \eqref{eq:deshrunken_oper_1} from the above bounds. Similarly, we can derive $$\norm{\calF_o^{d,\tau} H_o^{d,\tau}-\calF_o}_F \leq \norm{\calF_o^{d,\tau} H_o^{\tau}-\calF_o}_F \lesssim \frac{\sigma \sqrt{\max{N_o,T_o}}}{\psi_{\min,o}}||X_o||_F$$ 
which is \eqref{eq:deshrunken_frob}. For \eqref{eq:deshrunken_oper_2}, note that $$\norm{\calF_o^{d,\tau} H_o^{d,\tau}-\calF_o} \leq \norm{\calF_o^{d,\tau}}\norm{H_o^{d,\tau} - H_o^{\tau}} + \norm{\calF_o^{d,\tau} H_o^{\tau} - \calF_o}.$$ 
Then, by using Lemma 36 of \cite{ma2020implicit}, we have
$$
\norm{H_o^{d,\tau} - H_o^{\tau}} \lesssim \frac{1}{\psi_{\min,o}} \norm{\calF_o^{d,\tau}  - \calF_o^{\tau}} \norm{\calF_o} \lesssim \kappa_o \frac{\sigma \sqrt{\max\{N_o, T_o \}}}{\psi_{\min,o}}
$$
and it gives \eqref{eq:deshrunken_oper_2}. In addition, by the similar logic with Lemma \ref{lem:CCFMY_noncovex}, we can derive \eqref{eq:deshrunken_inco} also. For \eqref{eq:deshrunken_prox}, notice that
\begin{align*}
\norm{X_o^{d,\tau \top}X_o^{d,\tau}-Z_o^{d,\tau \top}Z_o^{d,\tau}}
&\leq \norm{(I_r + \lambda_o (X^{\tau\top}_o X^{\tau}_o)^{-1})^{\frac{1}{2}}}
\norm{X_o^{\tau \top}X_o^{\tau}-Z_o^{\tau \top}Z_o^{\tau}}
\norm{(I_r + \lambda_o (X^{\tau\top}_o X^{\tau}_o)^{-1})^{\frac{1}{2}}}\\
&\ \ + \norm{\Delta_{balance}^{\tau}} \norm{Z_o^{\tau \top}Z_o^{\tau}}
\norm{(I_r + \lambda_o (X^{\tau\top}_o X^{\tau}_o)^{-1})^{\frac{1}{2}}}\\
&\ \ + \norm{(I_r + \lambda_o (Z^{\tau\top}_o Z^{\tau}_o)^{-1})^{\frac{1}{2}}}
\norm{Z_o^{\tau \top}Z_o^{\tau}}\norm{\Delta_{balance}^{\tau}}.
\end{align*}
Then, by the above bounds, we can derive \eqref{eq:deshrunken_prox}. Using the similar methods of deriving \eqref{eq:deshrunken_oper_2} and \eqref{eq:deshrunken_inco}, we can derive \eqref{eq:deshrunken_leave_oper} and \eqref{eq:deshrunken_leave_inco} also. Lastly, we show \eqref{eq:deshrunken_leave_prox}. Set $\calF_0 = \calF_o$, $\calF_1 = \calF_o^{d,\tau} H_o^{\tau}$ and $\calF_2 = \calF_o^{d,\tau,(m)} Q_o^{\tau,(m)}$. Then, assumptions of Lemma \ref{LemmaB12} are satisfied as noted in Section I of \cite{chen:2019inference}. So, we can apply Lemma \ref{LemmaB12} to obtain
\begin{align*}
\norm{\calF_o^{d,\tau} H_o^{d,\tau}-\calF_o^{d,\tau,(m)}H_o^{d,\tau,(m)}} &\lesssim \kappa_o \norm{\calF_o^{d,\tau} H_o^{\tau}-\calF_o^{d,\tau,(m)} Q_o^{\tau,(m)}}\\ 
&\lesssim \kappa_o \norm{\calF_o^{\tau} H_o^{\tau}-\calF_o^{\tau,(m)} Q_o^{\tau,(m)}} \\
&\lesssim \kappa_o \frac{\sigma \sqrt{\max\{N_o \log N_o,T_o \log T_o\}}}{\psi_{\min,o}} \norm{\calF_o}_{2, \infty}.    
\end{align*}
\end{proof}
\bigskip

\begin{proof}[Proof of Lemma \ref{lem:leaveoneoutindependence}]
Define $$\Delta^{\tau,(m)} \coloneqq \bar{Z}_o^{d,\tau,(m)} \left( \bar{Z}_o^{d,\tau,(m)\top} \bar{Z}_o^{d,\tau,(m)} \right)^{-1} - Z_o \left( Z_o^{\top} Z_o \right)^{-1}$$ where $\bar{Z}_o^{d,\tau,(m)} = Z_o^{d,\tau,(m)} H_o^{d,\tau,(m)}$. Then, 
$$
\norm{e_m^\top  \calP_{\Omega_o}(\calE_o)\left[ \bar{Z}_o^{d,\tau,(m)} \left( \bar{Z}_o^{d,\tau,(m)\top} \bar{Z}_o^{d,\tau,(m)} \right)^{-1} - Z_o \left( Z_o^{\top} Z_o \right)^{-1} \right] }_2
= \norm{\sum_{t=1}^{T_o} \omega_{mt} \epsilon_{mt}\Delta^{\tau,(m)}_{t,\cdot}}_2.
$$
Note that $\bbE[\omega_{mt} \epsilon_{mt}\Delta^{\tau,(m)}_{t,\cdot} | \Delta^{\tau,(m)}_{t,\cdot}] = 0$ and $\{ \epsilon_{mt}\}_{t \leq T_o}$ are independent across $t$ conditioning on $\{ \Delta^{\tau,(m)}_{t,\cdot}\}_{t \leq T_o}$. Hence, we have by the matrix Bernstein inequality with Claim \ref{clm:leaveoneoutindependence} that
\begin{align*}
\norm{\sum_{t=1}^{T_o} \omega_{mt} \epsilon_{mt}\Delta^{\tau,(m)}_{t,\cdot}}_2 
&\lesssim \sqrt{\sigma^2 ||\Delta^{\tau,(m)}||_F^2 \max \{ \log N_o , \log T_o \}} + \sigma || \Delta^{\tau,(m)} ||_{2,\infty} \max \{ \log^2 N_o , \log^2 T_o \} \\
& \lesssim  \sigma \frac{\sqrt{r}}{\sqrt{\psi_{\min,o}}} \frac{\sigma}{\psi_{\min,o}}\sqrt{\kappa_o^3 \max \{N_o \log N_o, T_o \log T_o\}}.
\end{align*}

\begin{claim}\label{clm:leaveoneoutindependence}
With probability at least $1-O(\min\{N_o^{-10},T_o^{-10}\})$, we have for all $0 \leq \tau \leq \overbar{\tau}$ and $1 \leq m \leq N_o$,
\begin{align*}
||\Delta^{\tau,(m)}|| &\lesssim  \frac{1}{\sqrt{\psi_{\min,o}}} \frac{\sigma}{\psi_{\min,o}}\sqrt{\kappa_o^3 \max \{N_o , T_o \}},\ \ || \Delta^{\tau,(m)} ||_{2,\infty}\\
&\lesssim  \frac{1}{\sqrt{\psi_{\min,o}}} \frac{\sigma}{\psi_{\min,o}}\sqrt{\frac{\kappa_o^5 \mu_o r \max\{N_o \log N_o , T_o \log T_o \}}{\min \{N_o, T_o \}}}.
\end{align*}
\end{claim}

The proof for the part $$\norm{e_m^\top  \calP_{\Omega_o}(\calE_o)^\top \left[ \bar{X}_o^{d,\tau,(N_o + m)} \left( \bar{X}_o^{d,\tau,(N_o + m)\top} \bar{X}_o^{d,\tau,(N_o + m)} \right)^{-1} - X_o \left( X_o^{\top} X_o \right)^{-1} \right] }_2$$ is similar, and therefore omitted for brevity.
\end{proof}
\bigskip

\begin{proof}[Proof of Claim \ref{clm:leaveoneoutindependence}]
By Lemma 12 of \cite{chen:2019inference} with Lemma \ref{lem:deshrunken}, we have
\begin{align*}
||\Delta^{\tau,(m)}|| &\lesssim \max \Bigl\{ \norm{ Z_o (Z_o^\top  Z_o)^{-1} }^2,  \norm{\bar{Z}_o^{d,\tau,(m)} \left( \bar{Z}_o^{d,\tau,(m)\top} \bar{Z}_o^{d,\tau,(m)} \right)^{-1} }^2 \Bigr\} \norm{\bar{Z}_o^{d,\tau,(m)} - Z_o}\\
&\lesssim \frac{1}{\psi_{\min,o}} \frac{\kappa_o\sigma \sqrt{\max\{N_o,T_o\}}}{\psi_{\min,o}} \norm{X_o}\\
&\lesssim
\frac{1}{\sqrt{\psi_{\min,o}}} \frac{\kappa_o^{\frac{3}{2}}\sigma \sqrt{\max\{N_o,T_o\}}}{\psi_{\min,o}}.
\end{align*}
In addition, because
\begin{align*}
 \norm{\left( \bar{Z}_o^{d,\tau,(m)\top} \bar{Z}_o^{d,\tau,(m)} \right)^{-1} - (Z_o^\top  Z_o)^{-1}} 
&\leq \norm{\left( \bar{Z}_o^{d,\tau,(m)\top} \bar{Z}_o^{d,\tau,(m)} \right)^{-1} } \norm{\bar{Z}_o^{d,\tau,(m)\top} \bar{Z}_o^{d,\tau,(m)}  - Z_o^\top  Z_o} \norm{(Z_o^\top  Z_o)^{-1}}\\
&\lesssim \frac{1}{\psi_{\min,o}} \kappa_o^2 \frac{\sigma}{\psi_{\min,o}}\sqrt{\max\{N_o,T_o\}},   
\end{align*}
we have from Lemma \ref{lem:deshrunken} that
\begin{align*}
\norm{ \Delta^{\tau,(m)} }_{2,\infty} 
&\leq \norm{ \bar{Z}_o^{d,\tau,(m)} }_{2,\infty} \norm{\left( \bar{Z}_o^{d,\tau,(m)\top} \bar{Z}_o^{d,\tau,(m)} \right)^{-1} - (Z_o^\top  Z_o)^{-1}} +  \norm{ \bar{Z}_o^{d,\tau,(m)} - Z_o }_{2,\infty}\norm{ (Z_o^\top  Z_o)^{-1}}\\
& \lesssim \frac{1}{\sqrt{\psi_{\min,o}}} \frac{\sigma}{\psi_{\min,o}} \sqrt{\frac{\kappa_o^5 \mu_o r \max\{N_o \log N_o , T_o \log T_o \}}{\min \{N_o, T_o \}}} .
\end{align*}
\end{proof}

\subsection{Technical lemmas: Statistical properties of the nuclear norm penalized estimators and the corresponding non-convex estimator}\label{sec:nonconvex_property}

Lastly, we present the statistical properties of the non-convex estimator $(\breve{X}_o,\breve{Z}_o)$. Since this estimator is very close to the nuclear norm penalized estimator $\widetilde{M}_o$ as we will see in Lemma \ref{lem:smallgradient}, we can derive the convergence rates of the nuclear norm penalized estimator from this result. Besides, the statistical properties of the debiased estimators in the previous section are largely based on the result of the non-convex estimators in this section. 

Basically, the result in this section is the modification of \cite{chen2020noisy} for the case where missing is not at random and occurs only at one column. To save space, we omit the proofs of some lemmas if the proof is a simple modification of that in \cite{chen2020noisy}. We are willing to provide the full proofs upon request.

First, the following lemma shows the statistical properties of the nonconvex estimator which are used for the proofs in the previous sections. Remind that
\begin{align*}
	\calF_o^{\tau} \coloneqq \begin{bmatrix}
		X_o^{\tau}\\
		Z_o^{\tau}
	\end{bmatrix} \in \mathbb{R}^{(N_o+T_o) \times r},
	\quad
		\calF_o^{\tau,(m)} \coloneqq \begin{bmatrix}
		X_o^{\tau,(m)}\\
		Z_o^{\tau,(m)}
	\end{bmatrix} \in \mathbb{R}^{(N_o+T_o) \times r},
	\quad 
	\calF_o \coloneqq \begin{bmatrix}
		X_o\\
		Z_o
	\end{bmatrix} \in \mathbb{R}^{(N_o+T_o) \times r}.
\end{align*}

\begin{lemma}\label{lem:CCFMY_noncovex}
Suppose that Assumptions \ref{asp:apdx_error} - \ref{asp:apdx_groupandparameterssize} hold. With probability at least $1-O(\min\{N_o^{-11},T_o^{-11}\})$, the iterates $\{\calF_o^{\tau}\}_{0 \leq \tau \leq \overbar{\tau}}$ and $\{\calF_o^{\tau,(m)}\}_{0 \leq \tau \leq \overbar{\tau}}$ satisfy
\begin{align}
	&\norm{\calF_o^{\tau} H_o^{\tau}-\calF_o}_F \leq C_F \left(\frac{\sigma \sqrt{\max\{N_o,T_o\}}}{\psi_{\min,o}}+\frac{\lambda_o}{\psi_{\min,o}} \right) \norm{X_o}_F, \label{Prelim3}  \\
	&\norm{\calF_o^{\tau} H_o^{\tau}-\calF_o} \leq C_{op} \left(\frac{\sigma \sqrt{\max\{N_o,T_o\}}}{\psi_{\min,o}}+\frac{\lambda_o}{\psi_{\min,o}} \right) \norm{X_o},  \label{Prelim4} \\
	&\max_{1\leq m\leq N_o+T_o} \norm{\calF_o^{\tau} H_o^{\tau}-\calF_o^{\tau,(m)}Q_o^{\tau,(m)}}_F \leq  C_3  \left(\frac{\sigma \sqrt{\max\{N_o \log N_o,T_o \log T_o\}}}{\psi_{\min,o}}+\frac{\lambda_o}{\psi_{\min,o}} \right) \norm{\calF_o}_{2, \infty}, \label{Prelim5}  \\
	&\max_{1\leq m\leq N_o+T_o} \norm{\left(\calF_o^{\tau,(m)}H_o^{\tau, (m)}-\calF_o\right)_{m,\cdot}}_2 \leq C_4 \kappa_o \left(\frac{\sigma \sqrt{\max\{N_o \log N_o,T_o \log T_o\}}}{\psi_{\min,o}}+\frac{\lambda_o}{\psi_{\min,o}} \right) \norm{\calF_o}_{2, \infty}, \label{Prelim6}  \\
	&\norm{\calF_o^{\tau} H_o^{\tau}-\calF_o}_{2, \infty} \leq  C_{\infty}  \kappa_o \left(\frac{\sigma \sqrt{\max\{N_o \log N_o,T_o \log T_o\}}}{\psi_{\min,o}}+\frac{\lambda_o}{\psi_{\min,o}} \right) \norm{\calF_o}_{2, \infty}, \label{Prelim7} \\
    & \norm{X_o^{\tau+1} H_o^{\tau+1}- X_o}_{2, \infty} \leq C_{\infty,X}  r^{1/2} \kappa_o \left(\frac{\sigma \sqrt{\max\{N_o \log N_o,T_o \log T_o\}}}{\psi_{\min,o}}+\frac{\lambda_o}{\psi_{\min,o}} \right) \norm{X_o}_{2, \infty},\\
     & \norm{Z_o^{\tau+1} H_o^{\tau+1}- Z_o}_{2, \infty} \leq C_{\infty,Z}  r^{1/2} \kappa_o \left(\frac{\sigma \sqrt{\max\{N_o \log N_o,T_o \log T_o\}}}{\psi_{\min,o}}+\frac{\lambda_o}{\psi_{\min,o}} \right) \norm{Z_o}_{2, \infty},\\
	&\norm{X_o^{\tau \top}X_o^{\tau}-Z_o^{\tau \top}Z_o^{\tau}}_{F} \leq C_B \kappa_o \eta_o \left(\frac{\sigma \sqrt{\max\{N_o  ,T_o \}}}{\psi_{\min,o}}+\frac{\lambda_o}{\psi_{\min,o}} \right) \sqrt{r} \psi_{\max,o}^2 \leq C_B  \frac{\psi_{\max,o}}{\max\{N_o^5,T_o^5\}}, \label{Prelim8} \\
	&f(X_o^{\tau}, Z_o^{\tau}) \leq f(X_o^{\tau-1}, Z_o^{\tau-1})-\frac{\eta_o}{2}\norm{\nabla f(X_o^{\tau-1}, Z_o^{\tau-1}) }_F^2, \label{Prelim9} \\
	&\max_{1\leq m\leq N_o+T_o} \norm{\calF_o^{\tau} H_o^{\tau}-\calF_o^{\tau,(m)}H_o^{\tau,(m)}}_F \leq  5 C_3 \kappa_o \left(\frac{\sigma \sqrt{\max\{N_o \log N_o,T_o \log T_o\}}}{\psi_{\min,o}}+\frac{\lambda_o}{\psi_{\min,o}} \right) \norm{\calF_o}_{2, \infty}, \label{Prelim10} \\
	&\max_{1\leq m\leq N_o+T_o} \norm{\calF_o^{\tau,(m)}H_o^{\tau,(m)} - \calF_o} \leq  2 C_{op} \left(\frac{\sigma \sqrt{\max\{N_o,T_o\}}}{\psi_{\min,o}}+\frac{\lambda_o}{\psi_{\min,o}} \right) \norm{X_o}, \label{Prelim11} \\
	\nonumber &\max_{1\leq m\leq N_o+T_o} \norm{\calF_o^{\tau,(m)}Q_o^{\tau,(m)} - \calF_o}_{2, \infty} \\
 & \ \ \leq  (C_{\infty} \kappa_o + C_3 ) \left(\frac{\sigma \sqrt{\max\{N_o \log N_o,T_o \log T_o\}}}{\psi_{\min,o}}+\frac{\lambda_o}{\psi_{\min,o}} \right) \norm{\calF_o}_{2, \infty}, \label{Prelim12} 
\end{align}
where $C_F$, $C_{op}$, $C_3$, $C_4$, $C_{\infty}$,$C_{\infty,X}$,$C_{\infty,Z}$, $C_B>0$ are absolute constants, provided that $\eta_o \overset{c}{\asymp}  \frac{1}{\max\{N_o^6,T_o^6\}\kappa_o^3\psi_{\max,o}}$ and that $\overbar{\tau} = \max\{N_o^{23},T_o^{23}\}$.
\end{lemma}

\begin{proof}
Because the initial estimators, $(X_o^0,Z_o^0)$ and $(X_o^{0,(m)},Z_o^{0,(m)})$, are set to $(X_o,Z_o)$, \eqref{Prelim3} - \eqref{Prelim8} are satisfied when $\tau = 0$. Then, by the mathematical induction, Lemmas \ref{LemmaB3} - \ref{LemmaB8} with Lemmas \ref{LemmaB5}, \ref{LemmaBnew} show that the iterates $\{\calF_o^{\tau}\}_{o \leq \tau \leq \overbar{\tau}}$ and $\{\calF_o^{\tau,(m)}\}_{o \leq \tau \leq \overbar{\tau}}$ satisfy \eqref{Prelim3} - \eqref{Prelim8} with probability at least $1-O(\min\{N_o^{-11},T_o^{-11}\})$. In addition, \eqref{Prelim9} - \eqref{Prelim12} are derived from Lemmas \ref{LemmaB9} and \ref{LemmaB11}.
\end{proof}

The technical lemmas used in this proof are relegated to Section \ref{sec:proofforchen}. The following lemma shows the proximity between the non-convex estimator and the nuclear norm penalized estimator.

\begin{lemma}\label{lem:smallgradient}
Let $\tau^*_o = \argmin_{ 0 \leq \tau \leq \overbar{\tau}}|| \nabla f(X^\tau_o,Z^\tau_o)||_F$. Suppose that Assumptions \ref{asp:apdx_error} - \ref{asp:apdx_groupandparameterssize} hold. Then, with probability at least $1-O(\min\{N_o^{-11},T_o^{-11}\})$, we have
\begin{align}
\label{eq:smallgradient} &||\nabla f(X^{\tau^*_o}_o,Z^{\tau^*_o}_o)||_F \leq C_{gr}\frac{1}{\max \{N_o^5, T_o^5\}} \lambda_o \sqrt{\psi_{\min,o}},\\  
\label{eq:proximity}&\max\Bigl\{\norm{X^{\tau^*_o}_o Z^{\tau^*_o\top}_o-\widetilde{M}_o}_F, \norm{X^{\tau^*_o}_o Z^{\tau^*_o\top}_o- \calP_{r}(\widetilde{M}_o)}_F \Bigr\} \leq 4 C_{cvx} C_{gr} \frac{\lambda_o}{\max \{N_o^5, T_o^5\}},
\end{align}
where $C_{cvx}, C_{gr}>0$ are absolute constants.
\end{lemma}

\begin{proof}
The inequality \eqref{eq:smallgradient} comes from Lemma \ref{LemmaB2}. In addition, we have
$$
\norm{\calP_{r}(\widetilde{M}_o)-\widetilde{M}_o}_F \leq \norm{X^{\tau^*_o}_o Z^{\tau^*_o\top}_o-\widetilde{M}_o}_F \leq 2 C_{cvx} C_{gr} \frac{\lambda_o}{\max \{N_o^5, T_o^5\}}
$$ 
from Lemma \ref{LemmaA1} with Lemmas \ref{lem:CCFMY_noncovex}, \ref{LemmaA3} and \ref{LemmaA4}, and the inequality \eqref{eq:smallgradient} by setting $(\ddot{X}_o,\ddot{Z}_o) = (X^{\tau^*_o}_oH^{\tau^*_o}_o,Z^{\tau^*_o}_oH^{\tau^*_o}_o)$.
Besides, the inequality \eqref{eq:proximity} comes from this inequality.
\end{proof}
\bigskip

\section{Proofs of theorems and corollaries in the main text}

Using the tools from the previous section, we shall now prove the theorems and corollaries in the main text.

\subsection{Proofs for Section \ref{sec:convergencerate}}

\textbf{Proof of Theorem \ref{thm:consistency}}.

Note that $$\widetilde{M} - M = (\widetilde{M} - \breve{X} \breve{Z}^\top ) + (\breve{X} \breve{Z}^\top  - XZ^\top ).$$ Here, $(\breve{X}, \breve{Z})$ are the nonconvex estimator introduced in Section \ref{sec:notation} and $(X,Z)=( UD^\frac{1}{2},VD^\frac{1}{2})$ where $UDV^\top $ is the SVD of $M$. Note that Assumptions \ref{asp:apdx_error} - \ref{asp:apdx_groupandparameterssize} are satisfied since the number of missing entries $\vartheta_o$ is $|\Omega^c|$ in this case. Then, we have from Lemmas \ref{lem:CCFMY_noncovex} and \ref{lem:smallgradient} that
\begin{align*}
\norm{\widetilde{M} - M}_{\infty} &= \norm{\widetilde{M} - \breve{X} \breve{Z}^\top }_{\infty} + \norm{\breve{X} \breve{H} - X}_{2,\infty}\norm{\breve{H}^\top \breve{Z}^\top }_{2,\infty} + \norm{X}_{2,\infty}\norm{\breve{Z} \breve{H} - Z}_{2,\infty}
\\
&\lesssim  \frac{\lambda}{\max \{N^5, T^5\}} + \frac{\sigma \mu r^{\frac{3}{2}} \kappa^2 \sqrt{\max\{ \log N, \log T\}}}{\sqrt{\min\{N,T\}}}\\
&\lesssim \frac{\sigma \mu r^{\frac{3}{2}} \kappa^2 \sqrt{\max\{ \log N, \log T\}}}{\sqrt{\min\{N,T\}}},
\end{align*}
where $\lambda = C_\lambda \sigma \sqrt{\max\{N,T\}}$ for some constant $C_\lambda>0$, since we have by Lemma \ref{lem:CCFMY_noncovex}
\begin{align*}
&\norm{X}_{2,\infty}\norm{\breve{Z} \breve{H} - Z}_{2,\infty}, \norm{\breve{X} \breve{H} - X}_{2,\infty}\norm{\breve{H}^\top \breve{Z}^\top }_{2,\infty}\\
&\lesssim 
 \sqrt{r} \kappa \left(\frac{\sigma \sqrt{\max\{N \log N,T \log T\}}}{\psi_{\min}}+\frac{\lambda}{\psi_{\min}} \right) \norm{X}_{2, \infty}\norm{Z}_{2, \infty}\\
& \lesssim \frac{\sigma \mu r^{\frac{3}{2}} \kappa^2 \sqrt{\max\{ \log N, \log T\}}}{\sqrt{\min\{N,T\}}}.\ \ \square
\end{align*}
\bigskip

\noindent\textbf{Proofs of Corollaries \ref{cor:consistency_subgroup} and \ref{cor:groupclt_split_unit}}.

First, we prove Corollary \ref{cor:consistency_subgroup}. By Assumption (iii), we know $N_0 \leq N_l = N_0 + |\calG_l| \leq 2 N_0$. Then, we have by Assumptions (iii) and (iv) 
\begin{align}\label{eq:leftsingularvectorsum}
&\lambda_{\min} \left( \frac{1}{N_l} \sum_{i\leq N_l} \left( \sqrt{N} u_i \right) \left( \sqrt{N} u_i \right)^\top  \right) \\
\nonumber&\ \ \geq \lambda_{\min} \left( \frac{1}{N_l} \sum_{i\leq N_0} \left( \sqrt{N} u_i \right) \left( \sqrt{N} u_i \right)^\top  \right) - \norm{ \frac{1}{N_l} \sum_{i \in \calG_l} \left( \sqrt{N} u_i \right) \left( \sqrt{N} u_i \right)^\top  }\\
\nonumber&\ \ \geq \frac{c}{2} - \frac{\mu r |\calG_l|}{N_l}  \geq \frac{c}{4}.
\end{align}
Similarly, we can have $\lambda_{\max} \left( \frac{1}{N_l} \sum_{i\leq N_l} \left( \sqrt{N} u_i \right) \left( \sqrt{N} u_i \right)^\top  \right) \leq 4C$. Then, using Lemma \ref{lem:eigenrelation}, we can have $\mu_l \lesssim \mu \kappa^\frac{1}{2}$, $\kappa_l \lesssim \kappa$, and $\psi_{O,\min} \asymp \psi_{l,\min}$, where $\mu_l$ and $\kappa_l$ are the incoherence parameter and condition number of the submatrix $M_l$, and $\psi_{l,\min}$ is the smallest nonzero singular value of $M_l$. Using these relations, we can check that submatrix $M_l$ satisfies Assumptions \ref{asp:apdx_error} - \ref{asp:apdx_groupandparameterssize} under the assumptions of Corollary \ref{cor:consistency_subgroup}. Then, we can derive the bound of $\norm{\widetilde{M}_l - M_l}_{\infty}$ by the same way as in the proof of Theorem \ref{thm:consistency} from Lemmas \ref{lem:CCFMY_noncovex} and \ref{lem:smallgradient}. In addition, we replace $\mu_l$ and $\kappa_l$ in the bound of $\norm{\widetilde{M}_l - M_l}_{\infty}$ with $\mu \kappa^\frac{1}{2}$ and $\kappa$ using the above relations from Lemma \ref{lem:eigenrelation}, and replace $N_l$ with $N_0$ since $N_0 \leq N_l = N_0 + |\calG_l| \leq 2 N_0$. Lastly, the bound of $\norm{\widetilde{M} - M}_{\infty}$ trivially follows from that of $\norm{\widetilde{M}_l - M_l}_{\infty}$ since any entry of $M$ is included in at least one of $M_l$.

Symmetrically, we can prove Corollary \ref{cor:groupclt_split_unit} using the same way. So, we omit the proof. $\square$
\bigskip

\noindent\textbf{Proof of Corollary \ref{cor:consistency_subgroup_block}}

It is a simple extension of Corollary \ref{cor:consistency_subgroup} and the proof is same as that of Corollary \ref{cor:consistency_subgroup}. The only difference is that the dimension of the submatrix $M_l$ becomes $N_l \times T_l$ where $N_l = N_0 + |\calG_l|$ and $T_l = T_0 + 1$. Here, we have from Assumption (iv) that
\begin{align}\label{eq:rightsingularvectorsum}
\lambda_{\min} \left( \frac{1}{T_l} \sum_{t\leq T_l} \left( \sqrt{T} v_t \right) \left( \sqrt{T} v_t \right)^\top  \right) 
\nonumber&\geq \lambda_{\min} \left( \frac{1}{T_l} \sum_{t\leq T_0} \left( \sqrt{T} v_t \right) \left( \sqrt{T} v_t \right)^\top  \right) - \norm{ \frac{1}{T_l} \left( \sqrt{T} v_{t_o} \right) \left( \sqrt{T} v_{t_o} \right)^\top  } \\
& \geq \frac{c}{2} - \frac{\mu r}{T_l}  \geq \frac{c}{4}
\end{align}
and $\lambda_{\max} \left( \frac{1}{T_l} \sum_{t\leq T_l} \left( \sqrt{T} v_t \right) \left( \sqrt{T} v_t \right)^\top  \right) \leq 4C$. Then, by \eqref{eq:leftsingularvectorsum} and \eqref{eq:rightsingularvectorsum}, we can exploit Lemma \ref{lem:eigenrelation}. In the bounds of $\norm{\widetilde{M}_l - M_l}_{\infty}$, we replace $\mu_l$ and $\kappa_l$ with $\mu \kappa^\frac{1}{2}$ and $\kappa$ using the results of Lemma \ref{lem:eigenrelation}, and replace $N_l$ and $T_l$ with $N_0$ and $T_0$ since $N_0 \leq N_l = N_0 + |\calG_l| \leq 2 N_0$ and $T_l = T_0 +1$. In addition, the bound of $\norm{\widetilde{M} - M}_{\infty}$ trivially follows from that of $\norm{\widetilde{M}_l - M_l}_{\infty}$. $\square$
\bigskip

\noindent\textbf{Proof of Corollary \ref{cor:consistency_staggered_adoption}}

In the case of the estimation of missing entries in the matrix $M_{d,d' }$, the dimension of each submatrix is $N_l \times T_l$ where $N_l = N_{d' } + |\calG_l|$ and $T_l = T_{d} + 1$. By the similar way to \eqref{eq:leftsingularvectorsum} and \eqref{eq:rightsingularvectorsum}, we can show 
\begin{align*}
&\frac{c}{4} \leq \lambda_{\min} \left( \frac{1}{N_l} \sum_{i\leq N_l} \left( \sqrt{N} u_i \right) \left( \sqrt{N} u_i \right)^\top  \right) \leq \lambda_{\max} \left( \frac{1}{N_l} \sum_{i\leq N_l} \left( \sqrt{N} u_i \right) \left( \sqrt{N} u_i \right)^\top  \right) \leq 4C,\\
&\frac{c}{4} \leq \lambda_{\min} \left( \frac{1}{T_l} \sum_{t\leq T_l} \left( \sqrt{T} v_t \right) \left( \sqrt{T} v_t \right)^\top  \right) \leq \lambda_{\max} \left( \frac{1}{T_l} \sum_{t\leq T_l} \left( \sqrt{T} v_t \right) \left( \sqrt{T} v_t \right)^\top  \right) \leq 4C.
\end{align*}
Hence, we can exploit Lemma \ref{lem:eigenrelation} to replace $\mu \kappa^\frac{1}{2}$, $\kappa$, $\psi_{\min,O_{d,d' }}$ with $\mu_l$, $\kappa_l$, $\psi_{\min,l}$ and replace $N_d$ and $T_{d' }$ with $N_l$ and $T_l$ in our conditions and then, we can check that for each $l$, Assumptions \ref{asp:apdx_error} - \ref{asp:apdx_groupandparameterssize} are satisfied. Then, we derive the bounds of $\norm{\widetilde{M}_l - M_l}_{\infty}$ by the same way as in the proof of Theorem \ref{thm:consistency} using Lemmas \ref{lem:CCFMY_noncovex} and \ref{lem:smallgradient}. The bound of $\norm{\widetilde{M}_{d,d'} - M_{d,d'}}_{\infty}$ trivially follows from that of $\norm{\widetilde{M}_l - M_l}_{\infty}$. $\square$

\subsection{Proofs for Section \ref{sec:inference}}

\noindent\textbf{Proof of Theorem \ref{thm:groupclt_split_block}}

 First of all, by using the fact from Lemma \ref{lem:eigenrelation} that $\mu_l \lesssim \mu \kappa^\frac{1}{2}$, $\kappa_l \lesssim \kappa$, $\psi_{\min,l} \asymp \psi_{\min,O}$ and the relations that $N_0 \leq N_l = N_0 + |\calG_l| \leq 2N_0$ and $T_l = T_0 +1$ w.h.p., we can check that Assumptions \ref{asp:apdx_error} - \ref{asp:apdx_groupandparameterssize} are satisfied for each submatrix. Denote by $l(i)$ the group $0\leq l \leq L$ where the unit $i$ is included in. That is, $ i \in \calG_{l(i)}$. Then, by Proposition \ref{pro:decomposition}, we have the following decomposition:
\begin{align*}
& \frac{\calV_{\calG}^{-\frac{1}{2}}}{|\calG|}\sum_{i \in \calG} ( \widehat{m}_{it_0} - m_{{it_0}} ) \\
&= \underbrace{\frac{\calV_{\calG}^{-\frac{1}{2}}}{|\calG|}\sum_{i \in \calG} X_{l(i),i}^\top  \left( \sum_{j \leq N_0 } X_{l(i),j} X_{l(i),j}^\top  \right)^{-1}  \sum_{j \leq N_0} \epsilon_{jt_o} X_{l(i),j}}_{\coloneqq A}\\
&+\underbrace{\frac{\calV_{\calG}^{-\frac{1}{2}}}{|\calG|}\sum_{i \in \calG} Z_{l(i),t_o}^\top  \left( \sum_{s \leq T_0 } Z_{l(i),s} Z_{l(i),s}^\top  \right)^{-1}  \sum_{s \leq T_0} \epsilon_{is} Z_{l(i),s} }_{\coloneqq B}\\
&  + \underbrace{\frac{\calV_{\calG}^{-\frac{1}{2}}}{|\calG|}\sum_{i \in \calG_0} Z_{0,t_o}^\top \left[  \left( \sum_{s \leq T_0 } Z_{0,s} Z_{0,s}^\top + Z_{0,t_0} Z_{0,t_0}^\top  \right)^{-1} \left( \sum_{s \leq T_0} \epsilon_{is} Z_{0,s}  + \epsilon_{it_0} Z_{0,t_0}\right) - \left( \sum_{s \leq T_0 } Z_{0,s} Z_{0,s}^\top  \right)^{-1}  \sum_{s \leq T_0} \epsilon_{is} Z_{0,s}  \right]}_{\coloneqq \calR_1}  \\
& + \underbrace{\frac{\calV_{\calG}^{-\frac{1}{2}}}{|\calG|}\sum_{i \in \calG} \calR_{l(i),i}^{M}}_{\coloneqq \calR_2}.
\end{align*}
Here, $(X_l,Z_l) = (U_l D_l^{\frac{1}{2}}, V_l D_l^{\frac{1}{2}})$ where $U_l D_l V_l^\top $ is the SVD of $M_l$. $X_{l,j}$ is the transpose of the row of $X_{l}$ corresponding to the unit $j$ and $Z_{l,s}$ is the transpose of the row of $Z_{l}$ corresponding to the time period $s$. Because for each $0 \leq 1 \leq L$, there is an invertible matrix $H_l$ such that $u_j = H_l X_{l,j}$, we have
\begin{align*}
A = \frac{\calV_{\calG}^{-\frac{1}{2}}}{|\calG|}\sum_{i \in \calG} u_{i}^\top  \left( \sum_{j \leq N_o } u_{j} u_{j}^\top  \right)^{-1}  \sum_{j \leq N_o} \epsilon_{jt_0} u_j.
\end{align*}
Similarly, we can show that $$B = \frac{\calV_{\calG}^{-\frac{1}{2}}}{|\calG|}\sum_{i \in \calG} v_{t_o}^\top  \left( \sum_{s \leq T_0 } v_{s} v_{s}^\top  \right)^{-1}  \sum_{s \leq T_0} \epsilon_{is} v_s.$$ Note that
\begin{align*}
&\norm{a_j} \coloneqq \norm{\frac{\calV_{\calG}^{-\frac{1}{2}}}{|\calG|} \sum_{i \in \calG} u_{i}^\top  \left( \sum_{j \leq N_0 }  u_{j} u_{j}^\top  \right)^{-1} u_j} 
\leq \calV_{\calG}^{-\frac{1}{2}} \max_i ||u_i||^2 \psi_{\min}^{-1} \left( \sum_{j \leq N_0 }  u_{j} u_{j}^\top  \right) 
\leq \calV_{\calG}^{-\frac{1}{2}} \frac{\mu r}{N_0}.
\end{align*}
Hence, we have $$\norm{\sum_{j \leq N_0} \bbE[a_j^4 \epsilon_{jt_0}^4 ]} = \norm{\sum_{j \leq N_0} \bbE[ \epsilon_{jt_0}^4 ] a_j^4 }
\leq  \calV_{\calG}^{-2} \sigma^4 \frac{\mu^4 r^4 }{N_0^3}.$$ Then, for any $q>0$, we have by Cauchy-Schwarz and Markov inequalities that
\begin{align*}
 \Var(A)^{-1} \sum_{j \leq N_0} \bbE[(a_j \epsilon_{jt_0})^2 1_{\{|a_j \epsilon_{jt_0}| > q \Var(A)^{1/2} \}}] 
&\leq \frac{1}{\Var(A) q} \sqrt{\sum_{j\leq N_0} \bbE[(a_j \epsilon_{jt_0})^4 ]}
 \lesssim \frac{\mu^2 r^2 }{N_0^{\frac{1}{2}}} = o_p(1)
\end{align*}
since $$\Var(A) =  \calV_{\calG}^{-1} \sigma^2 \bar{u}_\calG^\top  \left( \sum_{j \leq N_0 }  u_{j} u_{j}^\top  \right)^{-1} \bar{u}_\calG \geq c \calV_{\calG}^{-1} \sigma^2 N_0^{-1}$$ for some constant $c>0$. Then, we have by Lindeberg theorem that $$\Var(A)^{-1/2} A \conD \calN(0,1).$$ In the same way, we can derive $$\Var(B)^{-1/2} B \conD \calN(0,1)$$ where $\Var(B) =  \calV_{\calG}^{-1} \frac{\sigma^2}{|\calG|} v_{t_0}^\top  \left( \sum_{s \leq T_0 }  v_{s} v_{s}^\top  \right)^{-1} v_{t_0}$. Then, because $A$ and $B$ are independent, by using the similar assertion in the proof of Theorem 3 of \cite{bai:2003}, we have
\begin{align*}
 A+ B = \Var(A)^{1/2} (\Var(A)^{-1/2} A) + \Var(B)^{1/2} (\Var(B)^{-1/2} B) \conD  \calN\left(0, 1 \right)   
\end{align*}
since $\Var(A) + \Var(B) =1$. 

In addition, note that the difference between $ \sum_{s \leq T_0 } Z_{0,s} Z_{0,s}^\top + Z_{0,t_0} Z_{0,t_0}^\top   $ and $ \sum_{s \leq T_0 } Z_{0,s} Z_{0,s}^\top $ is just one element $Z_{0,t_0} Z_{0,t_0}^\top $, and that between $ \sum_{s \leq T_0} \epsilon_{is} Z_{0,s} + \epsilon_{it_0} Z_{0,t_0}$ and $\sum_{s \leq T_0} \epsilon_{is} Z_{0,s}$ is just $\epsilon_{it_0} Z_{0,t_0}$. Hence, without difficulty, we can show that $\norm{\calR_1} = o_p(1)$. Moreover, note that since
$$
\calV_{\calG} =  \sigma^2 \bar{u}_\calG^\top  \left( \sum_{j \leq N_0 }  u_{j} u_{j}^\top  \right)^{-1} \bar{u}_\calG +  \frac{\sigma^2}{|\calG|} v_{t_0}^\top  \left( \sum_{s \leq T_0 }  v_{s} v_{s}^\top  \right)^{-1} v_{t_0} \geq c \sigma^2 \left(\frac{1}{N_0} + \frac{1}{|\calG|T_0} \right)
$$
for some constant $c>0$, we have $$\calV_{\calG}^{-\frac{1}{2}} \lesssim \min\{\sqrt{N_0}, \sqrt{|\calG|T_0}\}/\sigma.$$ Hence, by Proposition \ref{pro:decomposition}, we have with probability at least $1-O(\min\{N_0^{-7},T_0^{-7}\})$ that
\begin{align*}
\norm{\calR_2}
&\leq\calV_{\calG}^{-\frac{1}{2}} \max_{0\leq l \leq L}\max_{i\in \calG_l}||\calR_{l,i}^{M}|| \\
&\leq C_M' 
\left(
\max_{0\leq l \leq L} \frac{\sigma \kappa_l^5 \mu_l r \min\{\sqrt{N_0}, \sqrt{|\calG|T_0}\} \max\{N_l \log N_l,T_l \log T_l \}}{\psi_{\min,l}\min\{N_l , T_l \}} \right.\\
&\quad + \max_{0 \leq l \leq L} \frac{ \kappa_l^4 \mu_l^2 r^2 \min\{\sqrt{N_0}, \sqrt{|\calG|T_0}\}\max \{ \sqrt{N_l \log N_l}, \sqrt{T_l \log T_l} \}}{ \min \{N_l^{\frac{3}{2}}, T_l^{\frac{3}{2}} \}}\\
&\quad + \left. \max_{l \in [L]} \frac{ \mu_l^2 r^2 \kappa_l^3 |\calG_l|  \max \{\sqrt{N_l \log N_l}, \sqrt{T_l \log T_l} \}}{ \sqrt{N_l} \min \{N_l, T_l \}}
\right)
\end{align*}
for an absolute constant $C_M'  >0$. Then, by Assumptions (i), (ii), and (iii) with the relations that $\mu_l \lesssim \mu \kappa^\frac{1}{2}$, $\kappa_l \lesssim \kappa$, $N_0 \leq N_l \leq 2N_0$ and $T_l = T_0 + 1$ w.h.p., we have  $\norm{\calR_2} = o_p(1)$. Therefore,
$$\calV_{\calG}^{-\frac{1}{2}}\frac{1}{|\calG|}\sum_{i \in \calG} ( \widehat{m}_{it_0} - m_{{it_0}} ) \conD \calN(0,1). \ \  \square $$
\bigskip

\noindent\textbf{Proof of Theorem \ref{thm:groupclt_split_staggered_adoption}}

 By using the fact from Lemma \ref{lem:eigenrelation} that $\mu_l \lesssim \mu \kappa^\frac{1}{2}$, $\kappa_l \lesssim \kappa$, $\psi_{\min,l} \asymp \psi_{\min,O_{d_l}}$, and the relations that $N_0 \leq N_l =N_0 + |\calG_l| \leq 2N_0$ and $T_l = T_{d_l} + 1$ w.h.p., we can check that Assumptions \ref{asp:apdx_error} - \ref{asp:apdx_groupandparameterssize} are satisfied for each $N_l \times T_l$ submatrix $M_l$. Then, by Proposition \ref{pro:decomposition}, we have the following decomposition:
\begin{align*}
 \frac{\calV_{\calG}^{-\frac{1}{2}}}{|\calG|}\sum_{i \in \calG} ( \widehat{m}_{it_o} - m_{{it_o}} ) 
  &= \underbrace{\frac{\calV_{\calG}^{-\frac{1}{2}}}{|\calG|}\sum_{i \in \calG} X_{l(i),i}^\top  \left( \sum_{j \leq N_0 } X_{l(i),j} X_{l(i),j}^\top  \right)^{-1}  \sum_{j \leq N_0} \epsilon_{jt_o} X_{l(i),j}}_{\coloneqq A}\\
& \ \ +\underbrace{\frac{\calV_{\calG}^{-\frac{1}{2}}}{|\calG|}\sum_{i \in \calG} Z_{l(i),t_o}^\top  \left( \sum_{s \leq T_{d_{l(i)}} } Z_{l(i),s} Z_{l(i),s}^\top  \right)^{-1}  \sum_{s \leq T_{d_{l(i)}} } \epsilon_{is} Z_{l(i),s} }_{\coloneqq B}
 +\frac{\calV_{\calG}^{-\frac{1}{2}}}{|\calG|}\sum_{i \in \calG} \calR_{l(i),i}^{M}\\
  & = \underbrace{\frac{\calV_{\calG}^{-\frac{1}{2}}}{|\calG|}\sum_{i \in \calG} u_{i}^\top  \left( \sum_{j \leq N_0 } u_{j} u_{j}^\top  \right)^{-1}  \sum_{j \leq N_0} \epsilon_{jt_o} u_{j}}_{=A}  \\
  & \ \ +\underbrace{\frac{\calV_{\calG}^{-\frac{1}{2}}}{|\calG|}\sum_{i \in \calG} v_{t_o}^\top  \left( \sum_{s \leq T_{d_{l(i)}} } v_{s} v_{s}^\top  \right)^{-1}  \sum_{s \leq T_{d_{l(i)}} } \epsilon_{is} v_{s} }_{= B}
 +\frac{\calV_{\calG}^{-\frac{1}{2}}}{|\calG|}\sum_{i \in \calG} \calR_{l(i),i}^{M},
\end{align*}
with the convention that $d_0 = 0$.
Then, we can represent $A+B$ as
\begin{align*}
&A + B = \sum_{j \leq N}\sum_{s \leq T} \underbrace{\left( P 1_{\{j \leq N_0, s = t_o\}} + \sum_{0 \leq l \leq L} Q_l1_{\{j \in \calG_l , s \in T_{d_l} \}} 
\right)\epsilon_{js}}_{\coloneqq \calY_{js}},\\
&\text{where} \ \ P =\frac{\calV_{\calG}^{-\frac{1}{2}}}{|\calG|}\sum_{i \in \calG} u_{i}^\top  \left( \sum_{j \leq N_0 } u_{j} u_{j}^\top  \right)^{-1} u_{j}, \ \
Q_l =  \frac{\calV_{\calG}^{-\frac{1}{2}}}{|\calG|} v_{t_o}^\top  \left( \sum_{s \leq T_{d_{l}} }  v_{s} v_{s}^\top  \right)^{-1} v_s.
\end{align*}
Because $\{\epsilon_{js} \}_{j\leq N, s \leq T}$ is independent across $j$ and $s$, $A + B$ is a sum of independent random variables and so, we can use Lindeberg CLT. To check the Lindeberg condition, we first bound $\sum_{j,s} \bbE[\calY_{js}^4 ]$. Note that
\begin{align*}
\norm{P} = \frac{\calV_{\calG}^{-\frac{1}{2}}}{|\calG|} \norm{\sum_{i \in \calG} u_{i}^\top  \left( \sum_{j \leq N_0 }  u_{j} u_{j}^\top  \right)^{-1} u_j} 
\leq \calV_{\calG}^{-\frac{1}{2}} \max_i ||u_i||^2 \psi_{\min}^{-1} \left( \sum_{j \leq N_0 }  u_{j} u_{j}^\top  \right) 
\lesssim \calV_{\calG}^{-\frac{1}{2}}  \frac{\mu r  }{N_0}.
\end{align*}
Hence, we have
\begin{align*}
\norm{\sum_{j \leq N}\sum_{s \leq T} \bbE[P^4 1_{\{j \leq N_0, s = t_o\}} \epsilon_{js}^4]} = \norm{\sum_{j \leq N_0} \sum_{s = t_o} \bbE[ \epsilon_{js}^4] P^4 }
\lesssim \sigma^4 N_0 ||P||^4  \leq  \calV_{\calG}^{-2} \sigma^4 \frac{\mu^4 r^4 }{N_0^3}.
\end{align*}
 In addition, because $1_{\{j \in \calG_{l' } , s \leq T_{d_{l' }} \}}1_{\{j \in \calG_l , s\leq T_{d_{l}} \}} = 0$ when $l \neq l' $, we have
\begin{align*}
\sum_{j \leq N}\sum_{s \leq T} \bbE \left[  \left(\sum_{0 \leq l \leq L} Q_l1_{\{j \in \calG_l , s \leq T_{d_{l}} \}} \right)^4 \epsilon_{js}^4  \right] 
&= \sum_{j \leq N}\sum_{s \leq T} \sum_{0 \leq l \leq L} Q_l^4 1_{\{j \in \calG_l , s \leq T_{d_{l}} \}} \bbE \left[ \epsilon_{js}^4 \right]\\
&= \sum_{0 \leq l \leq L} \sum_{j \leq N}\sum_{s \leq T} Q_l^4 1_{\{j \in \calG_l , s \leq T_{d_{l}} \}} \bbE \left[ \epsilon_{js}^4 \right].
\end{align*}
For each $l$, we have
$$\sum_{j \in \calG_l}\sum_{ s \leq T_{d_{l}}} Q_l^4 \bbE \left[ \epsilon_{js}^4  \right] \lesssim \frac{|\calG_l|}{|\calG|} \calV_{\calG}^{-2} \sigma^4 \frac{1}{|\calG|^3} \frac{\mu^4 r^4}{T_{d_{l}}^{3}}
\leq \calV_{\calG}^{-2} \sigma^4 \frac{1}{(L+1)^3} \frac{\mu^4 r^4 }{T_{d_{l}}^{3}}$$
because 
$$||Q_l|| \leq \calV_{\calG}^{-\frac{1}{2}} \frac{1}{|\calG|} \max_t ||v_t||^2 \psi_{\min}^{-1} \left( \sum_{s \leq T_{d_{l}} }  v_{s} v_{s}^\top  \right) \leq \calV_{\calG}^{-\frac{1}{2}} \frac{1}{|\calG|} \frac{\mu r }{T_{d_{l}}}.$$ 
Then, we have
\begin{align*}
\sum_{0 \leq l \leq L}\sum_{j \in \calG_l}\sum_{ s \leq T_{d_{l}}} Q_l^4 \bbE \left[ \epsilon_{js}^4  \right] 
&\lesssim  \calV_{\calG}^{-2} \sigma^4 \mu^4 r^4 \frac{1}{(L+1)^3} \sum_{0 \leq l \leq L} \frac{1}{T_{d_{l}}^{3}} \\
& \leq \calV_{\calG}^{-2} \sigma^4 \mu^4 r^4 \left( \frac{1}{L+1} \sum_{0 \leq l \leq L} \frac{1}{T_{d_{l}}} \right)^3\\
&\lesssim \calV_{\calG}^{-2} \sigma^4 \mu^4 r^4 \bar{T}^{-3}
\end{align*}
where $\bar{T}^{-1} \coloneqq   \frac{1}{L+1} \sum_{0 \leq l \leq L} {T_{d_{l}}}^{-1}$. Therefore, we can reach
$$\sum_{j,s} \bbE[\calY_{js}^4 ] \lesssim \calV_{\calG}^{-2} \sigma^4 \frac{\mu^4 r^4}{N_0^3} + \calV_{\calG}^{-2} \sigma^4 \mu^4 r^4 \bar{T}^{-3}.$$ Then, for any $q >0$, we have by Cauchy-Schwarz and Markov inequalities with Claim \ref{clm:asympvariance_new},
\begin{align*}
 \Var(A+B)^{-1} \sum_{j,s} \bbE[\calY_{js}^2 1_{\{|\calY_{js}| > q \Var(A+B)^{1/2} \}}] 
\leq \frac{1}{\Var(A+B) q} \sqrt{\sum_{j,s} \bbE[\calY_{js}^4 ]} \lesssim \frac{\mu^2 r^2  }{N_0^{\frac{1}{2}}} + \frac{\mu^2 r^2  N_0}{\bar{T}^{\frac{3}{2}}}.
\end{align*}
Because $\bar{T} \geq \min_l T_{d_l} \coloneqq T_{\min}$, we have $\frac{\mu^2 r^2 }{N_0^{1/2}} + \frac{\mu^2 r^2 N_0}{\bar{T}^{3/2}} = o_p(1) $ by the Assumption (ii). Hence, the Lindeberg condition is satisfied.

\begin{claim}\label{clm:asympvariance_new}
(i) $\calV_{\calG}^{-1} \lesssim  \frac{N_0}{ \sigma^2}$ and (ii) $\Var(A+B) \conP 1$.
\end{claim}

\noindent Therefore, by using the Lindeberg CLT with Claim \ref{clm:asympvariance_new} (ii), we have $A+B \conD \calN(0,1)$. Next, we show that $\norm{\frac{\calV_{\calG}^{-\frac{1}{2}}}{|\calG|}\sum_{i \in \calG} \calR_{l(i),i}^{M}} = o_p(1)$. Since $\calV_{\calG}^{-\frac{1}{2}} \lesssim \frac{\sqrt{N_0} }{ \sigma} \asymp \frac{\sqrt{N_l}}{ \sigma} $ for all $l$, we have by Proposition \ref{pro:decomposition} with probability at least $1-O(\min\{N_0^{-7},T_{\min}^{-7}\})$ that
\begin{align*}
\norm{\frac{\calV_{\calG}^{-\frac{1}{2}}}{|\calG|}\sum_{i \in \calG} \calR_{l(i),i}^{M}} 
&\leq\calV_{\calG}^{-\frac{1}{2}} \max_{0 \leq l \leq L}\max_{i\in \calG_l}||\calR_{l,i}^{M}|| \\
& \leq C_M' 
\left(
\max_{0 \leq l \leq L} \frac{\sigma \kappa_l^5 \mu_l r \sqrt{N_l} \max\{N_l \log N_l,T_l \log T_l \}}{\psi_{\min,l}\min\{N_l , T_l \}} \right.\\
& \qquad \qquad  +\max_{0 \leq l \leq L} \frac{ \kappa_l^4 \mu_l^2 r^2 \sqrt{N_l} \max \{ \sqrt{N_l \log N_l}, \sqrt{T_l \log T_l} \}}{ \min \{N_l^{\frac{3}{2}}, T_l^{\frac{3}{2}} \}} \\
&\left. \qquad\qquad +   \max_{1 \leq l \leq L} \frac{ \mu_l^2 r^2 \kappa_l^3 |\calG_l|  \max \{\sqrt{N_l \log N_l}, \sqrt{T_l \log T_l} \}}{ \sqrt{N_l} \min \{N_l, T_l \}}
\right)
\end{align*}
for an absolute constant $C_M'  >0$. Then, by Assumptions (i), (ii), and (iii) with the relations that $\mu_l \lesssim \mu \kappa^\frac{1}{2}$, $\kappa_l \lesssim \kappa$, $N_0 \leq N_l \leq 2N_0$ and $T_l = T_{d_l} + 1$, we have  $\norm{\frac{\calV_{\calG}^{-\frac{1}{2}}}{|\calG|}\sum_{i \in \calG} \calR_{l(i),i}^{M}} = o_p(1)$. Therefore,
$$\calV_{\calG}^{-\frac{1}{2}}\frac{1}{|\calG|}\sum_{i \in \calG} ( \widehat{m}_{it_o} - m_{{it_o}} ) \conD \calN(0,1). \ \ \square $$
\bigskip

\begin{proof}[Proof of Claim \ref{clm:asympvariance_new}]
(i) We have $\calV_{\calG}^{-1} \lesssim  \frac{ N_0 }{ \sigma^2}$, because for some constant $c>0$,
\begin{align*}
\calV_{\calG}  
&\geq \sigma^2 \bar{u}_\calG^\top  \left( \sum_{j \leq N_0} u_j u_j^\top  \right)^{-1}\bar{u}_\calG \geq \sigma^2 \norm{\bar{u}_\calG}^2 \psi_{\min}\left( \left( \sum_{j \leq N_0} u_j u_j^\top  \right)^{-1}\right) \geq c \frac{\sigma^2}{N_0}. 
\end{align*}

\noindent (ii) A simple calculation shows that
\begin{align*}
\Var(A) =\calV_{\calG}^{-1}\sigma^2\bar{u}_\calG^\top  \left( \sum_{j \leq N_0} u_j u_j^\top  \right)^{-1}\bar{u}_\calG, \ \ 
\Var(B) =  \calV_{\calG}^{-1}  \frac{\sigma^2}{|\calG|} \sum_{0\leq l \leq L} \alpha_l v_{t_o}^\top \left( \sum_{s \leq T_{d_l}} v_s v_s^\top  \right)^{-1}v_{t_o}
\end{align*}
where $\alpha_l = \frac{|\calG_l|}{|\calG|}$. Hence, we have $\Var(A)+ \Var(B) = 1$. In addition, note that $\Cov(A,B) = \Cov(A,B^{(t_o)})$ where $$B^{(t_o)} \coloneqq \frac{\calV_{\calG}^{-\frac{1}{2}}}{|\calG|} \sum_{j \in \calG_0} v_{t_o}^\top  \left( \sum_{s \leq T_{0} }  v_{s} v_{s}^\top  \right)^{-1} \epsilon_{jt_o} v_{t_o}.$$ Then, we have
\begin{align*}
\norm{\Cov(A, B^{(t_o)})} &= \norm{\calV_{\calG}^{-1} \sigma^2 
v_{t_o}^\top  \left( \sum_{s \leq T_0} v_s v_s^\top  \right)^{-1} v_{t_o}   \bar{u}_{G}^\top \left( \sum_{j \leq N_0} u_j u_j^\top  \right)^{-1} \frac{1}{|\calG|} \sum_{j \in  \calG_0} u_j}\\
& \leq \calV_{\calG}^{-1} \sigma^2 \max_{s}\norm{v_s}^2 \max_j \norm{u_j}^2 \norm{\left( \sum_{s \leq T_0} v_s v_s^\top  \right)^{-1} } 
\norm{\left( \sum_{j \leq N_0} u_j u_j^\top  \right)^{-1} } \\
& \leq  \calV_{\calG}^{-1} \sigma^2 \frac{ \mu^2 r^2}{N_0 T_0} \conP 0.
\end{align*}
Hence, we have $\Var(A+B) =\Var(A) + \Var(B) + 2 \Cov(A,B) \conP 1$.    
\end{proof}
\bigskip

\noindent\textbf{Proof of Corollary \ref{coro:feasibleclt}}

From the proof of Claim \ref{clm:asympvariance_new} (ii), we know that $\calV_{\calG} = \Var( \tilde{A} ) + \Var( \tilde{B} )$
where $\widetilde{A} = \calV_{\calG}^{\frac{1}{2}} A$ and $\widetilde{B} = \calV_{\calG}^{\frac{1}{2}} B$. Note that
\begin{align*}
 \widetilde{A}
&= \sum_{j \leq N_0} \epsilon_{jt_o}  \left( \frac{1}{|\calG|}\sum_{i \in \calG} u_{i}^\top  \left( \sum_{k \leq N_0 }  u_{k} u_{k}^\top  \right)^{-1} u_{j}  \right)= \sum_{j \leq N_0} \epsilon_{jt_o} \left( \sum_{0 \leq l \leq L} \frac{|\calG_l|}{|\calG|}\frac{1}{|\calG_l|} \sum_{i \in \calG_l} u_{i}^\top  \left( \sum_{k \leq N_0 }   u_{k} u_{k}^\top  \right)^{-1} u_{j}  \right).
\end{align*}
Hence, we have
\begin{align*}
 \Var( \tilde{A} ) &= \sigma^2 \sum_{j \leq N_0} \left(\sum_{0 \leq l \leq L} \alpha_l  \bar{u}_{\calG_l}^\top  \left( \sum_{k \leq N_0}   u_{k} u_{k}^\top  \right)^{-1} u_j
 \right)^2 = \sigma^2 \sum_{j \leq N_0} \left(\sum_{0 \leq l \leq L} \alpha_l  \bar{X}_{l,\calG_l}^\top  \left( \sum_{k \leq N_0 }  X_{l,k} X_{l,k}^\top  \right)^{-1} X_{l,j}
 \right)^2,
\end{align*}
where $\bar{X}_{l,\calG_l} =\frac{1}{|\calG_l|} \sum_{i \in \calG_l} X_{l,i} $. In addition, as noted in the proof of Claim \ref{clm:asympvariance_new} (ii), we have
$$
\Var( \tilde{B} ) =  \frac{\sigma^2}{|\calG|} \sum_{0\leq l \leq L} \alpha_l Z_{l,t_o}^\top  \left( \sum_{s \in T_{d_l} } Z_{l,s} Z_{l,s}^\top  \right)^{-1}  Z_{l,t_o}.
$$
First, we show that $$\calV_\calG^{-1} \norm{ \widehat{\Var}( \tilde{A} ) - \Var( \tilde{A} ) } = o_p(1)$$ where 
$$
\widehat{\Var}( \tilde{A} ) = \widehat{\sigma}^2 \sum_{j \leq N_0} \left( \sum_{0 \leq l \leq L} \alpha_l  \widehat{\bar{X}}_{l,\calG_l}^\top  \left( \sum_{k \leq N_0 }   \widehat{X}_{l,k} \widehat{X}_{l,k}^\top  \right)^{-1} \widehat{X}_{l,j} \right)^2.
$$
Note that
\begin{align*}
\norm{ \widehat{\Var}( \tilde{A} ) - \Var( \tilde{A}) } 
&\lesssim \abs{\widehat{\sigma}^2 - \sigma^2} \sum_{j \leq N_0} \left(\sum_{0 \leq l \leq L} \alpha_l  \bar{X}_{l,\calG_l}^\top  \left( \sum_{k \leq N_0 }  X_{l,k} X_{l,k}^\top  \right)^{-1} X_{l,j}
 \right)^2\\
&\ \  + \sigma^2 \sum_{j \leq N_0} \norm{ \sum_{0 \leq l \leq L} \alpha_l  \bar{X}_{l,\calG_l}^\top  \left( \sum_{k \leq N_0 }  X_{l,k} X_{l,k}^\top  \right)^{-1} X_{l,j} } \\
& \qquad \times \norm{ \sum_{0 \leq l \leq L} \alpha_l \left( \widehat{\bar{X}}_{l,\calG_l}^\top  \left( \sum_{k \leq N_0}  \widehat{X}_{l,k} \widehat{X}_{l,k}^\top  \right)^{-1} \widehat{X}_{l,j} -  \bar{X}_{l,\calG_l}^\top  \left( \sum_{k \leq N_0 } X_{l,k} X_{l,k}^\top  \right)^{-1} X_{l,j}  \right) }. 
\end{align*}
Because
$$
\norm{\sum_{0 \leq l \leq L} \alpha_l  \bar{X}_{l,\calG_l}^\top  \left( \sum_{k \leq N_0 } X_{l,k} X_{l,k}^\top  \right)^{-1} X_{l,j}}
\leq \max_l \norm{\bar{u}_{\calG_l}^\top  \left( \sum_{k \leq N_0 } u_{k} u_{k}^\top  \right)^{-1} u_{j}} \leq \frac{\mu r }{N_0},
$$
we know by Claims \ref{clm:asympvariance_new} and \ref{clm:sigmaestimate} that
\begin{align*}
&\calV_\calG^{-1} \abs{\widehat{\sigma}^2 - \sigma^2} \sum_{j \leq N_0} \left(\sum_{0 \leq l \leq L} \alpha_l  \bar{X}_{l,\calG_l}^\top  \left( \sum_{k \leq N_0 }  X_{l,k} X_{l,k}^\top  \right)^{-1} X_{l,j}
 \right)^2 \\
& \quad \lesssim  \frac{\kappa^{5/2} \mu^3 r^2 \max\{\sqrt{N_{0} \log N_{0}} ,\sqrt{T_0 \log T_0} \} }{\min\{ N_{0}, T_0\}}
 = o_p(1).
\end{align*}

\begin{claim}\label{clm:sigmaestimate}
$\abs{\widehat{\sigma}^2 - \sigma^2} \lesssim \sigma^2 \frac{\kappa^{5/2} \mu r \max\{\sqrt{N_{0} \log N_{0}} ,\sqrt{T_0 \log T_0} \} }{\min\{ N_0, T_0\}}.$
\end{claim}

Next, we want to bound the following term: 
\begin{align*}
&\norm{ \sum_{0 \leq l \leq L} \alpha_l \left( \widehat{\bar{X}}_{l,\calG_l}^\top  \left( \sum_{k \leq N_0}  \widehat{X}_{l,k} \widehat{X}_{l,k}^\top  \right)^{-1} \widehat{X}_{l,j} -  \bar{X}_{l,\calG_l}^\top  \left( \sum_{k \leq N_0 }  X_{l,k} X_{l,k}^\top  \right)^{-1} X_{l,j}  \right) } \\
&\ \ \leq \max_l \norm{ \widehat{\bar{X}}_{l,\calG_l}^\top  \left( \sum_{k \leq N_0 }   \widehat{X}_{l,k} \widehat{X}_{l,k}^\top  \right)^{-1} \widehat{X}_{l,j} -  \bar{X}_{l,\calG_l}^\top  \left( \sum_{k \leq N_0 }   X_{l,k} X_{l,k}^\top  \right)^{-1} X_{l,j}}\\
&\ \  \leq \max_l \norm{ \widehat{X}_{l} \widehat{H}_{l} - X_l}_{2,\infty}
 \norm{   \left( \sum_{k \leq N_0 }  X_{l,k} X_{l,k}^\top  \right)^{-1} X_{l,j} } \\
& \quad + \max_l \norm{ X_l}_{2,\infty}^2 \norm{  \left( \sum_{k \leq N_0}   \widehat{H}_{l}^\top  \widehat{X}_{l,k} \widehat{X}_{l,k}^\top  \widehat{H}_{l} \right)^{-1}  -  \left( \sum_{k \leq N_0 }   X_{l,k} X_{l,k}^\top  \right)^{-1} }.
\end{align*}
As noted in the proof of Proposition \ref{pro:decomposition}, we have $$\norm{ \widehat{X}_{l} \widehat{H}_{l} - X_l}_{2,\infty}  \norm{X_l}_{2,\infty} \lesssim \sigma\frac{\kappa_{l}^{2} \mu_l r \max\{\sqrt{N_{l} \log N_{l}} ,\sqrt{T_{l} \log T_{l}} \} }{\min\{ N_{l}, T_{l}\}}.$$ In addition, because
$$
\norm{   \left( \sum_{k \leq N_0 } u_{l,k} u_{l,k}^\top  \right)^{-1}} \leq 
\norm{   \left( \sum_{k \leq N_0}  u_{l,k} u_{l,k}^\top  \right)^{-1} - I_r} + 1  = \norm{   \left( \sum_{k \leq N_0 }  u_{l,k} u_{l,k}^\top  \right)^{-1} -  (U_{l}^\top  U_{l})^{-1}} + 1
$$
and 
$$
 \norm{   \left( \sum_{k \leq N_0 } u_{l,k} u_{l,k}^\top  \right)^{-1} - (U_{l}^\top  U_{l})^{-1}}
 \lesssim \norm{  \sum_{k \in \calG_{l} } u_{l,k} u_{l,k}^\top } \leq |\calG_l| \frac{\mu_l r}{N_l} \ll 1,
$$
we have
$$
 \norm{   \left( \sum_{k \leq N_0 }  X_{l,k} X_{l,k}^\top  \right)^{-1}} \leq \psi_{\min,l}^{-1} \norm{   \left( \sum_{k \leq N_0 } u_{l,k} u_{l,k}^\top  \right)^{-1}} \lesssim \psi_{\min,l}^{-1}.
$$
Here, $U_l$ is the left singular vector of $M_l$ and $u_{l,k}^\top $ is the $k$-th row of it. Hence, we obtain
\begin{align*}
&\norm{  \left( \sum_{k \leq N_0 } \widehat{H}_{l}^\top  \widehat{X}_{l,k} \widehat{X}_{l,k}^\top  \widehat{H}_{l} \right)^{-1}  -  \left(  \sum_{k \leq N_0 }    X_{l,k} X_{l,k}^\top  \right)^{-1} }\\
&\ \ \lesssim
\norm{    \sum_{k \leq N_0 } \widehat{H}_{l}^\top  \widehat{X}_{l,k} \widehat{X}_{l,k}^\top  \widehat{H}_{l}   -  \sum_{k \leq N_0 } X_{l,k} X_{l,k}^\top   } \norm{ \left( \sum_{k \leq N_0 }  X_{l,k} X_{l,k}^\top  \right)^{-1} }^2 \\
&\ \ \lesssim \sigma\frac{\kappa_{l}^{2} \mu_l r  N_0 \max\{\sqrt{N_{l} \log N_{l}} ,\sqrt{T_{l} \log T_{l}} \} }{\psi_{\min,l}^2  \min\{ N_{l}, T_{l}\}},
\end{align*}
and
\begin{align*}
&\norm{ \sum_{0 \leq l \leq L} \alpha_l \left( \widehat{\bar{X}}_{l,\calG_l}^\top  \left( \sum_{k \leq N_0}  \widehat{X}_{l,k} \widehat{X}_{l,k}^\top  \right)^{-1} \widehat{X}_{l,j} -  \bar{X}_{l,\calG_l}^\top  \left( \sum_{k \leq N_0 }   X_{l,k} X_{l,k}^\top  \right)^{-1} X_{l,j}  \right) } \\ 
&\ \ \lesssim  \max_{0 \leq l \leq L} \sigma\frac{\kappa_{l}^{3} \mu_l^2 r^2  N_0 \max\{\sqrt{N_{l} \log N_{l}} ,\sqrt{T_{l} \log T_{l}} \} }{N_l \min\{ N_{l}, T_{l}\}\psi_{\min,l}}.
\end{align*}
Then, we have
\begin{align*}
& \calV_\calG^{-1} \sigma^2 \sum_{j \leq N_0} \norm{ \sum_{0 \leq l \leq L} \alpha_l  \bar{X}_{l,\calG_l}^\top  \left( \sum_{k \leq N_0 }  X_{l,k} X_{l,k}^\top  \right)^{-1} X_{l,j} } \\
& \qquad \qquad\qquad \times \norm{ \sum_{0 \leq l \leq L} \alpha_l \left( \widehat{\bar{X}}_{l,\calG_l}^\top  \left( \sum_{k \leq N_0 }  \widehat{X}_{l,k} \widehat{X}_{l,k}^\top  \right)^{-1} \widehat{X}_{l,j} -  \bar{X}_{l,\calG_l}^\top  \left( \sum_{k \leq N_0 }  X_{l,k} X_{l,k}^\top  \right)^{-1} X_{l,j}  \right) }   \\
& \ \  \lesssim  \max_{0 \leq l \leq L} \frac{\sigma}{\psi_{\min,l}}\frac{\kappa_{l}^{3}  \mu_l^3 r^3  N_l \max\{\sqrt{N_{l} \log N_{l}} ,\sqrt{T_{l} \log T_{l}} \} }{\min\{ N_{l}, T_{l}\}} = o_p(1)
\end{align*}
by Assumptions (i), (ii), and (iii) with the relations that $\mu_l \lesssim \mu \kappa^\frac{1}{2}$, $\kappa_l \lesssim \kappa$, $N_0 \leq N_l \leq 2N_0$ and $T_l = T_{d_l} + 1$. In the same token, we can also show that $$\calV_\calG^{-1} \norm{ \widehat{\Var}( \tilde{B} ) - \Var( \tilde{B}  ) } = o_p(1)$$ where $\widehat{\Var}( \tilde{B} ) = \frac{\widehat{\sigma}^2}{|\calG|} \sum_{0\leq l \leq L} \alpha_l \widehat{Z}_{l,t_o}^\top  \left( \sum_{s \in T_{d_l} } \widehat{Z}_{l,s} \widehat{Z}_{l,s}^\top  \right)^{-1}  \widehat{Z}_{l,t_o}$. Then, we have $\frac{\widehat{\calV}_\calG - \calV_\calG }{\calV_\calG} = o_p(1)$ and it implies that $\frac{\calV_\calG}{\widehat{\calV}_\calG} \conP 1$. Then, by the Slutsky's theorem with Theorem \ref{thm:groupclt_split_staggered_adoption}, we have the desired result. $\square$
\bigskip

\begin{proof}[Proof of Claim \ref{clm:sigmaestimate}]
Note that $$\abs{\widehat{\sigma}^2 - \sigma^2}  \leq \abs{\frac{1}{N_0 T_0}\sum_{i \leq N_0, t_ \leq T_0} \widehat{\epsilon}_{it}^2 - \epsilon_{it}^2} + \abs{\frac{1}{N_0 T_0} \sum_{i \leq N_0, t_ \leq T_0}  \epsilon_{it}^2 - \sigma^2}$$ where $\widehat{\epsilon}_{it}^2 =  \epsilon_{it}^2  + (m_{it} - \widehat{m}_{it})^2 + 2 \epsilon_{it} (m_{it} - \widehat{m}_{it})$. As noted in the proof of Proposition \ref{pro:decomposition}, we have
\begin{align*}
\max_{i \leq N_0, t_ \leq T_0} \norm{\widehat{m}_{it} - m_{it}} 
&\leq \max_{i \leq N_0, t_ \leq T_0} \norm{\widehat{m}_{it} - \widehat{X}_{0,i}^\top \widehat{Z}_{0,t}} + \max_{i \leq N_0, t_ \leq T_0} \norm{ \widehat{X}_{0,i}^\top \widehat{Z}_{0,t} - m_{it}}\\
&\lesssim \norm{ \widehat{X}_{0}\widehat{H}_{0} - X_{0}}_{2,\infty} \norm{  Z_{0}}_{2,\infty} + \norm{ \widehat{Z}_{0}\widehat{H}_{0} - Z_{0}}_{2,\infty} \norm{  X_{0}}_{2,\infty}\\
&\lesssim \sigma\frac{\kappa_{0}^{2} \mu_0 r \max\{\sqrt{N_{0} \log N_{0}} ,\sqrt{T_0 \log T_0} \} }{\min\{ N_{0}, T_0 \}}.
\end{align*}
Hence, we get
$$
\abs{\frac{1}{N_0 T_0}\sum_{i \leq N_0, t_ \leq T_0} \widehat{\epsilon}_{it}^2 - \epsilon_{it}^2} \lesssim \sigma \max_{i \leq N_0, t_ \leq T_0}  \norm{\widehat{m}_{it} - m_{it}} \lesssim \sigma^2 \frac{\kappa_{0}^{2} \mu_0 r \max\{\sqrt{N_{0} \log N_{0}} ,\sqrt{T_0 \log T_{0}} \} }{\min\{ N_{0}, T_0 \}}.
$$
Moreover, we have by concentration inequalities that 
$$
\abs{\frac{1}{N_0 T_0} \sum_{i \leq N_0, t_ \leq T_0}  \epsilon_{it}^2 - \sigma^2}
 \lesssim \sigma^2\left( N_0 T_0\right)^{-\frac{1}{2}}\log(N_0 T_0)^{1/2}.
 $$
Since the first term dominates the second term using the relations that $\mu_0 \lesssim \mu \kappa^\frac{1}{2}$, $\kappa_0 \lesssim \kappa$, we have the desired result.
\end{proof}
\bigskip

\subsection{Relations about eigenvalue and eigenvector between the full matrix and its submatrix}

Lastly, we present one lemma which shows the relations about eigenvalue and eigenvector between the full matrix and its submatrix.

\begin{lemma}\label{lem:eigenrelation}
(i) Let $M = (m_{it})_{1\leq i \leq N, 1 \leq t \leq T}$ be a $N \times T$ matrix of rank $r$ and $M_o = (m_{it})_{i \in \calI_o, t \in \calT_o}$ be a submatrix of $M$ where $|\calI_o| = N_o$ and $|\calT_o| = T_o$. The SVD of $M$ is $UDV^\top $, and the $i$-th row of $U$ is $u_i^\top $ and the $t$-th row of $V$ is $v_t^\top $. In addition, $\mu$, $\kappa$ denote the incoherence parameter and the condition number of $M$, and $\mu_o,\kappa_o$ denote those of $M_o$. If there are constants $C,c > 0 $ such that
\begin{align*}
&c \leq \psi_r\left(\frac{1}{N_o} \sum_{i \in \calI_o} \left(\sqrt{N}u_i \right) \left(\sqrt{N}u_i \right)^\top  \right) \leq \psi_1\left(\frac{1}{N_o} \sum_{i \in \calI_o} \left(\sqrt{N}u_i \right) \left(\sqrt{N}u_i \right)^\top  \right) \leq C ,\\
&c \leq \psi_r\left(\frac{1}{T_o} \sum_{t \in \calT_o} \left(\sqrt{T}v_t \right) \left(\sqrt{T} v_t \right)^\top  \right) \leq \psi_1\left(\frac{1}{T_o} \sum_{t \in \calT_o} \left(\sqrt{T}v_t \right) \left(\sqrt{T}v_t \right)^\top  \right) \leq C,
\end{align*}
we have $\mu_o \lesssim \mu \kappa^{1/2}$ and $\kappa_o \lesssim \kappa$.\\
\noindent (ii) Let $M_1 = (m_{it})_{i \in \calI_1, t \in \calT_1}$ and $M_2 = (m_{it})_{i \in \calI_2, t \in \calT_2}$ be submatrices of $M$ where $|\calI_1| = N_1$, $|\calI_2| = N_2$, $|\calT_1| = T_1$, and $|\calT_2| = T_2$. If 
there are constants $C,c > 0 $ such that for all $l \in \{1,2\}$,
\begin{align*}
&c \leq \psi_r\left(\frac{1}{N_l} \sum_{i \in \calI_l} \left(\sqrt{N}u_i \right) \left(\sqrt{N}u_i \right)^\top  \right) \leq \psi_1\left(\frac{1}{N_l} \sum_{i \in \calI_l} \left(\sqrt{N}u_i \right) \left(\sqrt{N}u_i \right)^\top  \right) \leq C ,\\
&c \leq \psi_r\left(\frac{1}{T_l} \sum_{t \in \calT_l} \left(\sqrt{T}v_t \right) \left(\sqrt{T} v_t \right)^\top  \right) \leq \psi_1\left(\frac{1}{T_l} \sum_{t \in \calT_l} \left(\sqrt{T}v_t \right) \left(\sqrt{T}v_t \right)^\top  \right) \leq C,
\end{align*}
we have $\frac{\sqrt{N_1 T_1}}{\psi_{1,\min}} \asymp \frac{\sqrt{N_2 T_2}}{\psi_{2,\min}}$ where $\psi_{l,\min}$ is the smallest singular value of $M_l$.
\end{lemma}

\begin{proof}[Proof of Lemma \ref{lem:eigenrelation}]
(i) Without loss of generality, assume that $\calI_o = \{1,\cdots,N_o\}$ and $\calT_o = \{1,\cdots,T_o\}$. Let the SVD of $M_o$ be $U_o D_o V_o^\top $. Then, we can say $$M_{it} = u_i^\top  D v_t = u_{o,i}^\top  D_o v_{o,t}$$ for $i \leq N_o$ and $t \leq T_o$. In addition, let $B_{sub} = U_{sub}D^{1/2}$ where $U_{sub} = [u_1,\dots,u_{N_o}]^\top $ and $F_{sub} = V_{sub}D^{1/2}$ where $V_{sub} = [v_1,\dots,v_{T_o}]^\top $. Then, we have $M_o = B_{sub} F_{sub}^\top $. Define
\begin{align*}
L^* &= \left( B_{sub}^\top B_{sub} \right)^{1/2} \left( F_{sub}^\top F_{sub} \right)\left( B_{sub}^\top B_{sub} \right)^{1/2} \\
&= D^{1/4}\left( U_{sub}^\top U_{sub} \right)^{1/2}D^{1/4}D^{1/2}\left( V_{sub}^\top V_{sub} \right)D^{1/2}D^{1/4}\left( U_{sub}^\top U_{sub} \right)^{1/2}D^{1/4}.
\end{align*}
Let $G_{L^*}$ be a $K \times K$ matrix whose columns are the eigenvectors of $L^*$ such that $\Lambda_{L^*} = G_{L^*}^\top L^* G_{L^*}$ is a descending order diagonal matrix of the eigenvalues of $L^*$. Define $$H_u = \left( B_{sub}^\top B_{sub} \right)^{-1/2} G_{L^*} = D^{-1/4}\left( U_{sub}^\top U_{sub} \right)^{-1/2}D^{-1/4} G_{L^*}.$$ Note that
\begin{align*}
\left(B_{sub} F_{sub}^\top F_{sub}B_{sub}^\top \right) B_{sub} H_u  &= B_{sub} \left( B_{sub}^\top B_{sub} \right)^{-1/2} \left( B_{sub}^\top B_{sub} \right)^{1/2}  \left( F_{sub}^\top F_{sub} \right)\left( B_{sub}^\top B_{sub} \right)^{1/2}\left( B_{sub}^\top B_{sub} \right)^{1/2} H_u\\
& = B_{sub} \left( B_{sub}^\top B_{sub} \right)^{-1/2} L^* G_{L^*}\\
&= B_{sub} \left( B_{sub}^\top B_{sub} \right)^{-1/2} G_{L^*} \Lambda_{L^*} \\
&= B_{sub} H_u \Lambda_{L^*}.
\end{align*}
In addition, we have
\begin{align*}
(B_{sub} H_u)^\top B_{sub} H_u &= H_u^\top B_{sub}^\top B_{sub} H_u 
= G_{L^*}^\top  \left( B_{sub}^\top B_{sub} \right)^{-1/2}B_{sub}^\top B_{sub}\left( B_{sub}^\top B_{sub} \right)^{-1/2}G_{L^*} = I_r.
\end{align*}
Hence, the column of $B_{sub} H_u$ are the eigenvector of $\left(B_{sub} F_{sub}^\top F_{sub}B_{sub}^\top \right) = M_oM_o^\top $ corresponding to the eigenvalue $\Lambda_{L^*}$. Hence, $B_{sub} H_u$ is the left singular vector of $M_o$, that is, $U_o$. Then, since
\begin{align}\label{eq:relation}
  U_o = B_{sub} H_u = U_{sub} D^{1/2} D^{-1/4} \left( U_{sub}^\top U_{sub} \right)^{-1/2} D^{-1/4} G_{L^*}= U_{sub} D^{1/4} \left( U_{sub}^\top U_{sub} \right)^{-1/2} D^{-1/4} G_{L^*}, 
\end{align}
we have the following incoherence condition for the submatrix:
\begin{align*}
 \max_i\norm{u_{o,i}} = \max_i\norm{e_i^\top U_o} 
\leq \max_i\norm{e_i^\top U_{sub}} \norm{ D^{1/4}} \norm{ D^{-1/4}} \norm{\left( U_{sub}^\top U_{sub} \right)^{-1/2}}
\leq \frac{\mu_o^{1/2}r^{1/2}}{\sqrt{N_o}} 
\end{align*}
where $\mu_o = C\mu \kappa^{1/2}$ for some constant $C>0$. Similarly, we can have $\max_t\norm{v_{o,t}} \leq \frac{\mu_o^{1/2}r^{1/2}}{\sqrt{T_o}}$ where $\mu_o = C \mu \kappa^{1/2}$ for some constant $C>0$. Hence, the incoherence parameter for the submatrix $M_o$ is $C\mu \kappa^{1/2}$ for some constant $C>0$.\\
Note that $$M_o = U_o D_o V_o^\top  = U_{sub} D V_{sub}^\top   \Longrightarrow D_o = U_o^\top U_{sub} D V_{sub}^\top  V_o.$$ Then, by using the relation \eqref{eq:relation}, we have
\begin{align*}
D_o &= U_o^\top (U_o G_{L^*}^\top  D^{1/4}\left( U_{sub}^\top U_{sub} \right)^{1/2}D^{-1/4}) D (D^{-1/4} \left( V_{sub}^\top V_{sub} \right)^{1/2} D^{1/4} G_{R^*} V_o^\top )V_o\\
&= G_{L^*}^\top  D^{1/4}\left( U_{sub}^\top U_{sub} \right)^{1/2} D^{1/2} \left( V_{sub}^\top V_{sub} \right)^{1/2} D^{1/4} G_{R^*},
\end{align*}
where $G_{R^*}$ is a $K \times K$ eigenvector matrix of $R^* = \left( F_{sub}^\top F_{sub} \right)^{1/2} \left( B_{sub}^\top B_{sub} \right)\left( F_{sub}^\top F_{sub} \right)^{1/2}$. Then, we have
\begin{align}\label{eq:relation_eigenvalue}
&\psi_1(D_o) \leq \norm{D^{1/4}}^2  \norm{D^{1/2}} \norm{\left( U_{sub}^\top U_{sub} \right)^{1/2}}\norm{\left( V_{sub}^\top V_{sub} \right)^{1/2}} \lesssim \psi_{1}(D)\frac{\sqrt{N_oT_o}}{\sqrt{NT}},\\
\nonumber&\psi_r(D_o) \geq \lambda^2_{\min}(D^{1/4}) \lambda_{\min}(D^{1/2}) \lambda_{\min}\left( \left( U_{sub}^\top U_{sub} \right)^{1/2} \right)\lambda_{\min}\left( \left( V_{sub}^\top V_{sub} \right)^{1/2} \right) \gtrsim \psi_{r}(D) \frac{\sqrt{N_oT_o}}{\sqrt{NT}}.
\end{align}
So, the condition number of the submatrix can be bounded like $\kappa_o = \frac{\psi_1(D_o)}{\psi_r(D_o)} \lesssim \frac{\psi_1(D)}{\psi_r(D)} = \kappa$.\\
(ii) By using the relation \eqref{eq:relation_eigenvalue} with the fact that $\kappa_2 \lesssim \kappa$, we know
$$
\psi_{1,\min}^{-1} \lesssim \psi_{\min}^{-1} \frac{\sqrt{NT}}{\sqrt{N_1T_1}} = \kappa^{-1}\psi_{\max}^{-1} \frac{\sqrt{NT}}{\sqrt{N_1T_1}} \lesssim \kappa^{-1}\psi_{2,\max}^{-1} \frac{\sqrt{N_2T_2}}{\sqrt{N_1T_1}} \lesssim \kappa_2^{-1}\psi_{2,\max}^{-1} \frac{\sqrt{N_2T_2}}{\sqrt{N_1T_1}} = \psi_{2,\min}^{-1} \frac{\sqrt{N_2T_2}}{\sqrt{N_1T_1}}.
$$
Similarly, we can show $\psi_{2,\min}^{-1} \lesssim \psi_{1,\min}^{-1} \frac{\sqrt{N_1T_1}}{\sqrt{N_2T_2}}$. Hence, we have that $\frac{\sqrt{N_1T_1}}{\psi_{1,\min}} \asymp \frac{\sqrt{N_2T_2}}{\psi_{2,\min}}$.

\end{proof}

\section{Formal inferential theory for the treatment effect estimation in Section \ref{sec:ticksizepilot}}\label{sec:treatment_CLT}

This section provides the formal inferential theory for the group averaged treatment effects, $\mu_{t_0}^{(d)}$ and $\theta_{t_0}^{(d)}$ in Section \ref{sec:ticksizepilot}. The assumption on the noise is the same as that in Section \ref{sec:convergencerate}, and the singular vectors of $M$ are incoherent in that there is a $\mu \geq 1$ such that $||U||_{2,\infty} \leq \sqrt{{\mu r}/{N}}$, $||V||_{2,\infty} \leq \sqrt{{\mu r}/{(T + 3 T_1)}}$.

Denote by $M_{O_{(d)}} = (m^{(0)}_{it})_{i \in \calI_d, t \leq T_0}$, and the smallest nonzero singular value of it by $\psi_{\min,O_{(d)}}$. In addition, denote $\{\calG_{(d),l}\}_{0 \leq l \leq L_d}$ by the subgroups of $\calG$ for the estimation of $\{m^{(d)}_{it_0}\}_{i \in \calG}$. Then, we have the following asymptotic normality of the group averaged estimator.

\begin{theorem}\label{thm:groupclt_treatment}
 Assume that for any $0 \leq d \leq 3$ and $l = 1, \cdots,L_d$,
\begin{itemize} 
\item[(i)] $\sigma \kappa^\frac{23}{4} \mu^\frac{3}{2} r^\frac{3}{2} \sqrt{N_d} \max\{{N_d\sqrt{\log{N_d}}},{T_0\sqrt{\log{T_0}}}\}=o_p\left(\psi_{\min,O_{(d)}} \min\{N_d,T_0\}\right)$;
\item[(ii)] $\kappa^\frac{11}{2} \mu^{3} r^{3}\sqrt{N_d} \max\{{\sqrt{N_d\log^3{N_d}}},{\sqrt{T_0\log^3{T_0}}}\}=o_p\left( \min\{N_d^\frac{3}{2},T_0^\frac{3}{2}\}\right)$;
\item[(iii)] $|\calG_{(d),l}| \kappa^\frac{17}{4} \mu^\frac{5}{2} r^\frac{5}{2} \max\{\sqrt{N_d\log{N_d}},\sqrt{T_0\log{T_0}}\}=o_p\left(\sqrt{N_d} \min\{N_d,T_0\}\right)$;
\item[(iv)] There are constants $C,c > 0 $ such that
        \begin{align*}
             &c \leq \lambda_{\min} \left( \frac{N}{N_d} \sum_{i \in \calI_d} u_{i}u_{i}^\top \right) \leq  \lambda_{\max}\left( \frac{N}{N_d} \sum_{i \in \calI_d} u_{i}u_{i}^\top \right)  \leq C, \\
             &c \leq \lambda_{\min} \left( \frac{T_M}{T_0} \sum_{t \leq T_0 } v_{t}v_{t}^\top \right) \leq \lambda_{\max}\left( \frac{T_M}{T_0} \sum_{t \leq T_0 } v_{t}v_{t}^\top \right)   \leq C,
        \end{align*}
        where $T_M = T + 3T_1$ is the number of columns of $M$;
\item[(v)] $\sqrt{N}\norm{\bar{u}_{\calG}} \geq c$ for some constant $c>0$ where $\bar{u}_{\calG} = |\calG|^{-1}\sum_{i \in \calG} u_{i}$.
\end{itemize}
Then, we have
	\begin{align*}
			\mathcal{V}_{\mu}^{-\frac{1}{2}}\left( \widehat{\mu}^{(d)}_{t_0} -   \mu^{(d)}_{t_0} \right) \overset{D}{\longrightarrow} \mathcal{N}(0,1),\ \ \mathcal{V}_{\theta}^{-\frac{1}{2}}\left( \widehat{\theta}^{(d)}_{t_0} -   \theta^{(d)}_{t_0} \right) \overset{D}{\longrightarrow} \mathcal{N}(0,1),
		\end{align*}
  $\mathcal{V}_{\mu} = \mathcal{V}_{\calG}(d,0)$ and $\mathcal{V}_{\theta} = \mathcal{V}_{\calG}(d,d-1)$ where
	\begin{align*}
	\mathcal{V}_{\mathcal{G}}(d,d') =&
		\sigma^2 \bar{u}_\calG^\top \left(\sum_{j \in \calI_d} u_j u_j^\top\right)^{-1}  \bar{u}_\calG +  \sigma^2 \bar{u}_\calG^\top \left(\sum_{j \in \calI_{d'}} u_j u_j^\top\right)^{-1}  \bar{u}_\calG\\
		 & \ \  + \frac{\sigma^2}{|\calG|} \left(v_{(d \cdot T_1+t_0)}-v_{(d' \cdot T_1+t_0)}\right)^\top \left(\sum_{s \leq T_0} v_{s} v_{s}^\top \right)^{-1} \left(v_{(d \cdot T_1+t_0)}-v_{(d' \cdot T_1+t_0)}\right).
	\end{align*} 
\end{theorem} 

For completeness, we provide the variance estimator. For each $0 \leq d \leq 3$ and $0 \leq l \leq L_d$, denote by $\left(\widehat{X}^{(d)}_l,\widehat{Z}^{(d)}_l \right)$ the debiased estimators derived from $\tilde{Y}^{(d)}_l$ which is the submatrix of $\tilde{Y}^{(d)}$ constructed for the estimation of $\{m^{(d)}_{it_0}\}_{i \in \calG_{(d),l}}$. In addition, $\widehat{X}^{(d)}_{l,j}$ denotes a row of $\widehat{X}^{(d)}_{l}$ which corresponds to the unit $j$ and $\widehat{Z}^{(d)}_{l,s}$ denotes a row of $\widehat{Z}^{(d)}_{l}$ which corresponds to the $s$-th column of $M$.

\begin{corollary}[Feasible CLT of Theorem \ref{thm:groupclt_treatment}] \label{coro:feasibletreatment}
Suppose the assumptions in Theorem \ref{thm:groupclt_treatment} hold. In addition, we have for all $0 \leq d \leq 3$, $\frac{\sigma}{\psi_{\min,O_{(d)}}} \frac{\kappa^{5} \mu^4 r^4 N_d \max\{\sqrt{N_d \log N_d} ,\sqrt{T_0 \log T_0} \} }{\min\{ N_d, T_0 \}} \conP 0$. Then,
\begin{align*}
			\widehat{\mathcal{V}}_{\mu}^{-\frac{1}{2}}\left( \widehat{\mu}^{(d)}_{t_0} -   \mu^{(d)}_{t_0} \right) \overset{D}{\longrightarrow} \mathcal{N}(0,1),\ \ \widehat{\mathcal{V}}_{\theta}^{-\frac{1}{2}}\left( \widehat{\theta}^{(d)}_{t_0} -   \theta^{(d)}_{t_0} \right) \overset{D}{\longrightarrow} \mathcal{N}(0,1),
		\end{align*}
  $\widehat{\mathcal{V}}_{\mu} = \widehat{\mathcal{V}}_{\calG}(d,0)$ and $\widehat{\mathcal{V}}_{\theta} = \widehat{\mathcal{V}}_{\calG}(d,d-1)$ where
\begin{align*}
 \widehat{\mathcal{V}}_{\mathcal{G}}(d,d')
	 & = \sum_{\delta \in \{d,d'\}} \widehat{\sigma}^2 \sum_{i \in \calI_\delta} \left(\sum_{0 \leq l \leq L_\delta} \alpha^{(\delta)}_l  \widehat{\bar{X}}_{\calG_{(\delta),l}}^{\top} \left( \sum_{j \in \calI_\delta}   \widehat{X}_{l,j}^{(\delta)} \widehat{X}_{l,j}^{(\delta)\top} \right)^{-1} \widehat{X}_{l,i}^{(\delta)}
 \right)^2\\
	 & \quad + \sum_{\delta \in \{d,d'\}} \frac{\widehat{\sigma}^2}{|\calG|}  \widehat{Z}_{0,(\delta \cdot T_1 +t_o)}^{(\delta)\top} \left( \sum_{s \leq T_0 } \widehat{Z}^{(\delta)}_{0,s} \widehat{Z}_{0,s}^{(\delta)\top} \right)^{-1}  \widehat{Z}_{0,(\delta \cdot T_1 +t_o)}^{(\delta)}\\
  & \quad - 2 \frac{\widehat{\sigma}^2}{|\calG|} \sum_{s \leq T_0} 
\left( \widehat{Z}_{0,(d \cdot T_1 +t_o)}^{(d)\top}   \left( \sum_{s \leq T_0 } \widehat{Z}_{0,s}^{(d)} \widehat{Z}_{0,s}^{(d)\top} \right)^{-1}  \widehat{Z}^{(d)}_{0,s} \right)
\left( \widehat{Z}_{0,s}^{(d')\top}   \left( \sum_{s \leq T_0 } \widehat{Z}_{0,s}^{(d')} \widehat{Z}_{0,s}^{(d')\top} \right)^{-1}  \widehat{Z}^{(d')}_{0,(d' \cdot T_1 +t_o)} \right),
\end{align*}
$\alpha_{l}^{(d)} = \frac{|\calG_{(d),l}|}{|\calG|}$, $\widehat{\sigma}^2 = \frac{1}{N T_0} \sum_{i \leq N, t \leq T_0} \widehat{\epsilon}_{it}^2$, $\widehat{\epsilon}_{it} = y_{it} - x_{it}^\top \beta - \widehat{m}^{(0)}_{it}$. In addition, $\widehat{\bar{X}}_{\calG_{(d),l}} = \frac{1}{|\calG_{(d),l}|} \sum_{i \in \calG_{(d),l}} \widehat{X}_{l,i}^{(d)} $.
\end{corollary}
\bigskip

\noindent\textbf{Proof of Theorem \ref{thm:groupclt_treatment}}

\noindent(i) Case 1 $(\widehat{\mu}_{t_0}^{(d)})$: Following the proof of Theorem \ref{thm:groupclt_split_block}, we have the decomposition:
\begin{align}\label{eq:treatmentdecomposition}
\nonumber&\frac{\calV_{\mu}^{-\frac{1}{2}}}{|\calG|}  \sum_{i \in \calG} \left( \widehat{m}_{it_o}^{(d)} - m_{it_o}^{(d)} \right)
 - \frac{\calV_{\mu}^{-\frac{1}{2}}}{|\calG|} \sum_{i \in \calG} \left( \widehat{m}_{it_o}^{(0)} - m_{it_o}^{(0)} \right)\\
\nonumber &= \underbrace{\calV_{\mu}^{-\frac{1}{2}}\bar{u}_\calG ^\top  \left(\sum_{j \in \calI_d} u_j u_j^\top  \right)^{-1} \sum_{j \in \calI_d} u_j \epsilon_{jt_o}}_{\coloneqq A^{(d)}}
 + \underbrace{\frac{\calV_{\mu}^{-\frac{1}{2}}}{|\calG|} \sum_{i \in \calG} v_{(d\cdot T_1 + t_o)}^\top  \left( \sum_{s \leq T_0}v_{s} v_{s}^\top  \right)^{-1} \sum_{s \leq T_0} v_{s} \epsilon_{is}}_{\coloneqq B^{(d)}}\\
&\ \ - \underbrace{\calV_{\mu}^{-\frac{1}{2}}\bar{u}_\calG ^\top  \left(\sum_{j \in \calI_0} u_j u_j^\top  \right)^{-1} \sum_{j \in \calI_0} u_j \epsilon_{jt_o}}_{\coloneqq A^{(0)}}
 - \underbrace{\frac{\calV_{\mu}^{-\frac{1}{2}}}{|\calG|} \sum_{i \in \calG} v_{t_o}^\top  \left( \sum_{s \leq T_0}v_{s} v_{s}^\top  \right)^{-1} \sum_{s \leq T_0} v_{s} \epsilon_{is}}_{\coloneqq B^{(0)}}
+ \calV_{\mu}^{-\frac{1}{2}} \calR
\end{align}
where $\calR$ is a residual term. First, we want to show the Lindeberg condition. Note that
\begin{align*}
& A^{(d)} + B^{(d)}
= \sum_{j \leq N}\sum_{s \leq T} \underbrace{\left( P 1_{\{j \in \calI_d, s = t_o \}}  +  \sum_{0\leq l\leq L_d} Q_{l}1_{\{j \in \calG_{(d),l}, s \leq T_0 \}}
\right)\epsilon_{js}}_{\coloneqq \calY^{(d)}_{js}},\\
&\text{where  }\ \ P = \calV_{\mu}^{-\frac{1}{2}} \bar{u}_\calG ^\top  \left(\sum_{j \in \calI_d} u_j u_j^\top  \right)^{-1} u_j,  \ \
Q_{l} = \frac{\calV_{\mu}^{-\frac{1}{2}}}{|\calG|} v_{(d\cdot \calI_d + t_o)}^{\top} \left( \sum_{s \leq T_0}v_{s} v_{s}^\top  \right)^{-1} \sum_{s \leq T_0} v_{s} .
\end{align*}
Using the same way in the proof of Theorem \ref{thm:groupclt_split_staggered_adoption}, we have 
$
\norm{P} \leq \calV_{\mu}^{-\frac{1}{2}} \frac{\mu r}{N_d}$ and
$\norm{Q_{l}} \leq \frac{\calV_{\mu}^{-\frac{1}{2}}}{|\calG|}  \frac{\mu r }{T_0}$.
Then, by the same token as the proof of Theorem \ref{thm:groupclt_split_staggered_adoption}, we have
$$
\sum_{j,s} \bbE[\calY_{js}^{(d)4} ] \lesssim \calV_{\mu}^{-2} \sigma^4 \frac{\mu^4 r^4 }{N_d^3} + \calV_{\mu}^{-2} \sigma^4 \mu^4 r^4 T_0^{ -3}.
$$ 
Similarly, we have
$$
\sum_{j,s} \bbE[\calY_{js}^{(0)4} ] \lesssim \calV_{\mu}^{-2} \sigma^4 \frac{\mu^4 r^4 }{N_0^3} + \calV_{\mu}^{-2} \sigma^4 \mu^4 r^4 T_0^{-3}
$$
where $A^{(0)} + B^{(0)} = \sum_{j \leq N}\sum_{s \leq T} \calY^{(0)}_{js}$. Then, for any $q >0$, we have by Cauchy-Schwarz and Markov inequalities with Claim \ref{clm:asympvariance_treatment},
\begin{align*}
&\Var(A+B)^{-1} \sum_{j,s} \bbE[\calY_{js}^2 1_{\{|\calY_{js}| > q \Var(A+B)^{1/2} \}}] \\
& \leq 2  \Var(A+B)^{-1} \left( \sum_{j,s} \bbE[\calY_{js}^{(d)2}  1_{\{|\calY_{js}| > q \Var(A+B)^{1/2} \}}] 
+ \sum_{j,s} \bbE[\calY_{js}^{(0)2}  1_{\{|\calY_{js}| > q \Var(A+B)^{1/2} \}}] \right)
 \\
&\leq \frac{2}{\Var(A+B) q} \sqrt{\sum_{j,s} \bbE[\calY_{js}^{(d)4} ]} + \frac{2}{\Var(A+B) q} \sqrt{\sum_{j,s} \bbE[\calY_{js}^{(0)4} ]}\\
& \lesssim \frac{\mu^3 r^3  }{N_d^{\frac{1}{2}}} + \frac{\mu^3 r^3  N_d}{T_0^{\frac{3}{2}}} + \frac{\mu^3 r^3  }{N^{\frac{1}{2}}_0} + \frac{\mu^3 r^3 N_0}{T_0^{\frac{3}{2}}},
\end{align*}
where $A = A^{(d)} - A^{(0)}$, $B = B^{(d)} - B^{(0)}$, and $\calY_{js} = \calY_{js}^{(d)} - \calY_{js}^{(0)}$. Because the last term is $o_p(1)$, the Lindeberg condition is satisfied.

\begin{claim}\label{clm:asympvariance_treatment}
(i) $\Var(A+B) = 1$ and (ii) $\calV_{\mu}^{-1} \lesssim \ \frac{\mu r \min\{N_0 ,N_d \}}{ \sigma^2}$. 
\end{claim}

\noindent Therefore, by Lindeberg CLT, we have $A+B \conD \calN(0,1)$. In addition, by the same token as in the proof of Theorem \ref{thm:groupclt_split_block}, we can show $\norm{\calV_{\mu}^{-\frac{1}{2}}\calR} = o_p(1)$. Therefore, 
$$
\calV_{\mu}^{-\frac{1}{2}} \left( \frac{1}{|\calG|}\sum_{i \in \calG} ( \widehat{m}^{(d)}_{it_o} - \widehat{m}^{(0)}_{it_o} ) - \frac{1}{|\calG|}\sum_{i \in \calG} ( m^{(d)}_{{it_o}} - m^{(0)}_{{it_o}} ) \right)
\conD \calN(0,1).
$$
\noindent(ii) Case 2 $(\widehat{\theta}_{t_0}^{(d)})$: The proof is the same as that of Case 1 if we change $A^{(0)},B^{(0)}$ to $A^{(d-1)},B^{(d-1)}$. Since it is a simple extension of the proof of Case 1, we omit it. $\square$

\begin{proof}[Proof of Claim \ref{clm:asympvariance_treatment}]
(i) Since $t_o > T_0$ and $\calI_d$ is disjoint with $\calI_0$, we have
$$
\Var(A+B) = \Var(A^{(d)}) + \Var(A^{(0)}) + \Var(B).
$$
A Simple calculations show that
\begin{align*}
&\Var(A^{(d)}) = \calV_{\mu}^{-1} \sigma^2 \bar{u}_\calG ^\top  \left(\sum_{j \in \calI_d} u_j u_j^\top  \right)^{-1} \bar{u}_\calG, \ \ 
\Var(A^{(0)}) = \calV_{\mu}^{-1} \sigma^2 \bar{u}_\calG ^\top  \left(\sum_{j \in \calI_0} u_j u_j^\top  \right)^{-1} \bar{u}_\calG,\\
&\Var(B) = \calV_{\mu}^{-1}\frac{\sigma^2}{|\calG|}  \left(v_{(d\cdot T_1 + t_o)} - v_{t_o}\right)^\top  \left( \sum_{s \leq T_0} v_{s} v_{s}^\top \right)^{-1} \left(v_{(d\cdot T_1 + t_o)} - v_{t_o}\right).
\end{align*}
Hence, we have $\Var(A+B)=1$.\\
\noindent (ii) Note that
\begin{align*}
\calV_{\mu} = \calV_{\mu}\Var(A) +\calV_{\mu}\Var(B) 
\geq \calV_{\mu}\Var(A) \geq \max \{ \calV_{\mu}\Var(A^{(d)}), \calV_{\mu}\Var(A^{(0)})\},
\end{align*}
since $\Var(A)= \Var(A^{(d)}) + \Var(A^{(0)})$. In addition, we have
$$
\calV_{\mu}\Var(A^{(d)}) = \sigma^2 \bar{u}_\calG ^\top  \left(\sum_{j \in \calI_d} u_j u_j^\top  \right)^{-1} \bar{u}_\calG 
\geq \sigma^2 \norm{\bar{u}_\calG}^2 \lambda_{\min}\left( \left(\sum_{j \in \calI_d} u_j u_j^\top  \right)^{-1} \right) \geq c \frac{\sigma^2}{\mu r N_d}
$$ for some constant $c>0$. Similarly, we have $\calV_{\mu}\Var(A^{(0)}) \geq c \frac{\sigma^2}{\mu r N_0}$ for some constant $c>0$. Therefore, we reach $\calV_{\mu}^{-1} \leq C \frac{\mu r \min \{N_0, N_d \}}{\sigma^2}$.
\end{proof}
\bigskip

\noindent\textbf{Proof of Corollary \ref{coro:feasibletreatment}}

\noindent(i) Case 1 $(\widehat{\mu}_{t_0}^{(d)})$: From the proof of Claim \ref{clm:asympvariance_treatment}, we know
$$
\calV_{\mu} = \Var(\tilde{A}^{(d)}) +  \Var(\tilde{A}^{(0)})  + \Var(\tilde{B}^{(d)}) + \Var(\tilde{B}^{(0)}) - 2 \Var(\tilde{B}^{(0)})
$$
where $\tilde{A}^{(\delta)} = \calV_{\mu}^\frac{1}{2} A^{(\delta)}$ and $\tilde{B}^{(\delta)} = \calV_{\mu}^\frac{1}{2} B^{(\delta)}$. Following the similar argument in the proof of Corollary \ref{coro:feasibleclt} with the definitions in \eqref{eq:treatmentdecomposition}, we have
$$
\Var(\tilde{A}^{(d)}) = \sigma^2 \sum_{i \in \calI_d} \left(\sum_{0 \leq l \leq L_d} \alpha^{(d)}_l  \bar{X}_{\calG_{(d),l}}^{\top} \left( \sum_{j \in \calI_d}   X_{l,j}^{(d)} X_{l,j}^{(d)\top} \right)^{-1} X_{l,i}^{(d)}
 \right)^2,
$$
where $\alpha^{(d)}_l = \frac{|\calG_{(d),l}|}{|\calG|}$, $\bar{X}_{\calG_{(d),l}} =\frac{1}{|\calG_{(d),l}|} \sum_{i \in \calG_{(d),l}} X_{l,i}^{(d)} $, and $X_{l}^{(d)} = U_l^{(d)} D_l^{(d)\frac{1}{2}}$. Here, $U_l^{(d)} D_l^{(d)} V_l^{(d)\top}$ are the SVD of $\tilde{M}^{(d)}_{l}$ which is the submatrix of $\tilde{M}^{(d)}$ constructed for the estimation of $\{m_{it_0}^{(d)}\}_{i \in \calG_{(d),l}}$. In addition, we have 
$$
\Var( \tilde{B}^{(d)} ) =  \frac{\sigma^2}{|\calG|} Z_{0,(d \cdot T_1 + t_o)}^{(d)\top} \left( \sum_{s\leq T_0 } Z^{(d)}_{0,s} Z_{0,s}^{(d)\top} \right)^{-1}  Z_{0,(d \cdot T_1 + t_o)}^{(d)},
$$
where $Z_{l}^{(d)} = V_l^{(d)} D_l^{(d)\frac{1}{2}}$.
Note that for all $\delta \in \{0, d \}$, $\calV_{\mu}^{-1} \lesssim \frac{\mu r N_\delta }{ \sigma^2}$ by Claim \ref{clm:asympvariance_treatment}. Then, by the same way as the proof of Corollary \ref{coro:feasibleclt}, we have 
$$
\calV_{\mu}^{-1} \norm{ \widehat{\Var}( \tilde{A}^{(\delta)} ) - \Var( \tilde{A}^{(\delta)} ) } = o_p(1), \ \ \calV_{\mu}^{-1} \norm{ \widehat{\Var}( \tilde{B}^{(\delta)} ) - \Var( \tilde{B}^{(\delta)} ) } = o_p(1).
$$ 
Similarly, we can show that $$
\calV_{\mu}^{-1} \norm{ \widehat{\Cov}( \tilde{B}^{(d)},\tilde{B}^{(0)} ) - \Cov(  \tilde{B}^{(d)},\tilde{B}^{(0)} ) } = o_p(1)
$$ 
where 
\begin{align*}
\Cov(  \tilde{B}^{(d)},\tilde{B}^{(0)} ) &= 
\frac{\sigma^2}{|\calG|} v_{(d\cdot T_1 + t_o)}^\top  \left( \sum_{s\leq T_0}v_{s}v_{s}^\top  \right)^{-1} v_{t_o}^\top  \\
&=\frac{\sigma^2}{|\calG|}\sum_{s\leq T_0}  v_{(d\cdot T_1 + t_o)}^\top \left( \sum_{s\leq T_0}v_{s}v_{s}^\top  \right)^{-1}v_{s}v_{s}^\top  \left( \sum_{s\leq T_0}v_{s}v_{s}^\top  \right)^{-1} v_{t_o}^\top  \\
&=\frac{\sigma^2}{|\calG|}\sum_{s\leq T_0}  Z_{0,(d\cdot T_1 + t_o)}^{(d)\top} \left( \sum_{s\leq T_0}Z_{0,s}^{(d)}Z_{0,s}^{(d)\top} \right)^{-1}Z^{(d)}_{0,s}Z_{0,s}^{(0)\top} \left( \sum_{s\leq T_0}Z_{0,s}^{(0)}Z_{0,s}^{(0)\top} \right)^{-1} Z_{0,t_o}^{(0)\top}.
\end{align*}
Hence, we have $\frac{\widehat{\calV}_{\mu} - \calV_{\mu}}{\calV_{\mu}} = o_p(1)$ and it implies that $\frac{\calV_{\mu}}{\widehat{\calV}_{\mu}} \conP 1$. Then, by the Slutsky's theorem with Theorem \ref{thm:groupclt_treatment}, we have the desired result. \\
(ii) Case 2 $(\widehat{\theta}_{t_0}^{(d)})$: The proof is the same as that of Case 1 if we change $A^{(0)},B^{(0)}$ to $A^{(d-1)},B^{(d-1)}$. Since it is a simple extension of the proof of Case 1, we omit it. $\square$

\section{Modification of results from \cite{chen2020noisy}}\label{sec:proofforchen}

Finally, we present technical tools used for proving Lemmas \ref{lem:CCFMY_noncovex} and \ref{lem:smallgradient} in Section \ref{sec:nonconvex_property}. These results are modifications of similar results from \cite{chen2020noisy} when missing is random. Indeed the overall architecture of our proof is the same as those in \cite{chen2020noisy}, and for brevity, we shall omit proofs of lemmas that are straightforward adaptation of those in \cite{chen2020noisy}.

\subsection{Proximity between the nonconvex estimator and the nuclear norm penalized estimator}\label{SectionA}

We begin by introducing further notations. For any matrix $G$, we denote by $G_{l,\cdot}$ (resp. $G_{\cdot,l}$) the $l$-th row (resp. column) of $G$. Let $G$ be a $N_o \times T_o$ matrix with rank $r$ and $L \Sigma R^\top $ be SVD of $G$. Then the tangent space of $G$, denoted by $T(G)$, is defined as
\[T(G) = \{D \in \mathbb{R}^{N_o \times T_o}  | D = A R^\top + L B^\top  \,\,\text{for some}\,\, A \in \mathbb{R}^{N_o \times r}\,\, \text{and} \,\, B \in \mathbb{R}^{T_o \times r} \}.\]
Let $\mathcal{P}_{T(G)}$ be the orthogonal projection onto $T(G)$, that is, $$\mathcal{P}_{T(G)}(E)=LL^\top E + ERR^\top  - LL^\top ERR^\top $$ for any $E \in \mathbb{R}^{N_o\times T_o}$. When there is no risk of confusion, we will simply denote by $T$ instead of $T(G)$. Let $T^{\perp}$ be the orthogonal complement of $T$ and $\mathcal{P}_{T^{\perp}}$ be the projection onto $T^{\perp}$. Note that
$\mathcal{P}_{T^{\perp}}(E)=(I-LL^\top )E(I-RR^\top )$ and $\mathcal{P}_T(E) + \mathcal{P}_{T^{\perp}}(E)=E$. Lastly, we define $\mathcal{P}_{\Omega_o}^{\text{diff}}(G) = \mathcal{P}_{\Omega_o}(G)- G$, for all $G \in \mathbb{R}^{N_o \times T_o}$. 

The following lemma plays a key role in showing the proximity between the nonconvex estimator $(\breve{X}_o,\breve{Z}_o)$ and the nuclear norm penalized estimator $\widetilde{M}_o$. We will eventually set $(\ddot{X}_o,\ddot{Z}_o) = (\breve{X}_o\breve{H}_o,\breve{Z}_o\breve{H}_o)$ where $(\breve{X}_o,\breve{Z}_o) = (X^{\tau^*_o}_o,Z^{\tau^*_o}_o)$ and $\breve{H}_o=H^{\tau^*_o}_o$.

\begin{condition}[Regularization parameter] \label{cond:regularization}
	The regularization parameter $\lambda_o$ satisfies (i) $\norm{\mathcal{P}_{\Omega_o}(\calE_o )}<\frac{7}{8} \lambda_o $ and (ii)  $\norm{\mathcal{P}_{\Omega_o}(\ddot{X}_o\ddot{Z}_o^\top -M_o) - \ddot{X}_o\ddot{Z}_o^\top -M_o}<\frac{1}{80} \lambda_o $.
\end{condition}

\begin{condition}[Injectivity]\label{cond:injectivity}
	Let $T$ be the tangent space of $\ddot{X}_o\ddot{Z}_o^\top $. There is a quantity $c_{\text{inj},o}>0$ such that
	$ \norm{\mathcal{P}_{\Omega_o}(H)}_F^2 \geq c_{\text{inj},o} \norm{H}^2_F$ for all $H \in T$.
\end{condition}

\begin{lemma}\label{LemmaA1}
	Suppose that $(\ddot{X}_o,\ddot{Z}_o)$ satisfies 
	\begin{align}\label{LemmaA1.1}
		\norm{\nabla f(\ddot{X}_o,\ddot{Z}_o)}_F \leq c \frac{\sqrt{c_{\inj,o} }}{\kappa_o} \lambda_o \sqrt{\psi_{\min,o}} 
	\end{align}
	for some sufficiently small constant $c>0$. Additionally, assume that any nonzero singular value of $\ddot{X}_o$ and $\ddot{Z}_o$ exists in the interval $[\sqrt{\frac{\psi_{\min,o}}{2}},\sqrt{2\psi_{\max,o}} ]$. Then, under Conditions \ref{cond:regularization} and \ref{cond:injectivity}, $\widetilde{M}_o$ satisfies
	\begin{align*}
		\norm{\ddot{X}_o\ddot{Z}_o^\top -\widetilde{M}_o}_F \leq C_{cvx} \frac{\kappa_o}{c_{\text{inj,o}}} \frac{1}{\sqrt{\psi_{\min,o}}} \norm{\nabla f (\ddot{X}_o,\ddot{Z}_o)}_F
	\end{align*}
	where $C_{cvx} > 0$ is an absolute constant.
\end{lemma}

\begin{proof}
This lemma is the simple modified version of Lemma 2 of \cite{chen2020noisy}. If we follow their proof by setting $p=1$ and considering our observation pattern with caution, we can get the result. To save space, we omit the proof.
\end{proof}
\bigskip

The following lemmas are used to show that our nonconvex estimator $(\breve{X}_o\breve{H}_o,\breve{Z}_o\breve{H}_o)$ satisfies Conditions \ref{cond:regularization} and \ref{cond:injectivity}. Lemma \ref{LemmaA3} shows Condition \ref{cond:regularization} (i) is satisfied when $\lambda_o = C_\lambda \sigma \sqrt{\max\{N_o,T_o\} }$ for a sufficiently large constant $C_\lambda>0$. In addition, Lemma \ref{LemmaA4} is used when we show Condition \ref{cond:regularization} (ii) and Condition \ref{cond:injectivity} are satisfied in the case $(\ddot{X}_o,\ddot{Z}_o) = (\breve{X}_o\breve{H}_o,\breve{Z}_o\breve{H}_o)$.

\begin{lemma}\label{LemmaA3}
	With probability at least $1-O(\min\{N_o^{-101},T_o^{-101}\})$, we have\\
 (i) $\norm{\calP_{\Omega_o}(\mathbf{1}\mathbf{1}^\top )- \mathbf{1}\mathbf{1}^\top }  \lesssim \sqrt{\max\{N_o,T_o\} }$, (ii) $\norm{\calP_{\Omega_o}(\calE_o)} \lesssim  \sigma \sqrt{\max\{N_o,T_o\} }$.
\end{lemma}

\begin{proof}
(i) All elements of $\calP_{\Omega_o}(\mathbf{1}\mathbf{1}^\top )- \mathbf{1}\mathbf{1}^\top $ excluding the elements of $\{(i,t_o)\}_{ i \in \calQ_o}$ are $0$. Because the elements of $\{(i,t_o)\}_{ i \in \calQ_o}$ are $-1$ and $|\calQ_o| \leq N_o$, it is trivial.\\
(ii) Denote a $N_o \times (T_o-1)$ matrix excluding the $t_o$-th column of $\calP_{\Omega_o}(\calE_o)$ by $\calP_{\Omega_o}(\calE_o)^{(-t_o)}$. By Theorem 5.39 of \cite{vershynin2010introduction}, we have $$||\calP_{\Omega_o}(\calE_o)^{(-t_o)}||=||\calE_o^{(-t_o)}|| \lesssim  \sigma \sqrt{\max\{N_o,T_o\}}$$ where $\calE_o^{(-t_o)}$ is the $N_o \times (T_o-1)$ matrix excluding the $t_o$-th column of $\calE_o$. In addition, it is trivial that $$||\calP_{\Omega_o}(\calE_o)_{\cdot,t_o}||\leq ||(\calE_o)_{\cdot,t_o}|| \lesssim  \sigma \sqrt{\max\{N_o,T_o\}}$$ where $\calP_{\Omega_o}(\calE_o)_{\cdot,t_o}$ and $(\calE_o)_{\cdot,t_o}$ are the $t_o$-th column of $\calP_{\Omega_o}(\calE_o)$ and $\calE_o$, respectively.
\end{proof}
\bigskip

\begin{lemma}\label{LemmaA4}
	Suppose that 
 \begin{align*}
  &\frac{\sigma}{\psi_{\min,o}}\sqrt{\frac{\max\{N_o^2, T_o^2\}}{\min\{N_o, T_o\} }} \ll \frac{1}{\sqrt{\kappa_o^4 \mu_o r \max \{\log N_o, \log T_o\}}}, \ \ \vartheta_o \mu_o r \ll \min\{N_o,T_o\},\\
  &\min\{N_o^2, T_o^2\}  \gg  \kappa_o^4 \mu_o^2 r^2 \max \{N_o\log N_o, T_o\log T_o\}.
 \end{align*}
Assume that $\lambda_o = C_{\lambda} \sigma \sqrt{\max\{N_o,T_o\}}$ for some large constant $C_{\lambda}>0$. Further, let $T$ denote the tangent space of $\ddot{X}\ddot{Z}_o^\top $. Then, with probability at least $1-O(\min\{N_o^{-100}, T_o^{-100}\}$,
	\begin{align*}
		&\norm{\mathcal{P}_{\Omega_o}(\ddot{X}_o\ddot{Z}_o^\top -M_o)-\ddot{X}_o\ddot{Z}_o^\top -M_o} < \frac{1}{80}\lambda_o  \quad \text{(Condition \ref{cond:regularization} (ii))} \\
		&\norm{\mathcal{P}_{\Omega_o}(H)}_F^2 \geq \frac{1}{32\kappa_o} \norm{H}_F^2 \quad \text{for all $H \in T$} \quad \text{(Condition \ref{cond:injectivity} with $c_{\inj,o}=1/(32\kappa_o)$)}
	\end{align*}
	hold uniformly for all $(\ddot{X}_o, \ddot{Z}_o)$ satisfying
	\begin{align}
	\nonumber &\max\Big\{\norm{\ddot{X}_o-X_o}_{2, \infty}, \norm{\ddot{Z}_o-Z_o}_{2, \infty}\Big\}\\
  &\ \ \leq C \kappa_o \left( \frac{\sigma \sqrt{\max\{N_o \log N_o , T_o \log T_o\}}}{\psi_{\min,o}} +\frac{\lambda_o}{\psi_{\min,o}} \right)    \max\Big\{\norm{X_o}_{2, \infty}, \norm{Z_o}_{2, \infty}\Big\} \label{LemmaA4.1}
	\end{align}
	for some constant $C >0$.
\end{lemma}

\begin{proof}
	It follows immediately from Lemma \ref{LemmaA5} and Lemma \ref{LemmaA6}.
\end{proof}

\begin{lemma}\label{LemmaA5}
	Assume that $\min\{N_o, T_o\}  \gg  \mu_o r \max \{\log N_o, \log T_o\}$ and $\vartheta_o \mu_o r \ll \min\{N_o,T_o\}$. Let $T$ denote the tangent space of $\ddot{X}_o\ddot{Z}_o^\top $. Then, with probability at least $1-O(\min\{N_o^{-100}, T_o^{-100}\}$,
	\begin{align*}
	\norm{\mathcal{P}_{\Omega_o}(H)}_F^2 \geq  \frac{1}{32\kappa_o} \norm{H}_F^2 \quad \text{for all $H \in T$} \quad \text{(Condition \ref{cond:injectivity} with $c_{\inj}=1/(32\kappa_o)$)}
	\end{align*}
	holds uniformly for all $(\ddot{X}_o, \ddot{Z}_o)$ such that
	\begin{align*}
		\max\Big\{\norm{\ddot{X}_o-X_o}_{2, \infty}, \norm{\ddot{Z}_o-Z_o}_{2, \infty}\Big\} \leq \frac{c}{\kappa_o\sqrt{\max\{N_o,T_o\}}}\norm{X_o} 
	\end{align*}
	where $c>0$ is some sufficiently small constant.
\end{lemma}

\begin{proof}
This Lemma is the simple modification of Lemma 7 of \cite{chen2020noisy}. If we follow their proof by considering our observation pattern cautiously, we can get the result. Importantly, we use Lemma \ref{LemmaCR091} which is the modified version of Corollary 4.3 of \cite{candes2009exact} in the place where \cite{chen2020noisy} use Corollary 4.3 of \cite{candes2009exact}. To save space, we omit the proof.
\end{proof}

\begin{lemma}\label{LemmaA6}
	Assume that 
	$$\frac{\sigma \sqrt{\max\{N_o \log N_o,T_o \log T_o\}}}{\psi_{\min,o}} \ll \frac{1}{\kappa_o}, \ \ 
 \min\{N_o^2, T_o^2\}  \gg   \kappa_o^4 \mu_o^2 r^2 \max \{N_o\log N_o, T_o\log T_o\}.$$ 
 Let $\lambda_o =C_{\lambda} \sigma \sqrt{\max\{N_o,T_o\}}$ for some large constant $C_{\lambda}>0$. Then, with probability at least $1-O(\min\{N_o^{-100}, T_o^{-100}\})$, 
 $$
 \norm{\mathcal{P}_{\Omega_o}(\ddot{X}_o\ddot{Z}_o^\top -M_o)-\ddot{X}_o\ddot{Z}_o^\top -M_o} \lesssim \sigma \sqrt{\max\{N_o, T_o \}} \sqrt{\frac{\kappa_o^4\mu_o^2 r^2 \max\{N_o \log N_o, T_o \log T_o\}}{\min\{N_o^2, T_o^2 \}}} < \frac{1}{80}\lambda_o
 $$
 holds uniformly for all $(\ddot{X}_o, \ddot{Z}_o)$ satisfying (\ref{LemmaA4.1}).
\end{lemma}

\begin{proof}
This Lemma is the simple modification of Lemma 8 of \cite{chen2020noisy}. To save space, we omit the proof.
\end{proof}

\subsection{Quality of non-convex estimates}

Before we proceed, we introduce some notations. Define an augmented loss function $f_{\aug}(A, B)$ to be
\[f_{\aug} \coloneqq \frac{1}{2} \norm{\mathcal{P}_{\Omega_o}(AB^\top -Y_o)}_F^2 +\lambda_o \norm{A}_F^2+\lambda_o \norm{B}_F^2 +\frac{1}{8}\norm{A^\top A-B^\top B}_F^2. 
\]
Then, the gradient of $f_{\aug}(\cdot, \cdot)$ is given by
\begin{align*}
&\nabla_X f_{\aug}(A, B) =  \mathcal{P}_{\Omega_o}(AB^\top -Y_o)B+\lambda_o A+\frac{1}{2}A (A^\top A-B^\top B),\\
&\nabla_Z f_{\aug}(A, B) =  \mathcal{P}_{\Omega_o}(AB^\top -Y_o))^\top A+\lambda_o B+\frac{1}{2}B(B^\top B-A^\top A).
\end{align*}
The difference between gradients of $\nabla f(A, B)$ and $\nabla f_{\aug}(A, B)$ are
$$
\nabla_X f_{\diff}(A, B) = -\frac{1}{2}A(A^\top A-B^\top B), \ \ \nabla_Z f_{\diff}(A, B) =  -\frac{1}{2}B(B^\top B-A^\top A) .
$$
In addition, note that we have the following properties of $\calF_o$: 
\begin{align*}
	\psi_1(\calF_o) = \norm{\calF_o}=\sqrt{2 \psi_{\max,o}}, \ \ \psi_r(\calF_o)=\sqrt{2\psi_{\min,o}}, \ \
	\norm{\calF_o}_{2, \infty} \leq \sqrt{\frac{\mu r \psi_{\max,o}}{\min\{N_o,T_o\}}}.
\end{align*}

The following Lemma is the one of the main parts where we require the condition $\min \{N_o, T_o\} \gg \vartheta_o \kappa_o^2 \mu_o r$. While it is the modified version of Lemma 12 in \cite{chen2020noisy}, the proof of it is quite different from theirs. Hence, we provide the full proof.

\begin{lemma}\label{LemmaB5}
Suppose that $\lambda_o=C_{\lambda} \sigma \sqrt{\max\{N_o,T_o\}}$ for some large constant $C_{\lambda}>0$, $0< \eta_o \ll 1/(\kappa_o^{2}\psi_{\max,o} \min\{N_o,T_o\})$, $\min \{N_o, T_o\} \gg \vartheta_o \kappa_o^2 \mu_o r$, and
$$\frac{\sigma}{\psi_{\min,o}}\sqrt{\frac{\max\{N_o^2, T_o^2\}}{\min\{N_o,T_o\}}} \ll \frac{1}{\sqrt{\kappa_o^4 \mu_o r \max \{\log N_o, \log T_o\}}},
$$
$$
\min \{N_o, T_o\} \gg \kappa_o \mu_o r \max \{\log^3 N, \log^3 T\} .
$$
Suppose also that the iterates satisfy (\ref{Prelim3})-(\ref{Prelim8}) at the $\tau$-th iteration, then with probability at least $1-O(\min\{N_o^{-99}, T_o^{-99}\}),$ we have
	\[
	\max_{1\leq m\leq N_o+T_o}  \norm{\calF_o^{\tau+1} H_o^{\tau+1}-\calF_o^{\tau+1,(m)}Q_o^{\tau+1,(m)}}_F \leq C_3  \left(\frac{\sigma \sqrt{\max\{N_o \log N_o,T_o \log T_o\}}}{\psi_{\min,o}}+\frac{\lambda_o}{\psi_{\min,o}} \right) \norm{\calF_o}_{2, \infty}
	\]
	where $C_3$ is some sufficiently large constant.
\end{lemma}

\begin{proof}
	Fix $1 \leq m \leq N_o+T_o$. The definition of $Q_o^{\tau+1, (m)}$ and the unitary invariance of Frobenius norm yield
	$$\norm{\calF_o^{\tau+1}H_o^{\tau+1}-\calF_o^{\tau+1,(m)}Q_o^{\tau+1,(m)}}_F \leq \norm{\calF_o^{\tau+1}H_o^{\tau}-\calF_o^{\tau+1,(m)}Q_o^{\tau, (m)}}_F.$$ By the gradient update rules (\ref{alg:nonconvex}) and (\ref{alg:loo}), we obtain
	\begin{align*}
& \calF_o^{\tau+1}H_o^{\tau}-\calF_o^{\tau+1,(m)}Q_o^{\tau, (m)}\\
  & \ \ = \left( \calF_o^{\tau} - \eta_o \nabla f(\calF_o^{\tau})\right) H_o^{\tau}-\left(\calF_o^{\tau,(m)} - \eta_o \nabla f^{(m)}(\calF_o^{\tau,(m)}) \right) Q_o^{\tau,(m)} \\
		& \ \  = \calF_o^{\tau} H_o^{\tau}-\eta_o \nabla f(\calF_o^{\tau}H_o^{\tau})- \left(\calF_o^{\tau,(m)}Q_o^{\tau,(m)} - \eta_o \nabla f^{(m)}(\calF_o^{\tau,(m)}Q_o^{\tau,(m)})\right) \\
		&\ \   = \underbrace{\left(\calF_o^{\tau} H_o^{\tau}-\calF_o^{\tau,(m)}Q_o^{\tau,(m)}
			\right) - \eta_o \left( \nabla f_{\aug}(\calF_o^{\tau} H_o^{\tau})-\nabla f_{\aug}(\calF_o^{\tau,(m)}Q_o^{\tau,(m)}) \right)
		}_{\coloneqq A_1}\\
		& \quad -\underbrace{\eta_o \left(\nabla f_{\diff}(\calF_o^{\tau} H_o^{\tau})-\nabla f_{\diff}(\calF_o^{\tau,(m)}Q_o^{\tau,(m)}) \right)}_{\coloneqq A_2} +\underbrace{\eta_o \left(\nabla f^{(m)}(\calF_o^{\tau,(m)} Q_o^{\tau,(m)})-\nabla f (\calF_o^{\tau,(m)}Q_o^{\tau,(m)}) \right)}_{\coloneqq A_3},
	\end{align*}
	Here, we use the facts that $\nabla f(\calA)O=\nabla f(\calA O)$ and $\nabla f^{(m)}(\calA)O=\nabla f^{(m)}(\calA O)$ for any $(N_o+T_o) \times r$ matrix $\calA$ and any orthonormal matrix $O \in \mathcal{O}^{r \times r}$. Hereinafter, we control $A_1, A_2$ and $A_3$ separately. The way of bounding $A_1$ and $A_2$ are the same as the proof of Lemma 12 in \cite{chen2020noisy} while the way of bounding $A_3$ is quite different.
	\begin{enumerate}
		\item The first term $A_1$ can be bounded using the same derivation as $\alpha_1$ in the proof of Lemma 10 of \cite{chen2020noisy}:
		$$ \norm{A_1}_F \leq  \left(1- \frac{\psi_{\min,o}}{20}\eta_o \right) \norm{\calF_o^{\tau}H_o^{\tau}-\calF_o^{\tau,(m)}Q_o^{\tau,(m)}}_{\rm F}$$ with probability at least $1-O(\min\{N_o^{-100},T_o^{-100}\})$. Here, we use the assumptions 
		$$
  \frac{\sigma}{\psi_{\min,o}} \sqrt{\frac{\max \{N_o^2, T_o^2\}}{\min\{N_o,T_o\}}} \ll \frac{1}{\sqrt{\kappa_o^4\mu_o r \max\{\log N_o, \log T_o\}}},$$ 
  $$ \min \{N_o , T_o \} \gg  \kappa_o  \mu_o  r  \max \{  \log^2 N_o,   \log^2 T_o\},$$ and $0 \leq \eta_o \ll 1/(\kappa_o^{2}\psi_{\max,o} \min\{N_o,T_o\})$.
		
		\item Regarding $A_2$, the triangle inequality gives us, with probability at least $1-O(\min\{N_o^{-100},T_o^{-100}\})$, $$\norm{A_2}_F \leq \eta_o \norm{\nabla f_{\diff}(\calF_o^{\tau} H_o^{\tau})}_F+\eta_o \norm{\nabla f_{\diff}(\calF_o^{\tau,(m)} Q_o^{\tau,(m)})}_F.$$ Following the bound of $\alpha_2$ in the proof of Lemma 10 of \cite{chen2020noisy}, we obtain $$ \eta_o \norm{\nabla f_{\diff}(\calF_o^{\tau} H_o^{\tau})}_F \leq 2 C_B \kappa_o \eta_o^2 \left( \frac{\sigma\sqrt{\max \{N_o  , T_o \}}}{\psi_{\min,o}} +\frac{\lambda_o}{\psi_{\min,o}}\right) \sqrt{r} \psi_{\max,o}^2 \norm{X_o}.$$ Additionally, Lemma \ref{LemmaB8} and the argument for bounding $\alpha_2$ in the proof of Lemma 10 of \cite{chen2020noisy} together give us $$\eta_o \norm{\nabla f_{\diff}(\calF_o^{\tau,(m)} Q_o^{\tau,(m)})}_F \leq 2 C_B \kappa_o \eta_o^2 \left( \frac{\sigma \sqrt{\max \{N_o  , T_o \}}}{\psi_{\min,o}} +\frac{\lambda_o}{\psi_{\min,o}}\right) \sqrt{r} \psi_{\max,o}^2 \norm{X_o}.$$
		The three inequalities together allow us to have
		\begin{align*}
			\norm{A_2}_F &\leq 4C_B \kappa_o \eta_o^2 \left( \frac{\sigma \sqrt{\max \{N_o  , T_o \}}}{\psi_{\min,o}}+\frac{\lambda_o}{\psi_{\min,o}}\right) \sqrt{r} \psi_{\max,o}^2 \norm{X_o} \\
   &\leq \eta_o \left( \sigma \sqrt{\max \{N_o  , T_o \}}+\lambda_o \right) \norm{\calF_o}_{2, \infty},
		\end{align*}
		with probability at least $1-O(\min\{N_o^{-100},T_o^{-100}\})$, where the last inequality follows from the assumption $\eta_o \ll \frac{1}{\min\{N_o,T_o\}\kappa_o^2 \psi_{\max,o}}$.
		
		\item For bounding $A_3$, observe that
		\begin{align*}
			A_3 &= \eta_o \begin{bmatrix}
				\underbrace{\left(\mathcal{P}_{m, \cdot}( X_o^{\tau,(m)}Z_o^{\tau,(m) \top}-M_o )-\mathcal{P}_{\Omega_{m,\cdot}}(X_o^{\tau,(m)}Z_o^{\tau,(m) \top}-M_o) \right)Z_o^{\tau,(m)}Q_o^{\tau,(m)}}_{\coloneqq B_1} \\
				\underbrace{\left(\mathcal{P}_{m, \cdot}( X_o^{\tau,(m)}Z_o^{\tau,(m) \top}-M_o )-\mathcal{P}_{\Omega_{m,\cdot}}(X_o^{\tau,(m)}Z_o^{\tau,(m) \top}-M_o) \right)^\top X_o^{\tau,(m)}Q_o^{\tau,(m)}}_{\coloneqq B_2} 
			\end{bmatrix}\\
  & \ \  +\eta_o \begin{bmatrix}
					\underbrace{\mathcal{P}_{\Omega_{m,\cdot}}(\calE_o) Z_o^{\tau,(m)}Q_o^{\tau,(m)}}_{\coloneqq C_1}\\
			  	\underbrace{\left(\mathcal{P}_{\Omega_{m,\cdot}}(\calE_o) \right)^\top X_o^{\tau,(m)}Q_o^{\tau,(m)}}_{\coloneqq C_2}
			\end{bmatrix}
		\end{align*}
		
		We invoke the following three claims to control $B_1$, $B_2$, and $C_1$, $C_2$, whose proofs are provided after the proof of Lemma \ref{LemmaB5}.
		
		\begin{claim}\label{ClaimB2}
			Assume that $$ \frac{\sigma \sqrt{\max\{N_o, T_o\}} }{\psi_{\min,o}}\ll \frac{1}{\sqrt{\kappa_o^2   \max \{\log N_o, \log T_o\}}}.$$ Then, for each $1 \leq m \leq N_o+T_o$, we have
			\[\norm{B_1}_F \lesssim  \frac{\sqrt{\vartheta_o} \mu_o r}{\min\{N_o,T_o\}}  \psi_{\max,o} \norm{ \calF_o^{\tau,(m)}Q_o^{\tau,(m)} -\calF_o }_{2,\infty}.
			\]
		\end{claim}
		
		\begin{claim}\label{ClaimB3}
			Assume that $$\frac{\sigma \sqrt{\max\{N_o, T_o\}}}{\psi_{\min,o}} \ll \frac{1}{\sqrt{\kappa_o^2   \max \{\log N_o, \log T_o\}}}.$$ Then, for each $1 \leq m \leq N_o+T_o$, we have
			\[\norm{B_2}_F \lesssim  \frac{\vartheta_o \mu_o r}{\min\{N_o,T_o\}}  \psi_{\max,o}  \norm{ \calF_o^{\tau,(m)}Q_o^{\tau,(m)} -\calF_o }_{2,\infty}.
			\]
		\end{claim}
		
		\begin{claim}\label{ClaimB4}
			Assume that $$\frac{\sigma \sqrt{\max\{N_o, T_o\}}}{\psi_{\min,o}} \ll \frac{1}{\sqrt{\kappa_o^2 \max \{\log N_o, \log T_o\}}}, \ \ \max \{N_o, T_o\} \gg   \max\{  \log^3 N_o,\log^3 T_o\}.$$ Then, for each $1 \leq m \leq N_o+T_o$, we have
			\[\max\{\norm{C_1}_F, \norm{C_2}_F\} \lesssim \sigma \sqrt{\max\{ N_o \log N_o,T_o \log T_o\}} \norm{ \calF_o}_{2, \infty},
			\]
			with probability at least $1-O(\min\{N_o^{-100}, T_o^{-100}\})$.
		\end{claim}
		 Then, the triangle inequality yields, with probability at least $1-O(\min\{N_o^{-100},T_o^{-100}\})$
		\begin{align*}
			 \norm{A_3}_F   &\leq \eta_o (\norm{B_1}_F+\norm{B_2}_F+\norm{C_1}_F+\norm{C_2}_F) \\
			&   \lesssim \eta_o \sigma \sqrt{\max\{ N_o \log N_o,T_o \log T_o\}} \norm{ \calF_o}_{2, \infty}
			+ \eta_o \frac{\vartheta_o \mu_o r}{\min\{N_o,T_o\}}  \psi_{\max,o}  \norm{ \calF_o^{\tau,(m)}Q_o^{\tau,(m)} -\calF_o }_{2,\infty}\\
			& \leq \eta_o \sigma \sqrt{\max\{ N_o \log N_o,T_o \log T_o\}} \norm{ \calF_o}_{2, \infty} \\
			 & \ \ + \eta_o \frac{\vartheta_o \mu_o r}{\min\{N_o,T_o\}}  \psi_{\max,o} (C_{\infty} \kappa_o+C_3) \left( \frac{\sigma \sqrt{\max\{N_o \log N_o,T_o \log T_o\}}}{\psi_{\min,o}} + \frac{\lambda_o}{\psi_{\min,o}}\right)\norm{\calF_o}_{2, \infty} .
		\end{align*}
		Here, the last inequality follows from Lemma \ref{LemmaB11} (i). 
	\end{enumerate}
	Combining the bounds on $A_i$, $i=1,2,3$, we reach, with probability at least $1-O(\min\{N_o^{-100},T_o^{-100}\})$,
	\begin{align*}
		&\norm{\calF_o^{\tau+1}H_o^{\tau+1}-\calF_o^{\tau+1, (m)}Q_o^{\tau+1, (m)}}_F \\
  &\leq \norm{A_1}_F+\norm{A_2}_F+\norm{A_3}_F \\
  &\leq    \left(1- \frac{\psi_{\min,o}}{20}\eta_o \right) \norm{\calF_o^{\tau}H_o^{\tau}-\calF_o^{\tau,(m)}Q_o^{\tau,(m)}}_F \\
		&  \ \ +\eta_o \left( \sigma \sqrt{\max \{N_o  , T_o \}}+\lambda_o \right) \norm{\calF_o}_{2, \infty} + \widetilde{C}\eta_o \sigma \sqrt{\max\{ N_o \log N_o,T_o \log T_o\}} \norm{ \calF_o}_{2, \infty}\\
		& \ \ + \widetilde{C}\eta_o \frac{\vartheta_o \mu_o r}{\min\{N_o,T_o\}}  \psi_{\max,o} (C_{\infty} \kappa_o+C_3) \left( \frac{\sigma \sqrt{\max\{N_o \log N_o,T_o \log T_o\}}}{\psi_{\min,o}} + \frac{\lambda_o}{\psi_{\min,o}}\right)\norm{\calF_o}_{2, \infty} \\
		& \leq  \left(1- \frac{\psi_{\min,o}}{20}\eta_o \right) C_3 \left( \frac{\sigma \sqrt{\max\{N_o \log N_o,T_o \log T_o\}}}{\psi_{\min,o}} + \frac{\lambda_o}{\psi_{\min,o}}\right)\norm{\calF_o}_{2, \infty} \\
  & \ \ + \eta_o \left( \sigma \sqrt{\max \{N_o  , T_o \}}+\lambda_o \right) \norm{\calF_o}_{2, \infty}
		  + \widetilde{C}\eta_o \sigma \sqrt{ \max\{ N_o \log N_o,T_o \log T_o\} } \norm{\calF_o}_{2, \infty} \\
		&  \ \  +  \widetilde{C}\eta_o \frac{\vartheta_o \mu_o r}{\min\{N_o,T_o\}}  \psi_{\max,o} (C_{\infty} \kappa_o+C_3) \left( \frac{\sigma \sqrt{\max\{N_o \log N_o,T_o \log T_o\}}}{\psi_{\min,o}} + \frac{\lambda_o}{\psi_{\min,o}}\right)\norm{\calF_o}_{2, \infty}\\
		&  \leq C_3 \left( \frac{\sigma \sqrt{\max\{N_o \log N_o,T_o \log T_o\}}}{\psi_{\min,o}} + \frac{\lambda_o}{\psi_{\min,o}}\right)\norm{\calF_o}_{2, \infty}
	\end{align*}
	for some constant $\widetilde{C}>0$. The penultimate inequality uses the induction hypothesis (\ref{Prelim5}), and the last inequality holds provided that $C_3$ is sufficiently large and $\min \{N_o, T_o\} \gg \vartheta_o \kappa_o^2 \mu_o r $. Therefore, with probability at least $1-O(\min\{N_o^{-99}, T_o^{-99}\}),$ we have
	$$
	\max_{1\leq m\leq N_o+T_o}  \norm{\calF_o^{\tau+1} H_o^{\tau+1}-\calF_o^{\tau+1,(m)}Q_o^{\tau+1,(m)}}_F \leq C_3 \left( \frac{\sigma \sqrt{\max\{N_o \log N_o,T_o \log T_o\}}}{\psi_{\min,o}} + \frac{\lambda_o}{\psi_{\min,o}}\right)\norm{\calF_o}_{2, \infty}.
	$$
	
\end{proof}

\begin{proof}[Proof of Claim \ref{ClaimB2}] Assume that $m \leq N_o$ and define $C \coloneqq X_o^{\tau,(m)}Z_o^{\tau,(m) \top}-X_oZ_o^\top$ and $\calX \coloneqq \mathcal{P}_{\Omega_{m,\cdot}}(C)-\mathcal{P}_{m, \cdot}(C)$. Using the unitary invariance of Frobenius norm, we have $\norm{B_1}_F=\Big\lVert\calX Z_o^{\tau,(m)}\Big\rVert_F$. First of all, if $m \notin \calQ_o$, $\calX=0$. Hence, $\norm{B_1}_F = 0$. If $m \in \calQ_o$, $\calX$ has only one nonzero element $-C_{lt_o}$. So, we have $$\norm{B_1}_F=\norm{\calX Z_o^{\tau,(m)}}_F \leq \norm{C_{l t_o} Z^{\tau,(m)}_{o,t_o,\cdot}}_2  \leq \norm{C}_{\infty} \norm{Z_o^{\tau,(m)}}_{2, \infty} \leq 2\norm{C}_{\infty} \norm{Z_o}_{2, \infty}$$ where $\norm{\cdot}_{\infty}$ is the max norm, and the last inequality follows from Lemma \ref{LemmaB11} (iv) provided that $$ \frac{\sigma\sqrt{\max\{N_o , T_o \}} }{\psi_{\min,o}}\ll \frac{1}{\sqrt{\kappa_o^2   \max \{\log N_o, \log T_o\}}}.$$ 
	Additionally, observe that Lemma \ref{LemmaB11} (iv) gives
	\begin{align*}
		\norm{C}_{\infty} & = \norm{X_o^{\tau,(m)}Q_o^{\tau,(m)}\left(Z_o^{\tau,(m)}Q_o^{\tau,(m)} \right)^\top  -X_oZ_o^\top }_{\infty}\\
		&\leq \norm{ X_o^{\tau,(m)}Q_o^{\tau,(m)} -X_o }_{2,\infty}\norm{   Z_o^{\tau,(m)}Q_o^{\tau,(m)}  }_{2, \infty} +\norm{X_o}_{2, \infty}\norm{ Z_o^{\tau,(m)}Q_o^{\tau,(m)} -Z_o }_{2,\infty}\\
  &\leq 3\norm{\calF_o}_{2, \infty}\norm{ \calF_o^{\tau,(m)}Q_o^{\tau,(m)} -\calF_o }_{2,\infty}. 
	\end{align*}
	Finally, we have
	\begin{align*}
		\norm{B_1}_F 
		\lesssim \norm{C}_{\infty} \norm{Z_o}_{2, \infty} 
		\lesssim \norm{\calF_o}_{2, \infty}^2\norm{ \calF_o^{\tau,(m)}Q_o^{\tau,(m)} -\calF_o }_{2,\infty} 
		\lesssim  \frac{\mu_o r}{\min\{N_o,T_o\}}  \psi_{\max,o}  \norm{ \calF_o^{\tau,(m)}Q_o^{\tau,(m)} -\calF_o }_{2,\infty}. 
	\end{align*}
	
	Now, assume that $N_o+1 \leq m \leq N_o+T_o$ and define $\breve{\calX} \coloneqq  \mathcal{P}_{\Omega_{\cdot, (m-N_o)}}(C)-\mathcal{P}_{\cdot,(m-N_o)}(C)$. First, if $m \neq N_o + t_o$, $\breve{\calX}=0$. If $m = N_o + t_o$, we have
	\begin{align*}
		&\norm{B_1}_F =\norm{\breve{\calX} Z_o^{\tau,(m)}}_F 
  = \norm{ \begin{bmatrix}
			 (\omega_{1t_o}-1)C_{1,t_o} \\
				\vdots \\
				(\omega_{N_o t_o}-1)C_{N_o,t_o}
			\end{bmatrix} Z^{\tau,(m)}_{o,t_o,\cdot}}_F 
    \leq   \norm{\begin{bmatrix}
			 (\omega_{1t_o}-1)C_{1,t_o} \\
				\vdots \\
				(\omega_{N_ot_o}-1)C_{N_o,t_o}
			\end{bmatrix}}_2 \norm{Z^{\tau,(m)}_{o,t_o,\cdot}}_2,
   \end{align*}
 Then, since
   \begin{align*}
 \norm{\begin{bmatrix}
			 (\omega_{1t_o}-1)C_{1,t_o} \\
				\vdots \\
				(\omega_{N_ot_o}-1)C_{N_o,t_o}
			\end{bmatrix}}_2 \norm{Z^{\tau,(m)}_{o,t_o,\cdot}}_2 \leq \sqrt{\sum_{i \in \calQ_o} C_{i,t_o}^2} \norm{Z^{\tau,(m)}_{o,t_o,\cdot}}_2  \leq \sqrt{\vartheta_o} \norm{C}_{\infty} \norm{Z_o}_{2,\infty},
\end{align*}
we can obtain
$$
\norm{B_1}_F \lesssim \frac{\sqrt{\vartheta_o} \mu_o r}{\min\{N_o,T_o\}}  \psi_{\max,o} \norm{ \calF_o^{\tau,(m)}Q_o^{\tau,(m)} -\calF_o }_{2,\infty}.
$$

\end{proof}

\begin{proof}[Proof of Claim \ref{ClaimB3}]
	First, assume that $m \leq N_o$. We follow the notation in the proof of Claim \ref{ClaimB2}. When $m \notin \calQ_o$, $X=0$. If $m \in \calQ_o$, we have
	\begin{align*}
		 \norm{B_2}_F&=\norm{\breve{\calX}^\top  X_o^{\tau,(m)}}_F = \norm{ \begin{bmatrix}
				0 \\
				\vdots \\
				-C_{mt_o}\\
				\vdots \\
				0
			\end{bmatrix} X^{\tau,(m)}_{o,m,\cdot}}_F 
			\leq 2\norm{C}_{\infty} \norm{X_o}_{2,\infty}.
	\end{align*}
So, we have	$$\norm{B_2}_F \lesssim  \frac{\mu_o r}{\min\{N_o,T_o\}}  \psi_{\max,o}  \norm{ \calF_o^{\tau,(m)}Q_o^{\tau,(m)} -\calF_o }_{2,\infty}.$$ Assume now that $N_o+1 \leq m \leq N_o+T_o$. Using the unitary invariance of Frobenius norm, we have $\norm{B_2}_F=\Big\lVert \breve{\calX}^\top X_o^{\tau,(m)}\Big\rVert_F$. If $m \neq N_o + t_o$, then $\breve{\calX}= 0$. In addition, if $m = N_o + t_o$, we obtain $$\norm{B_2}_F=\norm{\breve{\calX}^\top  X_o^{\tau,(m)}}_F = \norm{\sum_{i=1}^{N_o} \breve{\calX}_{i, \cdot}^\top  X^{\tau,(m)}_{o,i,\cdot}}_F =\norm{\sum_{i\in\calQ_o}-C_{i,t_o} X^{\tau,(m)}_{o,i,\cdot} }_2  \leq 2 \vartheta_o \norm{C}_{\infty} \norm{X_o}_{2,\infty}.$$
Therefore, we have $$\norm{B_2}_F \lesssim  \frac{\vartheta_o \mu_o r}{\min\{N_o,T_o\}}  \psi_{\max,o}  \norm{ \calF_o^{\tau,(m)}Q_o^{\tau,(m)} -\calF_o }_{2,\infty}.$$
\end{proof}

\begin{proof}[Proof of Claim \ref{ClaimB4}]
	First, we bound $C_1$. Assume that $m \leq N_o$. Since Frobenius norm is unitary invariant, we have
$$
\norm{C_1}_F = 
		 \norm{\mathcal{P}_{\Omega_{m,\cdot}}(\calE_o) Z_o^{\tau,(m)}}_F =
		\norm{\sum_{t=1}^{T_o} \underbrace{\omega_{mt} \epsilon_{mt} Z^{\tau,(m)}_{o,t, \cdot}}_{u_{mt}} }_2.
$$ 
  By the way of construction of leave-one-out estimates, $\{\epsilon_{mt}\}_{1\leq t \leq T_o}$ are independent of $Z_o^{\tau,(m)}$. Therefore, we have $\bbE\left[  \epsilon_{mt}
	\left| Z_o^{\tau,(m)} \right. \right]=\bbE\left[   \epsilon_{mt} \right]  =0$, and conditioning on $Z_o^{\tau,(m)}$, $\{\epsilon_{mt}\}_{1\leq t \leq T_o}$ are independent across $t$. Hence, conditioning on $Z_o^{\tau,(m)}$, we can exploit the  matrix Bernstein inequality \citep[][Proposition 2]{koltchinskii:2011}. Note that 
 $$ 
 \norm{\norm{u_{mt}}_2}_{\subE} \leq \norm{Z_o^{\tau,(m)}}_{2, \infty} \norm{\omega_{mt}\epsilon_{mt}}_{\subE} \lesssim \sigma \norm{Z_o^{\tau,(m)}}_{2, \infty},
 $$
 where $\norm{\cdot}_{\subE}$ denotes the sub-exponential norm; see \cite{koltchinskii:2011, tropp2015introduction}. Further, we can see that
	\[
	\norm{ \sum_{t=1}^{T_o} \omega_{mt}^2 \bbE\left[\epsilon_{mt}^2 \left|       Z_o^{\tau,(m)}   \right.  \right] Z^{\tau,(m)}_{o,t, \cdot}Z^{\tau,(m) \top}_{o,t, \cdot} } \lesssim
	\sigma^2 \norm{ \sum_{t=1}^{T_o} Z^{\tau,(m)}_{o,t, \cdot}Z^{\tau,(m) \top}_{o,t, \cdot}} = \sigma^2 \norm{Z_o^{\tau,(m)}}_F^2.
	\]
	Then, the matrix Bernstein inequality reveals that, with probability at least $1- O (\min \{N_o^{-100}, T_o^{-100}\})$,
	\begin{align*}
		\norm{\sum_{t=1}^{T_o} u_{mt}}_2  &\lesssim \sqrt{\sigma^2 \norm{Z_o^{\tau,(m)}}_F^2 \max \{\log N_o, \log T_o\}}+\sigma \norm{Z_o^{\tau,(m)}}_{2, \infty} \max\{\log^2 N_o, \log^2 T_o\}\\
  &\lesssim \sigma \sqrt{ \max \{N_o\log N_o, T_o\log T_o\}}\norm{Z_o^{\tau,(m)}}_{2, \infty},
	\end{align*}
	where the last relation uses the assumption $\max \{N_o,T_o\} \gg \max \{\log^3 N_o , \log^3 T_o\}$. 
	Applying Lemma \ref{LemmaB11} (iv) with the assumption $$ \frac{\sigma \sqrt{\max\{N_o , T_o \}}}{\psi_{\min,o}} \ll \frac{1}{\sqrt{\kappa_o^2   \max \{\log N_o, \log T_o\}}},$$ we reach, with probability at least $1-O\left( \min \{N_o^{-100}, T_o^{-100}\}\right)$, $$\norm{C_1}_F \lesssim \sigma \sqrt{ \max \{N_o\log N_o, T_o\log T_o\}} \norm{ \calF_o}_{2, \infty}.$$ 
	
	Now, we consider the case of $m \geq N_o+1$. Since Frobenius norm is unitary invariant and only the $(m-N_o)$-th column of the matrix $\mathcal{P}_{\Omega_{\cdot, (m-N_o)}}(\calE_o)$ has nonzero elements,
	\begin{align*}
		\norm{C_1}_F &= 
		\norm{\mathcal{P}_{\Omega_{\cdot,(m-N_o)}}(\calE_o) Z_o^{\tau,(m)}}_F= \norm{  \begin{bmatrix}
				 \omega_{1,(m-N_o)}\epsilon_{1,(m-N_o)} \\
				\vdots \\
				 \omega_{N_o,(m-N_o)}\epsilon_{N_o,(m-N_o)}
			\end{bmatrix} Z^{\tau,(m)}_{o,(m-N_o),\cdot}}_F \\
		& \leq  \norm{\sum_{i =1}^{N_o} \underbrace{e_i \omega_{i,(m-N_o)}\epsilon_{i,(m-N_o)}Z^{\tau,(m)}_{o,(m-N_o), \cdot}}_{\coloneqq u_{i,(m-N_o)}}}_F. 
	\end{align*}
Similarly, conditioning on $\{Z_o^{\tau,(m)}\}$, we can exploit the  matrix Bernstein inequality \citep[][Proposition 2]{koltchinskii:2011}. Note that 
$$ 
\norm{\norm{u_{i,(m-N_o)}}_F}_{\subE} \leq \norm{Z_o^{\tau,(m)}}_{2, \infty}\norm{ \epsilon_{i,(m-N_o)}}_{\subE} \lesssim \sigma \norm{Z_o^{\tau,(m)}}_{2, \infty} \text{  and  }
$$ 
	\begin{align*}
		\norm{ \sum_{i =1}^{N_o}  \omega_{i,(m-N_o)} \bbE\left[    \epsilon_{i,(m-N_o)}^2 \left|  Z_o^{\tau,(m)}   \right.  \right] e_i Z^{\tau,(m)}_{o,(m-N_o), \cdot}Z^{\tau,(m) \top}_{o,(m-N_o), \cdot} e_i^\top  } \lesssim
		N_o \sigma^2 \norm{Z_o^{\tau,(m)}}_{2, \infty}^2.
	\end{align*}
	Then, the matrix Bernstein inequality reveals that, with probability at least $1- O (\min \{N_o^{-101}, T_o^{-101}\})$,
	\begin{align*}
		\norm{\sum_{i=1}^{N_o}  u_{i,(m-N_o)}}_F 
		& \lesssim \sqrt{\sigma^2 \norm{Z_o^{\tau,(m)}}_{2,\infty}^2 \max \{N_o\log N_o, T_o\log T_o\}}+\sigma \norm{Z_o^{\tau,(m)}}_{2, \infty} \max\{\log^2 N_o, \log^2 T_o\} \\
		& \lesssim \sigma \sqrt{\max \{N_o\log N_o, T_o\log T_o\}}\norm{Z_o^{\tau,(m)}}_{2, \infty},
	\end{align*}
	where the last relation uses the assumption $\max \{N_o,T_o\} \gg \max \{\log^3 N_o , \log^3 T_o\}$. 
	Applying Lemma \ref{LemmaB11} (iv) with the assumption $$ \frac{\sigma \sqrt{\max\{N_o , T_o \}}}{\psi_{\min,o}} \ll \frac{1}{\sqrt{\kappa_o^2   \max \{\log N_o, \log T_o\}}},$$ we reach, with probability at least $1-O\left( \min \{N_o^{-100}, T_o^{-100}\}\right)$,
	$$\norm{C_1}_F \lesssim \sigma \sqrt{   \max\{ N_o \log N_o,T_o \log T_o\} } \norm{ \calF_o}_{2, \infty}.$$
	
	We turn to $C_2$. Assume $m \leq N_o$. Since Frobenius norm is unitary invariant, we have
	\begin{align*}
		\norm{C_2}_F &= 
		 \norm{\left(\mathcal{P}_{\Omega_{m,\cdot}}(\calE_o)\right)^\top  X_o^{\tau,(m)}}_F= \norm{ \begin{bmatrix}
				 \omega_{m1}\epsilon_{m1} \\
				\vdots \\
				\omega_{mt}\epsilon_{mt}
			\end{bmatrix} X^{\tau,(m)}_{o,m,\cdot}}_F 
		 = \norm{\sum_{t=1}^{T_o} \underbrace{e_t  \omega_{mt}\epsilon_{mt}X^{\tau,(m)}_{o,m, \cdot}}_{\coloneqq u_{mt}}}_F. 
	\end{align*}
Similarly, conditioning on $X_o^{\tau,(m)}$, we can exploit the  matrix Bernstein inequality. Note that $ \norm{\norm{u_{mt}}_F}_{\subE} \lesssim \sigma \norm{X_o^{\tau,(m)}}_{2, \infty}$ and $$\norm{ \sum_{t=1}^{T_o}  \omega_{mt} \bbE\left[  \epsilon_{mt}^2 \left| X_o^{\tau,(m)}   \right.  \right] e_t X^{\tau,(m)}_{o,m, \cdot}X^{\tau,(m) \top}_{o,m, \cdot} e_t^\top  }  \lesssim
		\sigma^2 \norm{ \sum_{t=1}^{T_o} X^{\tau,(m)}_{o,m, \cdot}X^{\tau,(m) \top}_{o,m, \cdot}} \leq   T_o \sigma^2 \norm{X_o^{\tau,(m)}}_{2, \infty}^2.$$
	Then, the matrix Bernstein inequality reveals that, with probability at least $1- O (\min \{N_o^{-100}, T_o^{-100}\})$,
	\begin{align*}
	\norm{\sum_{t=1}^{T_o} u_{t}}_F 
		& \lesssim \sqrt{\sigma^2 \norm{X_o^{\tau,(m)}}_{2,\infty}^2 \max \{N_o\log N_o, T_o\log T_o\}}+  \sigma \norm{X_o^{\tau,(m)}}_{2, \infty} \max\{\log^2 N_o, \log^2 T_o\} \\
		& \lesssim \sigma \sqrt{ \max \{N_o\log N_o, T_o\log T_o\}}\norm{X_o^{\tau,(m)}}_{2, \infty},
	\end{align*}
	where the last relation uses the assumption $\max \{N_o,T_o\} \gg \max\{\log^3 N_o , \log^3 T_o\}$. 
	Applying Lemma \ref{LemmaB11} (iv) with the assumption $$\frac{\sigma \sqrt{\max\{N_o , T_o \}}}{\psi_{\min,o}} \ll \frac{1}{\sqrt{\kappa_o^2   \max \{\log N_o, \log T_o\}}},$$ we reach, with probability at least $1-O\left( \min \{N_o^{-100}, T_o^{-100}\}\right)$,
	$$\norm{C_2}_F \lesssim \sigma \sqrt{\max\{ N_o \log N_o,T_o \log T_o\}} \norm{ \calF_o}_{2, \infty}.$$
	
	Now, assume that $m \geq N_o+1$. Since Frobenius norm is unitary invariant and only $(m-N_o)$-th column of the matrix $\mathcal{P}_{\Omega_{\cdot, (m-N_o)}}(\calE_o)$ has nonzero elements,
	\begin{align*}
		\norm{C_2}_F &= 
		\norm{\left(\mathcal{P}_{\Omega_{\cdot, (m-N_o)}}(\calE_o)\right)^\top  X_o^{\tau,(m)}}_F = \norm{\sum_{i=1}^{N_o} \underbrace{ \omega_{i,(m-N_o)} \epsilon_{i,(m-N_o)} X^{\tau,(m)}_{o,i, \cdot}}_{\coloneqq u_{i,(m-N_o)}}}_2.
	\end{align*}
	Conditioning on $X_o^{\tau,(m)}$, the matrix Bernstein inequality reveals that, with probability at least $1- O (\min \{N_o^{-100}, T_o^{-100}\})$,
	\begin{align*}
		\norm{\sum_{i=1}^{N_o}  u_{i,(m-N_o)}}_2 
		& \lesssim \sqrt{\sigma^2 \norm{X_o^{\tau,(m)}}_F^2 \max \{\log N_o, \log T_o\}}+\sigma \norm{X_o^{\tau,(m)}}_{2, \infty} \max\{\log^2 N_o, \log^2 T_o\} \\
		& \lesssim \sigma \sqrt{  \max \{N_o\log N_o, T_o\log T_o\}}\norm{X_o^{\tau,(m)}}_{2, \infty},
	\end{align*}
	where the last relation uses the assumption $\max \{N_o,T_o\} \gg \max \{\log^3 N_o , \log^3 T_o\}$. 
	Applying Lemma \ref{LemmaB11} (iv) with the assumption $$ \frac{\sigma \sqrt{\max\{N_o , T_o \}} }{\psi_{\min,o}}\ll \frac{1}{\sqrt{\kappa_o^2   \max \{\log N_o, \log T_o\}}},$$ we reach, with probability at least $1-O\left( \min \{N_o^{-100}, T_o^{-100}\}\right)$, $$\norm{C_2}_F \lesssim \sigma \sqrt{\max\{ N_o \log N_o,T_o \log T_o\}} \norm{ \calF_o}_{2, \infty}.$$
\end{proof}
\bigskip

The following two lemmas are the modifications of Section 4.2 of \cite{candes2009exact} for our missing pattern. The way of proof is different from that of \cite{candes2009exact} since we assume missing not at random. These lemmas are used in many parts of proofs.

\begin{lemma}\label{LemmaCR091}
	Define $\mathcal{P}_{T^*}(A)=U_oU_o^\top A+AV_oV_o^\top -U_oU_o^\top AV_oV_o^\top $. Assume that $ \frac{ \vartheta_o \mu_o r}{ \min\{N_o,T_o\}} \ll 1$. Then, we have
	\begin{align}
		\sqrt{\frac{99}{100}}\norm{\mathcal{P}_{T^*}(A)}_F \leq \norm{\mathcal{P}_{\Omega_o} \mathcal{P}_{T^*} (A)}_F \leq \sqrt{\frac{101}{100}}\norm{\mathcal{P}_{T^*}(A)}_F.  
	\end{align}
\end{lemma}

\begin{proof}
	We have by Lemma \ref{LemmaCR093}
	\begin{align*}
	 \big|\norm{\mathcal{P}_{\Omega_o} \mathcal{P}_{T^*} (A)}_F^2 - \norm{\mathcal{P}_{T^*}(A)}_F^2 \big| &=
	\big|\langle \mathcal{P}_{\Omega_o} \mathcal{P}_{T^*}(A),\mathcal{P}_{\Omega_o} \mathcal{P}_{T^*}(A) \rangle - \langle \mathcal{P}_{T^*}(A), \mathcal{P}_{T^*}(A) \rangle \big|
		\\
		&=\big| \langle (\mathcal{P}_{\Omega_o}\mathcal{P}_{T^*}-\mathcal{P}_{T^*})(A), \mathcal{P}_{T^*} (A) \rangle \big|\\
		&=\big|\langle (\mathcal{P}_{T^*}  \mathcal{P}_{\Omega_o}\mathcal{P}_{T^*}-\mathcal{P}_{T^*})(A),\mathcal{P}_{T^*}(A) \rangle\big|\\ 
		&\leq \norm{(\mathcal{P}_{T^*} \mathcal{P}_{\Omega_o}\mathcal{P}_{T^*}-\mathcal{P}_{T^*})(A)}_F\norm{\mathcal{P}_{T^*}(A)}_F\\
  &\leq \norm{\mathcal{P}_{T^*}  \mathcal{P}_{\Omega_o}\mathcal{P}_{T^*}-\mathcal{P}_{T^*}}\norm{\mathcal{P}_{T^*}(A)}_F^2\\
  &\leq 0.01 \norm{\mathcal{P}_{T^*}(A)}_F^2.
	\end{align*}
\end{proof}

\begin{lemma}\label{LemmaCR093}
	Under the incoherence assumption, we have
	\[
	\norm{\mathcal{P}_{T^*} \mathcal{P}_{\Omega_o}\mathcal{P}_{T^*}-\mathcal{P}_{T^*}} \leq \frac{2 \vartheta_o \mu_o r}{ \min\{N_o,T_o\}}.
	\]
\end{lemma}

\begin{proof}
	Let $(e_i^{N_o})_{i \in [N_o]}, (e_t^{T_o})_{t \in [T_o]}$ be the standard basis vectors for $\mathbb{R}^{N_o}$ and $\mathbb{R}^{T_o}$, respectively. Then $A \in \mathbb{R}^{N_o \times T_o}$ can be written as $A = \sum_{(i,t) \in [N_o] \times [T_o]}\langle A, e^{N_o}_i e^{T_o \top}_t \rangle e^{N_o}_i e^{T_o \top}_t$.
	Further, we can readily obtain
	\begin{align*}
		&\mathcal{P}_{T^*}(A)= \sum_{i,t}\langle \mathcal{P}_{T^*}(A), e^{N_o}_i e^{T_o \top}_t\rangle e^{N_o}_i e^{T_o \top}_t = \sum_{i,t} \langle A, \mathcal{P}_{T^*}(e^{N_o}_ie^{T_o \top}_t)\rangle e^{N_o}_ie^{T_o \top}_t, \\
		&\mathcal{P}_{\Omega_o}\mathcal{P}_{T^*}(A)= \sum_{i,t}\omega_{it} \langle A, \mathcal{P}_{T^*}(e^{N_o}_ie^{T_o \top}_t)\rangle e^{N_o}_ie^{T_o \top}_t,
		\mathcal{P}_{T^*} \mathcal{P}_{\Omega_o}\mathcal{P}_{T^*}(A) = \sum_{i,t} \omega_{it} \langle A, \mathcal{P}_{T^*}(e^{N_o}_ie^{T_o \top}_t)\rangle \mathcal{P}_{T^*}(e^{N_o}_ie^{T_o \top}_t).
	\end{align*}
	By defining an outer product $\otimes$ as $(A \otimes B)(C)= \langle B,C \rangle A$, we also have 
$$
\mathcal{P}_{T^*} \mathcal{P}_{\Omega_o}\mathcal{P}_{T^*} =\sum_{i,t} \omega_{it}  \mathcal{P}_{T^*}(e^{N_o}_ie^{T_o \top}_t) \otimes \mathcal{P}_{T^*}(e^{N_o}_ie^{T_o \top}_t)
$$ 
and $\mathcal{P}_{T^*}  = \sum_{i,t}   \mathcal{P}_{T^*}(e^{N_o}_ie^{T_o \top}_t) \otimes \mathcal{P}_{T^*}(e^{N_o}_ie^{T_o \top}_t)$. Hence, we have
	\begin{align*}
		\mathcal{P}_{T^*} \mathcal{P}_{\Omega_o}\mathcal{P}_{T^*}-\mathcal{P}_{T^*}
		 =\sum_{i,t}(\omega_{it}-1)\mathcal{P}_{T^*}(e^{N_o}_ie^{T_o \top}_t) \otimes \mathcal{P}_{T^*}(e^{N_o}_ie^{T_o \top}_t) 
		 = \sum_{i \in \calQ_o}\mathcal{P}_{T^*}(e^{N_o}_{i}e^{T_o \top}_{t_o}) \otimes \mathcal{P}_{T^*}(e^{N_o}_{i}e^{T_o \top}_{t_o}).
	\end{align*}
	By the definition of $\mathcal{P}_{T^*}$, 
 $$ \norm{\mathcal{P}_{T^*}(e^{N_o}_{i}e^{T_o \top}_{t_o})}_F^2 = \langle \mathcal{P}_{T^*}(e^{N_o}_{i}e^{T_o \top}_{t_o}), e^{N_o}_{i}e^{T_o \top}_{t_o} \rangle = \norm{U_oU_o^\top e^{N_o}_{i}}^2+\norm{V_oV_o^\top e^{T_o \top}_{t_o}}^2-\norm{U_oU_o^\top e^{N_o}_{i}}^2\norm{V_oV_o^\top e^{T_o \top}_{t_o}}^2 .$$
 Due to the incoherence condition, $\norm{U_oU_o^\top e_{i}^{N_o}}^2 \leq \mu_o r /N_o$ and $\norm{V_oV_o^\top e_{t_o}^{T_o}}^2\leq \mu_o r / T_o$. Then, we have $$\norm{\mathcal{P}_{T^*}(e^{N_o}_{i}e^{T_o \top}_{t_o})}_F^2 \leq 2 \mu_o r / \min\{N_o,T_o\}.$$ 
 Note that 
 $$
 \norm{ \mathcal{P}_{T^*}(e^{N_o}_{i}e^{T_o \top}_{t_o}) \otimes \mathcal{P}_{T^*}(e^{N_o}_{i}e^{T_o \top}_{t_o})} = \sup \langle B_1,  \mathcal{P}_{T^*}(e^{N_o}_{i}e^{T_o \top}_{t_o}) \rangle\langle \mathcal{P}_{T^*}(e^{N_o}_{i}e^{T_o \top}_{t_o}), B_2 \rangle
 $$
 where the supremum is taken over a countable collection of matrices $B_1$ and $B_2$ such that $\norm{B_1}_F \leq 1$ and $\norm{B_2}_F \leq 1$. Then, for all $i \in \calQ_o$, we have
\begin{align*}
\norm{\mathcal{P}_{T^*}(e^{N_o}_{i}e^{T_o \top}_{t_o}) \otimes \mathcal{P}_{T^*}(e^{N_o}_{i}e^{T_o \top}_{t_o})}
    &\leq |\langle B_1,\mathcal{P}_{T^*}(e^{N_o}_{i}e^{T_o \top}_{t_o}) \rangle| |\langle \mathcal{P}_{T^*}(e^{N_o}_{i}e^{T_o \top}_{t_o}),B_2 \rangle|\\
    &\leq  \norm{\mathcal{P}_{T^*}(e^{N_o}_{i}e^{T_o \top}_{t_o})}_F^2\\
    &\leq \frac{2 \mu_o r}{ \min\{N_o,T_o\}}.
\end{align*}
Hence, we have
	\begin{align*}
		\norm{ \mathcal{P}_{T^*} \mathcal{P}_{\Omega_o}\mathcal{P}_{T^*}-\mathcal{P}_{T^*}} 
		&\leq  \sum_{i \in \calQ_o}\norm{\mathcal{P}_{T^*}(e^{N_o}_{i}e^{T_o \top}_{t_o}) \otimes \mathcal{P}_{T^*}(e^{N_o}_{i}e^{T_o \top}_{t_o})}\\
  &\leq \vartheta_o \max_{i \in \calQ_o}\norm{\mathcal{P}_{T^*}(e^{N_o}_{i}e^{T_o \top}_{t_o}) \otimes \mathcal{P}_{T^*}(e^{N_o}_{i}e^{T_o \top}_{t_o})}\\
  &\leq \frac{2 \vartheta_o \mu_o r}{ \min\{N_o,T_o\}}.
	\end{align*}
\end{proof}
\bigskip

The following lemma is a simple modification of Lemma \ref{LemmaB7}. Using this lemma, we can change $\norm{\calF_o}_{2, \infty}$ with $\norm{X_o}_{2, \infty}$ and $\norm{Z_o}_{2, \infty}$ at the cost of having an additional term $\sqrt{r}$.

\begin{lemma}\label{LemmaBnew}
	Suppose that $\lambda_o=C_{\lambda} \sigma \sqrt{\max\{N_o,T_o\}}$ for some large constant $C_{\lambda}>0$, $0< \eta_o \ll 1/(\kappa_o^{2}\psi_{\max,o} \min\{N_o,T_o\})$, $\min \{N_o, T_o\} \gg \vartheta_o \kappa_o^2 \mu_o r$, and
 \begin{align*}
\frac{\sigma}{\psi_{\min,o}}\sqrt{\frac{\max\{N_o^2, T_o^2\}}{\min\{N_o,T_o\}}} \ll \frac{1}{\sqrt{\kappa_o^4 \mu_o r \max \{\log N_o, \log T_o\}}},  \ \ \min \{N_o, T_o\} \gg \kappa_o \mu_o r \max \{\log^3 N, \log^3 T\} .
 \end{align*}
Suppose also that the iterates satisfy (\ref{Prelim3})-(\ref{Prelim8}) at the $\tau$-th iteration, then with probability at least $1-O(\min\{N_o^{-99}, T_o^{-99}\}),$ we have
	\begin{align*}
	   & \norm{X_o^{\tau+1} H_o^{\tau+1}- X_o}_{2, \infty} \leq C_{\infty,X}  r^{1/2} \kappa_o \left(\frac{\sigma \sqrt{\max\{N_o \log N_o,T_o \log T_o\}}}{\psi_{\min,o}}+\frac{\lambda_o}{\psi_{\min,o}} \right) \norm{X_o}_{2, \infty},\\
     & \norm{Z_o^{\tau+1} H_o^{\tau+1}- Z_o}_{2, \infty} \leq C_{\infty,Z}  r^{1/2} \kappa_o \left(\frac{\sigma \sqrt{\max\{N_o \log N_o,T_o \log T_o\}}}{\psi_{\min,o}}+\frac{\lambda_o}{\psi_{\min,o}} \right) \norm{Z_o}_{2, \infty},
	\end{align*}
where $C_{\infty,X}$ and $C_{\infty,Z}$ are some sufficiently large constants.
\end{lemma}

\begin{proof}
By some modification of Lemma \ref{LemmaB5}, we can have
 \begin{align*}
   &\max_{1\leq m\leq N_o}  \norm{\calF_o^{\tau+1} H_o^{\tau+1}-\calF_o^{\tau+1,(m)}Q_o^{\tau+1,(m)}}_F \leq C_{3,X} \sqrt{r} \left(\frac{\sigma \sqrt{\max\{N_o \log N_o,T_o \log T_o\}}}{\psi_{\min,o}}+\frac{\lambda_o}{\psi_{\min,o}} \right) \norm{X_o}_{2, \infty},\\
   &\max_{N_o+1\leq m\leq T_o}  \norm{\calF_o^{\tau+1} H_o^{\tau+1}-\calF_o^{\tau+1,(m)}Q_o^{\tau+1,(m)}}_F\\
   & \ \ \leq C_{3,Z} \sqrt{r} \left(\frac{\sigma \sqrt{\max\{N_o \log N_o,T_o \log T_o\}}}{\psi_{\min,o}}+\frac{\lambda_o}{\psi_{\min,o}} \right) \norm{Z_o}_{2, \infty}.
 \end{align*}
In addition, by some modification of Lemma \ref{LemmaB6}, we have
\begin{align*}
	&\max_{1\leq m\leq N_o}  \norm{\left(\calF_o^{\tau+1,(m)}H_o^{\tau+1,(m)}-\calF_o\right)_{m,\cdot}}_2 \leq C_{4,X} \kappa_o \left(\frac{\sigma \sqrt{\max\{N_o \log N_o,T_o \log T_o\}}}{\psi_{\min,o}}+\frac{\lambda_o}{\psi_{\min,o}} \right)\norm{X_o}_{2, \infty},\\
 &\max_{N_o+ 1 \leq m\leq T_o}  \norm{\left(\calF_o^{\tau+1,(m)}H_o^{\tau+1,(m)}-\calF_o\right)_{m,\cdot}}_2 \leq C_{4,Z} \kappa_o \left(\frac{\sigma \sqrt{\max\{N_o \log N_o,T_o \log T_o\}}}{\psi_{\min,o}}+\frac{\lambda_o}{\psi_{\min,o}} \right)\norm{Z_o}_{2, \infty}.
  \end{align*}
Then, when $1\leq m\leq N_o$, we have with probability at least $1-O(\min\{N_o^{-99},T_o^{-99}\})$
	\begin{align}
		\norm{\left(X_o^{\tau+1} H_o^{\tau+1}-X_o\right)_{m, \cdot}}_2 &\leq \norm{\left(\calF_o^{\tau+1} H_o^{\tau+1}-\calF_o\right)_{m, \cdot}}_2 \nonumber\\
		& \leq \norm{\left(\calF_o^{\tau+1} H_o^{\tau+1}-\calF_o^{\tau+1,(m)} H_o^{\tau+1,(m)}\right)_{m, \cdot}}_2 +\norm{\left(\calF_o^{\tau+1,(m)} H_o^{\tau+1,(m)}-\calF_o\right)_{m, \cdot}}_2 \nonumber\\
		& \leq \norm{ \calF_o^{\tau+1} H_o^{\tau+1}-\calF_o^{\tau+1,(m)} H_o^{\tau+1,(m)} }_F \nonumber\\
  &\ \ + C_{4,\infty}  \kappa_o \left(\frac{\sigma \sqrt{\max\{N_o \log N_o,T_o \log T_o\}}}{\psi_{\min,o}}+\frac{\lambda_o}{\psi_{\min,o}} \right) \norm{X_o}_{2, \infty}, \label{LemmaB7.1}
	\end{align}
	For the first term, use Lemma \ref{LemmaB11} to have, with probability at least $1-O(\min\{N_o^{-99},T_o^{-99}\})$
	\begin{align}
		\norm{ \calF_o^{\tau+1} H_o^{\tau+1}-\calF_o^{\tau+1,(m)} H_o^{\tau+1,(m)} }_F &\leq 5 \kappa_o \norm{ \calF_o^{\tau+1} H_o^{\tau+1}-\calF_o^{\tau+1,(m)} Q_o^{\tau+1,(m)} }_F \nonumber\\
		&\leq 5 \kappa_o C_{3,X} \sqrt{r} \left(\frac{\sigma \sqrt{\max\{N_o \log N_o,T_o \log T_o\}}}{\psi_{\min,o}}+\frac{\lambda_o}{\psi_{\min,o}} \right) \norm{X_o}_{2, \infty} \label{LemmaB7.2}
	\end{align}
	Then, (\ref{LemmaB7.1}) and (\ref{LemmaB7.2}) collectively reveal that, with probability at least $1-O(\min\{N_o^{-99},T_o^{-99}\})$,
	\begin{align*}
		\norm{\left(X_o^{\tau+1} H_o^{\tau+1}-X_o\right)_{m, \cdot}}_2  \leq C_{\infty,X} \sqrt{r} \kappa_o \left(\frac{\sigma \sqrt{\max\{N_o \log N_o,T_o \log T_o\}}}{\psi_{\min,o}}+\frac{\lambda_o}{\psi_{\min,o}} \right) \norm{X_o}_{2, \infty}
	\end{align*}
	under the assumption that $C_{\infty,X} \geq 5 C_{3,X} + C_{4,X}$. Similarly, we can show the bound for $\norm{\left(Z_o^{\tau+1} H_o^{\tau+1}-Z_o\right)_{m, \cdot}}_2$.
\end{proof}
\bigskip

The following lemmas are the simple modified versions of the lemmas in \cite{chen2020noisy}. With the aids of Lemmas \ref{LemmaCR091} and \ref{LemmaCR093}, if we follow their proofs by setting $p=1$ while considering our observation pattern cautiously, we can get the results. To save space, we omit the proofs. However, we are willing to provide the full proofs upond request.

 \begin{lemma}\label{LemmaB2}
	Suppose that  $\lambda_o=C_{\lambda} \sigma \sqrt{\max\{N_o,T_o\}}$ for some large constant $C_{\lambda}>0$, $\overbar{\tau} = \max\{N_o^{23}, T_o^{23}\}$ and $\eta_o \overset{c}{\asymp}  1/\max\{N_o^6,T_o^6\}\kappa_o^3 \psi_{\max,o}$. Suppose also that
$$\frac{\sigma}{\psi_{\min,o}}\sqrt{\frac{\max\{N_o^2, T_o^2\}}{\min\{N_o, T_o\}}} \ll \frac{1}{\sqrt{\kappa_o^4 \mu_o r \max \{\log N_o, \log T_o\}}}, \ \ \min\{N_o,T_o\} \gg \mu_o r \kappa_o \max \{\log^2N_o, \log^2T_o\},$$ and the induction hypotheses (\ref{Prelim3})-(\ref{Prelim8}) hold for all $0 \leq \tau \leq \overbar{\tau}$ and (\ref{Prelim9}) holds for all $1 \leq \tau \leq \overbar{\tau}$. Then there is a constant $C_{gr}>0$ such that
	\[
	\min_{0\leq \tau < \overbar{\tau}}\norm{\nabla f(X_o^{\tau},Z_o^{\tau})}_F \leq C_{gr}\frac{1}{\max \{N_o^5, T_o^5\}} \lambda_o \sqrt{\psi_{\min,o}}.
	\]
\end{lemma}

\begin{lemma}\label{LemmaB3}
	Suppose that $\lambda_o=C_{\lambda} \sigma \sqrt{\max\{N_o,T_o\}}$ for some large constant $C_{\lambda}>0$, $$\frac{\sigma}{\psi_{\min,o}}\sqrt{\frac{\max\{N_o^2, T_o^2\}}{\min\{N_o, T_o\}}} \ll \frac{1}{\sqrt{\kappa_o^4 \mu_o r \max \{\log N_o, \log T_o\}}}, \ \ \min \{N_o, T_o\} \gg \mu_o r \kappa_o \max\{ \log^2N_o,  \log^2T_o\}$$ and $0< \eta_o \ll 1/(\kappa_o^{5/2}\psi_{\max,o})$. Suppose also that the iterates satisfy (\ref{Prelim3})-(\ref{Prelim8}) at the $\tau$-th iteration, then with probability at least $1-O(\min\{N_o^{-100}, T_o^{-100}\}),$ 
	\[
	\norm{\calF_o^{\tau+1}H_o^{\tau+1}-\calF_o}_F \leq C_F \left( \frac{\sigma \sqrt{\max\{N_o,T_o\}}}{\psi_{\min,o}} + \frac{\lambda_o}{\psi_{\min,o}}\right)\norm{X_o}_F,
	\]
	where $C_F>0$ is large enough.
\end{lemma}

\begin{lemma}\label{LemmaB4}
	Suppose $\lambda_o=C_{\lambda} \sigma \sqrt{\max\{N_o,T_o\}}$ for some large constant $C_{\lambda}>0$, 
$$\frac{\sigma \sqrt{\max\{N_o, T_o\}}}{\psi_{\min,o}} \ll \frac{1}{\sqrt{\kappa_o^4  \max \{\log N_o, \log T_o\}}}, \ \ \min \{N_o^2, T_o^2\} \gg \kappa_o^4 \mu_o^2 r^2 \max\{ N_o\log N_o,  T_o \log T_o\},$$ and $0< \eta_o \ll 1/(\kappa_o^{3}\psi_{\max,o} \sqrt{r})$. Suppose also that the iterates satisfy (\ref{Prelim3})-(\ref{Prelim8}) at the $\tau$-th iteration, then with probability at least $1-O(\min\{N_o^{-100}, T_o^{-100}\}),$ 
	\[ \norm{\calF_o^{\tau+1} H_o^{\tau+1}-\calF_o} \leq C_{op} \left(\frac{\sigma \sqrt{\max\{N_o,T_o\}}}{\psi_{\min,o}}+\frac{\lambda_o}{\psi_{\min,o}} \right) \norm{X_o}\]
	provided that $C_{op}$ is sufficiently large.
\end{lemma}

\begin{lemma}\label{LemmaB6}
	Suppose that $\lambda_o=C_{\lambda} \sigma \sqrt{\max\{N_o,T_o\}}$ for some large constant $C_{\lambda}>0$,
	$$
 \frac{\sigma \sqrt{\max\{N_o, T_o\}}}{\psi_{\min,o}} \ll \frac{1}{\sqrt{\kappa_o^2   \max \{\log N_o, \log T_o\}}}, \ \ \min \{N_o^2, T_o^2\} \gg  \kappa_o^2 \mu_o^2 r^2 \max\{ N_o\log  N_o,  T_o \log  T_o \},$$ and $0< \eta_o \ll 1/(\kappa_o^{2}\sqrt{r}\psi_{\max,o} )$. Suppose also that the iterates satisfy (\ref{Prelim3})-(\ref{Prelim8}) at the $\tau$-th iteration, then with probability at least $1-O(\min\{N_o^{-99}, T_o^{-99}\}),$ 
	\[
	\max_{1\leq m\leq N_o+T_o}  \norm{\left(\calF_o^{\tau+1,(m)}H_o^{\tau+1,(m)}-\calF_o\right)_{m,\cdot}}_2 \leq C_4 \kappa_o \left(\frac{\sigma \sqrt{\max\{N_o \log N_o,T_o \log T_o\}}}{\psi_{\min,o}}+\frac{\lambda_o}{\psi_{\min,o}} \right)\norm{\calF_o}_{2, \infty}.
	\]
\end{lemma}

\begin{lemma}\label{LemmaB7}
	Suppose that $\lambda_o=C_{\lambda} \sigma \sqrt{\max\{N_o,T_o\}}$ for some large constant $C_{\lambda}>0$, 
 \begin{align*}
   &\frac{\sigma}{\psi_{\min,o}}\sqrt{\frac{\max\{N_o^2, T_o^2\}}{\min\{N_o,T_o\}}} \ll \frac{1}{\sqrt{\kappa_o^4 \mu_o r \max \{\log N_o, \log T_o\}}}, \\ 
   &\min \{N_o^2, T_o^2\} \gg  \kappa_o^4 \mu_o^2 r^2 \max \{N_o \log^2 N_o, T_o \log^2 T_o\}.  
 \end{align*}
 Suppose also that the iterates satisfy (\ref{Prelim3})-(\ref{Prelim8}) at the $\tau$-th iteration, then with probability at least $1-O(\min\{N_o^{-98}, T_o^{-98}\}),$ 
	\[
	\norm{\calF_o^{\tau+1} H_o^{\tau+1}-\calF_o}_{2, \infty} \leq C_{\infty}  \kappa_o \left(\frac{\sigma \sqrt{\max\{N_o \log N_o,T_o \log T_o\}}}{\psi_{\min,o}}+\frac{\lambda_o}{\psi_{\min,o}} \right) \norm{\calF_o}_{2, \infty}.
	\]
	holds as long as $C_{\infty} \geq 5C_3+C_4$.
\end{lemma}

\begin{lemma}\label{LemmaB8}
	Suppose $\lambda_o=C_{\lambda} \sigma \sqrt{\max\{N_o,T_o\}}$ for some large constant $C_{\lambda}>0$,
	$$ 
 \frac{\sigma \sqrt{\max\{N_o, T_o\}}}{\psi_{\min,o}} \ll \frac{1}{\sqrt{\kappa_o^2    \max \{\log N_o, \log T_o\}}}, \ \ \min \{N_o^2, T_o^2\} \gg  \kappa_o^2 \mu_o^2 r^2 \max\{ N_o\log  N_o,  T_o \log T_o\},$$ 
 and $0<\eta_o <1/\psi_{\min,o}$. Suppose also that the iterates satisfy (\ref{Prelim3})-(\ref{Prelim8}) at the $\tau$-th iteration, then with probability at least $1-O(\min\{N_o^{-100}, T_o^{-100}\}),$ 
	\begin{align*}
		&\norm{X_o^{\tau+1 \top}X_o^{\tau+1}-Z_o^{\tau+1 \top}Z_o^{\tau+1}}_{F} \leq C_B \kappa_o \eta_o \left(\frac{\sigma \sqrt{\max\{N_o  ,T_o \}}}{\psi_{\min,o}}+\frac{\lambda_o}{\psi_{\min,o}} \right) \sqrt{r} \psi_{\max,o}^2 \\
		&\max_{1\leq m \leq N_o+T}\norm{X_o^{\tau+1,(m) \top}X_o^{\tau+1,(m)}-Z_o^{\tau+1,(m) \top}Z_o^{\tau+1,(m)}}_{F} \leq C_B \kappa_o \eta_o \left(\frac{\sigma \sqrt{\max\{N_o  ,T_o \}}}{\psi_{\min,o}}+\frac{\lambda_o}{\psi_{\min,o}} \right) \sqrt{r} \psi_{\max,o}^2
	\end{align*}
	holds true as long as $C_B \gg C_{op}^2$.
\end{lemma}

\begin{lemma}\label{LemmaB9}
	Suppose that $\lambda_o=C_{\lambda} \sigma \sqrt{\max\{N_o,T_o\}}$ for some large constant $C_{\lambda}>0$, $$ \frac{\sigma}{\psi_{\min,o}}\sqrt{\frac{\max\{N_o^2, T_o^2\}}{  \min\{N_o,T_o\}}} \ll \frac{1}{\sqrt{\kappa_o^4 \mu_o r \max \{\log N_o, \log T_o\}}},$$ and $0<\eta_o \ll 1/(q \psi_{\max,o}\max\{N_o,T_o\})$. Suppose also that the iterates satisfy (\ref{Prelim3})-(\ref{Prelim8}) at the $\tau$-th iteration, then with probability at least $1-O(\min\{N_o^{-99}, T_o^{-99}\}),$
	\begin{align*}
		f(X_o^{\tau+1}, Z_o^{\tau+1}) \leq f(X_o^{\tau}, Z_o^{\tau})-\frac{\eta_o}{2}\norm{\nabla f(X_o^{\tau}, Z_o^{\tau}) }_F^2.
	\end{align*}
\end{lemma}

\begin{lemma}\label{LemmaB11}
	Throughout the set of results, we assume that the $\tau$-th iterates satisfy the induction hypotheses (\ref{Prelim3})-(\ref{Prelim8}).
	\begin{enumerate}
		\item[(i)] Suppose that $\min\{N_o,T_o\} \gg \mu_o r \max\{\log N_o, \log T_o\}$. Then, we obtain
		\begin{align*}
			&\norm{\calF_o^{\tau,(m)}Q_o^{\tau,(m)}-\calF_o}_{2, \infty} \leq (C_{\infty} \kappa_o+C_3)\left( \frac{\sigma \sqrt{\max\{N_o \log N_o,T_o \log T_o\}}}{\psi_{\min,o}}  + \frac{\lambda_o}{\psi_{\min,o}}\right)\norm{\calF_o}_{2, \infty},   \\ 
			&\norm{\calF_o^{\tau,(m)}Q_o^{\tau,(m)}-\calF_o} \leq 2C_{op} \left( \frac{\sigma \sqrt{\max\{N_o,T_o\}}}{\psi_{\min,o}}  + \frac{\lambda_o}{\psi_{\min,o}}\right)\norm{X_o}.  
		\end{align*}
		\item[(ii)] Suppose that $ \frac{\sigma \sqrt{\max\{N_o,T_o\} }}{\psi_{\min,o}}  \ll \frac{1}{\kappa_o \sqrt{\max\{\log N_o, \log T_o\}}}$. Then, we have
		\begin{gather}
			\norm{\calF_o^{\tau} H_o^{\tau}-\calF_o} \leq \norm{X_o}, \quad \norm{\calF_o^{\tau} H_o^{\tau}-\calF_o}_F \leq \norm{X_o}_F, \quad \norm{\calF_o^{\tau} H_o^{\tau}-\calF_o}_{2, \infty} \leq \norm{\calF_o}_{2, \infty}, \label{LemmaB11.3} \\
			\norm{\calF_o^{\tau}} \leq 2 \norm{X_o}, \quad \norm{\calF_o^{\tau}}_F \leq 2 \norm{X_o}_F, \quad
			\norm{\calF_o^{\tau}}_{2, \infty} \leq 2 \norm{\calF_o}_{2, \infty}. \label{LemmaB11.4}
		\end{gather}
		\item[(iii)] Suppose that  $ \frac{\sigma \sqrt{\max\{N_o,T_o\} }}{\psi_{\min,o}}  \ll \frac{1}{\kappa_o\sqrt{\max\{\log N_o, \log T_o\}}}$ and $\sqrt{\frac{\mu_o r}{\min\{N_o,T_o\}}} \ll 1$. Then, we have
		\[
		\norm{\calF_o^{\tau} H_o^{\tau} -\calF_o^{\tau,(m)}H_o^{\tau,(m)}}_F \leq 5 \kappa_o \norm{\calF_o^{\tau} H_o^{\tau} -\calF_o^{\tau,(m)}Q_o^{\tau,(m)}}_F.
		\]
		\item[(iv)] Suppose that $\frac{\sigma \sqrt{\max\{N_o,T_o\} }}{\psi_{\min,o}}  \ll \frac{1}{\kappa_o \sqrt{\max\{\log N_o, \log T_o\}}}$ and $\min\{N_o,T_o\} \geq \kappa_o \mu_o$. Then (\ref{LemmaB11.3}), (\ref{LemmaB11.4}) also hold for $\calF_o^{\tau,(m)}H_o^{\tau,(m)}.$ Additionally, we have
		\[\psi_{\min,o}/2 \leq \psi_{\min} \left( (Z_o^{\tau,(m)} H_o^{\tau,(m)})^\top Z_o^{\tau,(m)} H_o^{\tau,(m)}\right) \leq \psi_{\max} \left( (Z_o^{\tau,(m)} H_o^{\tau,(m)})^\top Z_o^{\tau,(m)} H_o^{\tau,(m)}\right) \leq 2 \psi_{\max,o}.\]
	\end{enumerate}
\end{lemma}

\begin{lemma}\label{LemmaB12}
	Suppose $\calF_1, \calF_2, \calF_0 \in \mathbb{R}^{(N_o+T_o) \times r}$ are three matrices such that $\norm{\calF_1-\calF_0}\norm{\calF_0} \leq \psi_{\min}^2(\calF_0)/2$ and $\norm{\calF_1-\calF_2}\norm{\calF_0} \leq \psi_{\min}^2(\calF_0)/4$. Denote $$R_1 \coloneqq \argmin_{O \in \mathcal{O}^{r \times r}}\norm{\calF_1O-\calF_0}_F, R_2 \coloneqq \argmin_{O \in \mathcal{O}^{r \times r}}\norm{\calF_2O-\calF_0}_F.$$ Then we have
	\[
	\norm{\calF_1R_1-\calF_2R_2} \leq 5 \frac{\psi_{\max}^2(\calF_0)}{\psi_{\min}^2(\calF_0)}\norm{\calF_1-\calF_2} \quad \text{and} \quad \norm{\calF_1R_1-\calF_2R_2}_F \leq 5 \frac{\psi_{\max}^2(\calF_0)}{\psi_{\min}^2(\calF_0)}\norm{\calF_1-\calF_2}_F.
	\]
\end{lemma}
\begin{proof}
	This is the same as Lemma 37 in \cite{ma2020implicit}.
\end{proof}

\section{Additional empirical findings: comparison with the two-way fixed effect model in \cite{chung2020tick}}\label{sec:additionalempirical}

Finally, we provide further details of the comparison between our model and the two-way fixed effect model in \cite{chung2020tick} which is omitted in the main text to save space. Denote the quote, trade and trade-at-rule dummy variables by $\calQ_i$, $\calT_i$, and $\calT\calA_i$, respectively, and the pilot period dummy variable by $Pilot_t$. \cite{chung2020tick} consider the following two-way fixed effect model:
$$
y_{it} = (\calQ_i \times Pilot_t ) \theta^{(1)} +  (\calT_i \times Pilot_t ) \theta^{(2)} + (\calT\calA_i \times Pilot_t ) \theta^{(3)} + x_{it}^\top \beta + \alpha_i + \delta_t + \epsilon_{it},
$$
where $x_{it}$ is the set of $\left( \calQ_i \times Pilot_t \times TBC_{it} \right)$, $\left( Pilot_t \times TBC_{it} \right)$ and other control variables like stock prices and trading volumes. Since $y_{it} = \sum_{0 \leq d \leq 3} \Upsilon_{it}^{(d)} y_{it}^{(d)}$, where $\Upsilon_{it}^{(d)} = 1$ if and only if unit $i$ receives treatment $d$ at time $t$, and zero otherwise, with the convention that the treatment $0$ is the control, this model can be represented as Model \eqref{eq:theirmodel}. On the other hand, our model can be represented as:
$$
y_{it} = (\calQ_i \times Pilot_t ) \theta^{(1)}_{it} +  (\calT_i \times Pilot_t ) \theta^{(2)}_{it} + (\calT\calA_i \times Pilot_t ) \theta^{(3)}_{it} + x_{it}^\top \beta + \zeta_i^{\top} \eta_t^{(0)} + \epsilon_{it}.
$$
As noted in the main text, the above model is nested in our model and highly likely to be misspecified.

\begin{table}[h]
	\centering
	\begin{tabular}{ccccccc}
		\hline
		\hline
		& $\beta_1$ & $\beta_2$ & $\theta^{(1)}$ & $\theta^{(2)}$ & $\theta^{(3)}$ & $R^2$ \\
		\hline
		\multirow{2}[2]{*}{Our model} & 2.20 *** & -1.46 *** &  $\hat{\theta}^{(1)}_{it}$      &       $\hat{\theta}^{(2)}_{it}$ &  $\hat{\theta}^{(3)}_{it}$      & 0.79 \\
		& (0.10) & (0.07) & [mean: -0.40] & [mean: 0.86] & [mean: -0.98] &  \\
		\hline
		\multirow{2}[2]{*}{Two-way} & 3.68*** & -0.75 *** & -0.27 *** & 0.28 *** & -0.99 *** & 0.67 \\
		& (0.09) & (0.06) & (0.05) & (0.05) & (0.05) &  \\
		\hline
	\end{tabular}%
	\caption{Estimation results: `Two-way' denotes the two-way fixed effect model. Numbers in the parenthesis ( ) are standard errors. `mean' denotes the average of $\hat{\theta}^{(d)}_{it}$ over all treated stocks in the pilot periods.}
	\label{tab:estimationresult}%
\end{table}%

Table \ref{tab:estimationresult} provides estimates for both models. $\beta_1$ and $\beta_2$ are the coefficients for $\left( \calQ_i \times Pilot_t \times TBC_{it} \right)$ and $(Pilot_t \times TBC_{it})$, respectively. Note that the positive $\beta_1$ means that a larger TBC results in a larger treatment effect of the Q rule. It shows that as the minimum quoted spread increases from 1 cent to 5 cents under the Q rule, the effective spread increases, and this effect increases when the extent to which the new tick size (\$0.05) is a binding constraint on quoted spreads is larger.

It is worth noting that the treatment effect of the Q rule is $\theta_{it}^{(1)} + \beta_1 \cdot TBC_{it}$ since $$\bbE[y_{it} | \calQ_i =1, Pilot_t = 1] - \bbE[y_{it} | \calQ_i =0 , Pilot_t = 1] = \theta_{it}^{(1)} +  \beta_1 \cdot TBC_{it}$$ while that of the T rule and the TA rule are $\theta_{it}^{(2)}$ and $\theta_{it}^{(3)}$, respectively. Figure \ref{fig:Qrule_dynamics} shows the dynamics of the cross-sectional average of the treatment effects of the Q rule.

\begin{figure}[h!]
	\centering
	\includegraphics[width=0.9\textwidth]{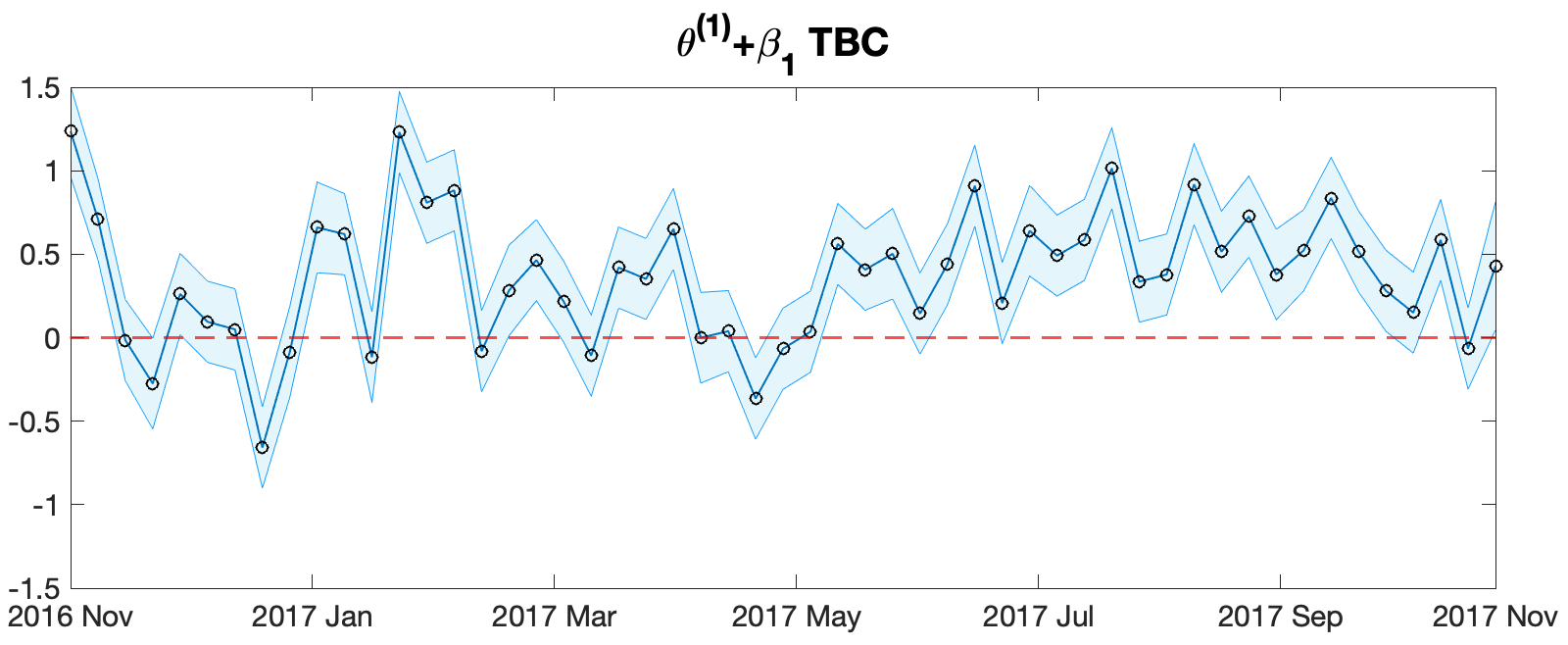}
	\caption{The dynamics of the cross-sectional average of the effect of Q rule: For the confidence band, we use the 95\% uniform critical value, $\Phi^{-1}(1 - 0.025/53)$. The dots denote the weekly average.}
	\label{fig:Qrule_dynamics}
\end{figure}

Note also that the sign of treatment effect of the Q rule is determined by the magnitudes of the positive effect of TBC and the negative effect of $\theta^{(1)}_{it}$, the effect of coarser quotable prices. To see why the Q rule results in coarser quotable prices, consider, for example, if the quoted spread is 17 cents without the Q rule. It may change to 15 cents or 20 cents under the Q rule. This effect is different from the effect related to the minimum quoted spread captured by $TBC$. $\theta^{(1)}_{it}$ can capture the effect of coarser quotable prices. Most of the time, the positive effect of TBC is greater than the negative effect of $\theta^{(1)}_{it}$, and therefore the treatment effect of Q rule is positive. Especially, as time passes, the negative effect of $\theta^{(1)}_{it}$ becomes weaker, and the treatment effect of Q rule becomes more positive.

\newpage
\bibliographystyle{apalike}
\bibliography{reference}

\end{document}